\documentclass[a4paper,11pt]{article}
\usepackage[margin=1in]{geometry}
\usepackage[T1]{fontenc}
\usepackage[ttscale=.85]{libertine}
\usepackage{amsmath,amssymb,amsthm}
\usepackage[dvipsnames]{xcolor}
\usepackage{graphicx}
\usepackage{xspace}
\usepackage{mathtools}
\usepackage[colorlinks=true, allcolors=blue]{hyperref}
\usepackage{thmtools}
\usepackage{cleveref}
\usepackage{enumitem}
\usepackage{thm-restate}
\usepackage[sort]{cite}

\theoremstyle{plain}

\newtheorem{theorem}{Theorem}[section]
\newtheorem{claim}[theorem]{Claim}
\newtheorem{lemma}[theorem]{Lemma}
\newtheorem{corollary}[theorem]{Corollary}
\newtheorem{observation}[theorem]{Observation}
\theoremstyle{definition} 
\newtheorem{definition}[theorem]{Definition}
\newtheorem{remark}[theorem]{Remark}

\newenvironment{claimproof}[1][Proof of claim.]{\begin{proof}[#1]}{\end{proof}}

\usepackage{tikz} 

\newcommand{\parcon}{%
  \mathord{
    \tikz[baseline=0.4ex]{%
      \draw[line width=0.05em] (0,0) -- (0.5em,0.9em);

      \fill (0.5em/3, 0.9em/3) circle[radius=0.1em];

      \fill (2*0.5em/3, 2*0.9em/3) circle[radius=0.1em];
    }
  }
}

\newcommand{\vcon}{\parcon}

\newcounter{ctr}
\loop
 \stepcounter{ctr}
 \expandafter\edef\csname c\Alph{ctr}\endcsname{\noexpand\mathcal{\Alph{ctr}}}
 \expandafter\edef\csname b\alph{ctr}\endcsname{\noexpand\mathbf{\alph{ctr}}}
 \expandafter\edef\csname b\Alph{ctr}\endcsname{\noexpand\mathbf{\Alph{ctr}}}
\ifnum\thectr<26
\repeat

\newcommand{\diam}{\mathrm{diam}}
\newcommand{\dist}{\mathrm{dist}}
\newcommand{\compl}{\mathrm{compl}}
\newcommand{\Area}{\mathrm{Area}}

\newcommand{\tw}{\mathrm{tw}}

\renewcommand{\bd}{\partial}

\DeclareMathOperator{\inter}{Int}
\DeclareMathOperator{\exter}{Ext}
\DeclareMathOperator{\Clo}{Cl}

\DeclareMathOperator{\Vor}{Vor}
\DeclareMathOperator{\rank}{rank}
\newcommand{\skel}{\widetilde{\mathrm{Skel}}}

\newcommand{\Cols}{\mathrm{Cols}}
\newcommand{\Top}{\mathrm{Top}}
\newcommand{\MIS}{\mathrm{MIS}}
\newcommand{\orig}{\mathrm{orig}}
\newcommand{\ord}{\mathrm{ord}}

\newcommand{\poly}{\mathrm{poly}}

\renewcommand{\leq}{\leqslant}
\renewcommand{\geq}{\geqslant}
\renewcommand{\rho}{\varrho}

\newcommand{\eps}{\varepsilon}
\renewcommand{\to}{\rightarrow}

\newcommand{\myin}{\mathrm{in}}
\newcommand{\myout}{\mathrm{out}}

\newcommand{\Reals}{\mathbb{R}}
\newcommand{\Nats}{\mathbb{N}}
\newcommand{\Ints}{\mathbb{Z}}
\newcommand\eqdef{%
\mathrel{\overset{\makebox[0pt]{\mbox{\normalfont\tiny\sffamily def}}}{=}}}
\newcommand{\opt}{\mbox{{\sc opt}}\xspace}  
\newcommand{\myopt}{\mathrm{opt}}           
\newcommand{\cc}{\mathrm{cc}}
\newcommand{\ccomp}{\mathrm{CComp}}
\newcommand{\join}{\mathrm{join}}
\newcommand{\Lip}{\mathrm{Lip}}
\newcommand{\Inn}{\mathrm{Inn}}
\newcommand{\Radi}{\mathrm{Radial}}
\newcommand{\Jum}{\mathrm{Jump}}
\newcommand{\AW}{\mathrm{AW}}
\newcommand{\NAW}{\mathrm{NAW}}

\newcommand{\Aout}{A_{\myout}}
\newcommand{\Ain}{A_{\myin}}
\newcommand{\Sopt}{S_{\myopt}}
\newcommand{\So}{\mathrm{South}}
\newcommand{\No}{\mathrm{North}}
\newcommand{\We}{\mathrm{West}}
\newcommand{\Ea}{\mathrm{East}}

\newcommand{\Perim}{\mathrm{Perim}}
\newcommand{\pl}[1]{\widetilde{#1}}
\newcommand{\re}[1]{\mathsf{#1}}
\newcommand{\wf}[1]{\widetilde{\mathsf{#1}}}

\let\oldemph\emph
\renewcommand{\emph}[1]{\oldemph{\textcolor{red!40!black}{#1}}}
\newcommand{\etal}{\oldemph{et~al.}}

\newcommand{\stsp}{\textsc{Subset} TSP\xspace}
\newcommand{\stein}{\textsc{Steiner Tree}\xspace}
\newcommand{\wstein}{\textsc{Vertex-Weighted Steiner Tree}\xspace}
\newcommand{\smt}{\mathrm{StTree}}
\newcommand{\tsp}{\mathrm{SubsetTSP}}

\newcommand{\dm}[1]{}
\newcommand{\skb}[1]{}

\makeatletter
\def\th@plain{%
  \thm@notefont{}
  \itshape 
}
\def\th@definition{%
  \thm@notefont{}
  \normalfont 
}
\makeatother

\title{Approximation Schemes for Subset TSP and Steiner Tree\\ on Geometric Intersection Graphs}
\author{S\'andor Kisfaludi-Bak\thanks{Aalto University, Espoo, Finland; e-mail: \texttt{sandor.kisfaludi-bak@aalto.fi}.\\ Supported by the Research Council of Finland, Grant 363444.} \and D\'aniel Marx\thanks{CISPA Helmholtz Center for Information Security, Germany, e-mail: \texttt{marx@cispa.de}}}
\date{}
\begin{document}

\maketitle
\thispagestyle{empty}

\begin{abstract}
We give approximation schemes for Subset TSP and Steiner Tree on unit disk graphs, and more
generally, on intersection graphs of similarly sized connected fat (not necessarily convex) polygons in
the plane. As a first step towards this goal, we prove spanner-type results: finding an induced subgraph of bounded size that is $(1+\eps)$-equivalent to the original instance in the sense that the optimum value increases only by a factor of at most $(1+\eps)$ when the solution can use only the edges in this subgraph.
\begin{enumerate}
\item For Subset TSP, our algorithms find a $(1+\eps)$-equivalent induced subgraph of size $\poly(1/\eps)\cdot\textup{OPT}$ in polynomial time, and use it to find a $(1+\eps)$-approximate solution in time $2^{\poly(1/\eps)}\cdot n^{O(1)}$. 
\item For Steiner Tree, our algorithms find a $(1+\eps)$-equivalent induced subgraph of size $2^{\poly(1/\eps)}\cdot\textup{OPT}$ in time $2^{\poly(1/\eps)}\cdot n^{O(1)}$, and use it to find a $(1+\eps)$-approximate solution in time $2^{2^{\poly(1/\eps)}}\cdot n^{O(1)}$.
  \item An improved algorithm finds a $(1+\eps)$-approximate solution for Steiner Tree in time $2^{\poly(1/\eps)}\cdot n^{O(1)}$.
  \end{enumerate}
  An easy reduction shows that approximation schemes for unit disks imply approximation schemes for planar graphs. Thus our results are far-reaching generalizations of analogous results of Klein [STOC'06] and Borradaile, Klein, and Mathieu [ACM TALG'09] for Subset TSP and Steiner Tree in planar graphs. We show that our results are best possible in the sense that dropping any of (i) similarly sized, (ii) connected, or (iii) fat makes both problems APX-hard.
\end{abstract} 

\clearpage

\tableofcontents
\thispagestyle{empty}
\setcounter{page}{0}
\clearpage
\section{Introduction}

Most of the basic combinatorial optimization problems defined on graphs are APX-hard, which means that there is a constant $c>1$ such that, assuming $\textup{P}\neq\textup{NP}$, there is no polynomial-time approximation algorithm for the problem with approximation ratio better than $c$.  In particular, this rules out the existence of a polynomial-time approximation scheme (PTAS), that is, an algorithm that takes and instance of the problem and an $\eps>0$ in the input, and outputs a $(1+\eps)$-approximate solution in time $n^{f(1/\eps)}$ for some function $f$. However, these problems can still admit a PTAS when restricted to some structured class of graphs. For example, there is a long line of research showing that several APX-hard problems admit approximation schemes when restricted to planar graphs (or perhaps even to bounded-genus or minor-free graphs) \cite{DBLP:journals/algorithmica/KleinMZ23,DBLP:conf/focs/Cohen-AddadFKL20,DBLP:journals/siamcomp/Cohen-AddadKM19,DBLP:conf/soda/Fox-EpsteinKS19,DBLP:conf/stoc/FoxKM15,DBLP:conf/soda/EisenstatKM12,DBLP:conf/soda/BateniHKM12,BorradaileKM09,DBLP:journals/siamcomp/Klein08,Klein06,DBLP:conf/soda/BateniFH19,DBLP:conf/stoc/BateniDHM16,DBLP:journals/jacm/BateniHM11,DBLP:conf/focs/Cohen-AddadLPP23,DBLP:conf/focs/Cohen-AddadFKL20,DBLP:journals/siamcomp/Cohen-AddadKM19,DBLP:conf/esa/Cohen-AddadPP19,DBLP:conf/focs/Cohen-AddadPP19,DBLP:conf/soda/AroraGKKW98,DBLP:conf/icalp/BergerG07,DBLP:conf/esa/BergerCGZ05,DBLP:conf/soda/CzumajGSZ04,DBLP:conf/focs/GrigniKP95,DBLP:conf/sea2/BorradaileLZ19,DBLP:conf/approx/BorradaileZ17,DBLP:journals/talg/BorradaileK16,DBLP:journals/algorithmica/BorradaileDT14,DBLP:conf/alenex/TazariM09,DBLP:conf/focs/AbbasiBBCGKMSS23,DBLP:conf/soda/Cohen-AddadVM18,DBLP:conf/stoc/Cohen-AddadVKMM16}. One of the most basic tools in the design of such approximation schemes is a shifting technique attributed to Baker~\cite{Baker94}, which is widely applicable for problems that are local in a certain sense, for example Independent Set or Dominating Set.
The shifting technique reduces a planar instance into several instances where the treewidth of the graph is bounded by a function of $1/\eps$. Such instances can usually be handled by standard dynamic-programming techniques on tree decompositions. Often, the shifting strategy leads to an \textit{efficient polynomial-time approximation scheme (EPTAS)}, which is a PTAS with running time of the form $f(1/\eps)\cdot n^{O(1)}$ for some function $f$. 

Klein \cite{DBLP:journals/siamcomp/Klein08} introduced a contraction-decomposition version of the shifting strategy. This technique can be used for problems where the solution needs to provide some form of connectivity, for example, in TSP (find a shortest closed walk visiting every vertex). Typically, the contraction-decomposition technique allows us to reduce the problem to instances with treewidth $O(1/\eps)$, at the cost of an additive error that is $\eps$ times the total size (weight) of the input graph. For problems where the value of the optimum is at least the size of the graph, such as in TSP on unweighted graphs, this additive error results in a multiplicative $(1+\eps)$-approximation. However, we cannot make this assumption for TSP on edge-weighted graphs, or for problems where the solution can be much smaller than the graph, for example, in Subset TSP (find a shortest closed walk visiting every vertex of the given terminal set $T$) or Steiner Tree (find a shortest tree containing every vertex of the given terminal set $T$). These problems were approached by first computing a ``spanner'': removing edges in a way that changes the optimum value $\textup{OPT}$ at most by a factor of $1+\eps$ and ensures that the remaining subgraph has total weight at most $f(1/\eps)\cdot \textup{OPT}$. This strategy was successful for weighted TSP, Subset TSP, Steiner Tree, Steiner Forest, and other problems \cite{DBLP:journals/siamcomp/Klein08,Klein06,DBLP:conf/stoc/BateniDHM16,DBLP:conf/soda/BateniHKM12,DBLP:journals/algorithmica/BateniH12,DBLP:journals/jacm/BateniHM11,DBLP:journals/algorithmica/KleinMZ23,DBLP:conf/approx/BorradaileZ17,DBLP:journals/talg/BorradaileK16,DBLP:journals/algorithmica/BorradaileDT14}.

In the realm of geometric problems (in the plane or in $\mathbb{R}^d$ for fixed $d$), PTASs have been achieved by leveraging the underlying spatial structure to simplify complex instances. For instance,  Euclidean TSP has a celebrated PTAS by Arora \cite{DBLP:journals/jacm/Arora98} that employs a hierarchical partitioning of the Euclidean plane into a grid-like structure, allowing the optimal tour to be approximated by a tour that interacts with grid boundaries only a bounded number of times. More broadly, a range of geometric optimization problems---including clustering, facility location, and network design---benefit either from similar spatial decomposition techniques or from other techniques exploiting the metric properties of the space~\cite{DBLP:journals/dcg/AlkemaBHK24,DBLP:journals/siamcomp/BergBKK23,DBLP:journals/siamcomp/Cohen-AddadKM19,DBLP:journals/talg/BorradaileKM15,DBLP:journals/algorithmica/BateniH12,DBLP:journals/jacm/Arora98,DBLP:conf/stoc/AroraRR98,DBLP:conf/focs/AbbasiBBCGKMSS23,DBLP:conf/stacs/HochbaumM84,DBLP:conf/focs/Cohen-AddadSS23,DBLP:conf/icalp/Cohen-AddadL19,DBLP:journals/jacm/Cohen-AddadFS21,DBLP:conf/soda/Cohen-Addad20,DBLP:conf/soda/Cohen-Addad18,DBLP:conf/focs/Kisfaludi-BakNW21}.

Besides planar graphs and point sets in $\mathbb{R}^2$, another equally natural and well-studied 2-dimensional setting is given by intersection graphs of various planar objects.
In these graphs, each vertex corresponds to a geometric object---such as disks, rectangles, or segments---and an edge represents the intersection between two objects. This spatial representation often imposes additional structure that can be exploited algorithmically. For instance, many NP-hard problems like Independent Set, Dominating Set, and Vertex Cover, when restricted to intersection graphs of simple geometric shapes (such as unit disks or similarly sized fat objects), admit PTASs \cite{DBLP:journals/jal/HuntMRRRS98}. The key insight is that the geometry enables a partitioning of the problem into regions with limited complexity, often through techniques analogous to the shifting strategy or local search methods. Moreover, for more complex settings where the objects vary in size, shape, or weight, additional layers of techniques are required \cite{DBLP:journals/siamcomp/ErlebachJS05,DBLP:journals/mp/KhanSW24,DBLP:journals/jacm/AdamaszekHW19,DBLP:journals/talg/HeydrichW19,DBLP:conf/esa/PilipczukLW18,DBLP:conf/isaac/AdamaszekCL09,DBLP:journals/jacm/BonamyBBCGKRST21}.

While many geometric problems on intersection graphs benefit from these techniques, there is a notable lack of PTAS results for connectivity problems. The most natural open question in this direction is settling if Steiner Tree and Subset TSP admit a (E)PTAS on geoemtric intersection graphs. Intuitively, the inherent non-locality of connectivity does not seem to allow any simple use of spatial decomposition techniques. Indeed, as we show later, an (E)PTAS for Steiner Tree and Subset TSP on unit disk intersection graphs implies an (E)PTAS for the same problem on planar graphs (Theorem~\ref{planar_to_udg_reduction}). Thus these geometric problems are {\em at least as hard} as their planar graph counterparts and hence an EPTAS for them would likely need to be at least as complex as the highly sophisticated spanner constructions for the planar versions \cite{BorradaileKM09,Klein06}. It is far from clear if such an EPTAS is possible:
there are cases where a problem on intersection graphs is strictly harder than on planar graphs. There is such a gap even for very basic  problems: for example, Independent Set and (Connected) Dominating Set admit an EPTAS on planar graphs, but there is no EPTAS for these problems in unit disk graphs (under suitable complexity assumptions) \cite{DBLP:conf/esa/Marx05,DBLP:conf/focs/Marx07a,DBLP:journals/corr/MarxS16,DBLP:journals/tcs/BergKW19,DBLP:journals/algorithmica/BergBK19a}. The intuitive similarity of Connected Dominating Set and Steiner Tree can be interpreted as a warning sign, suggesting that the substantial nonplanarity present in intersection graphs could potentially make Steiner Tree and Subset TSP strictly harder in, say, unit disk graphs compared to planar graphs.



Our main contribuitions are developing a set of tools to handle intersection graphs in various algorithmic and combinatorial contexts, and then using these tools to resolve the open questions by developping EPTASs for Subset TSP and Steiner Tree on geometric intersection graphs. Specifically, we consider three results on planar graph EPTASs and show that (unlike in the case of Independent Set etc.), analogous results for geometric intersection graphs do exist. First, for Subset TSP, Klein~\cite{Klein06} presented a way of obtaining a spanner of size at most $\poly(1/\eps)$ times the optimum, which, together with the contraction decomposition technique, gives an EPTAS.
\begin{theorem}[Planar Subset TSP EPTAS via spanner \cite{Klein06}]\label{th:planarsubsetsp}

    Given a planar instance $(G,T)$ of Subset TSP and an $\eps>0$, 
    \begin{enumerate}
    \item (Subset spanner) We can find a subgraph $G'\subseteq G$ such that $\tsp(G',T)\le (1+\eps)\tsp(G,T)$ and $|V(G')|\le \poly(1/\eps)\cdot \tsp(G,T)$ in polynomial time.
      \item (EPTAS) We can find an $(1+\eps)$-approximate solution of the instance $(G,T)$ in time $2^{\poly(1/\eps)}\cdot n^{O(1)}$.
      \end{enumerate}

\end{theorem}
Building on the Subset TSP result, Borradaile, Klein, and Mathieu~\cite{BorradaileKM09} obtained a spanner for Steiner Tree, but with a size exponential in $1/\eps$. Consequently, the resulting approximation scheme has double-exponential dependence on $1/\eps$.

\begin{theorem}[Planar Steiner Tree EPTAS via spanner \cite{BorradaileKM09}]\label{th:planarsteinertsp}
    Given a planar instance $(G,T)$ of Steiner Tree and an $\eps>0$,
    \begin{enumerate}
    \item (Steiner-tree spanner) We can find a subgraph $G'\subseteq G$ such that $\smt(G',T)\le (1+\eps)\smt(G,T)$ and $|V(G')|\le 2^{\poly(1/\eps)}\cdot \smt(G,T)$ in time $2^{\poly(1/\eps)}\cdot n^{O(1)}$.
      \item (EPTAS) We can find a $(1+\eps)$-approximate solution of the instance $(G,T)$ in time $2^{2^{\poly(1/\eps)}}\cdot n^{O(1)}$.
      \end{enumerate}

\end{theorem}
To obtain an EPTAS with better dependence on $1/\eps$, the structural insights obtained in the spanner construction can be used in a more efficient way. 
\begin{theorem}[Faster EPTAS for planar Steiner Tree \cite{BorradaileKM09}]\label{th:planarsteinertspfast}
Given a planar instance $(G,T)$ of Steiner Tree and an $\eps>0$, a $(1+\eps)$-approximate solution can be found in time $2^{\poly(1/\eps)}\cdot n^{O(1)}$.
\end{theorem}

\subsection{Our results}



Our main algorithmic results are counterparts of the EPTASs of Theorems~\ref{th:planarsubsetsp}--\ref{th:planarsteinertspfast} to the intersection graphs of unit disks and other objects. Part of our goal was to understand to what extent such EPTASs are possible, therefore we state the results in a more general form, for intersection graphs of ``similarly sized fat'' (not necessarily convex) polygons. (As we show in Lemma~\ref{lem:disktopolygon}, unit disk intersection graphs can be easily expressed by such objects.) We argue that this is a very reasonable and natural class of objects for these problems: we present a set of lower bound results that show that going beyond this class in various ways lead to APX-hard problems.


Formally, an object $p$ in $\Reals^2$ is called $\beta$-fat for some $\beta \in (0,1]$ if there are balls $B_1$, $B_2$ of radius $r_1,r_2$ such that $B_1\subseteq p \subset B_2$ and $r_1/r_2\geq \beta$. A finite collection $X$ of objects in $\Reals^2$ is called $\delta$-similarly-sized for some $\delta \in (0,1]$  if $\frac{\min_{x\in X} \diam(x)}{\max_{x\in X} \diam(x)}\geq \delta$. Throughout this article, we will assume that our objects are $\beta$-fat and the objects collections $\delta$-similarly-sized for some universal positive constants $\beta,\delta$.

Before presenting our generalizations, let us remark that in planar graphs the edge-weighted versions of Subset TSP and Steiner Tree can be easily reduced to the unweighted versions by subdividing each edge an appropriate number of times. In case of geometric intersection graphs, defining edge-weighted problems usually does not make much sense: large cliques can easily appear in intersection graphs and an edge-weighted clique can be used to represent arbitrary graphs. Vertex-weighted problems do make sense for intersection graphs, but such problems can be much more challenging compared to the unweighted versions. Already for planar graphs, much less is known about the vertex-weighted version \cite{DBLP:journals/talg/DemaineHK14}; in particular, Theorems~\ref{th:planarsubsetsp}--\ref{th:planarsteinertspfast} work only in case of uniform vertex weights. Therefore, all our results are for the unweighted versions of the problems. In all the algorithmic results, we assume that an intersection graph is given in the input with a representation. The input size $n$ is the total size of this representation, i.e., the total number of vertices of the polygons.

\begin{restatable}[Geometric Subset TSP EPTAS via spanner]{theorem}{TSPalg}
\label{thm:TSPalg}
Given $\eps>0$, an intersection graph $G$ of $\delta$-similarly sized $\beta$-fat polygons with its representation, and a set
$T\subseteq V(G)$ of terminals,
\begin{enumerate}
\item (Subset spanner) We can find an induced subgraph $G'\subseteq G$ such that $|V(G')|\le O_{\beta,\delta}(1/\eps^{10})\cdot\tsp(G,T)$ and $\tsp(G',T)\le (1+\eps)\tsp(G,T)$ in polynomial time.
  \item (EPTAS) We can find a $(1+\epsilon)$-approximate solution of the instance $(G,T)$ in time $2^{O_{\beta,\delta}(1/\eps^5\log(1/\eps))}\poly(n)$.
\end{enumerate}

\end{restatable}

\begin{restatable}[Geometric Steiner Tree EPTAS via spanner]{theorem}{SteinerSlow}
\label{thm:SteinerSlow}
Given $\eps>0$, an intersection graph $G$ of $\delta$-similarly sized $\beta$-fat polygons with its representation, and a set
$T\subseteq V(G)$ of terminals, 
\begin{enumerate}
  \item (Steiner-tree spanner) We can find an induced subgraph $G'\subseteq G$ such that $|V(G')|\le 2^{O_{\beta,\delta}(1/\eps^{7.5})}\cdot\smt(G,T)$ and $\smt(G',T)\le (1+\eps)\smt(G,T)$ 
    in time ${2^{O_{\beta,\delta}(1/\eps^{7.5})}}\poly(m)$.
  \item (EPTAS) We can find a $(1+\epsilon)$-approximate solution  of the instance $(G,T)$ in time  $2^{2^{O_{\beta,\delta}(1/\eps^{7.5})}}\poly(n)$.
  \end{enumerate}
\end{restatable}

\begin{restatable}[Faster EPTAS for geometric Steiner Tree]{theorem}{SteinerFast}
\label{thm:SteinerFast}
Given $\eps>0$, an intersection graph $G$ of $\delta$-similarly sized $\beta$-fat polygons with its representation, and a set $T\subseteq V(G)$ of terminals, we can find a $(1+\epsilon)$-approximate solution in time $2^{O_{\beta,\delta}(1/\eps^{33})}\poly(n)$.
\end{restatable}

We show that the EPTASs for Subset TSP and Steiner Tree on intersection graphs of connected similarly sized fat objects is optimal in the sense that dropping any of the three conditions makes both problems APX-hard problems. Additionally, moving to three dimensions also make the problem APX-hard, even for the simple case of unit balls.

\begin{restatable}[APX-hardness results]{theorem}{APXhard}
\label{thm:apxhard}
\textsc{Steiner Tree} and \textsc{Subset TSP} are APX-hard in intersection graphs of:
\begin{enumerate}[label=(\alph*)]
\item similarly sized fat (but potentially disconnected) objects, even for unit-diameter 1/3-fat objects with two connected components,
\item similarly sized connected (but potentially non-fat) objects, even for unit-diameter objects,
\item connected fat (but potentially differently sized) objects, even if each object is $0.18$-fat, and 
\item similarly sized connected fat objects in $\Reals^3$, even if each object is a unit-diameter ball.
\end{enumerate}
\end{restatable}

To formally show that approximation schemes for intersection graphs imply approximation schemes in planar graphs, we present reductions of an appropriate form that transfers the existence of (E)PTASs. This reduction shows that Theorems~\ref{thm:TSPalg}--\ref{thm:SteinerFast} can be seen as generalizations of Theorems~\ref{th:planarsubsetsp}--\ref{th:planarsteinertspfast} (but note that the exponents in the $\poly(1/\eps)$ terms are somewhat weaker in our results).

\begin{restatable}[Reduction from planar graphs to unit disk graphs]{theorem}{udgreduction}\label{planar_to_udg_reduction}
There is a PTAS reduction from \textsc{Planar Subset TSP} to \textsc{Unit Disk Subset TSP} and from \textsc{Planar Steiner Tree} to \textsc{Unit Disk Steiner Tree}.
\end{restatable}

\paragraph*{New concepts and techniques of independent interest}

As the setting of intersection graphs is siginficantly more general than planar graphs, we need first to develop a number of basic tools that were not available for intersection graphs. These techniques could be of independent interest for other problems involving interection graphs. Note that we state these results in a very general form, considering similarly sized connected fat polygons, making them potentially useful in a wider range of applications.
\begin{itemize}
\item The faces of a planar graph partition the plane into disjoint regions. Such a decomposition is crucially used in the proofs of Theorems~\ref{th:planarsubsetsp}--\ref{th:planarsteinertspfast}, but it is not clear what the intersection graph analog of such a clean partition would be. We introduce the notion of \emph{wireframes,} which are planar graph representations of intersection graphs. Wireframes will be the main structures guiding our algorithm by decomposing the problem into regions. We also introduce object frames, a planar intersection graph whose objects are parts of the original objects.
\item \emph{Contraction decomposition} is an important step in planar PTASs, allowing us to reduce the problem to bounded-treewidth graphs. We show a clean way of proving such results for intersection graphs: first, we show how a \emph{Lipschitz embedding} into planar graphs can be obtained, and then we show that contraction decomposition of the target planar graph can be lifted to the original intersection graph.
\item  One obvious difference compared to planar graphs is that intersection graphs of even very simple objects can contain large cliques. However, we can preprocess the instance with a \emph{sparsification} procedure, which, at the cost of $(1+\epsilon)$-factor increase of the optimum, reduces the maximum clique size to polynomial in $1/\eps$ (for Subset TSP) or exponential in $1/\eps$ (Steiner Tree).
  \item  Borradaile, Klein, and Mathieu~\cite{BorradaileKM09} improved Theorem~\ref{th:planarsteinertsp}(2) to Theorem~\ref{th:planarsteinertspfast} using an elaborate dynamic programming algorithm designed specifically for their setting. We argue that an alternative, more streamlined way of doing this is via a reduction to a \emph{vertex-weighted} Steiner Tree problem on bounded-treewidth graphs. This change of viewpoint is the starting point in our proof of Theorem~\ref{th:planarsteinertspfast}.
\end{itemize}

Having these fairly self-contained components at hand seems to be a prerequisite for working towards Theorems~\ref{thm:TSPalg}--\ref{thm:SteinerFast}. However, we want to emphasize that obtaining these results is only the start of the journey. The differences between planar graphs and  intersection graphs make it necessary to deviate from the proof of Theorems~\ref{th:planarsubsetsp}--\ref{th:planarsteinertspfast} at every step, introducing new proof ideas that handle the phenomena specific to intersection graphs. In Sections~\ref{sec:eptas-subset-tsp}--\ref{sec:faster-ptas-steiner}, the paragraphs marked with ``(NEW)'' overview these additional proof ideas.

It turns out that perhaps the most significant initial difference between planar graphs and intersection graphs is that in planar graphs whenever two paths cross, then they share a vertex. For contrast, consider the unit disk intersection graph in Figure~\ref{fig:unitdisk_vs_planar}(i). There is a unique shortest $a-b$ path and a unique shortest $c-d$ path. Even though these two paths cross twice, it is not possible to reroute any of the two paths to the other path without increasing its length (contrary to what we would expect in a planar graph). This difference invalidates the main principle behind the structural decompositions in the proof of Theorems~\ref{th:planarsubsetsp}--\ref{th:planarsteinertspfast}. This principle can be briefly formulated as follows (see Figure~\ref{fig:unitdisk_vs_planar}(ii)). Let $C$ be a cycle such that for any two vertices $x,y\in C$, the shorter of the two paths between $x$ and $y$ on $C$ is a shortest $x-y$ path. Then if $P$ is a path with endpoints outside $C$ and $P$ enters $C$, then $P$ can be rerouted, without increasing its length, to avoid the interior of $C$. Thus such a cycle can effectively make the interior irrelevant with respect to certain connections. In intersection graphs, the rerouting may incur additional cost, invalidating this argument. These additional costs have a cascading effect on the algorithm: they require extra layers of proof ideas to make certain steps more robust to such cost increases, which in turn require further extra layers of arguments in the later parts of the algorithm.

\begin{figure}
\centering
\includegraphics[width=\textwidth]{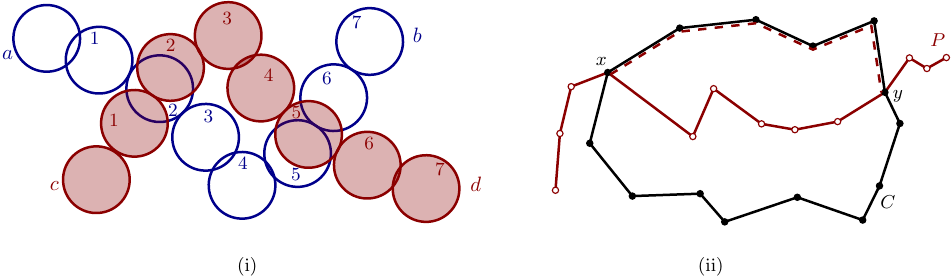}
\caption{(i) Two paths in a unit disk graphs that cannot be shortened: on the blue path from $a$ to $b$, the distance of the disk from $a$ is indicated with blue numbers, and on the red $c\rightarrow d$ path the red number is the distance from $c$. (ii) In a planar graph, if $C$ is a cycle where the distance of vertices on the cycle is less or equal to their graph distance, then we can shorten any path $P$ that intersects $C$ at least twice.}
\label{fig:unitdisk_vs_planar}
\end{figure}

\subsection{Organization}

The paper is organized as follows. Section~\ref{sec:overview} gives an overview of both previous techniques in planar graphs and our new techniques in intersection graphs to achieve the desired EPTASes. In Section~\ref{sec:prelim}, we recall standard notation for graphs and geometric objects, and define new concepts introduced in the paper, such as wireframes and object frames.

Sections~\ref{sec:connectivity}--\ref{sec:sparse} contain basic results that are crucial for the algorithms. As they are mostly used in black-box fashion, the reader may decide to only skim these sections at first, and jump directly to the algorithmic results starting in Section~\ref{sec:tspspanner}.
In more detail, Section~\ref{sec:connectivity} shows how a set of objects can be turned into an object frame while providing the same connectivity for terminal points. This result is crucial for the structural results for Steiner tree (Section~\ref{sec:SteinerStruct}), but a special case with only two terminals is used already in the construction of the so-called skeleton and the spanner in Section~\ref{sec:tspspanner}.
Section~\ref{sec:lipschitz} presents our results on Lipschitz embeddings into planar graphs. These results are used, among other places, in Section~\ref{sec:contract}, where we prove our contraction decomposition results for geometric objects. Section~\ref{sec:sparse} contains the sparsification results for both Subset TSP and Steiner Tree, showing that the clique size can be reduced to a function of $1/\eps$.

Our main result for Subset TSP (Theorem~\ref{thm:TSPalg}) is proved in Section~\ref{sec:tspspanner}, by first constructing the subset spanner and then using it with contraction decomposition and a bounded-treewidth algorithm.

Theorem~\ref{thm:SteinerSlow} for Steiner Tree is proved in Sections~\ref{sec:mortar} and \ref{sec:SteinerStruct}: the so-called mortar graph is constructed in Section~\ref{sec:mortar} and the main structural theorem for Steiner trees is stated in Section~\ref{sec:SteinerStruct}.
The faster PTAS for Steiner Tree (Theorem~\ref{thm:SteinerFast}) is proved in Section~\ref{sec:SteinerFast} using this structural understanding.

Section~\ref{sec:lower} proves our APX-hardness results (Theorem~\ref{thm:apxhard}), as well as the connection to planar problems (Theorem~\ref{planar_to_udg_reduction}).

Most of our results are proved for so-called $\alpha$-standard intersection graphs of polygons. Appendix~\ref{sec:graf_konvert} discusses how intersection graphs of similarly sized connected fat objects can be converted to such a representation.

\section{Overview}\label{sec:overview}
\subsection{EPTAS for Subset TSP}
\label{sec:eptas-subset-tsp}

We first review the main steps of the proof of Theorem~\ref{th:planarsubsetsp} by Klein~\cite{Klein06}, and then briefly explain the new ideas and concepts in our algorithm for the geometric intersection graph setting. In this and the following sections, we first go through the main algorithmic components of the previous algorithms and then introduce the new components that we developed. 

\subsubsection{Overview of Klein's techniques~\cite{Klein06} in planar graphs}
\paragraph*{Skeleton and subset spanner construction}
The first step in the algorithm of Klein~\cite{Klein06} is to build a \emph{subset spanner:} as stated in the first part of Theorem~\ref{th:planarsubsetsp}, this is a subgraph $G'$, containing all the terminals, such that $|V(G')|\le \poly(1/\eps)\tsp(G,T)$ and $\dist_{G'}(x,y)\le (1+\eps)\dist_G(x,y)$ for any $x,y\in T$. Let us start with any tree containing every vertex of $T$ and having size $O(\tsp(G,T))$; for example, a constant-factor approximate Steiner tree is suitable. This is likely not a correct subset spanner yet: a shortest $x-y$ path $P$ may use ``shortcuts'' outside  the tree. To obtain a spanner, we need to extend the tree (with a factor of $\poly(1/\eps)$ increase in total size) in a way that the shortest path $P$ can be rerouted (at the cost of at most a factor of $1+\eps$ increase in length). It is convenient for the description of the algorithm to change the setting: using the standard step of ``cutting open the tree'', we can assume that the boundary cycle $C$ of the infinite face contains all of $T$, and we want to extend $C$ into a subset spanner, increasing its size by at most a factor of $\poly(1/\eps)$.

As an intermediate step, we extend $C$ into a so-called skeleton (a subgraph of $G$), satisfying a weaker property. Let $F$ be a face of the skeleton; with slight abuse of notation, let $G[F]$ be the subgraph of $G$ containing every edge inside $F$ or on the boundary of $F$. The boundary of each (finite) face $F$ of the skeleton is divided into two paths, North and South, such that
\begin{itemize}
\item For any two vertices $x,y$ of North, the subpath $\textup{North}[x,y]$ is a shortest $x-y$ path in $G[F]$.
\item For any two vertices $x,y$ of South, the subpath $\textup{South}[x,y]$ is a $(1+\eps)$-approximate shortest $x-y$ path in $G[F]$.
\end{itemize}
We extend the cycle $C$ with the following iterative procedure. Suppose that $C$ has two vertices $x,y$ such that there is a ``shortcut'': an $x-y$ path $P$ inside $C$ such that $|P|$ is more than a factor of $(1+\eps)$ shorter than the $x-y$ subpath $P'$ of $C$. Let us choose $P$ such that $|P'|$ is minimum possible. If we extend $C$ by adding $P$ to it, then we are creating two faces (see Figure~\ref{fig:skeleton}):
\begin{itemize}
\item A face bounded by $P$ and $P'$, which satisfies the requirements when we define $P$ to be North (as it is a shortest path) and $P'$ to be South (where $(1+\eps)$-approximate minimality follows from the minimal choice of $P'$).
\item A face bounded by the cycle $C'$, which is obtained from $C$ by replacing $P$ with $P'$.
\end{itemize}
Let us continue this process with the cycle $C'$, stopping when no shortcut can be found in the cycle. A charging argument shows that the total length of all the added paths is $O(1/\eps)$ times the length of the initial cycle $C$.
\begin{figure}
\centering
\includegraphics{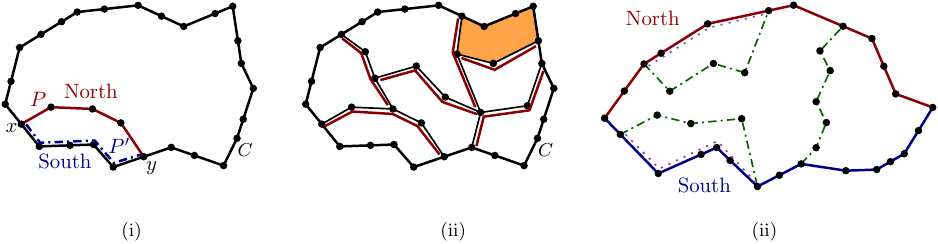}
  \caption{(i) Adding the shortcut $P$ to the skeleton creates two faces inside $C$. (ii) The final skeleton after adding multiple shortcuts (red). The final orange face has no shortcuts. (iii) A face of the skeleton, with a North-North, a South-South, and a North-South path going through it.}
  \label{fig:skeleton}
\end{figure}

Does the resulting graph have the subset spanner property? Consider a shortest $x-y$ path $P$ in $G$ and consider a subpath $P'$ that intersects the skeleton only in the endpoints $x'$ and $y'$, i.e., $P'$ goes inside some face $F$ (see Figure~\ref{fig:skeleton}(iii)). If $x'$ and $y'$ are both on the North part of $F$, then the shortest path property shows that $P'$ can be rerouted along North without any increase of length. Similarly, if $x'$ and $y'$ are both on South, then $P'$ can be rerouted along South with a $(1+\eps)$ factor increase of length. To handle the situation where $P'$ connects North and South, Klein~\cite{Klein06} constructed a spanner connecting two paths, in the following sense:

\begin{restatable}[Bipartite spanner, Klein~\cite{Klein06}]{lemma}{BipartiteSpan}\label{lem:bipartitespan}
Let $G$ be a graph, and let $Q$ be a shortest path in $G$ from $a$ to $b$, and
let $Q'$ be an arbitrary path from $a$ to $b$. Then for any $\eps>0$ there is
a subgraph $H$ with $|H|=O(1/\eps^3)|Q'|$ such that for any $q\in V(Q)$ and
$q'\in V(Q')$ we have $\dist_{H\cup Q \cup Q'}(q,q')\leq (1+\eps)\dist_G
(q,q')$.
\end{restatable}

Using Lemma~\ref{lem:bipartitespan} in every face of the skeleton with $Q=\textup{North}$ and $Q'=\textup{South}$ allows us to extend the skeleton such that $P'$ can be rerouted with a factor of $1+\eps$ increase whenever it connects North and South. Overall, we create the required spanner having size $O(1/\eps^4|V(C)|)$.

\paragraph*{Contraction decomposition and bounded-treewidth algorithms}
Using the subset spanner construction, we can assume that $|V(G)|,|E(G)|\le \poly(1/\eps)\tsp(G,T)$. Next we would like to decrease the treewidth of $G$ by contracting an appropriate small set of edges. To find this set of edges, we use the following contraction decomposition result (a similar result was implicitely proved by Klein\cite{DBLP:journals/siamcomp/Klein08}).

\begin{restatable}[Contraction Decomposition for Planar Graphs, Demaine~\etal~\cite{DemaineHM10}]{theorem}{ContractdecompPlanar} \label{thm:contractdecomp_planar}
For any integer $k\geq 2$, the edge set of any planar graph $G$ can be
partitioned into $k$ sets $E_1,\dots,E_k$ such that $\tw(G/E_i)=O(k)$ for all
$i\in [k]$. The partition can be found in $O(n^{3/2}\log n)$ time.
\end{restatable}

The smallest set $E_i$ has size at most $|E(G)|/k=\poly(1/\eps)\tsp(G,T)/k$. Thus by setting $k$ to be sufficiently large (but still polynomial in $1/\eps$), we can ensure that $|E_i|\le \eps\cdot\tsp(G,T)$. It can be shown that contracting a set $E$ of edges changes the optimum of Subset TSP by at most $2|E|$, hence it is sufficient to approximate $\tsp(G/E_i,T)$, where we know that $\tw(G/E_i)=O(k)=\poly(1/\eps)$. Such an instance can be solved in time $2^{\poly(1/\eps)\cdot n^{O(1)}}$ using standard dynamic-programming techniques.

\subsubsection{Solving Subset TSP in geometric intersection graphs}

\paragraph*{Cell partition} Intersection graphs can contain large cliques, making their structure significantly more complicated compared to planar graphs. However, it is well-known that in case of similarly sized fat objects, the intersection graph $G$ can be partitioned into cliques such that each clique is adjacent to a bounded number of other cliques. For the description of this partition, it is convenient to define similarly sized fat objects differently. We say that an intersection graph of simple polygons is $\alpha$-standard if (i) each component of the intersection of two polygons is a simple polygon, (ii) no three polygon edges meet at one point, (iii) each polygon contains a closed disk of radius 1 and has diameter at most $\alpha$. We denote by $\cI_\alpha$ the set of all $\alpha$-standard intersection graphs. It can be shown that an intersection graph of $\delta$-similarly sized $\beta$-fat polygons is in $\cI_\alpha$ for some $\alpha$ depending on $\beta$ and $\delta$ (see Theorem~\ref{thm:atalakitas}). From now on, we state the results for this class $\cI_\alpha$.

Formally, we define the \emph{cell partition} $\cP$ as follows. Consider the grid of points in the plane with integer coordinates. By item (iii) in the definition of $\alpha$-standard, every object has one such point in its interior. Let us assign each object to the lexicographically first grid point that is in its interior; let $\cP$ be the partition of the objects according to which point they are assigned to. Clearly, every class of $\cP$ induces a clique in $G$. We denote by $G_{\cP}$ the graph obtained by contracting each class of $\cP$ into a single vertex. Often, it is usefully to think of the structure of $G$ in terms of the graph $G_{\cP}$, as this graph better reveals its close-to-planar structure. Observe that two objects can intersect only if they are assigned to lattice points at distance at most $O(\alpha)$. Thus a vertex of $G_{\cP}$ can be adjacent only to vertices that represent points at distance $O(\alpha)$ in the plane; in particular $G_{\cP}$ has maximum degree $O(\alpha^2)$.

\paragraph*{(NEW) Sparsification} A clique in the partition $\cP$ can be large, but in Section~\ref{sec:sparse} we show that the instance can be preprocessed to make the size of these cliques bounded. First, if a clique contains many terminals, then (as they all have to be visited anyways), we may assume that they are mostly visited after each other. This observation allows us to decrease the number of terminals and assume that $|C\cap T|=O(1/\eps)$ for every $C\in \cP$. Let us select a distinguished vertex of each clique. A solution of Subset TSP consists of paths between terminals. If such a path is longer than $O(1/\eps)$, then we can modify it in such a way that after every $O(1/\eps)$ steps it visits a distinguished vertex, and this modification increases the cost only by a factor of $1+O(\eps)$. This means that we can assume that this $(1+\eps)$-approximate solution consists of paths of length $O(1/\eps)$ between terminals and distinguished vertices. Let us fix such a path for every pair. Consider a clique $C\in \cP$. There are $O(1/\eps^2)$ other cliques that are at distance $O(1/\eps)$, thus there are $O(1/\eps^3)$ terminals and distinguished vertices at distance $O(1/\eps)$. This means that there are at most $O(1/\eps^6)$ pairs of these vertices such that the fixed shortest path of length $O(1/\eps)$ between them go through $C$. As a shortest path can use only at most two vertices of a clique, this means that only $O(1/\eps^6)$ vertices of the clique are used by these paths and the rest of the vertices can be removed from the graph without affecting the $(1+\eps)$-approximate solution. Therefore, we can assume from now on that every clique has size $\poly(1/\eps)$.   

\paragraph*{(NEW) Wireframes and regions}
Our goal is again to find a subset spanner, as required by the first part of Theorem~\ref{thm:TSPalg}. In the planar case, the faces of the skeleton partitioned the instance into (internally) disjoint regions, and we could handle each region independently. The first problem we encounter in the case of intersection graphs is that the notion of faces, regions, being inside, being outside, etc.~are much less clear. To tackle this difficulty, we introduce a formalism that turns out to be essential for the robust and formal description of the algorithms (especially for Steiner Tree).

Given an intersection graph $G$, a \emph{wireframe} is a plane graph $\pl H$ where vertices are points in the plane and edges are polygonal curves. Each vertex $v$ of $\pl H$ has a parent object $p(v)$ of $G$ such that $v$ is located in the object $p(v)$. Moreover, for each edge $uv$ of $\pl H$, the corresponding polygonal curve can be split into a part inside $p(u)$ and a part inside $p(v)$ (see Definition~\ref{def:wireframe} for the precise definition). The graph $\pl H$ can have much fewer edges and vertices than $G$: it can be thought of as a planar subgraph of the intersection graph, with a certain form of planar representation. Note, however, that it is not necessarily a subgraph, as multiple points can have the same parent object.

\paragraph*{(NEW) Skeleton and subset spanner for intersection graphs.}
In Section~\ref{sec:tspspanner}, similarly to the planar algorithm, we first construct a 2-connected skeleton. In light of the sparsification result mentioned above, it is sufficient to bound by $\poly(1/\eps)\cdot \tsp(G,T)$ the number of \textit{cliques} of the cell partition $\cP$ that appear in the skeleton, as this also bounds the number of objects and hence proves Theorem~\ref{thm:TSPalg}(1).

Formally, now this skeleton is a wireframe of the original intersection graph $G$. As before, we want to divide the boundary of each face into North and South such that any subpath of them is (approximately) shorter than any path connecting the endpoints through objects ``inside'' the face. We use the following definition to remove any potential ambiguity of what it means that a path of objects is inside the face. 
Let $\re R$ be a region of the plane corresponding to a face of the skeleton. We define the intersection graph $G|_{\re R}$, where for every object $v\in V(G)$,  we introduce into $G|_{\re R}$ every component of $v\cap \re R$ as a separate object (so the objects in $G|_{\re R}$ are also connected polygons). Note that, due to these smaller objects created by the intersection with the boundary, the polygons in $G|_{\re R}$ are no longer similarly sized fat objects. 
When constructing the skeleton, we want to ensure that the boundary of each face is divided into two parts, North and South, such that between any two vertices $x,y$ on North, the subpath of North is not longer than the distance between $x$ and $y$ in $G|_{\re R}$ (and similarly for South with a factor of $1+\eps$).

We can build the skeleton with a similar iterative procedure as in the planar case. There is one technical challenge that appears when adding the shortcut and moving from a larger region $\re R$ to a smaller region $\re R'$: because we are cutting the objects at the boundary, the total number/complexity of the polygons in $G|_{\re R'}$ may be larger than in $G|_{\re R}$. Thus repeated applications of this step may lead to an exponential number of objects. We overcome this issue by an appropriate (nontrivial) choice of the polygonal curve that we are cutting along.

After constructing the skeleton, for every face bounding some region $\re R$, we use Theorem~\ref{lem:bipartitespan} on $G|_{\re R}$ to add connections between North and South (note that Theorem~\ref{lem:bipartitespan} works in arbitrary graphs). Does the resulting set of objects have the spanning property? Consider a path $P$ between two terminals and let $P'$ be a subpath that enters and leaves region $\re R$ by crossing North of $\re R$. Here is the point where the issue about crossing paths demonstrated in Figure~\ref{fig:unitdisk_vs_planar} comes into play: it is possible that the subpath $P'$ is not using any object of the path North. This means that if we want to use a subpath of North for rerouting a path, then the length is increased by one when we ``jump'' to North, and then once more when we ``jump'' back to the original path. 
Therefore, even though North has the required shortest path property, rerouting $P'$ on North may increase length by 2. If $P'$ is of length $\Omega(1/\eps)$, this increase by 2 is introduces only $1+O(\eps)$ multiplicative error, and hence tolerable. However, if $|P'|$ is $o(1/\eps)$, then the increase by 2 is significant.

Let us notice that if $|P'|=O(1/\eps)$, then $P'$ is fully contained in the distance-$O(1/\eps)$ neighborhood of the skeleton. Thus if we include in the subset spanner the distance-$O(1/\eps)$ neighborhood of the skeleton, then every such short path $P'$ is fully included in the subset spanner. Consequently, if the subpath $P'$ leaves the subset spanner, then it has to be of length $\Omega(1/\eps)$, hence it can be rerouted along North with a factor $1+O(\eps)$ increase. We can argue similarly for subpaths entering and leaving by crossing South, or for subpaths connecting North and South. Notice that there are $O(1/\eps^2)$ cliques in the cell partition that are at distance at most $O(1/\eps)$ from an object.
Thus adding the distance-$O(1/\eps)$ neighborhood of the skeleton increases the number of cells only by a factor of $\poly(1/\eps)$ and hence we can ensure that the obtained subset spanner is the union of $\poly(1/\eps)\cdot \tsp(G,T)$ cliques of the cell partition, as we wanted. Moreover, our earlier sparsification ensured that each clique contains $\poly(1/\eps)$ objects, hence in fact the subset spanner consists of $\poly(1/\eps)\cdot \tsp(G,T)$ objects, proving the first part of Theorem~\ref{thm:TSPalg}.

\paragraph*{(NEW) Contraction decomposition via Lipschitz embedding}

Having constructed the subset spanner, we can assume that $|V(G)|\le \poly(1/\eps)\tsp(G,T)$.
To reduce the problem to an instance with bounded treewidth, we would like to use an analog of Theorem~\ref{thm:contractdecomp_planar} for intersection graphs. Let us observe that the intersection graph itself can be a clique, thus it is impossible to reduce treewidth to a constant by contracting a small part of the graph.  However, by the sparsification procedure described above, we can assume that there is a cell partition where every clique has size $\poly(1/\eps)$. Therefore, in Section~\ref{sec:contract}, we prove the analog of Theorem~\ref{thm:contractdecomp_planar} under a bound on the clique size. As a minor notational difference, instead of edge contractions, it will be convenient to consider contracting vertex sets. Given a set $X$ of vertices, we let $G\vcon X=G/E(G[X])$ to be the graph obtained by contracting each component of $G[X]$ into a single vertex. The following statement is proved by Theorem~\ref{thm:contractdecomp}(ii):
\begin{theorem}\label{thm:introcontract} 
 Let $G\in \cI_\alpha$ be an intersection graph of similarly sized fat objects with cell partition $\cP$ where each clique $C\in \cP$ has size at most $|C|\leq \ell$.
 Given an integer $k$, we can compute in polynomial time $k$ sets $X_1,\dots, X_k$ of vertices such that $\sum_{i=1}^k|X_i|=O(|\cP|\ell)$ and $\tw(G\vcon X_i)=O(k \ell)$.
\end{theorem}
We remark that in case of unit disk graphs, there are now known contraction decomposition theorems~\cite{UDGcontractdecompSODA19,TrueContractDecomp} that can handle cliques, but due to our more general objects these results are insufficient for our purposes. (In fact, they are insufficient for our Steiner tree algorithm even if the input graph is a unit disk graph, as we need to modify our starting objects via region restrictions and other changes.)

With Theorem~\ref{thm:introcontract} at hand and using the bounds on $|\cP|$ and on the maximum clique size, we can solve the instance in time $2^{\poly(1/\eps)}\cdot n^{O(1)}$ using a standard dynamic programming algorithm solving the problem on bounded-treewidth graphs. 

It is natural to approach the proof of Theorem~\ref{thm:introcontract} via the graph $G_{\cP}$: we would like to find sets of vertices $X'_i$ in $G_{\cP}$ such that $G_{\cP}\vcon X'_i$ has treewidth $O(k)$. For each $X'_i$, there is a set $X_i$ of at most $|X'_i|\ell$ vertices in $G$, corresponding to the cliques in $X'_i$. Furthermore, if we replace each vertex of $G_{\cP}\vcon X'_i$ by a clique of size $\ell$, then we get a supergraph of $G\vcon X_i$. Thus the treewidth bound of $O(k)$ on $G_{\cP}\vcon X'_i$ would imply the required bound of $O(k\ell)$ on $G\vcon X_i$.

How can we find these sets $X'_i$ in $G_{\cP}$? As $G_{\cP}$ is somewhat close to being planar, we can try to approximate it with a genuinely planar graph $H$ and then invoke Theorem~\ref{thm:contractdecomp_planar} on $H$. However, the requirements of this approximation is fairly delicate and naive ideas (such as representing $G_{\cP}$ by an induced subgraph of a sufficiently fine grid) fail. The reason is that $H$ has to strike a balance between two conflicting requirements:
\begin{itemize}
\item $H$ should not have too many edges: otherwise, contracting a small set in $H$ does not correspond to the contraction of a small set in $G_{\cP}$.
  \item $H$ should not have too few edges: otherwise, small treewidth after contracting a set in $H$ does not mean that contracting the corresponding set in $G_{\cP}$ results in small treewidth.
  \end{itemize}
It turns out that the notion of Lipschitz embedding satisfies these requirements.   
  We say that a mapping $f:V(G)\rightarrow V(H)$ is
\emph{$c$-Lipschitz} for some constant $c\geq 1$ if for any pair of vertices $x,y\in V(G)$
we have that
\begin{equation*}
\frac{1}{c}\cdot \dist_G(x,y)\leq \dist_H(f(x),f(y))\leq c\cdot \dist_G(x,y).
\end{equation*}
For unit disk intersection graphs, a construction of a $c$-Lipschitz embedding into a planar graph is known for some constant $c$: in fact, they admit planar $c$-spanners \cite{CCV10planarhopspan,Biniaz2020}. In Section~\ref{sec:lipschitz}, we construct $c$-Lipschitz embeddings for the intersection graphs of connected similarly sized fat objects. The fact that we are handling non-convex objects makes the construction of the embedding significantly more challenging and requires a different strategy. (See \Cref{thm:getplanarclique} for a more detailed theorem statement.)
\begin{theorem}\label{thm:introgetplanarclique}
For any $\alpha>0$, there exists a positive constant $c=O(\alpha^{32})$ such that
for any $G\in \cI_\alpha$ given by its representation the following hold. Let
$\cP$ be the cell partition of $G$. Then in polynomial time one can construct
a plane graph $H$ with the same vertex set as $G_{\cP}$ such that the identity mapping  is $c$-Lipschitz from $G_\cP$ to~$H$.
\end{theorem}
Very briefly, we want to represent edges of $G_{\cP}$ by polygonal curves that are covered by two adjacent objects. Then we introduce a new vertex at each point where two such curves intersect each other. If the objects are convex, then the curves can be chosen to be simple and it is easy to bound the number of these intersections. However, if the objects are not convex, then already two curves can have lots of intersections. In case two curves intersect multiple times, we can try to reroute one curve on the other to reduce the number of intersections, at the cost of making some curves longer. However, doing this iteratively can make curves arbitrarily longer. In order to avoid this, we split the curves into a constant number of nonconflicting sets, perform many reroutings in parallel, and then repeat this in a constant number of rounds. Theorem~\ref{thm:introgetplanarclique} requires that $H$ and $G_{\cP}$ have the same vertex set, but our construction introduced many additional vertices at the intersections. To get rid of these intersection vertices, we can contract the Voronoi regions of the original vertices, retaining the Lipschitz propeperty and ensuring that the planar graph $H$ has the same vertex set as $G_{\cP}$.

The constructed planar graph $H$ will be used as a proxy for $G_{\cP}$ when trying to reduce treewidth using contractions. We invoke Theorem~\ref{thm:contractdecomp_planar} to find a small set $X$ in $H$ whose contraction reduces treewidth. To complete the proof of Theorem~\ref{thm:introcontract}, we need to argue that an analogous contraction reduces treewidth in $G_{\cP}$ as well. This requires a delicate argument that modifies a tree decomposition of a contraction of $H$ to obtain a tree decomposition of a contraction of $G_{\cP}$. In particular, it is not sufficient to contract $X$ in $G_{\cP}$, but we need to contract a small neighborhood of $X$.

\begin{restatable}{lemma}{raisecontract}
\label{lem:twcompare}
Let $G$ and $H$ be graphs with the same vertex set $V$,
so that the identity map of $V$ is $c$-Lipschitz from $G$ to $H$. Let $X\subset V$ be
arbitrary, and let $X^r$ be the set of vertices whose endpoints have
$H$-distance at most $r$ from some vertex in $X$, i.e., we set $X^r=N_H
(X,r)$. Assume moreover that $H$ has degree at most $\Delta$. Then $\tw(G\vcon X^{c^2})= O(\Delta^c) \cdot \tw(H\vcon X)$.
\end{restatable}


\subsection{EPTAS for Steiner Tree}

\subsubsection{Overview of the techniques of Borradaile et al.~\cite{BorradaileKM09} in planar graphs}
Let us review the EPTAS of Borradaile, Klein, and Mathieu~\cite{BorradaileKM09} for Steiner Tree on planar graphs and then explain what new algorithmic components are needed in case of intersection graphs (for example, techniques to obtain noncrossing structures of various types, which can be taken for granted in the case of planar graphs). The algorithm for planar graphs starts with the same idea as in the case of Subset TSP: constructing a \emph{Steiner-tree spanner,}\footnote{The term
\textit{banyan} is a different name for the same object~\cite{DBLP:conf/stoc/RaoS98,DBLP:conf/stoc/BartalG21}.} that is, a subgraph that admits a $(1+\eps)$-optimal Steiner tree of the terminals in $T$. However, the size of the Steiner-tree spanner obtained by the algorithm can be bounded only by $2^{\poly(1/\eps)}\cdot \smt(G,T)$, hence the contraction decomposition can reduce the treewidth only to $2^{\poly(1/\eps)}$, resulting in $2^{2^{\poly(1/\eps)}}\cdot n^{O(1)}$ running time after performing contraction decomposition and solving the resulting bounded-treewidth instance.

\paragraph*{Columns and the mortar graphs}
The algorithm for Steiner Tree starts similarly to the Subset TSP algorithm: with the construction of the skeleton of size $\poly(1/\eps)\smt(G,T)$. But then each face $F$ is further partitioned, the following way. Consider a face of the skeleton, with its boundary partitioned into South and North. Let $s_0$ be an endpoint of South and then for $i=1,2,\dots$, vertex $s_i$ is defined to be the first vertex $s_i$ on South such that the distance from $s_i$ to North inside the face $F$ is at most $|\textup{South}[s_{i-1},s_i]|/\eps$. Let column $C_i$ be a shortest path from $s_i$ to North. In a planar graph, we can assume that the columns do not cross. It is easy to observe that the total length of the columns in a face is at most $|\textup{South}|/\eps$. Thus the total length of all columns in all faces of the skeleton is at most $O(1/\eps)$ times the total size of the skeleton, hence it is also $\poly(1/\eps)\cdot \smt(G,T)$.

We fix a constant $\kappa=\poly(1/\eps)$, and in each face of the skeleton we extend the skeleton by every $\kappa$th column. These columns partition each face into multiple \emph{bricks,} and the columns form the East and West sides of these bricks (see Figure~\ref{fig:brick_in_skeleton}). The skeleton extended by these columns is called the \emph{mortar graph.}
\begin{figure}
  \centering
  \includegraphics{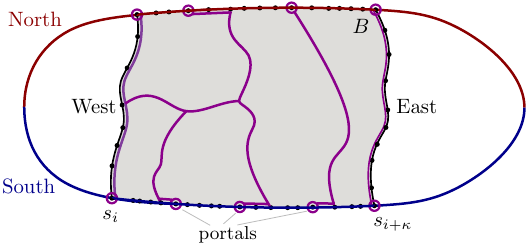}
  \caption{The brick $B$ inside a skeleton face, bounded by columns starting at $s_i$ and $s_{i+\kappa}$. The solution inside $B$ (purple forest) includes $\Ea(B)$ and $\We(B)$, and it is connected to the nearest portals (purple circled nodes) along $\No(B)$ and $\So(B)$.}\label{fig:brick_in_skeleton}
\label{fig:brickforest}
\end{figure}

More precisely, we select an $i\in [\kappa]$ and, starting with the $i$th column, we include every $\kappa$th column in the mortar graph . By an averaging argument, there is an $i$ such that the total length of the added columns is at most a $1/\kappa$ fraction of the total length of all columns. If $\kappa$ is sufficiently large, then the total length is $O(\eps \cdot \smt(G,T))$. Therefore, we can afford to ``buy'' these columns and can be assumed to be part of the solution. This can be expressed by contracting them, thus they will not play much role in the later arguments. The purpose of extending the skeleton this way is that now we can assume that the sequence $s_1$, $s_2$, $\dots$ has length at most $\kappa$ in each brick. This structural property is heavily exploited in the next step.

\paragraph*{Simplifying a forest}
We can think of the solution as having two parts: a subset of the mortar graph and a forest inside each face of the mortar graph. If a tree in a brick is attached only to the North or South side, then it can be replaced by a subpath of that side, increasing cost only by at most a $1+\eps$ factor. But, in general, a tree can connect several vertices of the North side with several vertices of the South side, and then it is not clear if the shortest path properties of North and South can be exploited for any simplification. Nevertheless, Borradaile, Klein, and Mathieu~\cite[Thm.~10.7]{BorradaileKM09} showed by a careful simplification argument that parts of the tree can be replaced by subpaths of the boundary in such a way that the size of the solution increases only by at most a factor of $(1+\eps)$ and the forest in each brick is joined to the boundary of the brick only at $\poly(1/\eps)$ vertices.

\paragraph*{Portals}
We want to ensure that the inside of a brick has limited and well controlled interaction with the boundary, that is, the solution enters the inside of the brick via a limited number of portals. Let us consider a brick $B$ an let us designate $\theta=\poly(1/\eps)$ portals on the boundary of $B$ at roughly equal distance from each other. That is, if $\ell(B)$ is the length of the boundary of $B$, then the portals are at distance roughly $\ell(B)/\theta$ from each other. Whenever the solution has a tree inside a brick that is attached to a vertex $v$ of the boundary, then let us extend the solution by connecting $v$ to the closest portal. For each such vertex $v$, this modification increases the cost by at most $\ell(B)/\theta$. As we have assumed that the forest inside $B$ is attached to the boundary at $\poly(1/\eps)$ vertices, the total increase incurred by brick $B$ is at most $\ell(B)\cdot\poly(1/\eps)/\theta$. As the total size of the mortar graph created so far was $\poly(1/\eps)\cdot\smt(G,T)$, the total increase of length is $\poly(1/\eps)\cdot\smt(G,T)/\theta$ when summed over every brick $B$. If $\theta=\poly(1/\eps)$ is sufficiently large, then this increase is $\eps\cdot\smt(G,T)$. Thus after some slight modification of the graph $G$ (as demonstrated in Figure~\ref{fig:brick}), we can assume that the inside of each brick is connected only to the $\theta$ portals on the boundary.

\begin{figure}
  \centering
  \includegraphics{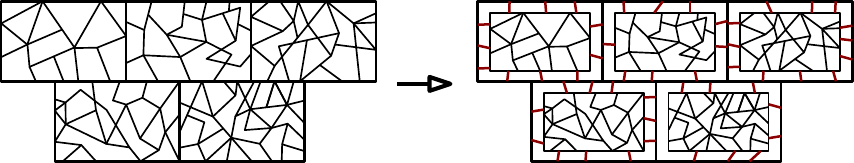}
  \caption{The graph modification that restricts interactions between bricks to the portals.}\label{fig:brick}
\end{figure}

\paragraph*{Constructing the Steiner-tree spanner}
Inside each brick $B$, the only role of the solution is to connect certain subsets of the $\theta$ portals on the boundary of $B$. We need to ensure that the this can be done as efficiently in the Steiner-tree spanner as in the original graph. For this purpose, for every subset of the $\theta$ portals of $B$, we compute an optimal Steiner tree in $B$ and add this tree to the Steiner-tree spanner. Note that an optimal Steiner tree with at most $\theta$ terminals can be computed in time $2^{O(\theta)}\cdot n^{O(1)}=2^{\poly(1/\eps)}\cdot n^{O(1)}$ in any graph. Furthermore, in our case we can exploit that each brick is planar and the terminals are on the boundary, in which case the problem is polynomial-time solvable \cite{DBLP:journals/mor/EricksonMV87}. Clearly, the size of each such tree is at most the length $\ell(B)$ of the boundary of $B$, as the boundary clearly connects all these terminals. Thus the total length we add for each brick $B$ is at most $2^{\theta}\cdot\ell(B)=2^{\poly(1/\eps)}\cdot\ell(B)$. Summing for every brick $B$, this increases the total length of the mortar graph by a factor of $2^{\poly(1/\eps)}$, resulting in a Steiner-tree spanner of length $2^{\poly(1/\eps)}\smt(G,T)$. 

\subsubsection{Constructing a Steiner tree spanner in geometric intersection graphs}

\paragraph*{(NEW) Noncrossing columns via Lipschitz embedding}
Section~\ref{sec:mortar} is devoted to constructing an analog of the mortar graph in the intersection graph setting.
In the case of intersection graphs, if we define the columns as shortest paths from each $s_i$ to North, then these paths may cross and therefore they are not suitable for defining the East/West sides of the brick. One could try to uncross these paths, but in intersection graphs jumping from one path to the other can increase the length, and it is hard to control the total increase of length after iteratively uncrossing paths. To overcome this issue in a robust way, we use a Lipschitz embedding of the union of the columns into a planar graph (more precisely, into a wireframe of the union of the columns). Using this embedding, we can find shortest paths that correspond to noncrossing curves in the wireframe, hence they can be used as the East and West sides of the brick. The Lipschitz property of the embedding ensures that the columns defined this way have the required length and distance properties.

There is a technical complication due to the fact that we have Lipschitz-embedding results only for similarly sized fat objects. However, the objects in $G|_{\re R}$ can be non-fat and have smaller diameter, as $G|_{\re R}$ contains objects intersected with the boundary of $\re R$. We solve this issue by considering the union of the first few objects in each column as a single object, ensuring that at least one of these objects is disjoint from the boundary of $\re R$ (and similarly for the last few objects). Note that the union of a constant number of similarly sized fat objects changes diameter and fatness only by a constant factor, thus we can still consider this set of objects as similarly sized fat. A further technical complication that appears here is that the created objects can have holes. However, we can modify them by ``puncturing the holes'' in a way that the intersection graph does not change.

Another technical difference compared to the case of planar graphs is that we have to be more careful about the columns going on or close to the boundary of the region. It is even possible that multiple columns overlap with the northern or southern boundary. In the clean world of planar graphs, such shared subpaths are inconsequential. On the other hand, for intersection graphs, handling this issue introduces additional technicalities. In particular, we need to allow that on some part of the boundary South and East overlap, complicating the definitions considerably.

\paragraph*{(NEW) Sparsification}
In Section~\ref{sec:sparse}, we show that, similarly to Subset TSP, large cliques can be reduced in size with a factor of at most $(1+\eps)$ change in the optimum value. However, for Steiner Tree, we can guarantee only an exponential $2^{\poly(1/\eps)}$ upper bound on the reduced clique size. The crucial difference is that the role of a vertex in a clique is not just to be part of a path of length $O(1/\eps)$ connecting two vertices at distance at most $O(1/\eps)$, but to be part of a subtree connecting a subset of up to $O(1/\eps)$ vertices at distance at most $O(1/\eps)$. There are exponentially many such subsets and for each subset potentially a different vertex of the clique is needed, resulting in the exponential upper bound.

\paragraph*{(NEW) Uncrossing trees and object frames}
In Section~\ref{sec:SteinerStruct}, we would like to use the forest simplification result \cite[Thm.~10.7]{BorradaileKM09} from the planar algorithm as a black box. This seems reasonable: even though intersection graphs can be more complicated than planar graphs, the solution inside a brick is still just a forest. The forest simplification result considers only forests, so one could argue that it is irrelevant whether this forest is taken from a planar graph or from an intersection graph. However, it is more precise to say that the simplification is done on a planar graph that consists of a boundary cycle and a forest inside it; then the simplification argument replaces some parts of the forest with subpaths of the boundary. Therefore, in order to invoke this result, we need to define such a planar graph based on the solution in the intersection graph. The problem is that even if an intersection graph is a tree, the objects can cross each other and hence it might not be possible to draw the tree without crossing edges in the plane if we want to respect respect the topology of the original objects on the boundary (see Figure~\ref{fig:crossingtree}(i)). Thus the union of the forest and the boundary can be highly nonplanar.
\begin{figure}
\centering
  \includegraphics{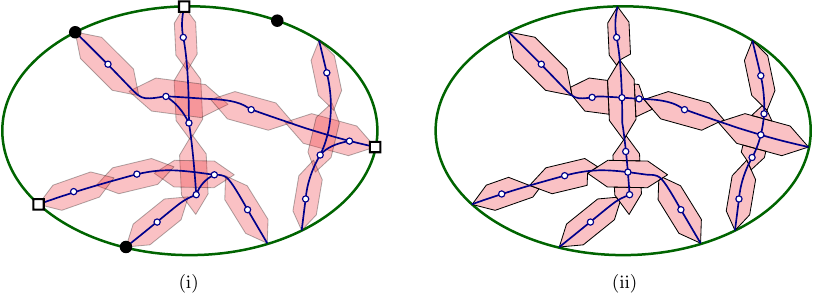}
  \caption{(i) The intersection graph of the 15 red objects and the 3+3 objects on the boundary is a tree (shown by blue). However, the union of this tree and the green boundary is nonplanar, as it contains the subdivision of a $K_{3,3}$. (ii) Replacing the 15 red objects with an object frame of 18 subobjects. Union of the intersection graph of these objects and the boundary is now planar.}
  \label{fig:crossingtree}
\end{figure}  

In order to be able to use \cite[Thm.~10.7]{BorradaileKM09}, we create a ``really planar'' representation of the forest inside the brick. An \emph{object frame} of an intersection graph $G$ of polygons is an intersection graph $H$, where each object $v$ of $H$ is a subpolygon of a parent object $p(v)$ of $G$ and the polygons of $H$ are internally disjoint (plus some further nondegeneracy conditions). Note that $p(u)=p(v)$ is possible for $u,v\in V(H)$, that is, the object frame can use multiple subpolygons of a polygon. We show in Section~\ref{sec:connectivity} (Theorem~\ref{thm:tracespanningtree}) that, given an intersection graph $G$ that covers a set of $t$ points in the plane, we can construct an object frame $H$ of $G$ such that (i) $H$ is a forest, (ii) $|V(H)|\le |V(G)|+O(t)$, (iii) $H$ provides the same connectivity of the points as $G$, that is, if two of the $t$ points are covered by a component of $G$, then it is also connected by a component of $H$. With this result at hand, we replace the forest of the solution inside a brick by an object frame (see Figure~\ref{fig:crossingtree}(ii)). Now the forest defined by this object frame together with the boundary form a planar graph, hence \cite[Thm.~10.7]{BorradaileKM09} can be invoked.

\paragraph*{(NEW) Solution generators}
In the case of planar graphs, we considered the solution to be composed of two disjoint parts: a subgraph of the mortar graph and forests inside the bricks. However, unlike in planar graphs, it is not straightforward to modify the instance such that the interior of a brick is connected only to the portals of the boundary. Moreover, the situation is even more complicated: as it was the case for Subset TSP, we need to include in the Steiner-tree spanner the distance-$\poly(1/\eps)$ neighborhood of the mortar graph, making it even less clear how the interior is supposed to be connected only to the portals.

For intersection graphs, we need to adopt a slightly different viewpoint. Instead of saying that the solution is partitioned into a subgraph of the mortar graph and forests, we say that the solution provides two things: a subgraph $S_0$ in the neighborhood of the mortar graph and connections between some subsets of the portals. Of course, the connections between the portals are realized by trees and these trees have to go through the neighborhood of the mortar graph, but we do not care what exactly these trees do there. In particular, we do not care if these trees use objects that are already in $S_0$: we can achieve $(1+\eps)$-approximate solution size even if we are double counting such objects. In Section~\ref{sec:SteinerStruct}, we prove our main structural result for Steiner Tree, which shows that there is an $(1+\eps)$-approximate solution that arises this way. Proving this requires modifying a solution in a way that components that connect portals are not needed to connect other parts of the boundary.

For the robust handling of the solutions inside bricks, we introduce a more abstract formalism, which will be particularly useful for the next result, the faster PTAS (Theorem~\ref{thm:SteinerFast}). A \emph{portal-respecting generator} consists of a subset $S_0$ of the distance-$\poly(1/\eps)$ neighborhood of the mortar graph and for each brick $B$ a collection $\cF$ of subsets of portals. The generator is feasible if $S_0$ connects all the terminals in $T$ if we extend $S_0$ with ``virtual edges'' connecting, for every $P\in \cF$, the vertices of $P$ to each other. Intuitively, a generator represents solutions that consist of $S_0$ and for each set $P$ of portals in $\cF$, an arbitrary tree connecting $P$. If the generator is feasible, then any set obtained this way is indeed a solution. We define the cost of the generator to be $|S_0|$, plus, for every $P\in \cF$,  the cost of the minimum Steiner tree connecting $P$. Clearly, any solution represented by a feasible generator has size at most the cost of the generator. Our main structural result can be formulated as saying that there is a feasible generator whose cost is at most $(1+\eps)\smt(G,T)$.

\subsection{Faster PTAS for Steiner Tree}
\label{sec:faster-ptas-steiner}
Borradaile, Klein, and Mathieu~\cite{BorradaileKM09} presented a PTAS for planar graphs that, instead of using the Steiner-tree spanner, uses dynamic programming on the mortar graph (Theorem~\ref{th:planarsteinertspfast}). Our proof of Theorem~\ref{thm:SteinerFast} uses a very different approach. Therefore, instead of giving an overview of the technique of Borradaile, Klein, and Mathieu~\cite{BorradaileKM09}, we first sketch a different, more streamlined algorithm for planar graphs. The ideas in this sketch form the basis for an analogous algorithm for intersection graphs presented in Section~\ref{sec:SteinerFast}. However, the difficulties arising in the intersection graph setting lead to major technical challenges and makes it necessary to substantially redesign this approach.

\paragraph*{(NEW) PTAS via weighted tree-representative vertices in planar graphs}
Using contraction decomposition, we can find a set $X$ of $O(\eps\cdot \smt(G,T))$ vertices whose contraction reduces the treewidth of the mortar graph $M$ to $\poly(1/\eps)$. We can afford to include these vertices into the solution, which is effectively the same as contracting them in the input. Thus we may assume that the mortar graph has treewidth $\poly(1/\eps)$. This of course does not mean that the input graph has bounded treewidth: the parts of the graph inside the bricks can still be complicated. However, this part of the graph is used only for trees connecting terminals; the idea is to represent each possible tree with a single vertex, making what appears inside each brick bounded.

It is convenient to switch at this point to a vertex-weighted problem, where each vertex $v$ has a nonnegative weight $w(v)$ and the task is to minimize the total weight of the solution. (Note that we are not aiming to solve the vertex-weighted version in general, but the specific instance appearing in the rest of the algorithm can be conveniently described and solved as a weighted instance.) We extend the mortar graph $M$ the following way. Consider the $\theta=\poly(1/\eps)$ portals of a brick $B$ and for any subset $P$ of them, introduce a vertex $v_P$ that is adjacent to every portal in $P$; let the weight of $v_P$ be the size of the minimum Steiner tree containing $P$. We repeat this step for every brick $B$ and every subset of the portals of $P$, this way we are introducing at most $2^\theta=2^{\poly(1/\eps)}$ new vertices for each brick $B$ (note that the resulting graph $M^*$ can be highly nonplanar). It is clear that any solution of weight $W$ of this new weighted instance $(M^*,T)$ can be turned into a solution of size $W$ of the original unweighted instance $(G,T)$. Also, it should be clear that if the original instance $(G,T)$ has a solution of size $W$ that consists of a subgraph of the mortar graph $M$ and trees connecting subsets of portals in the bricks, then the new instance $(M^*,T)$ has a solution of weight $W$. Thus it is sufficient to solve this new, weighted instance $(M^*,T)$.

How can we bound the treewidth of $M^*$? It is known that introducing a new vertex into each face of a planar graph and connecting it arbitrarily with vertices of the face changes treewidth only by a constant factor. We are introducing at most $2^{\poly(1/\eps)}$ vertices per face, which then increases treewidth by at most a factor of $2^{\poly(1/\eps)}$. To see this treewidth bound, let us first introduce a vertex $v_F$ in each face $F$ such that $v_F$ is adjacent to every vertex of $F$, and consider a tree decomposition of this graph. Whenever $v_F$ appears in a bag, let us replace it with a block of $2^{\poly(1/\eps)}$ copies. It is clear that the resulting tree decomposition has treewidth at most $O(2^{\poly(1/\eps)})$ times the original treewidth.
The standard dynamic programming algorithm for Steiner Tree can be easily adapted to the vertex-weighted version. However, as this algorithm is exponential in treewidth, this would still give only a  $2^{2^{\poly(1/\eps)}}\cdot n^{O(1)}$ time algorithm, which is weaker than our goal.

To improve the running time, let us observe that the dynamic programming can be made faster on such decompositions. In each bag of the decomposition, we have at most $\poly(1/\eps)$ original vertices of $M$, plus at most $\poly(1/\eps)$ blocks of tree-representative vertices, with each block having size at most $2^{\poly(1/\eps)}$. The important observation is that if a solution for the original instance $(G,T)$ uses at most $\theta=\poly(1/\eps)$ trees in each brick, then there is a corresponding solution of the new instance $(M^*,T)$ that uses at most $\poly(1/\eps)$ tree-representative vertices from each block. Therefore, there are $(\poly(1/\eps)\cdot 2^{\poly(1/\eps)})^{\poly(1/\eps)}=2^{\poly(1/\eps)}$ different subsets of a bag that can potentially appear in such a solution and we can restrict our search to solutions with this property. The running time of (weighted) Steiner tree mainly depends on the number of subproblems defined in the dynamic programming, which in turn mainly depends on the number of possible subsets of vertices that need to be considered at each bag. Given our bound on the number of relevant subsets, it follows that standard dynamic programming can solve the problem in $2^{\poly(1/\eps)}\cdot n^{O(1)}$ time.

\begin{figure}
\centering  
  \includegraphics{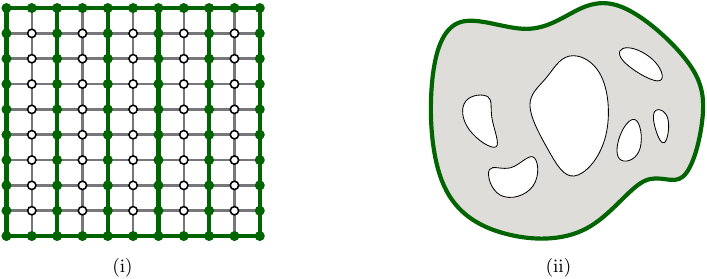}
  \caption{(i) Even if a subgraph (green) has constant treewidth, its neighborhood can induce a a graph of arbitrary large treewidth. (ii) Extending a single brick of the mortar graph with its distance-$\poly(1/\eps)$ neighborhood (gray shaded region) can create multiple ``islands''.}
  \label{fig:mortartreewidth}
\end{figure}  

\paragraph*{(NEW) Contracting through bricks}
If we try to follow the same strategy to speed up our PTAS for intersection graphs, then the main source of the difficulty is that the solution is not composed from a subgraph of the mortar graph plus trees in the bricks, but from a subgraph of the distance-$\poly(1/\eps)$ neighborhood of the mortar graph plus trees in the bricks. Let us observe that even if the mortar graph has bounded treewidth, it very well may be that the neighborhood of the mortar graph has large treewidth (see Figure~\ref{fig:mortartreewidth}(i) for an example). Inside a brick, if we consider a distance-$\poly(1/\eps)$ neighborhood of the boundary, then this may create connections between different parts of the boundary and can split the inside of the brick into multiple ``islands''. By their construction, the bricks have lots of structural properties and we may try to argue that including the neighborhood cannot create so many extra connections to increase treewidth significantly. While it may be possible to prove this, we found potential arguments very tedious and therefore chose a different strategy.

In Section~\ref{sec:SteinerFast}, instead of using contraction decomposition to find a set whose contraction decreases the treewidth of the mortar graph, we use contraction decomposition on the distance-$\poly(1/\eps)$ neighborhood of the mortar graph to reduce the treewidth of this neighborhood to $\poly(1/\eps)$. Note that now the set $X$ we contract can go through a brick (see Figure~\ref{fig:mortarkontrakt}). Therefore, it is difficult to determine where we should put a tree-representative vertex corresponding to a brick. Instead, we put tree-representative vertices in each island. As each island is in a single brick, these vertices should represent a part of one of the $2^{\poly(1/\eps)}$ possible trees that can appear in this brick. (Here it is a technical complication that the restriction of a tree to an island can consist of many components.)
\begin{figure}
\centering
  \includegraphics[width=0.6\linewidth]{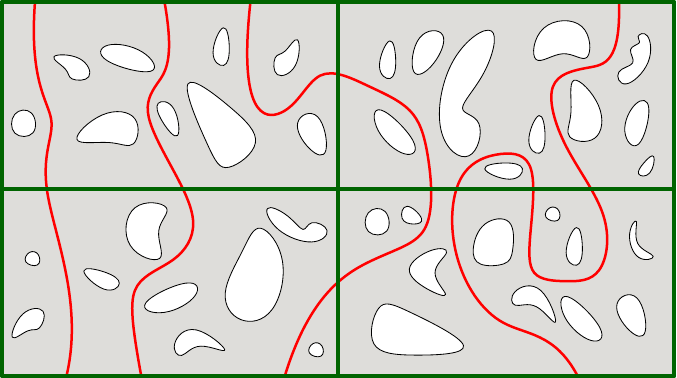}
  \caption{The mortar graph (green), its distance-$\poly(1/\eps)$ neighborhood (shaded gray), and the set $X$ (red) whose contraction reduces the treewidth of this neighborhood.}
  \label{fig:mortarkontrakt}
\end{figure}  

If the graph were planar, then the scheme described in the previous paragraph could be made to work, with techniques similar to what we described in the planar case. In the case of intersection graphs, the set $X$ for contraction needs to be found in a planar constant-Lipschitz representation of the graph (or more precisely, of the cell partition). Moreover, we need to insert the tree-representative vertices into the faces of this planar representation. In the end, we argue that certain properties (such as treewidth) of this planar representation has implications on the original graph (similarly to Lemma~\ref{lem:twcompare} above, which shows how the effect of contraction in the planar representation can be lifted back to the original graph). In this situation, this lifting is extremely delicate and nontrivial, especially if the set $X$ is close to some of the islands. Nevertheless, we overcome these difficulties and manage to construct a $(1+\eps)$-approximately equivalent weighted instance that has a tree decomposition where the optimal solution has only $2^{\poly(1/\eps)}$ different possible intersections with each bag. Then the  $2^{\poly(1/\eps)}\cdot n^{O(1)}$ running time follows by dynamic-programming techniques.

\section{Preliminaries}\label{sec:prelim}

\paragraph*{Graphs, walks, paths, distances}
We use standard graph notation~\cite{Diestel} unless stated otherwise. The graph $G$ has vertex set $V(G)$ and edge set $E(G)$.
A \emph{walk} $P$ is an alternating sequence of vertices and edges, starting and ending with a vertex. A walk is \emph{closed} if its start- and endpoint are the same vertex, it is a \emph{path} if each of its vertices appears exactly once in $P$, and it is a \emph{cycle} if it is a closed walk where all vertices are non-repeating except for the start and end vertex being identical. The length of $P$ is denoted by $|P|$, and it is the number of edges in the sequence of $P$. In particular, we allow for length-$0$ paths that consist of a single vertex. For a cycle, path, or walk $P$ the shortest path (or walk) from $u$ to $v$ along $P$ is denoted by $P[u,v]$. We will also use $P(u,v)$ to refer to the subpath/subwalk $P[u,v]$ with the first and last vertex removed.

The shortest-path distance of $u,v\in V(G)$ is denoted by $\dist_G(u,v)$ and is simply the number of edges on a shortest path with endpoints $u$ and $v$. For a vertex set $X\subset V(G)$ let $N(X,k)$ be the set of vertices within
distance $k$ from $X$, that is, $N(X,k)=\{v\in V(G) \mid \exists x\in
X: \dist_G(x,v)\leq k\}$. We also use the notation $N(v,k)$ when talking about the neighborhood of one vertex, and $N(X)$ (resp., $N(v)$) is a shorthand for $N(X,1)$ (resp., $N(v,1)$), i.e., the \emph{closed neighborhood} of $X$ (resp., the closed neighborhood of $v$).

\paragraph*{Partitions, contractions, and consolidations} 
Let $G$ be a graph and let $\cP$ be a partition of $V(G)$. For a vertex $v\in V(G)$ or vertex set $X\subseteq V(G)$, let $v_\cP$ or $X_\cP$ denote the partition class of $v$ or the set of partition
classes in $\cP$ intersected by $X$, respectively.
Given a vertex set $X$ and a partition $\cQ$ of some vertex set $Y\subseteq V(G)$, we denote by $X\parcon \cQ$ the set where vertices of $X$ that fall in the same class of $\cQ$ are identified. The consolidation has a corresponding \emph{consolidation map} $\psi:X \rightarrow X\parcon \cQ$ where $\psi(u)=\psi(v)$ if and only if $u=v$ or $u$ and $v$ are in the same class of $\cQ$.

For a graph $G$ and a partition $\cQ$ of some subset $Y$ of $V(G)$, let $G_\cQ$ denote the \emph{consolidation} that one gets by
identifying the vertex set $C$ for each $C\in \cQ$, and deleting
loops and parallel edges. More precisely, $V(G\parcon \cQ) := V(G)\parcon \cQ$, and
\[E(G\parcon Q)=
\{\psi(C)\psi(C') \mid C,C'\in \cQ^1 \wedge \exists uv\in E(G): (u\in C \wedge v\in C')\},\]
where $\cQ^1$ is the partition of $V(H)$ that extends $\cQ$ into a partition of $V(G)$ by adding each vertex of $V(G)\setminus Y$ as a singleton, and $\psi$ is the consolidation map for $V(G)\parcon \cQ$. We will usually work with a partition $\cP$ of $V(G)$ itself where each class of $\cP$ induces a clique in $G$. We will then abbreviate $G\parcon \cP$ as $G_\cP$, and usually think of $\cP$ as the vertex set of $G_\cP$.

If $F$ is an edge set in $G$, then we denote by $G/F$ the \emph{contraction} of the edges of $G$, where we remove loops and parallel edges. If $\cQ_F$ denotes the partition of $G[F]$ into its connected components, then $G/F=G\parcon \cQ_F$. For a graph $G$ and a vertex set $X\subset V(G)$, let $G\vcon X$ denote the graph
where the edges induced by $X$ are contracted, i.e., $G/E(G[X])$. For a vertex set $A\subset G$, let $V(G\vcon X)|_A$ denote the vertices of $G\vcon X$ whose pre-image contains some vertex of $A$.  


\paragraph*{Treewidth and $\cP$-flattened treewidth.}

A tree decomposition of a graph $G$ is a tree graph $\cT$ whose nodes are subsets of $V(G)$ called \emph{bags}, with the following properties:
\begin{enumerate}[label=(\roman*)]
\item $\bigcup_{B\in V(\cT)} B = V(G)$, i.e., each vertex of $G$ is represented in some bag of $\cT$
\item For each edge $uv\in E(G)$ there is some $B\in V(\cT)$ with $u,v\in B$, that is, each edge of $G$ is induced by some bag
\item For each $v\in V(G)$ the bags $\{B\in V(\cT) \mid v\in B\}$ form a connected subtree of $\cT$.
\end{enumerate}

The width of the tree decomposition $\cT$ is $\max_{B\in V(\cT)} |B| -1$, and the \emph{treewidth} of $G$ is the minimum width among its tree decompositions, and it is denoted by $\tw(G)$.

Let $\omega:V(G)\rightarrow \Reals$ be a weight function. The weighted width of a tree decomposition $\cT$ of $G$ is $\max_{B\in V(\cT)} \omega(B)$. The weighted treewidth of $(G,\omega)$ is the minimum weighted width among its tree decompositions.

Suppose that $\cP$ is a partition of $V(G)$. The \emph{$\cP$-flattened treewidth}~\cite{BergBKMZ20} of $G$ is the weighted treewidth of $G_\cP$ under the weight function $\omega:\cP\rightarrow \Reals_{\geq 0}$ defined as $\omega(C):=\log(|C|+1)$ for each $C\in \cP$. We denote $\cP$-flattened treewidth of $G$ by $\tw_\cP(G)$.

Note that in an $n$-vertex graph $G$ and any partition $\cP$ of $V(G)$ we have $\tw_\cP(G)=O(\log n)\cdot \tw(G\parcon \cP)$.

\begin{observation}\label{obs:connected_treedecomp}
Let $\cT$ be a candidate tree decomposition of $G$ where all vertices
and edges of $G$ are represented in some bag of $\cT$. Then $\cT$ is a valid tree decomposition of $G$ if and only if for every
connected subgraph $H\subset G$ the set of bags containing some vertex of $H$
form a connected subtree of $T$.
\end{observation}

\begin{proof}
Clearly if any connected subgraph $H$ has this porperty, then in particular one-vertex edgeless subgraphs also have this property, which directly yields property (iii) in the definition. On the other hand, if $\cT$ is a valid tree decomposition, then let us use induction on $V(H)$ to prove that the vertices of $V(H)$ form a connected subtree. The claim holds for $|V(H)|=1$ by condition (iii) of tree decompositions. Suppose that $v\in V(H)$, and let $\cT_v$ be the subtree of $\cT$ where $v$ appears in the bag. By induction, there is also a subtree $\cT_{H-v}$ that contains the bags where some vertex of $H-v$ appears. Let $u\in N_H(v)$. By property (ii) of tree decompositions, there is some bag $B_{uv}$ containing both $u$ and $v$. Since $B_{uv}$ is a vertex of both $\cT_v$ and $\cT_{H-v}$, we get that the union of these trees is a connected subtree of $H$, concluding the proof.
\end{proof}

\paragraph*{Subset TSP and Steiner tree.}

We formally define our optimization problems as follows.

\begin{quote}
\stsp\\
\textbf{Input:} A graph $G$ and a vertex set $T\subseteq V(G)$ called \emph{terminals}.\\
\textbf{Output:} A shortest closed walk $W$ of $G$ where $T\subset V(W)$.
\end{quote}

\begin{quote}
\stein \\
\textbf{Input:} A graph $G$ and a vertex set $T\subseteq V(G)$ called \emph{terminals}.\\
\textbf{Output:} A minimum size vertex set $S \subseteq V(G)$ where $T\subseteq S$ and $G[S]$ is connected.
\end{quote}

Notice that our definition of \stein is non-standard, as one normally has an edge-weighted graph and one is looking for a subgraph of edge weight $k$ that contains all the terminals. In our unweighted setting these two optimization problems are equivalent from the perspective of approximation: indeed, a solution of size $|S|=k$ yields has a spanning tree of size $k-1$, and a Steiner tree with $k-1$ edges has $k$ vertices. Thus (up to changing $\eps$ by a constant factor) an approximate solution for one problem is yields an approximate solution for the other with a linear time overhead.

\paragraph*{Objects, regions, and curves}

For a closed planar set (object or region) $\re R\subset \Reals^2$ we denote by $\bd \re R$ its boundary. We denote by $\inter \re R$ and $\exter \re R$ the interior and exterior of $\re R$, respectively. For an arbitrary planar set $X$ we denote by $\Clo(X)$ its closure. A continuous function $\gamma:[0,1]\rightarrow \Reals^2$ is a \emph{curve} between its endpoints $\gamma(0)$ and $\gamma(1)$. We will often abuse notation and refer to the picture of $\gamma$ (i.e., $\gamma([0,1])$) as $\gamma$. A curve $\gamma$ is \emph{piecewise linear} if it is the concatenation of closed segments at endpoints shared between consecutive segments. A curve is a \emph{closed curve} when $\gamma(0)=\gamma(1)$. The curve is \emph{self-intersecting} if $\gamma(x)=\gamma(y)$ for some $x\neq y$ where $\{x,y\}\neq \{0,1\}$. A non-self-intersecting closed curve is a \emph{Jordan curve}.

\paragraph*{Polygons and intersection graphs}
 
A \emph{simple polygon} is a closed bounded region of $\Reals^2$ whose boundary is a piecewise linear closed Jordan curve. In particular, a segment is not considered to be a simple polygon. A non-self-intersecting polygon is a connected bounded region of $\Reals^2$ whose boundary consists of pairwise disjoint piecewise linear Jordan curves. If $P$ is a non-self-intersecting polygon, then the closures of the bounded components of $\Reals^2\setminus P$ are called the \emph{holes} of $P$; since the holes are bounded by piecewise linear Jordan curves, they are simple polygons. 
We observe that a component of the union of some non-self-intersecting polygons where no three polygon vertices are collinear is a non-self-intersecting polygon.

The \emph{complexity} of a non-self-intersecting polygon is its number of vertices, and the complexity of an intersection graph that is represented by non-self-intersecting polygons is the total number of vertices in the representing polygons. For a polygon $v$ and graph $G$ let $\compl(v)$ and $\compl(G)$ denote their complexity, respectively. For a set of polygons the \emph{arrangement} given by the polygons is the partition of the plane into a set of faces, edges and vertices given by the drawing of the polygon boundaries in the plane. The arrangement of an intersection graph $G$ of polygons is simply the arrangement of its polygons, and it is denoted by $\cA(G)$. The complexity of an arrangement is the complexity of its elements, i.e., the total number of vertices, edges and faces, and is denoted by $\compl(\cA(G))$. Notice that a set of $\compl(G)$ segments can create at most $\binom{\compl(G)}{2}$ intersections, so Euler's formula implies that $\compl(\cA(G))=O(\compl(G)^2)$. We use $V(\cA(G))$ to refer to the set of vertices in the arrangement, i.e., this set includes both the polygon vertices and the intersection points of their edges.

An \emph{intersection graph} $G$ is graph whose vertices are objects (sets) in some fixed metric space ---typically the plane---, and its edges are the unordered pairs of intersecting objects. Objects are typically denoted by lowercase letters, and subgraphs in $G$ are thought of as a collection of vertices and edges. The intersection graphs in this article are given via their representation, that is, by some straightforward description of the objects $V(G)$: in particular, unit disks are defined via the coordinates of their disk centers, and non-self-intersecting polygons are defined by listing the coordinates of the polygon vertices on each boundary component in clockwise order. 

We will also work with graphs that are represented as plane graphs, i.e., vertices are points (or sometimes identical to points), and its edges are associated with piecewise linear curves between the corresponding endpoints. For a graph $G$ with such a representation, we typically use the notation~$\pl G$. If the curves representing the edges are piecewise linear, then $\compl(\pl G)$ denotes the total number of internal vertices on its edge curves plus the number of vertices in $G$; similarly, if $\gamma$ is a piecewise linear curve, then $\compl(\gamma)$ is the number of vertices and edges in $\gamma$.

\paragraph*{Fair and $\alpha$-standard intersection graphs.}

\begin{definition}[Fair and $\alpha$-standard intersection graphs]
We say that an intersection graph $G$ of simple polygons in the plane is \emph{fair} if for any two polygons $v,w\in V(G)$ each component of the intersection $v\cap w$ is a simple polygon and no three polygon edges (from the polygons in $V(G)$) meet at one point.
If additionally each polygon $v\in V(G)$ contains a closed disk of radius $1$ and has diameter $\diam(v)\leq \alpha$ for some fixed $\alpha>1$, then $G$ is called \emph{$\alpha$-standard}.
\end{definition}

We denote by $\cI_\alpha$ the class of intersection graphs with an $\alpha$-standard embedding. Notice that a polygon of an $\alpha$-standard graph contains a closed axis-aligned unit square in its interior, since this already holds for the disk of radius $1$ contained in the polygon. Since an axis-aligned unit square contains some grid point from $\Ints^2$, we can associate each object with the lexicographically smallest\footnote{The \emph{lexicographic order} of $\Ints^2$ is a total ordering where we order first by $x$-coordinates and then by $y$-coordinates, that is, $(a,b)<(c,d)$ if and only if $a<c$ or $a=c$ and $b<d$.} point of $\Ints^2$  contained in the interior of the polygon.

\begin{definition}[Cell partition]
The \emph{cell partition} or \emph{clique partition} of a graph $G\in \cI_\alpha$ (given by its representation) is the partition $\cP$ of $V(G)$ where $u,v\in V(G)$ are in the same class of $\cP$ if and only if the lexicographically smallest grid point contained in the interior of $u$ and $v$ are the same grid point. 
\end{definition}

In particular, all objects assigned to the same class are stabbed by a single point, thus for any $C\in \cP$ we have that $G[C]$ is a clique. Occasionally, we will refer to the grid point associated with a class $C$ of $\cP$ as $\dot C$, and the set of points associated with the classes of $\cP$ as $\dot \cP$.

In Appendix~\ref{sec:graf_konvert} we show the following theorem about converting intersection graphs into $\alpha$-standard graphs.

\begin{restatable}{theorem}{Atalakitas}
\label{thm:atalakitas}
Let $G$ be an intersection graph of (a) $\delta$-similarly sized disks or (b) $\delta$-similarly sized $\beta$-fat polygons given by their representation. Then in polynomial time, we can construct an $\alpha$-standard intersection graph $G'$ that is isomorphic to $G$ whose representation complexity is polynomial in $n$, where $\alpha=36/(\beta\delta)$ and $\beta=1$ for disks.
\end{restatable}

\subsection{Object frames and wireframes}

A \emph{frame} can be thought of as a planar graph that relates to some underlying intersection graph but can use vertices of said graph multiple times. More precisely, we define two types of frames as follows.

\begin{definition}[Object Frame]\label{def:objectframe}
An \emph{object frame} of a fair intersection graph $G$ is an intersection graph $H$ that is planar together with a function $p:V(H)\rightarrow V(G)$ that satisfies the following properties.
\begin{enumerate}
\item Each vertex $v\in V(H)$ is a simple polygon, and has a parent polygon $p(v)\in V(G)$ such that $v \subseteq p(v)$.
\item For any $v,w\in V(H)$ they are either disjoint or $v\cap w$ is a segment of positive length.
\item The common intersection of any three distinct polygons $v_1,v_2,v_3\in V(H)$ is empty.
\end{enumerate}
\end{definition}

\begin{definition}[Wireframe]\label{def:wireframe}
A \emph{wireframe} of a fair intersection graph $G$ is a plane graph $\pl H$ and a function $p:V(\pl H)\rightarrow V(G)$ that satisfies the following properties.
\begin{enumerate}
\item Each vertex $\dot v\in V(\pl H)$ is a point located in its parent polygon $p(\dot v)\in V(G)$ that is not a vertex of any polygon $w\in V(G)$.\footnote{We explicitly allow for $\dot v$ to be on the boundary $\bd v$ of its parent as long as it is not a vertex of any polygon.}
\item For any $\dot u \dot v\in E(\pl H)$ the edge is represented by a polygonal curve $\gamma_{\dot u \dot v}:[0,1]\rightarrow \Reals^2$, and either $p(\dot u)=p(\dot v)$ and $\gamma_{\dot u \dot v}\subset \inter p(\dot u)$, or $p(\dot u)\neq p(\dot v)$ and for some $t\in (0,1)$ we have $\gamma_{\dot u \dot v}(0,t) \subset \inter p(\dot u)$ and $\gamma_{\dot u \dot v}(t,1) \subset \inter p(\dot v)$.
\item For each edge curve $\gamma$ of $E(H)$ and each $v\in V(G)$ the intersection $\bd v\cap \gamma$ is a finite point set (consequently, it does not contain any segments) and it contains no vertex of the arrangement~$\cA(G)$.
\end{enumerate}
\end{definition}

We note here that when constructing a wireframe $\pl H$ for some subgraph or object frame $H$ of $G$, it is convenient to disregard condition 3 of wireframes during the construction, as one can perturb the vertices of a wireframe (while staying in the same face of the arrangement $\bigcup_{v\in V(G)} \bd v$) so that the edges of the piecewise-linear curves in $E(\pl H)$ do not align with any of the edges of $\bd v$ (when $v\in V(G)\setminus V(H)$ or pass through any of the vertices of $\cA(G)$. This insight can be summarized as follows.

\begin{observation}[Perturbation principle]\label{obs:perturb}
Let $H$ be a subgraph or object frame of $G$, and let $\pl H$ be a wireframe of~$H$. Then in $O(\compl(G)+\compl(\pl H))$ time we can perturb the vertices of the edge curves of $\pl H$ to get a wireframe $\pl H'$ of \textit{both} $H$ and $G$ that is isomorphic to $\pl H$ and has the same complexity. Moreover, if the vertices of $\pl H$ are disjoint from the vertices of the polygons of $V(G)$, then $V(\pl H)=V(\pl H')$, i.e., it is sufficient to perturb some of the internal vertices of the edge curves of $\pl H$.
\end{observation}

We can convert between wireframe and object frame representations using the following lemmas.

\begin{lemma}\label{lem:obj2wire}
Let $G$ be an $n$-vertex fair intersection graph hosting an object frame $H$, and for each $v\in V(H)$ let $\dot v\in v$ be a given designated point that is not a polygon vertex for any polygon of $V(H)\cup V(G)$,  and not on the shared boundary of any pair of objects $v,w\in V(H)$. Then in polynomial time we can build a wireframe $\pl H$ of $G$ on the vertex set $\{\dot v \mid v\in V(H)\}$ isomorphic to $H$ where $\compl(\pl H)\leq n^2\cdot \compl(H)$.
\end{lemma}

\begin{proof}
For each edge $vw\in E(H)$ we wish to fix a simple polygonal curve $\gamma(\dot v \dot w)$ with at most $\compl(v)+\compl(w)$ vertices from $\dot v$ to $\dot w$ that is covered by $v\cup w$ and intersects the segment $(\bd v\cap \bd w)$ in exactly one internal point. The object frame properties ensure that the objects neighboring $v$ will intersect $v$ in exactly one segment somewhere along $\bd v$.

\begin{figure}[t]
\centering
\includegraphics{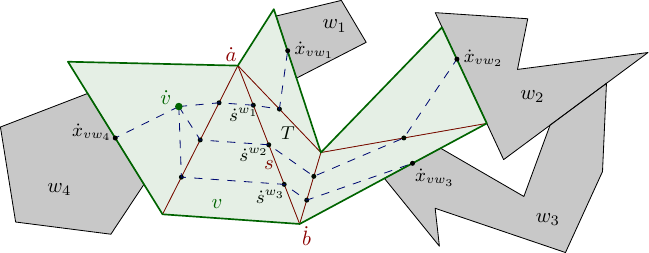}
\caption{The green polygon $v$ is triangulated using the red edges. We place points on the segment $s=\dot a \dot b$ according to the appearance of the neighbors of $v$ along $\bd v$. The construction creates the blue dashed half-edge-curves: these will form the wireframe edges incident to $\dot v$ together with their other halves in the neighboring polygons $w_i$.}\label{fig:obj2wire}
\end{figure}

Fix some triangulation of $v$, and consider each segment $s=\dot a \dot b$ of this triangulation that is not on the boundary of $v$, see \Cref{fig:obj2wire}. We will also assume that $\dot v \not \in s$ and handle this case near the end of the proof. Note that $s$ splits $v$ into two polygons; let $v_s$ denote the polygon among these that does not contain $\dot v$. Let $w_1,\dots,w_{k_s}$ be the sequence of neighbors of $v$ in order of their appearance on $\bd v_s\setminus s$ when traversing $\bd v_s\setminus s$ from $\dot a$ to $\dot b$.
We subdivide $s$ into $k_s+1$ equal length segments, and label the endpoints from $\dot a$ to $\dot b$ with $\dot a, \dot s^{w_1},\dots,\dot s^{w_{k_s}},\dot b$. Furthermore, for each $w\in N_H(v)$ we fix some point~$\dot x_{vw}$ that is an internal point of the segment~$v\cap w$. We can now define the first half of the edge curve of~$\dot v \dot w$ from $\dot v$ to $\dot x_{vw}$ as the piecewise linear curve given by the vertex sequence $\dot v,\dot s^w_1,\dot s^w_2,\dots, \dot x_{vw}$, where $s_1,s_2,\dots$ are the unique sequence of triangulation edges separating $v$ from the triangle containing~$v\cap w$.  (Note that when $v$ is a triangle, then this curve consists of a single segment from $\dot v$ directly to~$\dot x_{vw}$)
Notice that there are at most $\compl(v)-3$ internal edges in the triangulation, thus the half-edge described above has at most $\compl(v)-1$ vertices; together with the other half of the edge in $\dot w$, this gives at most $\compl(v)+\compl(w)-3$ vertices (as $\dot x_{vw}$ appears on both half-edges).

It remains to show that this results in a plane graph, i.e., that there are no intersections between edges. Suppose the contrary, there is an intersection point. Notice that no segment of an edge curve will align with a triangulation edge, and the edge curves use different vertices. We can also notice that each half-edge-curve consists of segments in distinct triangles, thus no self-intersection occurs on these curves. Thus an intersection can only happen between two edges connecting distinct boundary points of some triangle $T=\Delta(\dot a \dot b \dot c)$ of the triangulation of $v$. Let $s=\dot a \dot b$ be the edge of $T$ that separates $\dot v$ from the the interior of $T$. Let $w_1,w_2$ be the neighbors of $v$ whose half-edge-curves intersect in $T$, and assume without loss of generality that $w_1$ appears before $w_2$ on $\bd v_s\setminus s$ when traversing $\bd v_s\setminus s$ from $\dot a$ to $\dot b$.
If both $v\cap w_1$ and $v\cap w_2$ are separated from $\dot v$ by $s':=\dot a\dot c$, then the segments $\dot s^{w_1}(\dot s')^{w_1}$ and $\dot s^{w_2}(\dot s')^{w_2}$ cannot intersect as $w_1\cap v$ and $w_2\cap v$ will appear in the same order on $\bd v_{s'}\setminus s'$ as it did on $\bd v_s\setminus s$, so these points appear in the same order on the segments $\dot a \dot b$ and $\dot a \dot c$. The analogous argument works if both $v\cap w_1$ and $v\cap w_2$ are separated from $\dot v$ by $s'':=\dot c\dot b$. Finally, due to the order of appearance of $v\cap w_1$ and $v\cap w_2$ on $\bd v$ the only remaining option is that $v\cap w_1$ is separated from $\dot v$ by $s'=\dot a \dot c$ and $v\cap w_2$ is separated from $\dot v$ by $s''=\dot c \dot b$. Again this means that $\dot s^{w_1}(\dot s')^{w_1}$ and $\dot s^{w_2}(\dot s'')^{w_2}$ are disjoint.

Recall that the above proof assumed that $\dot v$ does not fall on any internal edge of the triangulation of $v$. If this happens, then we can simply remove the edge of the triangulation containing $\dot v$ and proceed with the proof unchanged: if there is some $w\in N_H(v)$ that intersects $\bd v$ on one of the triangles incident to $\dot v$, then $\dot v$ will be directly connected to $\dot x_{vw}$ inside that triangle, and such half-edge-curves cannot introduce intersections with other edges starting from $\dot v$.

The resulting plane graph $\pl H$ is a wireframe of $H$. We apply the perturbations of \Cref{obs:perturb} to ensure that the result also satisfies condition 3 of wireframes with respect to $G$. This results in a wireframe of complexity at most $\compl(\pl H)< E(H)\cdot \compl(H)< n^2\cdot \compl(H)$.
\end{proof}

\begin{lemma}
Let $G$ be an $n$-vertex fair intersection graph hosting a wireframe $\pl H$. Then in polynomial time we can construct an object frame $H$ isomorphic to $\pl H$ of complexity $\compl(H)\leq n^2\cdot \compl(\pl H)$. 
\end{lemma}

\begin{proof}
Consider cutting the curve of each edge $\gamma(\dot v \dot w)$ of $\pl H$ at some point $\dot x_{\dot v \dot w}$ such that the portion of the edge curve from $\dot v$ to $\dot x_{\dot v \dot w}$ is covered by $p(\dot v)$ and the rest is covered by $p(\dot w)$; let us call these parts half-edge curves. The star at $\dot v$ is defined as the union of the half-edge curves incident to $\dot v$. Since $\deg(\dot v)\leq n-1$ and $\dot v \in \inter p(\dot v)$, we can construct a polygon $u_{\dot v}$ of complexity $n\cdot \compl(\pl H)$ for $\dot v$ in some small $\eps$-neighborhood of the star of $\dot v$ such that $u_{\dot v} \subset p(\dot v)$ and for each of the points $\{\dot x_{\dot v \dot w}\mid \dot w \in N_{\pl H}(\dot v)\}$ there is a side $s_{\dot v \dot w}\subset p(\dot v)\cap p(\dot w)$ of $u_{\dot v}$ which contains $\dot x_{vw}$ as an inner point. It is straightforward to check that if $\eps$ is small enough, then these polygons are either disjoint or intersect in segments $s_{\dot v \dot w}$, thus they form an object frame of $G$ that is isomorphic to $\pl H$. The resulting object frame $H$ has complexity at most $\compl(H)\leq n^2\cdot \compl(\pl H)$.
\end{proof}

Finally, we will often restrict our intersection graphs to various regions as follows.

\begin{definition}[$\re R$-restriction]
Let $G$ be a fair interection graph, and let $\re R$ be a simple polygon such that for any $v\in V(G)$ each component of the intersection $v\cap R$ is a simple polygon. Then the \emph{$\re R$-restriction} of $G$ is the intersection graph induced by the simple polygons $\bigcup_{v\in V(G)} \ccomp(v\cap \re R)$, where $\ccomp(X)$ is the set of connected components of $X$. We denote the $\re R$-restriction of $G$ by $G|_{\re R}$. For a vertex $v \in G|_{\re R}$ we denote by $\hat v$ the vertex of $G$ such that $v \in \ccomp(\hat v \cap R)$.
\end{definition}

For a wireframe $\pl H$ the \emph{drawing} of $\pl H$, denoted by $\gamma(\pl H)$, is the union of all curves in the fixed plane representation of $\pl H$.

If $\re R$ is a region that is bounded by a wireframe cycle $\wf R$ of $G$, then there is a natural restriction of the parent relation $p(.)$ to $\re R$: for $\dot v \in \wf R$ we define $p|_{\re R}(\dot v)$ to be the unique object $v \in V(G|_{\re R})$ such that $\dot v \in v$ and $\hat v = p(\dot v)$. We say that a wireframe $\wf W$ of $G$ is a \emph{boundaried wireframe} of $R=G|_{\re R}$ if its representation is in the closed region $\re R$ and it has no vertices on $\bd \re R$ except the vertices of $\wf R$.

\section{Constructing a connectivity-preserving object frame}
\label{sec:connectivity}

We start by defining an intersection graph that is a mix between a fair intersection graph and an object frame.
An intersection graph is \emph{fair-frame} if its objects are simple polygons, and for any pair $u,v$ of objects, either (\textit{i}) each component of $u\cap v$ is a simple polygon or (\textit{ii}) $u\cap v$ is a segment of positive length. 

In order to prove our next theorem, we need the following technical lemma, which can be considered as a conversion step from a fair graph to an object frame.

\begin{lemma}\label{lem:polygonsubtract}
Let $G$ be a fair-frame graph with a fixed vertex $v$, and consider the connected components of $G-v$. Let $U_1,\dots,U_k$ denote the unions of the objects in these connected components, and assume that $k\geq 2$. Then $v$ has a collection $v_1,\dots,v_t$ of subpolygons for some $t\leq k-1$ such that $V(G-v)\cup\{v_1,\dots,v_t\}$ induces a connected fair-frame graph where all $v_i$ are pairwise disjoint, and for each $v_i$ and $u\in V(G)-v$ the intersection $v_i\cap u$ is either empty or a segment of positive length. Finally, the new objects have total complexity $\sum_{i=1}^t \compl(v_i) \leq 3\compl(v\setminus (\bigcup_{j=1}^k U_k))$.
\end{lemma}

\begin{proof}
Consider the drawing of the boundaries of $\bd v$ and $\bd U_i$ for all $i$. These boundaries slice the plane into several connected regions, see \Cref{fig:subpolyslicer}. Notice moreover that in the fair-frame graph it is possible that $\bd U_i\subset \bd v$; in this case we think of the shared segment as a distinct degenerate region. We label each region with the set of objects among $v,U_1,\dots,U_k$ that cover the region. Thus there are at most $2k+2$ distinct labels. These regions then define a planar graph $H$ whose vertices correspond to regions and edges correspond to region pairs whose shared boundary contains a segment of positive length. Since $v$ is connected, we have that the planar graph $H$ restricted to the regions inside $v$ is also connected; let $H_v$ denote this graph. Notice that since each $U_i$ intersects $v$, the graph $H_v$ must contain a vertex of label $\{v,U_i\}$ for each $U_i$. Next, for each $U_i$, we identify the set of vertices in $H_v$ that have the label $\{v,U_i\}$ into a single vertex, and denote by $u_i$ the newly created vertex. We remove any multiple edges created by the operation. (We note that loops are not created by the identification as regions of equal label are not adjacent.) Consider some spanning tree of the resulting graph, and remove all leaves of label $\{v\}$, resulting in the tree~$T$.
  
\begin{figure}[t]
\centering
\includegraphics[width=\textwidth]{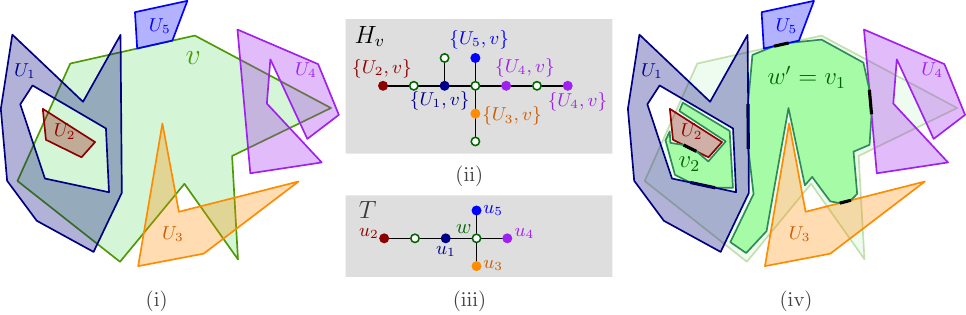}
\caption{(i) The setup of \Cref{lem:polygonsubtract}: the connected components of $G-v$ cover the regions $U_i$. (\textit{ii}) The graph $H_v$ labelled with the regions covering each vertex. Vertices of label ${v}$ are denoted by empty green circle nodes. (\textit{iii}) The tree $T$ obtained by identifying labels $\{U_i,v\}$ for each $i$ and removing leaves of label $\{v\}$. (\textit{iv}) The subpolygons $v_1,v_2$ replacing $v$ and keeping connectivity.}\label{fig:subpolyslicer}
\end{figure}

We claim that $T$ has at most $k-1$ vertices that are labeled with $\{v\}$. Indeed, $T$ has at most $k$ leaves (namely, each vertex $u_i$ may be a leaf), and a vertex of label $\{v\}$ is adjacent only to new vertices, while new vertices are only adjacent to vertices of label $\{v\}$, as the regions $U_i$ are pairwise disjoint. Consequently, $T$ is bipartite where one part consists of the vertices of label $\{v\}$, and each of these vertices have degree at least two, while the other part consists of the $k$ vertices $u_i$ for $i=1,\dots,k$. This means that if there are $t$ vertices with label $\{v\}$, then the edge count is $t+k-1\geq 2t$, which implies $t\leq k-1$, as claimed.

Let $w$ be a vertex of $T$ with label $\{v\}$, and for each adjacent vertex $u_j$ we fix a positive length segment $s_j$ that is on the intersection of the region of $w$ and (one of) the region(s) corresponding to $u_j$. (More precisely, we choose $s_j$ to be a segment inside $v$ that is part of the boundary of a single object in $U_j$.) Notice that the region of $w$, denoted by $\re R(w)$ is a connected component in $(v\setminus (\bigcup_{j=1}^k U_k))$. Observe that $\re R(w)$ contains a subpolygon $w'$ that has the middle third of each of these $s_j$ segments on its boundary, and it is otherwise disjoint from the boundary of the region of $w$. We will select $w'$ so that it follows the boundary of $\re R(w)$ in the interior of $\re R(w)$, while its complexity is at most three times the complexity of the region of $w$. We select at most one vertex for $w'$ for each vertex of $\re R(w)$ and at most two vertices for $w'$ for each edge of $\re R(w)$, as illustrated in \Cref{fig:subpolyslicer}. It is routine to check that the subpolygons $v_1,\dots,v_t$ defined this way for each of the $t$ vertices of $T$ labeled with $\{v\}$ are pairwise disjoint, their intersections with all vertices of $G-v$ satisfy object frame property 2, and the resulting intersection graph is fair-frame. Moreover, we have $\sum_{i=1}^t \compl(v_i) \leq 3\compl(v\setminus (\bigcup_{j=1}^k U_k))$.
\end{proof}

\begin{theorem}\label{thm:tracespanningtree}
Let $G$ be a fair intersection graph, and let $X\subset V(G)$ be a set of terminals. Moreover, each terminal $x\in X$ has a designated point $\dot x\in x$ in the plane such that the designated points of distinct terminals do not coincide. Then there exists an object frame $G'$ of $G$ where for each $x\in X$ there is a corresponding terminal vertex $x'$ in $V(G')$ (where $p(x')=x$), and the following hold:
\begin{enumerate}
\item For each $x\in X$ we have $\dot x \in x'$.
\item Terminals $u,v\in X$ are in the same connected component of $G$ if and only if $u',v'$ are in the same connected component of $G'$.
\item $|V(G')| - |V(G)|\leq 16|X|$.
\item The object frame $G'$ is a forest.
\item $\compl(G')=O(\compl(\cA(G)))$.
\end{enumerate}
Finally, when $|X|\leq 2$, then $|V(G')|\leq |V(G)|$.
\end{theorem}

\begin{proof}
If $|X|=1$, then the object frame consisting of the single terminal
$x\in X$ satisfies all the conditions. For the rest of the proof, assume $|X|\geq 2$.
  
We will define a procedure that creates the desired graph $G'$ using several modifications to $G$; for convenience, we refer to the current graph as $G$, and assume without loss of generality that $G$ is a connected fair-frame graph, as we can do the following procedure on each of its connected components. (In turn, we will show that the resulting object frame $G'$ is a tree.)

It will be convenient to assume in the proof that $X$ consists of pairwise disjoint objects. This can be achieved the following way. If $|X|=2$ and the two objects $x,y\in X$ intersect, then the object frame satisfying the statement of the theorem can be found the following way. If $\dot x\in y$, then we set $x'$ to be a triangle in a small neighborhood of $\dot x$ (so that the triangle is in $x$ and it does not contain $\dot y$), and we set $y'$ to be a subpolygon of $y$ that contains $\dot y$ and shares some boundary segment with this triangle. Otherwise, let $y'=y$ and let  $x'$ be a subpolygon of $x$ that contains $\dot x$, internally disjoint from $y$, and shares a boundary segment with $y$. (This is possible as $x\cap y$ cannot be a  single point in a fair-frame graph.) In both cases, we get the required object frame.


If $|X|\ge 3$ and some terminals intersect, then we modify $X$ the following way. Let us place a new infinitesimally small polygon $x'$ in each polygon $x$, so that the created polygons are pairwise disjoint, and $x'$ covers $\dot x$.
We redefine $X$ to be the set of these new polygons, and consider them the new terminals; note that this increases the vertex count of $G$ by $|X|$ and retains the fairness of the representation, as well as the connectivity between terminals.

We define the rank of a vertex $v\in V(G)\setminus X$ (during any time in this procedure) as the number of connected components of $G-v$ that contain some vertex from $X$; we denote this number by $\rank(v)$. Furthermore, we fix some root vertex $r\in X$.
We then apply the following modifications exhaustively:
\begin{description}
\item[M1] If for some $v\in V(G)\setminus X$ the graph $G-v$ has a component that does not contain any vertex of $X$, then delete all vertices in this component.
\item[M2] If $\rank(v) = 1$ for some $v\in V(G)\setminus X$, then we remove $v$.
\item[M3] If $\rank(v)\geq 2$ for some $v\in V(G)\setminus X$ and some vertex $w\in V(G) - v$ intersects $\inter v$, then let $U_1,\dots,U_k$ be the components of $G-v$, where $U_1$ is the component containing the root $r$. By  \Cref{lem:polygonsubtract}, there exists a collection of at most $k-1$ subpolygons
  $v'_1,\dots,v'_{t}\quad (t\leq k-1)$ of $v$, which are pairwise disjoint, do not intersect the interior of any object of $G-v$, and these polygons
 can be used to connect the components $U_1,\dots,U_k$. We then replace $v$ with $v'_1,\dots,v'_{t}$ in $G$ to retain connectivity and the fair-frame property.
\item[M4] For any point $p\in \Reals^2$ that is on the boundary of least three polygons $v\in V(G)$, we remove a small triangle or concave quadrilateral around $p$ from all but one polygon that contains~it; note that this does not change the realized intersection graph.
\end{description}

After exhausting modifications M1, M2, and M3, the realizing polygons are interior-disjoint. After M4, the resulting intersection graph is an object frame of the original graph. Moreover, each polygon in $V(G)\setminus X$ is a cut vertex, as otherwise it could be removed by M2.

We now continue the procedure with the following modifications.
\begin{description}
\item[M5] If there is a cycle in the graph, then let $u$ and $v$ be two neighboring vertices along the cycle. We remove a small neighborhood of their shared boundary $u\cap v$ from $v$. (Note that if the segment $u\cap v$ contained a terminal point $\dot x$, then $\dot x$ is still covered by $u$.)
\end{description}
Note that after applying M5 exhaustively we get a tree. Finally, we do our final modifications:
\begin{description}
\item[M6] If there are polygons $v'_i,v'_j \in V(G) \setminus X$ that are subpolygons of the same original object $v$ such that the unique shortest path between them $u_0=v'_i,u_1,u_2,\dots, u_t=v'_j$ has the property that $u_1,u_2\in V(G)\setminus X$ and $\rank(u_1)=\rank(u_2)=2$, then we do the following operation. Let $C_{u_0}$ and $C_{u_3}$ be the two components (as sets of polygons) that we get after removing $u_1$ and $u_2$ from the graph. Let $U_0= \bigcup C_{u_0}$ and $U_3= \bigcup C_{u_3}$ denote the unions of these components, and apply \Cref{lem:polygonsubtract} on $v$ and the components $U_0$ and $U_3$. We get that $v$ has a subpolygon $v'$ that connects  $U_0$ and $U_3$. We then remove $u_1$ and $u_2$ and add $v'$. Notice that $v'$ may be incident to $v_i'$ or $v_j'$ (or even both); in such a case we take $v'$ and all incident subpolygons of $v$ and replace them with their union (as a single subpolygon of $v$).
\end{description}

\begin{figure}[t]
\centering
\includegraphics[width=\textwidth]{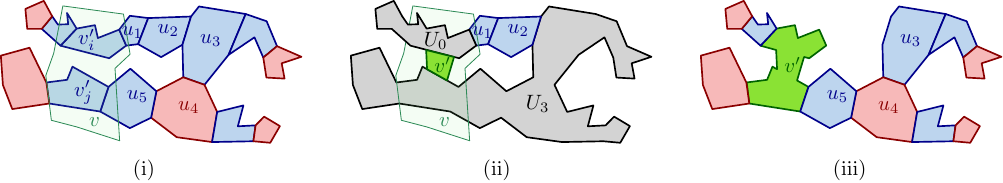}
\caption{Modification M6, where polygons of $X$ are red and other polygons of the tree $G$ are blue. (i) The setup of M6: the path between subpolygons $v'_i$ and $v'_j$ of $v$ passes through $u_1$ and $u_2$ that are of rank $2$ (thus not branching vertices) and not polygons of $X$. (\textit{ii}) The union-objects $U_0$ and $U_3$, with a subpolygon $v'$ of $v$ connecting them. (\textit{iii}) The final polygon $v'$ after incorporating the incident $v'_i$ and $v'_j$. The connectivity is preserved after $u_1$ and $u_2$ are deleted.}\label{fig:subpolyjump}
\end{figure}

Let $G'$ denote the resulting object frame.
Consider the planar subdivision given by the union of the original polygon boundaries. Modifications M1 and M2 do not increase the complexity. Each polygon created by the modifications M3 and M6 is a subpolygon of some connected collection of regions of this subdivision, in particular, each created polygon has at most three times as many vertices as the union of the corresponding set of regions by the complexity bound of \Cref{lem:polygonsubtract}. Modification M4 may double the number of sides of any given polygon, and M5 can be executed without increasing the complexity of the polygons in question. As a result, we have that $\compl(G')\leq \compl(G)+2\cdot 3\compl(\cA(G))=O(\compl(\cA(G))$.

Recall that after applying M5 exhaustively we get a tree. Each leaf of this tree must be a terminal from $X$. Thus, the tree has at most $|X|-2$ vertices of degree at least $3$, and there are at most $2|X|-2$ edge-disjoint paths whose internal vertices have degree $2$ and whose endpoints are leaves or branching in the tree. Consequently, there are at most $8|X|-8+2|X|-2<10|X|$ vertices that are within distance $2$ from a branching point. Furthermore, if $x\in X$ is not a branching point, then there are at most $5$ vertices that are within distance $2$ from $x$ and not within distance two of a branching, so in total there are at most $15|X|$ vertices in $G'$ that are within distance $2$ of a vertex of $X$ or a branching. Thus, after applying M6 exhaustively, each of the $n$ original objects has at most one subpolygon that is at distance at least $3$ from all branching vertices and all vertices from $X$. Consequently, there are at most $n+|X|+15|X|$ polygons. Thus the resulting object frame is a tree with at most $n+16|X|$ vertices, as required.

Finally, notice that when $|X|=2$, then either the terminals intersect and $|G'|=2$, or the the two terminals are disjoint. In this latter case, we did not introduce any new terminals, and by definition the rank of any object is at most $2$. In particular, M2 can replace a vertex $v$ of rank $k=2$ with $k-1=1$ new object, so the number of objects is not increased by the modifications.
\end{proof}

If we apply the above theorem on a shortest path (with the endpoints being the terminals) then we get an isomorphic representation.

\begin{corollary}\label{cor:shortestpathframe}
Let $P$ be a shortest path from $u$ to $v$ in $G$, where $\dot u\in u$ and $\dot v\in v$. Then there exists a wireframe $\pl P$ in $G$ of complexity $\compl(\pl P)=O(\compl(\cA(P)))$ inducing a path from $\dot u$ to $\dot v$ where the parent relation $p(.)$ gives an isomorphism to $P$.
\end{corollary}

\begin{proof}
We consider the intersection graph given by the objects along $P$ and set $X=\{u,v\}$ with designated points $\dot u$ and $\dot v$. Then \Cref{thm:tracespanningtree} gives an object frame of complexity $O(\cA(P))$ where $|V(G')|\leq |V(G_P)|$ where the end objects are connected with an object frame. Since all edges of the object frame must correspond to edges of $P$, the only way for this connection to occur is if there is a bijection between $E(P)$ and $E(G')$, which in turn implies that the graphs are isomorphic under the parent relation. This isomorphism is preserved when we switch to a wireframe representation using \Cref{lem:obj2wire}. Finally, we apply \Cref{obs:perturb} to make this wireframe of $P$ into a wireframe of $G$ without increasing its complexity.
\end{proof}

We will use the above corollary several times to define wireframe paths, but we will usually need a finer control over the complexity of the resulting paths. This is provided by the following lemma for a more specific setting. Essentially, we want to say that when we move from a region $\re R$ to a smaller region $\re R_2$ by cutting $\re R$ into two parts $\re R_1,\re R_2$ with a shortest path $P$, then the increase of the complexity of the boundary of $\re R_2$ compared to $\re R$ can be bounded by the complexity inside $\re R_1$.

\begin{lemma}\label{lem:pathcutcomplexity}
Let $\wf R$ be a wireframe of $G\in \cI_\alpha$ bounding the region $\re R$, let $\dot a,\dot b\in V(\wf R)$ and $P$ be a shortest path of $G|_{\re R}$ from $p|_{\re R}(\dot a)$ to $p|_{\re R}(\dot b)$. Let $V(\bd \re R)$ denote the set of vertices in $\cA(G|_{\re R})$ that appear on $\bd \re R$, and let $V(\inter \re R)$ be the set of vertices of $\cA(G|_{\re R})$ in $\inter \re R$. 

If $\dot a$ and $\dot b$ splits $\bd \re R$ into wireframe paths $\pl Q_1$ and $\pl Q_2$, then there exists a wireframe path~$\pl P$ from~$\dot a$ to~$\dot b$ isomorphic to~$P$ under the parent function $p|_{\re R}(.)$ such that 
\begin{itemize}
\item  $\pl P$ intersects $\bd \re R$ only at $\dot a$ and $\dot b$, and
\item the regions $\re R_1,\re R_2 \subset \re R$ bounded by $\pl Q_1 \cup \pl P$ and $\pl Q_2 \cup \pl P$, respectively, satisfy $|V(\bd R_2)|\leq |V(\bd R)|+8|V(\inter \re R_1)|$.
\end{itemize} 
\end{lemma}

\begin{proof}
We follow the proof of \Cref{cor:shortestpathframe} more carefully to establish a better complexity bound on $\pl P$ and the boundaries of the regions $\re R_i$.

Consider the intersection graph given by the objects along $P$ and set $X=\{p|_{\re R}(\dot a),p|_{\re R}(\dot b)\}$ with designated points $\dot a$ and $\dot b$. Then \Cref{thm:tracespanningtree} gives an object frame of complexity $O(\cA(P))$ and $|V(G')|\leq |V(G_P)|$, where the end objects are connected with an object frame, converted to a wireframe using \Cref{lem:obj2wire}, where the parent relation gives an isomorphism to $P$. We apply \Cref{obs:perturb} to make this wireframe of $P$ into a wireframe of $G$ without increasing its complexity. Let $F_1,\dots,F_k$ be the sequence of faces of the arrangement $\cA(P)$ traversed by $\pl P$ such that $\dot a \in \bd F_1$ and $\dot b\in \bd F_k$.

Let $\dot p_i$ be a point of $\bd F_i\cap \bd F_{i+1}$ that is not a vertex of $\cA(G|_{\re R})$ for each $i=1,\dots,k-1$, and set $\dot p_0:=\dot a$ and $\dot p_k:=\dot b$. Notice that we have freedom in choosing a path $\pl P$ connecting the points $\dot p_0,\dots,\dot p_k$ in this sequence, as long as the relative interior of $\gamma(\pl P)[\dot p_{i-1}, \dot p_i]\subset \inter F_i$, we preserve the isomorphism and the wireframe properties. We may also assume without loss of generality that each face occurs at most once in the sequence $F_i$, because we can directly traverse from the first entry into the face to the last exit of the face within $F_i$. See Figure~\ref{fig:pathslicer} for an illustration.

\begin{figure}[t]
\centering
\includegraphics{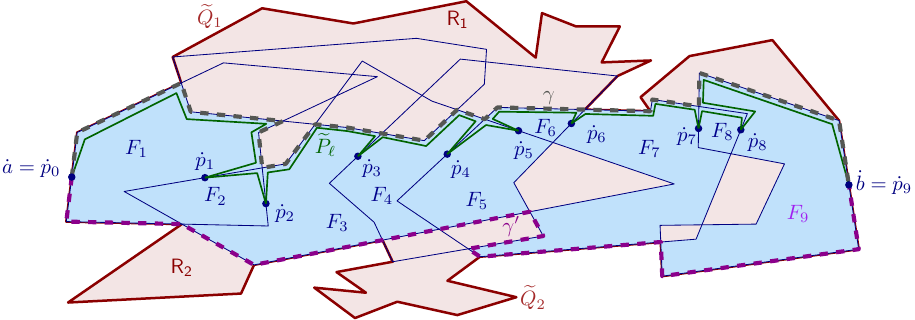}
\caption{Finding a path $\pl P$ (green) through the faces $F_1,\dots,F_k$ (blue background) of an arrangement. The path $\pl P$ slices the shaded (red and blue) region $\re R$ into $\re R_1$ and $\re R_2$. The curve $\pl P$ is defined to follow the curve $\gamma$ (gray, dashed). The curve $\gamma'$ is denoted by a magenta, dashed curve.}\label{fig:pathslicer}
\end{figure}

Let $\bd F$ denote the cycle of $\bd \bigcup_{i=1}^k F_i$ that contains $\dot p_0$ and $\dot p_k$: such a cycle exists as $\dot p_0=\dot a$ and $\dot p_k=\dot b$ are on the closed curve $\bd R$ that is disjoint from the connected set $\inter \bigcup_{i=1}^k F_i$. Notice that $\dot p_0$ and $\dot p_k$ splits the closed curve $\bd F$ into two subpaths: $\gamma$ and $\gamma'$.

We claim now that $\dot p_{i-1}$ and $\dot p_i$ are on the same component of $\bd F_i$. To see why, assume the contrary. Since $\pl P$ passes through $F_i$ only once, it follows that some boundary component $\mu$ of $\bd F_i$ which is disjoint from $\bd \re R$ is crossed by $\pl P$ exactly once, and thus $\mu$ separates $\dot a$ and $\dot b$. This would imply however that the closed curves $\mu$ and $\gamma(\pl P)\cup \gamma(\pl Q_1)$ intersect each other exactly once, which is a contradiction.

The component of $\bd F_i$ containing $\dot p_{i-1}$ and $\dot p_i$ is split into two parts by $\dot p_{i-1}$ and $\dot p_i$. Let $\gamma_i$ and $\gamma'_i$ denote these parts. Observe that $\gamma_i$ cannot intersect both $\gamma$ and $\gamma'$ (and the same holds for $\gamma'_i$) as in the union of the faces $\bigcup_{i=1}^k F_i$ the curve $\gamma(\pl P)$ separates $\gamma$ and $\gamma'$, and it also separates $\gamma_i$ and $\gamma'_i$. We can therefore define $\gamma_i$ and $\gamma'_i$ in such a way that $\gamma_i\cap \gamma'=\emptyset$ and $\gamma'_i\cap \gamma=\emptyset$.

Assume without loss of generality that $\gamma$ is inside $\re R_1$. We define $\pl P$ to follow the curve $\gamma_i$ inside each $\inter F_i$  from $\dot p_{i-1}$ to $\dot p_i$. As a result, the complexity of $\pl P$ is $\sum_{i=1}^k \compl(\gamma_i)$. Notice that $\pl P$ intersects $\bd \re R$ only at $\dot a$ and $\dot b$ by definition.

Consider now the plane graph $\pl A_{\re R}$ given by the vertices and edges of $\cA(G|_{\re R})$, that is $V(\pl A_{\re R})=V(\bd \re R)\cup V(\inter \re R)$. Since $G$ is $\alpha$-standard, we have that no three polygon edges of $G$ meet at one point, thus for any $v\in V(\inter \re R)$ we have $\deg_{\pl A_{\re R}}(v)\leq 4$.

Charge each edge of $\pl P$ to the corresponding edge of $\pl A_{\re R}$. Notice that each edge of $\pl A_{\re R}$ can be followed by at most two edges of $\pl P$ (one edge in each adjacent face of $\pl A_{\re R}$), except edges that are on $\bd \re R$ as such edges can be followed only once. Indeed, for an edge of $\pl A_{\re R}$ that is in $\bd \re R$ only one of the incident faces of $\pl A_{\re R}$ will lie in $\re R$, while $\pl P$ lies in $\re R$ and consequently can follow along only on the side of the edge covered by $\re R$.
More precisely, each edge of $\pl A_{\re R}$ that is incident to a vertex inside $\inter \re R_1$ is charged at most twice, while each edge of $\pl A_{\re R}$ that is on $\gamma(\pl Q_1)$ is charged at most once; no other edges of $\pl A_{\re R}$ are charged. Denoting the former edge set by $E_{\inter}$ and the latter by $E_{\bd}$, we can write:
\begin{align*}
\compl(\pl P) &\leq  |E_{\bd}| + 2|E_{\inter}|\\
&\leq \compl(\pl Q_1) + 2\sum_{v\in V(\pl A_{\re R})\cap \inter \re R_1} \deg_{\pl A_{\re R}}(v)\\
&\leq \compl(\pl Q_1) + 8|V(\inter \re R_1)|
\end{align*}
by our earlier degree bound. Adding $\compl(\pl Q_2)$ to both sides concludes the proof.
\end{proof}

\section{Lipschitz embeddings into plane graphs}
\label{sec:lipschitz}

Let $G,H$ be graphs. We say that a mapping $f:V(G)\rightarrow V(H)$ is
\emph{$c$-Lipschitz} for some constant $c\geq 1$ if for any pair of vertices $x,y\in V(G)$
we have that
\begin{equation}\label{eq:Lipsch}
\frac{1}{c}\cdot \dist_G(x,y)\leq \dist_H(f(x),f(y))\leq c\cdot \dist_G(x,y).
\end{equation}
Clearly any $c$-Lipschitz mapping is injective. 

\begin{observation}\label{obs:GtoGP}
The natural mapping $f$ from $G\in \cI_\alpha$ to $G_\cP$ satisfies
\begin{equation*}
\frac{1}{3}\dist_G(x,y)\leq \dist_{G_\cP}(f(x),f(y))\leq \dist_G(x,y)
\end{equation*}
for all $x,y$ with $f(x)\neq f(y)$.
\end{observation}

\begin{proof}
The upper bound of $\dist_{G_\cP}(f(x),f(y))\leq \dist_G(x,y)$ follows from the fact that $G_\cP$ is obtained from $G$ by edge contractions, which cannot increase distances.

For any edge $f(u)f(v)\in E(G_\cP)$ there is a path of length at most $3$
connecting $u$ and $v$ in $G$, and therefore $\dist_{G_\cP}(f(u),f
(v))\leq 3 \dist_G(u,v)$ holds. The same inequality follows for any $u,v$ where $f(u)\neq f(v)$ by applying the above inequality on all edges of the shortest $f(u)-f(v)$ path in $G_\cP$.
\end{proof}

The goal of this section is to show that $G_\cP$ can be Lipschitz-mapped to a
planar graph on the same vertex set as $G_\cP$, and then to use this to show
that $G$ itself can be Lipschitz-mapped to a planar graph with the same
vertex set as $G$. For a clique $C\in \cP$, let $\dot C$ denote the stabbing
point of $C$, and let $\dot\cP$ be the set of centers of
non-empty cells. Recall that we can also think of $G_\cP$ as an intersection graph, where the object corresponding to a clique $C\in \cP$ is the union of the objects in $C$. The core theorem of this section is the following.

\begin{theorem}\label{thm:getplanarclique}
For any $\alpha>0$, there exists a positive constant $c=O(\alpha^{32})$ such that
for any $G\in \cI_\alpha$ given by its representation the following holds. Let
$\cP$ be the cell partition of $G$. Then in polynomial time one can construct
a plane graph $H=(\dot\cP,E')$ such that natural mapping from $\cP$ to
$\dot\cP$ is $c$-Lipschitz from $G_\cP$ to~$H$. Moreover, $H$ has maximum degree $O(\alpha^{34})$, its edges have geometric diameter at most $O(\alpha^{17})$ and each edge of $H$ can be covered by at most $O(\alpha^{34})$ objects of $G_\cP$.
\end{theorem}

Before we begin the proof, note that the theorem implies the analogous
statement for $G$ itself.

\begin{corollary}\label{cor:getplanargraph}
For any $\alpha>0$ there exists a positive constant $c=O(\alpha^{32})$ such that
for any $G\in \cI_\alpha$ given by its representation, one can in polynomial
time construct a plane graph $G'=(V(G),E')$ such that the identity mapping of
$V(G)$ is $c$-Lipschitz from $G$ to $G'$.
\end{corollary}

\begin{proof}
Let $H$ be the plane graph on $\dot\cP$ created by
\Cref{thm:getplanarclique}, and to each vertex $\dot C$ of $H$ we attach
$|C|-1$ leaves, and assign the original vertex and the leaves arbitrarily to
distinct vertices of $C$. It is routine to check that with the resulting planar graph $G'$ the above map is $3c$-Lipschitz, where $c$ is the Lipschitz-constant guaranteed by
\Cref{thm:getplanarclique}.
\end{proof}

The following lemma will be used much later to help us get rid of unwanted
vertices using a contraction, while maintaining a Lipschitz map.

\begin{lemma}\label{lem:contractingextraverts}
Let $G$ be a graph  with vertex set $V$, and let $H$ be a graph with vertex
set $W$ such that $V\subseteq W$. Suppose moreover that the identity map
$f:V\rightarrow V$ is $c$-Lipschitz from $G$ to $H$. Then in polynomial time
(in terms of $|W|$) one can construct a set of edges $S$ in $H$ such that the
vertices of the contraction $H/S$ correspond to the vertices of $V$ and this
map is $c^2$-Lipschitz from $G$ to $H/S$.%

Moreover, if $H$ is a plane graph whose edges are piece-wise linear curves of geometric diameter at most $\delta$, then for any $\mu>0$ the graph $H/S$ can be realized as a plane graph whose edges have geometric diameter at most $\delta \cdot (c+2)$, and each edge of $H/S$ is contained in the geometric $\mu$-neighborhood of some path of $H$.
\end{lemma}

\begin{proof}
Let us show first that we can assume that each vertex $w\in W$ has distance at
most $c/2$ from some vertex in $V$.
We claim that the identity map remains
$c$-Lipschitz from $G$ to $H$ if we remove every vertex of $H$ that is at distance more than $c/2$ from every vertex in $V$. Indeed, each edge of $G$ can be realized with a
path of length at most $c$ in $H$ connecting two vertices of $V$, and no vertex on such a path was removed. Consequently, any shortest path of length $k$ in $G$
can be covered by the union of at most $k$ such paths of length $c$. On the
other hand, deleting vertices from $H$ can only increase pairwise distances,
so the inequality $\frac{1}{c}\dist_G(u,v)\leq \dist_H(u,v)$ remains true
after the deletion. Therefore, in the following we can assume that each vertex $w\in W$ has distance at
most $c/2$ from some vertex in $V$.

Let $V=\{v_1,\dots,v_n\}$, and define the following vertex sets in $H$. The
\emph{Voronoi cell} $\Vor(v_i)$ of a vertex $v_i\in V$ contains the vertices
 $w\in W$ that are closest to $v$ in $H$-distance among all the vertices of
 $V$ (breaking ties by the order of $i$). More precisely, we define
\begin{align*}
\Vor(v)=\big\{w\in W \mid &\phantom{\text{ and }}\forall j<i: \dist_H(v_i,w)<\dist_H(v_j,w)\\
&\text{ and } \forall j>i: \dist_H(v_i,w)\leq \dist_H(v_j,w)\big\}.
\end{align*}
 
We claim that $H[\Vor(v_i)]$ is connected. Suppose for the sake of contradiction that $H[\Vor(v_i)]$ is disconnected, and let $u\in \Vor(v_i)$ be in a different connected component of $H[\Vor(v_i)]$ than $v_i$. Let $P_u$ be a shortest path from $v_i$ to $u$, and let $zz'$ be an edge of $P_u$ where $z\not\in \Vor(v_i)$ and $z'\in \Vor(v_i)$ (where $z$ is closer to $u$ than $z'$). Now $z\in \Vor(v_j)$; suppose that $j<i$. Then by the definition of Voronoi diagrams we have (a) $\dist_H(z,v_i)\geq \dist_H(z,v_j)$ and (b) $\dist_H(z',v_i) < \dist_H(z',v_j)$. On the other hand, we have $\dist_H(z',v_i)=\dist_H(z,v_i) + 1$ and the edge $zz'$ ensures $\dist_H(z',v_j) \leq \dist_H(z,v_j)+1$; substituting these to (b) gives a contradiction with (a). The case when $j>i$ can be argued the same way, one just needs a strict inequality in (a) and a non-strict inequality in (b). This concludes the proof that $H[\Vor(v_i)]$ is connected.
 
Notice that since each vertex in $W$ has distance at most $c/2$ from some
 vertex in $V$, we have that $\Vor(v)$ has diameter at most $c$. Let $E_i$
 denote the set of edges of the subgraph of $H$ induced by $\Vor(v_i)$, and
 let $S=\bigcup_{i=1}^n E_i$.

Next, we contract all edges of each Voronoi cell in $H$ and remove any
multiple edges to get the graph $H'=H/S$, where we identify the vertex that
results from the contraction of $E_i$ with~$v_i$. 
Since contractions can only decrease the distances between vertices, we have
$\dist_{H'}(v_i,v_j)\leq \dist_{H}(v_i,v_j)
\leq c\cdot\dist_{G}(v_i,v_j)$. On the other hand, if $\rho$ is a shortest path from
 $v_i$ to $v_j$ in $H$ of length $\ell$, then it intersects at least
 $\ell/c$ distinct Voronoi cells (as each cell has diameter at most~$c$).
 Consequently, $\dist_{H'}(v_i,v_j)\geq \frac{1}{c}\dist_{H}(v_i,v_j)
\geq\frac{1}{c^2}\dist_{G}(v_i,v_j)$.

\noindent\textit{Planar realizations.}
Suppose now that $H$ is a plane graph whose edges have geometric diameter at most $\delta$, and $\mu>0$. We assume without loss of generality that $\mu$ is small enough, i.e., we assume that $\mu<\delta/2$, and in the plane realization of $H$ all vertices have pairwise distance at least $3\mu$, and all non-incident vertex-edge pairs have pairwise distance at least $3\mu$. Our first goal is to construct pairwise non-crossing short paths $P_{ij}$ for each edge $v_iv_j$ of $H/S$.

For each edge $v_iv_j$ of $H/S$ with $i<j$ let $P^0_{ij}$ be a path of length at most $c+1$ in $H$ constructed as follows: we fix some edge $w_iw_j$ where $w_i\in \Vor(v_i)$ and $w_j\in \Vor(v_j)$. Then we fix a shortest path from $v_i$ to $w_i$ and from $w_j$ to $w_j$. The concatenation of these two paths and the edge $w_iw_j$ forms a path $P^0_{ij}$ of length at most $c/2+1+c/2=c+1$.

Next, we \emph{uncross} the paths $P^0_{ij}$. Consider two distinct paths $P^0_{ij}$ and $P^0_{k\ell}$. Observe that when the paths do not share endpoints, then their vertices are covered by pairwise disjoint voronoi cells, thus they are vertex-disjoint. Suppose now that $P^0_{ij}$ and $P^0_{k\ell}$ share one endpoint, and assume without loss of generality that $i=k$. (Since the paths are distinct, they cannot share both their endpoints, as they would correspond to the same edge $v_iv_j$ of $H/S$.) It follows that $P^0_{ij}$ and $P^0_{i\ell}$ can only share vertices inside $\Vor(v_i)$. Let $w$ be their shared vertex that is at the largest $H$-distance from $v_i$. We then modify $P^0_{i\ell}$ by replacing $P^0_{i\ell}[v_i,w]$ with $P^0_{ij}[v_i,w]$. We do these modifications exhaustively, until any pair of paths are either vertex disjoint or they share some prefix/suffix; let $P_{ij}$ denote the collection of paths at the end of the procedure. \skb{esetleg abra?}

It is now routine to construct the desired planar realization of $H/S$. For each vertex $v_i$ and each neighbor $v_j$ with $v_iv_j\in E(H/S)$ consider the paths $P_{ij}$ from $v_i$. These paths form a plane tree in $H$ rooted at $v_i$ whose leaves are $N_{H/S}(v_i)$, i.e., the other endpoints of the paths $P_{ij}$. We create a new plane tree consisting of the root-to-leaf paths to remove all overlaps but keep all paths within distance $\mu$ of the original root-to-leaf paths.\skb{abra TODO} Since the overlapping parts are restricted to $\Vor(v_i)$, these modifications can be done in each Voronoi cell separately. The resulting collection of plane paths $P'_{ij}$ are pairwise disjoint other than possible shared endpoints, they are within distance $\mu$ of the original path $P_{ij}$ in $H$ of length at most $c+1$, and their geometric diameter is at most $(c+1)\delta+2\mu\leq (c+2)\delta$. Thus their curves form a plane realization of $H/S$ with the desired properties.
\end{proof}

\subsection{Approximate plane graph with extra vertices}

In order to prove \Cref{thm:getplanarclique}, we need to start with an
intermediate step, which resembles known contraction decomposition theorems for unit disks~\cite{UDGcontractdecompSODA19,TrueContractDecomp}: initially, we also allow additional vertices in the plane graph,
and later we will remove the extra vertices so that only the clique centers
$\dot\cP$ remain.

Let us fix a graph $G\in \cI_\alpha$ and its representation, as well as a
cell-partition $\cP$ based on $\Ints^2$. Each edge $CC'$ of $G_\cP$ can be associated with some
edge $uu'$ of $G$ where $u\in C$ and $u'\in C'$. If $\pl P$ is a single-edge wireframe of $G$ with endpoints $\dot C$ and $\dot C'$, where $p_{\gamma}(\dot C)\in C$ and $p_{\gamma}(\dot C')\in C'$, then we say that the curve $\gamma=\gamma(\pl P)$ is an \emph{edge curve} of $G$ with respect to the partition $\cP$, which corresponds to the edge $CC'$ of $G_\cP$. An edge curve can be constructed by applying \Cref{cor:shortestpathframe} to the single-edge shortest path $uu'$ of $G$ and the fixed endpoints $\dot C\in \inter u$ and $\dot C'\in \inter u'$. We construct for each edge of $G_\cP$ a corresponding edge curve, and denote by $\Gamma$ the resulting set of edge curves. We denote by $\cI(\Gamma)$ the intersection
graph induced by $\Gamma$. Observe that the curves are drawn in an open region, thus ---similarly to \Cref{obs:perturb}--- they can be defined in such way to satisfy the following general position requirements: (i) the
curves of $\Gamma$ avoid $\dot\cP$, that is, no curve passes through any
point $a\in \dot\cP$ other than at its endpoints, (\textit{ii}) distinct edge curves intersect each other in $O(\cA(G))$ points, and (\textit{iii}) no three edge curves meet at any point of $\Reals^2 \setminus \dot\cP$.

Since each object of $G$ has diameter at most $\alpha$, and an edge curve can be covered by the union of two objects that intersect, we have the following observations.
\begin{observation}\label{obs:diameterofgamma}
The edge curves of\, $\Gamma$ satisfy the following properties.
\begin{enumerate}[label=\textit{(\roman*)}]
\item Each point $\dot C$ is incident to at most $O(\alpha^2)$ edge curves.
\item Each edge curve $\gamma\in \Gamma$ has geometric diameter at most $2\alpha$.
\item Each edge curve intersects at most $O(\alpha^4)$ other edge curves.
\end{enumerate}
\end{observation}

We can now prove the following.

\begin{theorem}\label{thm:getplanarextraverts}
For any $\alpha>0$ there exists a positive constant $c=c_{\alpha}$ such that
for any $G\in \cI_\alpha$ given by its representation the following hold. Let
$\cP$ be the cell partition of $G$. Then in polynomial time one can construct
a plane graph $H$ with $\dot\cP\subseteq V(H)$ such that the natural mapping from $\cP$ to $\dot\cP$
is $c$-Lipschitz from $G_\cP$ to~$H$ for some $c=O(\alpha^{16})$. Moreover, (\textit{i}) $H$ has maximum degree at most $O(\alpha^2)$, (\textit{ii}) for each clique $C\in \cP$ the union $\bigcup_{v\in C} v$ covers at most $O(\alpha^8)$ vertices of $H$, and (\textit{iii}) each edge of $H$ is a section of an edge curve of $G$ and thus each edge of $H$ has geometric diameter at most~$2\alpha$.
\end{theorem}

\begin{proof}
Consider the curve set $\Gamma$ in the plane, and add each
intersection of pairs of curves as new vertices, resulting in a plane graph, where each vertex is either a point $\dot C\in \dot \cP$ of degree at most $O(\alpha^2)$ (by \Cref{obs:diameterofgamma}(\textit{i})) or it is an intersection thus a vertex of degree $4$. The graph $H$ will be defined as a subgraph of this graph, so it will automatically satisfy conditions (i) and (\textit{iii}). Unfortunately, each curve of $\Gamma$ could intersect any other curve of
$\Gamma$ in $O(\cA(G))$ many points, therefore we have to decrease the number
of intersections involved. We partition $\Gamma$ into classes so that edge curves that are in the same class are pairwise distant in~$\cI(\Gamma)$, and we will add the classes of curves one by one in our construction.

\begin{claim}\label{cl:partitioncurves}
For any positive integer $d$ there is a partition of $\Gamma$ into
$O(d^4\alpha^4)$ sets such that if $\gamma,\gamma'$ are in the same partition
class then $\dist_{\cI(\Gamma)}(\gamma,\gamma')\geq d$ and
$\min_{x\in \gamma, y\in \gamma'}\dist(x,y)\geq d\alpha$. 
\end{claim}

\begin{claimproof}
\Cref{obs:diameterofgamma} implies that if $\dist(\gamma(0),\gamma'
(0))>2\alpha \cdot(d+1)$, then $\dist_{\cI(\Gamma)}(\gamma,\gamma')\geq d$. It also follows that any point of $\gamma$ and any point of $\gamma'$ has distance more than $d\alpha$. Set
$t=\lceil 2(d+1)\alpha\rceil$, and let us now define the partition class
$\Gamma_{i,j,k,\ell}$ for $i,j\in [t]\times[t]$, as follows:
\begin{align*}
\Gamma_{i,j,k,l}= \big\{ \gamma\in \Gamma \mid
&  (\gamma(0))_x \equiv i \mod t,\\
& (\gamma(0))_y \equiv j \mod t \\
& (\gamma(1))_x \equiv k \mod t,\\
& (\gamma(1))_y \equiv \ell \mod t \big\}
\end{align*}
This is a partition of size $t^2=O(d^4\alpha^4)$ that satisfies
the desired property.
\end{claimproof}

Let $\Gamma_1,\Gamma_2,\dots$ be a partition of
$\Gamma$ according to \Cref{cl:partitioncurves} into $O (\alpha^4)$ 
classes so that curves in the same class have distance at least
$d=3$ in $\cI(\Gamma)$. Our strategy is to add these classes one by one, reducing
the number of newly created intersections after each step. On the planar graph $H$ that we build,
each edge will be labeled by the curve $\gamma\in \Gamma$ that the edge came from, and we also
maintain a set $\Pi$ of paths: for each previously introduced curve $\gamma$ with endpoints $\gamma(0)$ to $\gamma(1)$ we ensure that $\Pi$ contains a path $\pi(\gamma)$ connecting these endpoints in the current graph.
\medskip

In each phase $i$ we will add the curves of $\Gamma_i$ one by one to the existing planar graph, updating the graph $H$ with new vertices and edges, as well as updating the path collection $\Pi$. 
When adding a new curve $\gamma\in \Gamma_i$, we add the endpoints $\gamma(0)$ and $\gamma(1)$ as vertices if they are not yet $H$-vertices, and we add $\gamma$ as the realization of the edge between $\gamma(0)$ and $\gamma(1)$. We also add the single-edge path $\pi(\gamma)$ realized by this edge to $\Pi$. Next, we need to address the intersections of $\gamma$ with the existing graph $H$ to maintain planarity. Suppose that $\gamma$ intersects some existing edge of $H$, and label each intersection on $\gamma$ with the curve of $\Gamma_1\cup\dots\cup\Gamma_{i-1}$ that is intersected there. If there is a label that repeats, then we use the following \emph{slicing} procedure recursively on $\gamma$.

Let $\gamma(a)$ and $\gamma(b)$ be a pair of
maximally distant intersections (along $\gamma$) of the same label $\gamma'$, that is, choose $\gamma'$ so that for any label $\hat\gamma\in \Gamma_1\cup \dots \cup \Gamma_{i-1}$ that appears on $\gamma$ at least twice at $\gamma(\hat a)$ and $\gamma(\hat b)$,
it holds that $[a,b]\setminus [\hat a,\hat b] \neq \emptyset$. We add $\gamma(a)$ and $\gamma(b)$ as new vertices, and remove the edge represented by $\gamma|_{(a,b)}$. We update the part of $\pi(\gamma)$ between $\gamma(a)$ and $\gamma(b)$ to detour along $\pi(\gamma')$ and apply the slicing procedure recursively on the curves\footnote{Technically, we have $\gamma^1,\gamma^2:[0,1]\rightarrow \Reals^2$ defined as
$\gamma^1(x)=\gamma(ax)$ and $\gamma^2(x)=\gamma(b+
(1-b)x)$.} $\gamma^1=\gamma|_{[0,a]}$ and $\gamma^2=\gamma|_{[b,1]}$. By the end of the slicing procedure, the remaining subcurves of $\gamma$ will be intersected with distinct labels; for each such intersection we add a new vertex as described above.

Let $H$ denote the final graph.
We note here that at the end of the procedure, for any given pair of edge curves $\gamma,\gamma'$ there are at most two intersections among $\gamma\cap \gamma'$ that appear as vertices in $H$. Thus property (\textit{ii}) follows: all objects of a given clique can be covered by a square of area $O(\alpha^2)$ that intersects $O(\alpha^4)$ edge curves, so there can be at most $O(\alpha^8)$ pairs of edge curve pairs intersecting in this region, each contributing at most $2$ vertices to $H$.

Let $\Pi$ denote the final paths in $H$, and let $\pi(\gamma)$ be the path of $\Pi$ corresponding to $\gamma$. We denote by $H_i$ the graph and by $\Pi_i$ the set of paths after all curves of $\Gamma_1\cup\dots\cup\Gamma_i$ have been introduced.
For a curve $\gamma\in \Gamma_1\cup \dots \cup \Gamma_i$, let $\pi_i
(\gamma)$ denote the corresponding path in $\Pi_i$. We denote the embedded curve of $\pi_i(\gamma)$ by $\tilde\pi_i(\gamma)$.

\begin{claim}\label{cl:slicedetours}
For any $\gamma\in \Gamma_i$ the path $\pi_i(\gamma)$ has at most $i+1$ edges that are represented by subcurves of $\gamma$.
\end{claim}

\begin{claimproof}
The curve $\gamma$ can intersect at most one curve from
each set $\Gamma_1,\dots,\Gamma_i$, as otherwise
there would be two curves in some $\Gamma_j$ whose graph distance in $\cI(\Gamma)$ is at most $2$. Therefore, $\gamma$ will have at most $i$
disjoint maximally distant intersection pairs in the slicing procedure. Consequently, $\pi_i(\gamma)$ makes at most $i$ detours from $\gamma$, thus $\pi_i(\gamma)$ retains at most $i+1$ subcurves of $\gamma$.
\end{claimproof}

\begin{claim}\label{cl:curvedist}
The maximum (Euclidean) distance of a curve point $x\in \gamma$ and a curve point
$y\in \tilde\pi_i(\gamma)$ is $O(i\alpha)$.
\end{claim}

\begin{claimproof}
We use induction on $i$. Clearly the statement holds for $i=1$ since for any
$\gamma\in \Gamma_1$ we have $\gamma=\tilde\pi_1(\gamma)$, and thus the maximum
distance between them is $\diam(\gamma)\leq 2\alpha$ by \Cref{obs:diameterofgamma}(ii).
Suppose now that $\max\{\dist(x,y)\mid x\in \gamma, y\in \tilde\pi_i
(\gamma)\}\leq 4i\alpha$.

Notice that if $\gamma\in \Gamma_1 \cup \dots \cup \Gamma_i$, then the
curves $\tilde\pi_i(\gamma)$ and $\tilde\pi_{i+1}(\gamma)$ are identical, so it is
sufficient to prove that for any $\gamma\in \Gamma_{i+1}$ it holds that
$\max\{\dist_{\Reals^2}(x,y)\mid x\in \gamma, y\in \pl\pi_{i+1}(\gamma)\}\leq 2
(i+1)\alpha$. Each
detour that $\pi_{i+1}(\gamma)$ takes from $\gamma$ will be along a portion
of some curve $\pl\pi_i(\gamma')$ for some
$\gamma'\in \Gamma_1\cup \dots \cup \Gamma_i$. By induction, we have that
$\pl\pi_i(\gamma')$ has maximum distance $2i\alpha$ from $\gamma'$, and
$\gamma'$ intersects $\gamma$, therefore the maximum distance from $\gamma$
is at most $4i\alpha+\diam(\gamma)+\diam(\gamma')\leq 4i\alpha+4\alpha
=4(i+1)\alpha$, which concludes the proof.
\end{claimproof}

We occasionally refer to the curves corresponding to the paths in
$\Pi$ by $\tilde\Pi$. Since there are $O(\alpha^4)$ phases, \Cref
{cl:curvedist} implies that the geometric diameter of each path
$\tilde\pi \in \tilde\Pi$ is $O(\alpha^4\cdot \alpha)=O(\alpha^5)$. Moreover, we have the following.

\begin{claim}\label{cl:pilength}
Each curve $\tilde\pi\in \tilde\Pi$ has diameter $O(\alpha^5)$, and the corresponding path $\pi(\gamma)$ in $H$ has at most $O(\alpha^{16})$ edges.
\end{claim}

\begin{claimproof}
Consider a path $\pi \in \Pi$ in $H$ and its edge labels. Recall that by
definition any pair of distinct curves in $\Gamma_i$
have minimum geometric distance at least $3\alpha$. Since the geometric diameter of $\tilde\pi$
is $O(\alpha^5)$, it may  intersect other curves within a region of size $O(\alpha^5)\times O(\alpha^5)$, which contains at most $O(\alpha^{10})$ endpoints. Each endpoint has $O(\alpha^2)$ edge curves ending there, so altogether at most $O(\alpha^{12})$ distinct labels can occur in
$\pi$.

We now bound the number of occurrences for a fixed label $\gamma\in \Gamma_i$ along $\pi$.
By \Cref{cl:slicedetours} we have that $\gamma$ will create at most $i+1=O(\alpha^4)$ subcurves
(i.e. edges labeled with $\gamma$) after its slicing in phase $i$. The later phases $j>i$ can however add further subdividing vertices on these edges. In each new phase $j>i$, the union of the label-$\gamma$ subcurves of $\pi$ can receive at most two new subdividing vertices, as at most one curve among $\Gamma_j$
can intersect $\gamma$ from each later phase $j>i$, and after its slicing it will intersect $\gamma$ at most twice.

Thus any fixed curve label can
occur on at most $O(\alpha^4)$ edges of $\pi$. Consequently, there are $O
(\alpha^{12})$ distinct edge labels on $\pi$, each occurring $O(\alpha^4)$ times, so the length of $\pi$ in $H$ is $O(\alpha^{16})$. 
\end{claimproof}

It remains to show the $c$-Lipschitz property. Since each edge of $G_\cP$ is
represented by a curve $\gamma\in \Gamma$ that has a corresponding path
$\pi\in\Pi$ in $H$, our bound on the paths in $\Pi$ imply that $\dist_{H}
(\dot{u},\dot{v})\leq O(\alpha^{16})\dist_{G_\cP}(u,v)$. For the other
inequality, observe that each edge of $H$ can be covered by at most two
objects, and therefore any path of length $k$  in $H$ corresponds to a path in $G$ of length at most $2k$, which directly maps
to a path of length at most $2k$ in $G_\cP$. Consequently, 
\begin{equation}\label{eq:Hlower}
\dist_{H}(\dot{u},\dot{v})\geq \frac12\dist_{G_\cP}(u,v),
\end{equation}
which concludes the proof of \Cref{thm:getplanarextraverts}.
\end{proof}

We are now ready to finish the proof of \Cref{thm:getplanarclique}.

\begin{proof}[Proof of \Cref{thm:getplanarclique}]
Let $H_0$ be the planar graph with extra vertices given by \Cref{thm:getplanarextraverts}. We set $\mu$ small enough so that for any edge curve $\gamma$ of $H$ covered by a minimal collection $T$ of objects it holds that the geometric $\mu$-neighborhood of $\gamma$ is also covered by $T$. We apply \Cref{lem:contractingextraverts} to $G_\cP$, $H_0$ and $\mu$. This yields an edge set $S$ and the plane graph $H_0/S$ so that $H_0/S$ and $G_\cP$ are $O(\alpha^{16})$-Lipschitz, each edge of $H_0/S$ is covered by objects from $G_\cP$, and each edge of $H_0/S$ has at most $O(\alpha^{16})$ times the geometric diameter of the maximum edge diameter of $H_0$, which ---as an edge curve section--- had geometric diameter at most $O(\alpha)$. Thus $H:=H_0/S$ is a plane graph that can be created in
polynomial time and has the desired Lipschitz property with the Lipschitz-constant $\big(O(\alpha^{16})\big)^2=O(\alpha^{32})$, its edges have geometric diameter at most $O(\alpha^{17})$ and they are covered by objects of $G_\cP$. Since there are $O(\alpha^{34})$ objects in $G_\cP$ within distance $O(\alpha^{17})$ of any given point in the plane, we have that $O(\alpha^{34})$ objects of $G_\cP$ are sufficient to cover any edge of $H$.

To prove the degree bound on $H$, recall that $V(H)=\dot\cP$ is a subset of the integer grid, so we have that the number of points within distance $O(\alpha^{17})$ from a fixed vertex $\dot C\in V(H)$ is at most $O(\alpha^{34})$, thus the maximum degree of $H$ is at most $O(\alpha^{34})$, which concludes the proof.
\end{proof}

\subsection{Constructing a wireframe from a Lipschitz embedding}

By modifying the graph obtained in \Cref{thm:getplanarextraverts} we get the following theorem.

\begin{theorem}\label{thm:Lipschitz_wireframe}
For any $\alpha>0$ there exists a positive constant $c=c(\alpha)$ such that
for any $G\in \cI_\alpha$ given by its representation in polynomial time one can construct a wireframe $\pl G$ such that $p(V(\pl G))=V(G)$ and for any $\dot u, \dot v \in V(\pl G)$ where $p(\dot u)\neq p(\dot v)$ we have
\[\frac{1}{c} \dist_G\big(p(\dot u), p(\dot v)\big) \leq \dist_{\pl G}(\dot u, \dot v)\leq c\cdot\dist_G\big(p(\dot u), p(\dot v)\big).\]
Moreover, for each $v\in V(G)$ we have $|p^{-1}(v)|\leq c$, and the maximum degree of $\pl G$ is at most $c+\Delta(G)$ where $\Delta(G)$ is the maximum degree of $G$.
\end{theorem}

\begin{proof}
Let $\cP$ be the cell partition of $G$, and apply~\Cref{thm:getplanarextraverts} to obtain a plane graph $H$. We will use the drawing of $H$ to construct $\pl G$. For each $C\in \cP$ fix a vertex $v_C\in C$, and set $p(\dot C)=v_C$. We note that the vertices $v_C$ are distinct for different cliques $C$. Consider now a vertex $\dot a \in V(H)\setminus \dot \cP$, which is an intersection of the edge curves $\gamma$ and $\gamma'$. Since $\dot a \in \gamma$, we have that it must be contained either in the interior of $p_{\gamma}(\gamma(0))$ or in the interior of $p_{\gamma}(\gamma(1))$. Similarly, it is in the interior of at least one among $p_{\gamma'}(\gamma'(0))$ and $p_{\gamma'}(\gamma'(1))$. We set $p(\dot a)$ to be one of the objects $u \in \{p_{\gamma}(\gamma(0)), p_{\gamma}(\gamma(1)), p_{\gamma'}(\gamma'(0)), p_{\gamma'}(\gamma'(1))\}$ where $\dot a \in \inter u$.

Doing this assignment for all vertices of $H$ results in a planar graph that satisfies the first property of wireframes. We will need further modifications to satisfy the second property. Consider an edge $e=\dot x\dot y\in E(H)$, and suppose that the edge is represented as a subcurve of the edge curve $\gamma$ where $p_{\gamma}(\gamma(0))=a$ and $p_{\gamma}(\gamma(1))=b$. Let $\gamma^{e}$ denote the restriction of $\gamma$ to the part between $\dot x$ and $\dot y$, re-parameterized so that $\gamma^e(0)=\dot x$ and $\gamma^e(1)=\dot y$. The definition of $\gamma$ implies that $\gamma^e[0,1]$ is either entirely contained in the interior of some object $w\in \{a,b\}$, or there is some $t\in (0,1)$ such that $\gamma^e[0,t)\subset \inter a$ and $\gamma^e(t,1]\subset \inter b$. If $\gamma^e$ is contained
in $\inter w$ where $w\in \{a,b\}$, then we place a new vertex at its middle whose parent is $w$. In the latter case, let $\dot z\in \gamma^e$ be a point such that $\dot z \in \inter a \cap \inter b$. Then we place two new points $\dot z_a$ and $\dot z_b$ in a small neighborhood of $\dot z$ so that they are both in $\dot z \in \inter a \cap \inter b$, the initial part of $\gamma_a$ ending at $\dot z_1$ is covered by $\inter a$, and the final part of $\gamma^e$ from $\dot z_2$ is covered by $\inter b$. We set their parents as $p(\dot z_1) = a$ and $p(\dot z_2) = b$. Notice that we need to subdivide each edge of $H$ at most twice.

Next, to ensure that $p(V(\pl G))=V(G)$, consider some vertex $u\in V(G)$ that is not a parent of any vertex so far, and let $C$ be the clique of $u$. Then we add a new point $\dot u$ with $p(\dot u)=u$ to our planar graph in a small neighborhood of $\dot C$ (and inside $\inter u$) and connect $\dot u$ to $\dot C$.

At the end of this procedure, the mapping $p(.)$ defined above gives a valid wireframe of $G$ that we denote by $\pl G$. We note here the following about the above construction.
\begin{quote}
\begin{description}
\item[($\star$)] If $\gamma^e$ is an edge of $H$ that is part of the edge curve $\gamma$, then $\gamma^e$ contains at least one vertex of $V(\pl G)$ whose parent is either $p_\gamma(\gamma(0))$ or $p_\gamma(\gamma(1))$.
\end{description}
\end{quote}

We will prove the degree bound and the bound on $p^{-1}(v)$ next. Recall that by~\Cref{thm:getplanarextraverts} the maximum degree of $H$ is $O(1)$, and our subdivisions introduced only degree-$2$ vertices. We furthermore added up to $|C|-1$ neighbors to some of the vertices $\dot C$; it follows that the maximum degree of $\pl G$ is at most $O(1)+\max_{C\in \cP}|C|\leq O(1)+\Delta(G)$. \Cref{thm:getplanarextraverts} created a graph where the objects in any clique of $G$ contain $O(1)$ vertices of $H$ in total, and $H$ has degree $O(1)$. We have subdivided each edge of $H$ at most twice, thus there are $O(1)$ vertices of $\pl G$ intersected by any given clique (and thus by any given object) of $G$. Since an object $v\in V(G)$ can only be assigned as parent to vertices of $\pl G$ that fall  inside $v$, we have that $|p^{-1}(v)|=O(1)$.

It remains to show the distance property. Let $\dot x, \dot y\in V(\pl G)$ with $x=p(\dot x)$ and $y=p(\dot y)$, where $x\neq y$. First, notice that since $\pl G$ is a wireframe in $G$, we have that each edge of $\pl G$ corresponds to an edge (between the parents) in $G$; in particular, a shortest path in $\pl G$ can be mapped to a walk of equal length in $G$ connecting $x$ and $y$. Therefore
\[\dist_G(x,y)\leq  \dist_{\pl G}(\dot x, \dot y).\]

Next, we need to upper bound $\dist_{\pl G}(\dot x, \dot y)$. Notice that it is sufficient to prove that $\dist_{\pl G}(\dot x, \dot y)$ is bounded by a constant whenever $x$ and $y$ are adjacent in $G$; the general bound then follows from applying this bound on all edges of a shortest path.
Let $\gamma_u^e$ be the curve of an edge of $H$ containing (or incident to) $\dot x$, and let $\gamma_u$ be the edge curve of $G$ with respect to $\cP$ that $\gamma^e_u$ is a subcurve of. By ($\star$) we have that $\gamma^e_u$ contains a vertex $\dot u$ such that $p(\dot u)\in \{p_{\gamma_u}(\gamma_u(0)),p_{\gamma_u}(\gamma_u(1))\}$. Let $u:=p(\dot u)$ and let $C_u\in \cP$ such that $u\in C_u$ (and thus $\dot C_u\in \{\gamma_u(0),\gamma_u(1)\}$). We define $\dot v$, $v$ and $C_v$ symmetrically based on $\dot y$.

By the triangle inequality, we have
\begin{equation}\label{eq:plG_triangle}
\begin{split}
\dist_{\pl G}(\dot x, \dot y)
&\leq \dist_{\pl G}(\dot x, \dot u) + \dist_{\pl G}(\dot u, \dot C_u)\\
&\phantom{\leq} + \dist_{\pl G}(\dot C_u, \dot C_v)\\
&\phantom{\leq} + \dist_{\pl G}(\dot C_v, \dot v) + \dist_{\pl G}(\dot v, \dot y).
\end{split}
\end{equation}

Recall from the proof of \Cref{thm:getplanarextraverts} that $\gamma^e_u$ is a part of some curve $\tilde \pi_u\in \tilde \Pi$ with endpoints $\gamma_u(0)$ and $\gamma_u(1)$. The corresponding path $\pi_u$ of $H$ has length at most $c_\Pi=O(\alpha^{16})$ by \Cref{cl:pilength}. Since $\pl G$ was defined by subdividng each edge of $H$ at most twice, all distances are at most tripled; consequently the path in $\pl G$ corresponding to the curve $\tilde \pi_u$ has length at most $3c_\Pi$. In particular, we get that $\dist_{\pl G}(\dot u,\dot C_u)\leq 3c_\Pi$. The analogous argument gives $\dist_{\pl G}(\dot v,\dot C_v)\leq 3c_\Pi$. Since $\gamma^e_u$ contains both $\dot x$ and $\dot u$, and $\gamma^e_u$ was subdivided at most twice in $\pl G$, we also have $\dist_{\pl G}(\dot x,\dot u)\leq 3$. Applying the same arguments to the symmetric terms in~\eqref{eq:plG_triangle}, we obtain
\[\dist_{\pl G}(\dot x, \dot y)\leq 6+6c_\Pi+\dist_{\pl G}(\dot C_u, \dot C_v).\]

Next, we bound $\dist_{\pl G}(\dot C_u, \dot C_v)$. Since $\pl G$ is obtained by subdividing each edge of $H$ at most twice, we have
\[\dist_{\pl G}(\dot C_u, \dot C_v)\leq 3\dist_{H}(\dot C_u, \dot C_v)\leq 3c_H\dist_{G_\cP}(C_u, C_v),\]
where $c_H$ is the Lipschitz-constant of the mapping from $G_\cP$ to $H$ that is guaranteed in \Cref{thm:getplanarextraverts}.

Recall that $u\in C_u$ and $v\in C_v$. We also have that $\dist_{\pl G}(\dot x,\dot u)\leq 3$ which implies that $\dist_{G}(x,u)\leq 3$; analogously, $\dist_{G}(y,v)\leq 3$. Since $x$ and $y$ are adjacent in $G$, we have that
\[\dist_{G_\cP}(C_u,C_v)\leq \dist_G(u,v)\leq \dist_G(u,x)+\dist_G(x,y)+\dist_G(y,v)\leq 7.\]

Putting it all together, we have 
\begin{align*}
\dist_{\pl G}(\dot x, \dot y) &\leq  6+6c_\Pi+\dist_{\pl G}(\dot C_u, \dot C_v)\\
& \leq 6+6c_\Pi + 3c_H\dist_{G_\cP}(C_u, C_v)\\
& \leq 6+6c_\Pi + 21c_H\\
& =O(1),
\end{align*}
which concludes the proof.
\end{proof}

\section{Contraction decompositions for Lipschitz-mapped graphs}
\label{sec:contract}
In order to create our algorithms, we will need to use Baker's
technique~\cite{Baker94} to be able to work on a bounded-treewidth graph.
Our starting point is the following \emph{contraction decomposition} from~\cite{DemaineHM10}.

\ContractdecompPlanar*


Based on the above theorem and our machinery for planarizing intersection graphs via Lipschitz embeddings, we get a contraction decompostion for intersection graphs. However, we need the following key lemma first, which allows us to compare the treewidth of a contraction of a Lipschitz-mapping of $H$ of $G$ to the treewidth of some (larger) contraction of $G$.
Recall that for a graph $G$ and a vertex set $X\subset V(G)$, let $G\vcon X$ denote the graph
where the edges induced by $X$ are contracted, i.e., $G/E(G[X])$.

\raisecontract*


\begin{proof}
  \dm{Ezt a bizonyitast at kene irni $V(G\vcon X)|_A$ helyett consolidation mapre es kulon observation nelkul. De talan nem most.}\skb{Atirtam}
Let $\breve X$ denote the set of contracted vertices in $H\vcon X$, that is, $V
(H\vcon X)$ is the disjoint union of $V\setminus X$ and $\breve X$. Let $\breve X^
{c^2}$ denote the set of vertices contracted in $G\vcon X^{c^2}$, that is, $V(G\vcon X^{c^2})$ is the disjoint union of $V\setminus X^{c^2}$ and $\breve X^{c^2}$.

We claim that if $a,b\in V$ are contracted to the same vertex $\breve x\in \breve X$, then they are also contracted to the same vertex $\breve x^{c^2}\in \breve X^{c^2}$. Indeed, since $a$ and $b$ are both contracted to $\breve x$ in $H$,
there exists a path from $a$ to $b$ in $H[X]$. By the Lipschitz property,
each edge $uv$ of this path can be represented by a path $P_{uv}$ of length
at most $c$ in $G$. The vertices of $P_{uv}$ are in $N_G(\{u,v\},c/2)$, so by
the Lipschitz property, they are in $N_H(\{u,v\},c^2/2)\subset X^{c^2}$ and
thus the concatenation of the paths $P_{uv}$ forms a connected subgraph of $H[X^{c^2}]$ that contains both $u$ and $v$. This concludes the proof of the
claim. For a given $\breve x\in \breve X$, let $\phi(\breve x)$ be the
corresponding contracted vertex in $\breve X^{c^2}$. The previous claim implies
that $\phi:\breve X\rightarrow \breve X^{c^2}$ is well-defined and surjective.

Let $\cT$ be a tree decomposition of $H\vcon X$ of width $w$. Using this
decomposition we now create a tree decomposition $\cT'$ of $G\vcon X^{c^2}$ with the same underlying tree. In what follows,
$B_H$ denotes a bag of the tree decomposition of $H\vcon X$, and the corresponding
bag in the new decomposition for $G\vcon X^{c^2}$ is denoted by~$B_G$. We will furthermore let $\psi_G:V(G)\rightarrow V(G \vcon X^{c^2})$ and $\psi_H:V(H)\rightarrow V(H\vcon X)$ denote the consolidation maps of the contractions.

For each vertex $v\in B_H\cap (V\setminus X)$, place all vertices of $\psi_G(N_H(v,c))$ in~$B_G$. For each vertex
$v\in \breve X\cap B_H$ we place $\phi(v)$ in $B_G$.  Thus the resulting bags $B_G$
have increased their size by a factor of at most 
\begin{equation}\label{eq:sloppy1}
\max_v |N_H(v,c)|\leq \sum_{i=0}^c \Delta^i <\Delta^{c+1},
\end{equation}
so if the result is a
valid tree decomposition of $G\vcon X^{c^2}$, then its treewidth is at most $
(w+1)\cdot \Delta^{c+1}-1=O(\Delta^{c+1})\cdot w$. Thus, it remains to show that the resulting decomposition
$\cT'$ is a valid tree decomposition of $G\vcon X^{c^2}$.  

Observe that all vertices of $G\vcon X^{c^2}$ are included in some bag by the surjectivity of $\phi$. Let $uv$ be an
edge of $G\vcon X^{c^2}$. If both $u$ and $v$ are vertices of $V\setminus X^
{c^2}$, then $\dist_H(u,v)\leq c\cdot\dist_G(u,v)=c$, thus $u\in \psi_G(N_H(v,c))$ and the vertices $u,v$ are present in all bags $B_G$ where the
original bag $B_H$ contained $v$. If $u\in \breve X^{c^2}$, then let $u'v$ be
the edge of $G$ corresponding to $uv$. (Note that $v$ cannot be in $\breve X^
{c^2}$ as the contraction $G\vcon X^{c^2}$ would have contracted the edge $uv$ as
well.) Observe that $u'\in V\setminus X$ as otherwise if $u'\in X$, then
$v\in N_G(u',1)\subseteq N_H(u',c)$ so both $u'$ and $v$ would contract to
the same vertex in $G\vcon X^{c^2}$. Consequently, $v\in \psi_G(N_H(u',c))$ and the edge is represented in any bag $B_G$ where the bag $B_H$
contained $u'$.

Finally, we need to show that the set of bags that contain a given vertex
$v\in V(G\vcon X^{c^2})$ form a connected subtree of $\cT'$.
Suppose first that $v\in V\setminus X^{c^2}$. Then $v$ appears in a bag
$B_G$ if and only if the corresponding bag $B_H$ contains some vertex of $N_H
(v,c)\setminus X$. Note that since $v\not\in X^{c^2}$, we have that $N_H
(v,c)$ is disjoint from $X$, so the condition simplifies to $B_H$ containing
some vertex of $N_H(v,c)$. By \Cref{obs:connected_treedecomp}, the
vertices of $\psi_H(N_H(v,c))$ appear in a connected subtree of
$\cT$ as they induce a connected subgraph of $H\vcon X$, thus the same holds for the set of bags where $v$ appears in.

Suppose now that $v\in \breve X^{c^2}$. Let $Y$ be the set of vertices that are
contracted to $v$ by the contraction $G\vcon X^{c^2}$. Clearly $G[Y]$ is
connected, so by the Lipschitz property, we have that the connected
components of $H[Y]$ can be connected by length $c$ paths, hence, $H[N_H
(Y,c)]$ is connected. By \Cref{obs:connected_treedecomp} we have that the
vertices $\psi_H(N_H(Y,c))$ appear in a connected subtree of $\cT$ as they induce a connected graph. Notice however that $v$ is added to a bag $B_G$ if and
only if some vertex of $\psi_H(N_H(Y,c))$ appears in the corresponding bag
$B_H$. Consequently, the set of vertices containing $v$ also form a connected
subtree in~$\cT'$.
\end{proof}

We are now ready to state our contraction decomposition theorem.

\begin{theorem}[Contraction Decomposition for Intersection Graphs]
\label{thm:contractdecomp}
Let $G$ be an intersection graph of similarly sized fat objects with cell partition $\cP$ where each clique $C\in \cP$ has size at most $|C|\leq \ell$.
Then for any $k$ there is a collection of sets $\cX_1,\dots,\cX_k\subset \cP$ such that $\sum_{i=1}^k |\cX_i|=O(|\cP|)$ and the corresponding vertex sets $X_i:= \bigcup_{C\in \cX_i} C$ satisfy the following bounds about treewidth and $\cP$-flattened treewidth.
\begin{enumerate}[label=\textbf{(\roman*)}]
\item $\tw(G_\cP\parcon\cX_i)=O(k)$
\item $\tw(G\vcon X_i)=O(k\ell)$
\item $\tw_{\cP\vcon X_i}(G\vcon X_i)=O(k\log(\ell+1))$, where $\cP\vcon X_i$ is the partition of $V(G\vcon X_i)$ containing the cliques of $\cP\setminus \cX_i$, and the contracted vertices are added as singletons.
\end{enumerate}
\end{theorem} 

\begin{proof}
Apply \Cref{thm:getplanarclique} to construct a plane graph $H$ with
the same vertex set as $G_\cP$ (namely, $V(H)=V(G_\cP)=\cP$) such that
the identity map is $c$-Lipschitz from $G_\cP$ to $H$. We apply
\Cref{thm:contractdecomp_planar} on $H$ with $k$ classes. We get a
partition $E_1,\dots,E_k$ of $E(H)$ where $\tw(H/E_i)=O(k)$. For each
$E_i$ let $\cX_i\subseteq \cP$ denote the vertex set in $G_\cP$ that is
within $H$-distance $c^2$ from $V(E_i)$, that is, set $\cX_i=N_H(V
(E_i),c^2)$. Let $X_i$ be the set of vertices in $V(G)$ that appear in the
cliques $\cX_i$. Since $H$ has constant-bounded degree $\Delta_H$, each
$c^2$-neighborhood in $H$ has constant size, thus \begin{equation}\label{eq:sloppy2}
|\cX_i|=O(|E_i|\cdot \Delta_H^{c^2+1})=O(|E_i|).
\end{equation}
Consequently, 
\[\sum_{i=1}^k |\cX_i|=O(\sum_{i=1}^k |E_i|)=O(E(H))=O(V(H))=O(|\cP|).\]

By \Cref{lem:twcompare} each of
$G_\cP\parcon\cX_i$ has treewidth $O(\Delta_{H}^{c+1})\cdot \tw(H/E_i)=O(k)$. Consequently, $\tw(G\vcon X_i)= O(k)\cdot \max_{C\in \cS} |C| = O(k\ell)$. Since contracted vertices of $G\vcon X_i$ are singletons of the partition $\cP\vcon X_i$, the maximum clique size is obtained on some clique of $\cP\setminus \cX_i$. Thus the $(\cP\vcon X_i)$-flattened treewidth of $G\vcon X_i$ can be bounded as $\tw_\cP(G\vcon X_i)=O(k)\cdot \max_{C\in \cP} \log(|C|+1)=O(k\log(\ell+1))$.
\end{proof}

\begin{remark}
The bounds \eqref{eq:sloppy1} and \eqref{eq:sloppy2} are the only reason why our running time is doubly (or triply, in case of the slower Steiner tree algorithm) exponential as a function of $\alpha$. To get single-exponential dependence on $\alpha$, one can observe that the size of a $c$-neighborhood in $H$ is only quadratic in $c$, as all the graphs $H$ where \Cref{lem:twcompare} is applied as well as the graph $H$ in \Cref{thm:contractdecomp} are not only of degree at most~$c$, but have further porperties: their vertices are grid points and its edges have length bounded by some polynomial of $\alpha$. Thus the size of the $c$-neighborhood of a vertex of $H$ is $\poly(\alpha)c^2$.
\end{remark}

\section{Sparsifying cliques}
\label{sec:sparse}

The goal of this section is to simplify the underlying graph so that each
clique has $f(1/\eps)$ vertices, which will be required for our
algorithms. The simplification is done in several steps, and uses some
problem-specific arguments.

\subsection{Sparsifying cliques for \stsp}

\begin{lemma}\label{lem:subsetTSPsolTerminals}
Let $G$ be an intersection graph with cell partition $\cP$ and let
$T\subset V(G)$ be a set of terminals. Then there is an optimum subset
TSP walk $W$ for $T$ that uses $O(1)$ non-terminal vertices from each
cell $C\in \cP$, i.e., $|(V(W)\setminus T)\cap C|=O(1)$ for each $C\in \cP$.
Moreover, all vertices are used $O(1)$ times by the walk, and all but $O(1)$ vertices of $T\cap C \cap V(W)$ are incident only to edges induced by $C$. 
\end{lemma}

\begin{proof}
We start by simplifying the tour so that no pair of neighboring cells\footnote{The cells $C$ and $C'$ are neighboring if they are neighbors in $G_\cP$, that is, they may not be ``grid-neighbors''.} has more
than two edges of $W$ going between them. Let $C,C'$ be two cells such that
there are at least three edges of $W$ going between them. We pick an
orientation on $W$ to make it a directed closed walk. As a result, there are at
least two arcs of $W$ going in the same direction between $C$ and $C'$;
suppose without loss of generality that $uu'$ and $vv'$ are edges of $W$ that
both go from $C$ to $C'$. We can then replace the edges $uu'$ and $vv'$ in
$W$ with the edges $uv$ and $u'v'$ (and reverse the direction of the $u'$ to
$v$ subpath of $W$) to get a closed walk $W'$ of the same length that has fewer inter-cell steps. Repeated applications of this step result in a tour
$W$ that has at most two edges going between any pair of neighboring cells.

If a vertex $x\in C\setminus T$ has all incident edges ending in $C$, then
$x$ could be bypassed and we would get a shorter walk spanning all the
terminals; therefore, each vertex of $C\setminus T$ must be incident to at
least one edge that goes to a different cell. Since there are only $O
(1)$ cells that a given vertex can be connected to, and at most two edges in
$W$ lead to each of these cells, there are only $O(1)$ non-terminals in $C$
that are on $W$. The last claim also follows.
\end{proof}

Next, we remove some vertices from the terminal set $T$ (but not from the
graph) so that in the new instance each cell has only a few terminals from
the new terminal set.

\begin{lemma}\label{lem:subsetTSPsparsifyT}
Given $\eps>0$ and an instance $(G,T)$ of \stsp, we can construct a vertex set
$X\subset T$ to be removed from the set of terminals in $\poly(n)$ time such
that the following hold for $X$ and the new terminal set $T'=T\setminus X$.
\begin{enumerate}[label=(\roman*)]
	\item Each $C\in \cP$ contains only a few terminals:
              $|C\cap T'|\leq 1/\eps$
	\item If $(G,T)$ has a subset TSP of length $\ell$ then $(G,T')$ has
              a subset TSP of length $\big(1+O(\eps)\big)\ell-|X|$. 
	\item Given a subset TSP of $(G,T')$ of length $\ell'$, a subset TSP
              of $(G,T)$ of length $\ell'+|X|$ can be constructed in $O(n^2)$
              time.
\end{enumerate}
\end{lemma}

\begin{proof}
Let $\cP$ be a cell partition of $G$. In each clique $C$, we select up to
$1/\eps$ arbitrary terminals per clique to be put into $T'$; in cliques with
at most $1/\eps$ terminals, all terminals will be put into $T'$. We let
$X=T\setminus T'$ bet he set of terminals to be deleted. Clearly this
satisfies~(\textit{i}).

To prove (\textit{ii}), let $W$ be a subset TSP walk of $(G,T)$ of length $\ell$ that is
simplified with respect to $\cP$ in the sense of
\Cref{lem:subsetTSPsolTerminals}. Consequently, there are $O(1)$ edges of
$W$ leaving each clique $C\in \cP$; let us call such edges \emph{external}, and let
$W_{ext}$ denote their multiset. Edges of $W$ whose endpoints are in the same
clique of $\cP$ are called \emph{internal}, and their multiset is denoted by
$W_{int}$. Note that if a vertex of $W$ is not incident to any external
edges, then it must be a terminal and it is visited only once in $W$. Indeed,
a non-terminal that is incident to only internal edges can be removed from
$W$, and multiple visitations of a terminal with internal edges can also be
simplified to one visitation. Since we have $O(1)$ external edges incident to
any clique $C\in \cP$, we have that there are at most $|C\cap T|$ internal
edges incident to~$C$.

We say that $C$ is \emph{light} if $|C\cap T|\leq 1/\eps$, and otherwise we
say that $C$ is \emph{heavy}.
We can now construct a subset TSP walk $W'$ for the instance $(G,T')$ as follows:
we follow the walk $W$, and if in some heavy clique $C$ we have that both the arc entering and exiting $v\in X\cap V(W)\cap C$ is internal, then we skip~$v$. Note that the skipping is possible since the predecessor and successor of $v$ are connected in $C$.
\[
|W'| = |W|-\sum_{C \text{ heavy}} (|X\cap C| - O(1))
= |W|-|X| + O(\eps |T|)=\big(1+O(\eps)\big)\ell-|X|. 
\]

To prove (\textit{iii}), let $W'$ be an optimum tour for the instance $(G,T')$ that is
simplified using \Cref{lem:subsetTSPsolTerminals}. We construct a tour
$W$. First, let $W_{ext}=W'_{ext}$. To construct the internal edges of $W$,
consider a clique $C\in \cP$. If $C$ is light, then we use the same internal
edges from $G[C]$ for $W'$ as we did in $W$. If $C$ is heavy, then notice
that $W'$ must have an internal edge in $C$. Indeed, $C$ has $1/\eps$
terminals from $T'$ and there are only $O(1)$ edges of $W'$ leave $C$. Let $
(u,v)\in W'_{int}[C]$ be a an edge of $W'$ . We replace the edge $(u,v)$ with
a $u\to v$ path whose internal vertices are $C\cap X$, which lengthens $W'$
by $|C\cap X|$. We apply this modification to each clique $C$, and the
resulting tour $W$ is a subset TSP tour of all the terminals in $T=T' \cup
X$, and has length
\[|W|=|W'|+\sum_{C\text{ heavy}} (|C\cap X|)=\ell' +|X|.\]
Finally, notice that both \Cref{lem:subsetTSPsolTerminals} and the
modification above is algorithmic and can be done in $O(n^2)$ time.
\end{proof}

\begin{lemma}\label{lem:sparsenTSP}
Let $G$ be an intersection graph with cell partition $\cP$ and terminal set
$T$ such that $|C\cap T|\leq k$ for each $C\in \cP$. Then for any $\eps>0$
there is a vertex set $V'\subseteq V$ such that for any $C\in \cP$ we have
$|V'\cap C| = O(k^2/\eps^4)$, and for any $u,v \in T$ we have $\dist_{G[V']}
(u,v)\leq (1+\eps)\dist_G(v,v')$. Moreover, given $G$ and its cell
partition $\cP$, the set $V'$ can be computed in polynomial time.
\end{lemma}

\begin{proof}
From each $C\in \cP$ choose a single vertex $a_C$ to be the \emph{hub} of $C$,
and collect the hubs in the vertex set $A$. Then for any pair of vertices
$u,v\in T\cup A$ of distance at most $4/\eps+2$ we add to $V'$ all vertices
of a shortest $u \to v$ path.

If a vertex $w \in C$ is added to $A$, then it is either a hub, a terminal, or
it is an internal vertex of a shortest path connecting some $u,v\in T\cup A$
with $\dist_G(u,v)\leq 4/\eps+2$. Consequently, $\dist_G(w,u)$ and $\dist_G
(w,v)$ are both less than $4/\eps+2$. The number of vertices in $T\cup A$ at
distance at most $4/\eps+2$ from $w$ is at most $O(k/\eps^2)$, therefore the
number of such shortest paths is $O(k^2/\eps^4)$, and thus $|V'\cap C|=O
(k^2/\eps^4)$.

To show the second property, let $u,v\in T$ be vertices whose distance is
$\ell\geq 4/\eps$. Let $u=u_0,u_1,\dots,u_k=v$ be a sequence of vertices on
a shortest $u\to v$ path such that $k=\lceil \ell\eps/4 \rceil$ and $\dist_G
(u_i,u_{i+1})\leq 4/\eps$. Let $a_i$ be the hub of the cell $C_i\in \cP$
which contains $u_i$. Since $\dist_G(u_i,u_{i+1})\leq 4/\eps$, we have that
$\dist_G(a_i,a_{i+1})\leq 4/\eps+2$, therefore $V'$ contains a shortest path
from $a_i,a_{i+1}$. Concatenating these paths we get a $u\to v$ path whose
internal vertices are in $V'$ and has length at most $\ell+2k=\ell +
2\lceil \ell\eps/4 \rceil < \ell+2\cdot(\ell\eps/2)=(1+\eps)\ell$.

Finally, notice that the above construction can be repeated using a
polynomial algorithm.
\end{proof}

\begin{lemma}\label{lem:sparseTSP}
Given $\eps>0$ and an instance $(G,T)$ of \stsp with a fixed cell partition
$\cP$ of $G$, in polynomial time we can construct an instance $(G',T')$ that
satisfies the following conditions.

\begin{itemize}
\item $V(G')\subseteq V(G)$, $T'\subseteq T$
\item Each clique $C\in \cP$ satisfies
      $|C\cap V(G')|\leq O(1/\eps^6)$ and $|C\cap T'|\leq 1/\eps$.
\item The shortest subset TSP
tour of $(G',T')$ is within a $1+O(\eps)$ factor from the optimum of $
(G,T)$, and given a $(1+O(\eps))$-approximate tour for $(G',T')$ we can in
polynomial time create a $(1+O(\eps))$-approximate tour for $(G,T)$.
\end{itemize}
\end{lemma}

\begin{proof}
Fix a cell partition $\cP$ of $G$. By \Cref{lem:subsetTSPsparsifyT} we
can create a terminal set $T'$ where for each $C\in \cP$ the number of
terminals is at most $|C\cap T'|\leq 1/\eps$, and the length of the optimum
tour increases by at most a $1+O(\eps)$ factor. We can then apply \Cref
{lem:sparsenTSP} on $(G,T')$ with $k=1/\eps$ to get a graph $G'$ where each
clique $C'$ has size at most $|C'|\leq 1/\eps^6$. Note that any subset TSP walk
$W$ can be regarded as a sequence of shortest paths between its consecutive
terminal visits, so by \Cref{lem:sparsenTSP} the shortest subset TSP
tour of $(G',T')$ is within a $1+O(\eps)$ factor from the optimum of $
(G,T)$, and given a $(1+O(\eps))$-approximate tour for $(G',T')$ we can in
polynomial time create a $(1+O(\eps))$-approximate tour for~$(G,T)$.
\end{proof}

\subsection{Sparsifying cliques for \stein}

The next two lemmas are similar to those that we have seen for \stsp.

\begin{lemma}\label{lem:SteinSolTerminals}
Let $G$ be an intersection graph with cell partition $\cP$ and let $T\subseteq V
(G)$ be a set of terminals. Then given any Steiner tree $H_0$ for $(G,T)$ we can construct in polynomial time a Steiner tree $H$ for $(G,T)$ where $V(H)\subset V(H_0)$, and 
$H$ uses $O(1)$ non-terminal vertices from each cell $C\in \cP$, i.e., $|(V
(H)\setminus T)\cap C|=O(1)$ for each $C\in \cP$. Moreover, there are at most $O(1)$ vertices of $T\cap C \cap V(H)$ that are incident to some edge of $H$ that has an endpoint outside $C$.
\end{lemma}

\begin{proof}
We start by simplifying $H_0$ into a tree $H_1$ so that no pair of neighboring cells has more than
one edge of $H_1$ going between them. Let $C,C'$ be two cells such that there
are at least two edges of $H_0$ going between them, let $(u,u')$ and $
(v,v')$ be these edges (where $u,v$ and $u',v'$ are not necessarily
distinct). We remove the edge $(u,u')$ from $H_0$, and as a result the tree falls
apart into two components, one of which contains the edge $(v,v')$. Without
loss of generality, $u$ and $v$ are in different components. Therefore we can
add the edge $(u,v)$ to form a tree of the same length, but which has fewer
inter-clique edges.

By applying the above modification exhaustively, we get a a tree $H_1$ where any
pair of neighboring cells have at most one edge going between them.
If a vertex $x\in C\setminus T$ has all its incident edges in $H_1$ ending in $C$, then
we bypass $x$ by adding a path connecting its neighbors, and remove $x$ from $H_1$.
Once no more such modifications are possible, the resulting tree $H$ has the property that each vertex of $H$ in $C\setminus T$ must be incident to at
least one edge that goes to a different cell. Since there are only $O
(1)$ cells that a given vertex can be connected to, and at most one edge in
$H$ lead to each of these cells, there are only $O(1)$ non-terminals in $C$
that are on $H$. Thus $H$ has the desired properties and it was obtained in polynomial time.
\end{proof}

\begin{lemma}\label{lem:SteinerSolSparsifyT}
Given $\eps>0$ and an instance $(G,T)$ of \stein, we can construct a
vertex set $X\subset T$ to be removed from the set of terminals in
$\poly(n)$ time such that the following hold for $X$ and the new terminal set
$T'=T\setminus X$.
\begin{enumerate}[label=(\roman*)]
  \item Each $C\in \cP$ contains only a few terminals:
              $|C\cap T'|\leq 1/\eps$
  \item If $(G,T)$ has a Steiner tree of length $\ell$ then $(G,T')$ has
              a Steiner tree of length $\big(1+O(\eps)\big)\ell-|X|$. 
  \item Given a Steiner tree of $(G,T')$ of length $\ell'$, a Steiner
              tree of $(G,T)$ of length $\ell'+|X|$ can be constructed in
              $O(n^2)$ time.
\end{enumerate}
\end{lemma}

\begin{proof}
Let $\cP$ be a cell partition of $G$. In each clique $C$, we select up to
$1/\eps$ arbitrary terminals per clique to be put into $T'$; in cliques with
at most $1/\eps$ terminals, all terminals will be put into $T'$. We let
$X=T\setminus T'$ be the set of terminals to be deleted. Clearly this
satisfies~(\textit{i}).

To prove (\textit{ii}), let $H$ be an optimum Steiner tree of $(G,T)$ that is
simplified with respect to $\cP$ in the sense of
\Cref{lem:SteinSolTerminals}. As in
\Cref{lem:subsetTSPsparsifyT}, edges of $H$ leaving a clique $C\in \cP$
are external and denoted by $H_{ext}$, and the rest of the edges are $H_
{int}$. Since there are at most $O(1)$ external edges incident to $C$, and
each non-terminal in $C$ must be incident to at least one such edge, we have
that there are at most $|C\cap T|+O(1)$ internal edges incident to $C$.

We can now construct a Steiner tree $H'$ for the instance $(G,T')$. To
construct the internal edges of $H'$, let $C$ be some clique. We say that $C$
is \emph{light} if $|C\cap T|\leq 1/\eps$, and otherwise we say that $C$ is \emph{heavy}.
To construct a shorter tree, in each heavy clique $C$ we remove from $H$ any $v\in X\cap C$ that is only incident to internal edges. (We can connect the neighbors of $v$ inside the clique with a star after the removal.) Since there are only $O(1)$ external edges incident to $C$, this removes at least $|X\cap C|-O(1)$ vertices. Thus 
\[
|H'| = |H|-\sum_{C \text{ heavy}} (|X\cap C| - O(1))
= |H|-|X| + O(\eps |T|)=\big(1+O(\eps)\big)\ell-|X|. 
\]
%


To prove (\textit{iii}), let $H'$ be an optimum tree for the instance $(G,T')$. We can
then construct the desired tree $H$ by simply connecting each terminal of $X$
to some terminal of $T'$ in the same clique.
\end{proof}

\begin{lemma}\label{lem:sparsenSteinercliques}
Let $G$ be an intersection graph with cell partition $\cP$ and terminal set
$T$ such that $|C\cap T|\leq k$ for each $C\in \cP$. Then for any $\eps>0$
there is a vertex set $V'\subseteq V(G)$ such that for any $C\in \cP$ we have
$|V'\cap C| = 2^{O((1/\eps)\log(k/\eps))}$, and for any $T' \subset T$ we
have $|\smt(G[V'],T')|\leq (1+\eps)|\smt(G,T')|$. The set $V'$ can be computed
in $2^{O((1/\eps)\log(k/\eps))}\poly(n)$ time.
\end{lemma}

\begin{proof}
Similarly to \Cref{lem:sparsenTSP}, we define a hub vertex $a_C$ in each
clique $C\in \cP$, and collect the hubs in the set $A$. Next, for each subset
$T'$ of $T\cup A$ we want to find a minimum Steiner tree for the terminal
set $T'$, and take the union of the vertices of $\smt(G,T')$ if
$|\smt(G,T')|\leq 2/\eps$. More precisely, we set
\[V'=\bigcup_{\substack{T'\subseteq A \cup T\\|\smt(G,T')|\leq 2/\eps}}
  V(\smt(G,T')).\]

First, we show that $V'$ satisfies the desired properties. Let $H$ be an
optimum Steiner tree for $(G,T)$. We can partiton $H$ into $k$ edge-disjoint subtrees
$H_1,\dots,H_k$, each of size between $t\eqdef\lfloor 1/(2\eps) \rfloor$ and
$2t$. To see this, root $T$ at some arbitrary vertex $r$, and let $v$ be the
lowest vertex (i.e., most distant from $r$) whose subtree has at least
$t$ vertices. We order the children of $v$ in decreasing order of their
subtree sizes, and take the shortest prefix of the children that together
with $v$ induce a subtree of size at least $t$. Note that this subtree
cannot exceed $2t$ edges, since that would mean that the subtree of the first
child and $v$ would induce a subtree of size at least $t$. We say that a
vertex of $H$ is chopped if there are at least two distinct trees $H_i$
incident to it.

Note that the number of trees $k$ is between $|H|/t$ and $|H|/2t$, thus
$k=\Theta(\eps |H|)$. For each tree $H_i$ let $\cP_i$ be the set of cliques
where $H_i$ has a chopped vertex. The terminal set
$T_i$ is defined as the terminals $T\cap V(H_i)$ together with the hubs in the
cliques of $\cP_i$, that is, let
\[T_i\eqdef (T\cap V(H_i)) \cup \bigcup_{C\in\cP_i} a_C.\]
Note that since $|\cP_i|\leq |V(H)|\leq 1/\eps$, we have
$|T_i|\leq 2/\eps$, therefore $V'$ induces a minimum Steiner tree of $T_i$.
Consequently, $|\smt(G[V'],T_i)|= \smt(G,T_i) \leq |H_i|+|\cP_i|$.

We claim that the tree $H'$ defined as the spanning tree of the union of the
trees $\smt(G[V'],T_i)$ is a suitable approximation. To prove this, we can sum
up the inequalities above for all $i$:
\[|H'|\leq \sum_{i=1}^k |\smt(G[V'],T_i)| \leq
\sum_{i=1}^k (|H_i|+|\cP_i|)\leq |H|+\!\!\sum_{v\text{ chopped}}\!\! \deg_H(v)
=|H|+k-1 \leq (1+\eps)|H|,\]
since the total degree of the chopped vertices is one less than the number
$k$~of trees that $H$ was chopped into, and since each tree has at most
$2t\leq 1/\eps$ edges, we have $k\leq \eps|H|$ at most the number of trees that
we have chopped into.

In order to compute $V'$, we enumerate all sets of $S\subset T\cup A$ with
$|S|\leq 2/\eps$ and with diameter at most $2/\eps$ (where the diameter is
measured in the metric space induced by $G$). Since a given clique $C\in \cP$
has $O(1/\eps^2)$ cliques within distance $2/\eps$, we have that there are at
most $O(1/\eps^2)$ cliques from which $S$ can be built if its Steiner tree
contains some vertex of $C$. Each clique has at most $k$ terminals and one
hub, thus there are at most $((k+1)/\eps^2)^{O(1/\eps)}=2^{O(
(1/\eps)\log(k/\eps))}$ sets $S$ whose Steiner tree intersects $C$. For each
such set $S$, we use the algorithm of Dreyfus and Wagner~\cite{dw71} to compute
the minimum spanning tree of $S$ in $3^{|S|}\poly(n)$ time, and we add the
vertices of the computed tree to $V'$. Since there are at most $n$
neighborhoods in which we run the enumeration and the Steiner tree
computation, the running time is $2^{O((1/\eps)\log(k/\eps))}\poly
(n)$, as required.
\end{proof}

The equivalent of \Cref{lem:sparseTSP} for Steiner tree follows analgously, using \Cref{lem:SteinerSolSparsifyT} and \Cref{lem:sparsenSteinercliques}.

\begin{lemma}\label{lem:sparseSteiner}
Given $\eps>0$ and an instance $(G,T)$ of \stein with a fixed cell partition
$\cP$ of $G$, in $2^{O((1/\eps)\log(1/\eps))}\poly(n)$ time we can construct an instance $(G',T')$ that satisfies the following conditions.

\begin{itemize}
\item $V(G')\subseteq V(G)$, $T'\subseteq T$
\item Each clique $C\in \cP$ satisfies
      $|C\cap V(G')|\leq 2^{O((1/\eps)\log(1/\eps))}$ and $|C\cap T'|\leq 1/\eps$.
\item The shortest Steiner tree of $(G',T')$ is within a $1+O(\eps)$ factor from the optimum of $(G,T)$, and given a $(1+O(\eps))$-approximate tour for $(G',T')$ we can in polynomial time create a $(1+O(\eps))$-approximate tour for $(G,T)$.
\end{itemize}
\end{lemma}

\section{A subset spanner for intersection graphs}
\label{sec:tspspanner}

The main goal of this section is to prove the following theorem.
\begin{theorem}[Subset spanner]\label{thm:spanner}
Given $G\in \cI_\alpha$ with cell partition $\cP$, a vertex set $T\subset V(G)$ of terminals, and some $\eps\in (0,1/2)$, there exists a family $\cS\subseteq \cP$ of size $|\cS|=O(\frac{1}
{\eps^4}\smt(G,T))$ such that for any pair of vertices $x,y \in T$ we have
$\dist_{G[\bigcup \cS]}(x,y)\leq (1+\eps)\dist_G(x,y)$, and $\cS$ can be computed in polynomial time.
\end{theorem}

The first two steps of our proof mirror the subset spanner proof in weighted
planar graphs given by Klein~\cite{Klein06}. In Klein's construction, the spanner is created in three
phases: first, a single source spanner is created, which only
concerns distances from a single vertex $v$ to the vertices of a shortest
path $Q$ not containing $v$. Next, a bipartite spanner is defined that tracks
distances from a shortest $uv$ path $Q_2$ to a different $uv$ path $Q_1$
disjoint from $Q_2$. Finally, a spanner is made that approximates distances
between any pair of vertices on the outer face of a planar graph. It turns
out that the first two steps work in arbitrary, not necessarily planar graphs as well. More precisely Klein's proof of the first step (Lemma~\ref{thm:singlesource} below) does use planarity, but it can be avoided, while the proof of the second step (Lemma~\ref{lem:bipartitespan}) can be used verbatim for nonplanar graphs.

\begin{lemma}[Single source spanner]\label{thm:singlesource}
Let $G$ be a graph, and let $Q$ be a shortest path in $G$ between its
endpoints. Then for
any $\eps>0$ there is a subgraph $H$ with $|H|=O(1/\eps^2)|\dist_G(v,Q)|$ such that for
any $x\in V(Q)$ we have $\dist_{H\cup Q}(v,x)\leq (1+\eps)\dist_G(v,x)$. Moreover, $H$ can be computed in polynomial time.
\end{lemma}

\begin{proof}
Let $R$ be a shortest path from a vertex $v$ to $Q$.
Let $x_0$ be the endpoint of $R$ on $Q$, and first we consider the portion of
$Q$ that follows $x_0$. Let $Q'$ be the portion of $Q$ whose vertices are at
distance at most $(1/\eps)|R|$ distance from $x_0$. Starting at $x_0$ we
define a sequence of vertices $x_1,\dots$ on $Q'$ the following way. For
$i\geq 1$, let $x_i$ be the earliest vertex on $Q'$ where $\dist_G(v,x_{i-1}) + \dist_G(x_{i-1},x_i)>(1+\eps)\dist_G(v,x_i)$. Analogously, we do the
same procedure moving backwards on $Q'$ from $x_0$, defining the points $x_
{-1},x_{-2},\dots$. Let $H$ be the union of the shortest paths $v \to x_i$.

Let $x_k$ be the vertex of largest index assigned by the procedure. We
estimate the total length of paths from $v$ to $x_1,\dots,x_k$; the paths
whose endpoint has a negative index can be estimated analogously. By summing the inequalities $\dist_G(v,x_{i-1}) + \dist_G(x_{i-1},x_i)>(1+\eps)\dist_G(v,x_i)$  defining the $x_i$'s, we get
\[(1+\eps)\sum_{i=1}^k \dist_G(v,w)
< \sum_{j=0}^{k-1} \left(\dist_G(v,x_j)
+ \dist_G(x_{j},x_{j+1})\right)
=\sum_{j=0}^{k-1} \dist_G(v,x_j)+\dist_G(x_0,x_k)\]
Substracting $\sum_{i=1}^{k-1}\dist_G(v,x_i)$ from both sides and using that  $\dist(x_0,x_k)\leq (1/\eps) |R|$ holds by construction, we get: 
\[\eps\sum_{i=1}^k \dist_G(v,x_i) +\dist_G(v,x_k)
< \dist_G(v,x_0)+(1/\eps)|R|=(1+1/\eps)|R|,\]
which yields $\sum_{i=1}^k \dist_G(v,x_i)=O(|R|/\eps^2)$.
\end{proof}

Klein~\cite{Klein06} shows that applying the single source spanner several
times can yield a bipartite spanner.

\BipartiteSpan*

\subparagraph*{Perimeters in region-restricted graphs}
Let $G$ be an intersection graph, and let $\wf R$ be a wireframe of $G$ that is a cycle, where $\re R$ is either the bounded or unbounded region defined by $\gamma(\wf R)$. Consider now the graph $R:=G|_{\re R}$: for each vertex $\dot v\in V(\wf R)$ there is a natural corresponding polygon $v_{\re R} \in V(G|_{\re R})$, namely, the component of $p(\dot v)\cap \re R$ that contains the point $\dot v$.
Consider the polygon set $\{v_{\re R} \mid \dot v \in V(\wf R)\}$. Notice that because of the definition of wireframes, whenever $\dot v \dot v'$ is an edge of $\wf R$, then $v_{\re R}v'_{\re R}$ is an edge induced by $\{v_{\re R} \mid \dot v \in V(\wf R)\}$. Let $\Perim(R)$ denote the walk on $\{v_{\re R} \mid \dot v \in V(\wf R)\}$ corresponding to~$\wf R$.

\subparagraph*{Object duplication}\label{sec:objectduplication} Let us discuss a technical issue that can be avoided using a simple trick.
Starting here, we will often use a path $\pl P$ in some wireframe $\pl W$ of $G$ to define a path in $G$. This is not straightforward, since $\pl P$ would normally give rise to a \textit{walk} in $G$ through the parent relation: indeed, if $\dot u \dot v$ is an edge of $\pl P$, then $p(\dot u)p(\dot v)$ is an edge in $G$, but different vertices along $\pl P$ may have the same parent in $G$.
To avoid technicalities involving walks using an object multiple times, we pretend that we have multiple copies of an object available. In all cases where $v\in G$ appears $k$ times as a parent along the walk corresponding to $\pl P$, we add $k-1$ new copies of the polygon $v$ to $G$. This is also done in a restricted intersection graph $R:=G|_{\re R}$ whose boundary is a wireframe cycle $\wf R$: we think of the objects of $\Perim(R)$ as objects of $R$, each of which has been copied as many times as they occur as parents of $V(\wf R)$-vertices.
Consequently $\Perim(R)$ forms a cycle in $G|_{\re R}$ rather than a closed walk. Finally, we remark that adding such objects to $G$ does not impact the optimum of \textsc{Steiner Tree} (or even \textsc{Subset TSP}). To see why, notice that (sub)polygon(s) of an original object $v$ that appear in the optimum of the modified instance can be substituted by a single copy of $v$ in any feasible solution, whose total size can therefore not increase, while any feasible solution of the original instance is also feasible in the modified instance.

\subparagraph*{Building a skeleton}
Every subpath of a shortest path is a shortest path between its endpoints, but if a path is, say, twice as long as the distance between its endpoints, then a subpath can be still much longer than the distance between the endpoints of the subpath. We will often encounter paths that are approximately short in a sense that extends to subpaths as well:
\begin{definition}
A path $P$ in a graph $G$ from $u$ to $v$ is called an \emph{internally $c$-approximate shortest path} if for any $a,b\in V(P)$ where $a,b$ are not both endpoints of $P$\footnote{We do not include the full path $P$ for technical reasons; we will need this to ensure that the southern boundary paths of the faces in the skeleton are internally $(1+\eps)$-approximate rather than talking about their subpaths. See \Cref{def:shortcut} and \ref{it:sk_approximatepaths}.}, we have $|P[a,b]|\leq c\cdot\dist_G(a,b)$.
\end{definition}

The following lemma relates the graphs obtained by restricting $G$ to smaller regions. Essentially, it says in various forms that distances can only increase with further restrictions. It will be used much later.
Recall that for a cycle $W$ and $u,v \in V(W)$ we denote by $W[u,v]$ the shorter subpath of~$W$ connecting $u$ and~$v$.

\begin{lemma}\label{lem:restrictionmetric}
Let $\wf B,\wf R$ be wireframe cycles in $G$ where $\re B\subset \re R$, and let $B=G|_{\re B}$ and $R=G|_{\re R}$. 
Then the following hold.
\begin{enumerate}[label=(\roman*)]
\item Let $u,v\in V(B)$ such that $u\in \ccomp(\hat u\cap \re B)$ and $v\in \ccomp(\hat v\cap \re B)$ for some $\hat u, \hat v\in V(R)$. Then $\dist_B(u,v)\geq \dist_R(\hat u, \hat v)$. 
\item If $\pl P\subset \wf B$ is path such that the corresponding path $P_R$ in $R$ is a internally $c$-approximate shortest path for some $c\geq 1$, then the path of $\Perim(B)$ corresponding to $\pl P$ is a internally $c$-approximate shortest path in $B$.
\item If\, $\pl W$ is a boundaried wireframe of $B$ then for any $\dot u,\dot v \in V(\pl W)$ we have
\[\dist_{\pl W}(\dot u,\dot v)\geq \dist_B(p|_{\re B}(\dot u),p|_{\re B}(\dot v)).\]
\end{enumerate}
Moreover, each of the above statements hold when replacing $R$ with $G$.
\end{lemma}

\begin{proof}
(\textit{i}) For $v\in V(B)$ let $\hat v\in V(R)$ denote the unique corresponding object such that $v\in \ccomp(v\cap \re B$. Note that for any edge $ab$ of $B$ we have that $\hat a \hat b$ is an edge in $R$. Thus any path in $B$ from $u$ to $v$ has a corresponding walk in $R$ of the same length, which concludes the proof.\\
(\textit{ii}) Let $P_B$ denote the path on $\Perim(B)$ corresponding to $\pl P$. By definition, we have that for any proper subpath $Q_B=P_B[u,v]$ the corresponding path $Q_R=P_R[\hat u,\hat v]$ satisfies
\[|Q_B|=|Q_R|\leq c\cdot\dist_R(\hat u,\hat v) \leq c\cdot\dist_B(u,v),\]
where the last inequality used (\textit{i}).\\
(\textit{iii}) Observe that for any wireframe edge $\dot a \dot b\in E(\pl W)$ the objects $p|_{\re B}(\dot a),p|_{\re B}(\dot b)$ form an edge in $B$. Thus a shortest path in $\pl W$ has a corresponding walk of the same length in $B$ connecting the $p|_B$-parents of its endpoints.
\end{proof}

Our final ingredient in the proof of \Cref{thm:spanner} is the creation of a wireframe $\skel$ of $G$ called the \emph{skeleton}. The skeleton will be created by iteratively adding so-called \emph{shortcuts} to a starting graph.

\begin{definition}[Shortcut and its reach]\label{def:shortcut}
A \emph{shortcut} of $G|_{\re R}$ is a shortest path $P$ in $G|_{\re R}$
between some vertices $\breve u, \breve v \in V(\Perim(R))$, such that $|P|<\frac{|\wf  R[\dot u,\dot v]|}{1+\eps}$. The \emph{reach} of a shortcut
from $\breve u$ to $\breve v$ is the length of the shorter subpath of $\wf R$ from $\dot u$ to~$\dot v$, that is, $|\wf R[\dot u,\dot v]|=|\Perim(R)[\breve u, \breve v]|$.
\end{definition}

We need a skeleton with the following properties.

\begin{lemma}[Skeleton properties]\label{lem:skeletonprops}
Given $G,T,$ and $\eps\in (0,1/2)$, a wireframe $\skel$ of $G$ of complexity $\cA(\skel)=O(n^3/\eps)\cdot \compl(\cA(G))$ can be constructed in polynomial time with the following properties.
\begin{enumerate}[label=\textbf{Sk\arabic*}]
\item The graph $\skel$ is $2$-connected. \label{it:sk_2connected}
\item The wireframe $\skel$ has a unique irrelevant face with region $\re F_\textup{irr}\subset \Reals^2$ such that any object of $V(G)$ intersecting $\re F_\textup{irr}$ also intersects $\Reals^2\setminus \re F_\textup{irr}$. (Other faces are called \emph{relevant}.)\label{it:sk_irrelevant}
\item The cycle $\wf R$ of $\skel$ on the boundary of any relevant face region $\re R\subset \Reals^2$ is the concatenation of two paths, $\No(\wf R)$ and $\So(\wf R)$, with corresponding paths $\No(G|_{\re R})$ and $\So(G|_{\re R})$ in~$\Perim(R)$.\label{it:sk_compass}
\item For each $v\in T$ there is at least one point $\dot v\in V(\So(\pl R))$ with $p(\dot v)=v$ for some relevant face $\re R$ of $\skel$.\label{it:sk_terminals}
\item For any relevant face $\re R$ the paths $\No(G|_{\re R})$ is a shortest path and $\So(G|_{\re R})$ is an internally $(1+\eps)$-approximate shortest path in $G|_{\re R}$.\label{it:sk_approximatepaths}
\item The total length of the northern and southern paths of all relevant faces is $O(1/\eps)\smt(G,T)$.\label{it:sk_northsouthbound}
\end{enumerate}
\end{lemma}

\begin{proof}
  First we build a wireframe that is a cycle going through every terminal, as follows.
Consider a 2-approximate Steiner tree $W_0$ of the terminal set $T$ in $G$. Such a tree can be computed by finding a minimum spanning tree of $T$ in the graph metric $\dist_G$ in polynomial time. We use \Cref{thm:tracespanningtree} on $W_0$ with terminal set $T$ to get an object frame that is a tree. We have that $\compl(W_0)=O(\compl(\cA(G)))$.
Use \Cref{lem:obj2wire} to obtain the corresponding wireframe $\pl W_0$ that is a plane tree of complexity $\compl(\pl W_0)=O(n^2)\cdot \compl(\cA(G))$.

\begin{figure}
\centering
\includegraphics{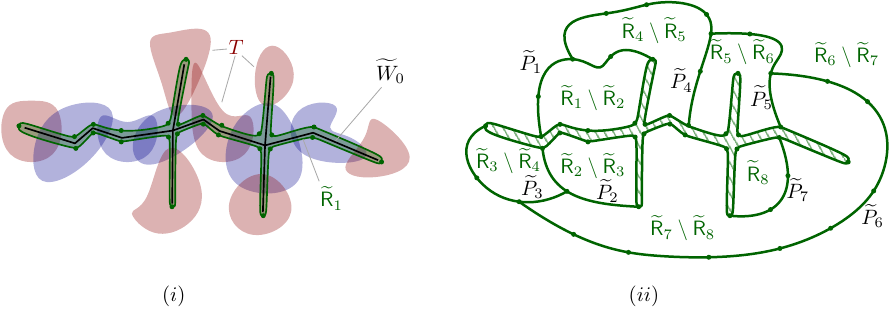}
\caption{(i) The objects of the approximate Steiner tree, terminals are shaded red. In a small neighborhood of the black wireframe $\pl W_0$ we create a wireframe cycle $\wf R_1$. (ii) The unbounded region of $\wf R_1$ is sliced into smaller regions with paths $\pl P_i$. The region bounded by $\wf R_1$ (shaded) is the irrelevant face. }
\label{fig:skeleton_construct}
\end{figure}

We can create a wireframe cycle $\wf{R}_1$ that is realized as a closed curve that goes around the plane drawing of $\gamma(\pl W_0)$ inside an infinitesimally small neighborhood of $\gamma(\pl W_0)$; see \Cref{fig:skeleton_construct}(i).
That is, any object of $V(G)$ intersecting the bounded region of this curve also intersects the unbounded region. When $\dot v\in V(W_0)$ has degree $k$, then the plane cycle $\wf R_1$ has $k$ vertices $\dot v_1,\dots \dot v_k$ corresponding to $\dot v\in V(W_0)$ placed between consecutive edges in a small neighborhood of $\dot v$ in $\pl W_0$. 
Notice that $|V(\wf R_1)| \leq 2|V(W_0)| \leq 4\smt(G,T)$. We define $\re R_1$ to be the unbounded region defined by the closed curve $\gamma(\wf R_1)$, and $F_\textup{irr}$ is the bounded region; it satisfies Property~\ref{it:sk_irrelevant}.
The above operation increases the complexity only by a constant factor compared to that of $\pl W_0$, thus $\compl(\wf R_1)=O(n^2)\cdot \compl(\cA(G))$.

Let $G_1=G|_{\re R_1}$ be our initial intersection graph and $
\wf R_1$ be our initial wireframe constructed above, and we add the vertices of
$p(V(\wf R_1))$ to a set $S$, and furthermore we also add $\wf R_1$ to an auxiliary
\emph{skeleton} $\skel$ that is constructed as a wireframe in $G$ that we will build along the way. Then we will recursively
define shortcuts $P_i$ and wireframe cycles $\wf R_i$ around regions $\re R_i$. See \Cref{fig:skeleton_construct}(ii) for an illustration.

While there is a shortcut $P_i$ between two vertices of $\wf R_i$ in $G|_{\re R_i}$, we find a shortcut  with the smallest reach that satisfies $
(1+\eps)|P_i|< |\wf R_i[u,v]|$,  where $u$ and $v$ are the endpoints of $P_i$. Note that the vertices of $P_i$ form a shortest path in $G|_{\re R_i}$, thus a wireframe path corresponding to $P_i$ can be created using \Cref{lem:pathcutcomplexity} on the object set $V(P_i)$ with terminals $u,v$ and designated points $\dot u \in u$ and $\dot v \in v$ (where $\dot u, \dot v \in V(\wf R_i)$); let $\pl P_i$ be the wireframe corresponding to this object frame.
Notice that since $P_i$ is a shortest path, its vertices induce a path in $G|_{\re R_i}$, and thus there is a one-to-one correspondence between the vertices and edges in $P_i$ and $\pl P_i$.
Adding $\pl P_i$ to $\skel$ generates a new wireframe, where $\re R_i$ is split into two sub-regions by $\gamma(\pl P_i)$, where one region $\re R_{i+1}\subset \re R_i$ does not contain the edges of $\wf R_i[\dot u,\dot v]$. Let $\wf R_{i+1}$ be the wireframe cycle given by the concatenation of $\pl P_i$ and $E(\wf R_i)\setminus E(\wf R_i[u,v])$.
By \Cref{lem:pathcutcomplexity} we have that 
\begin{equation}\label{eq:compl_P_i}
\compl(\wf R_{i+1})\leq  \compl(\wf R_i) + 8|V(\cA(G))\cap(\inter (\re R_{i}\setminus \re R_{i+1}))|.
\end{equation}

The procedure stops when there are no more shortcuts satisfying $(1+\eps)|P_i|<|\wf R_i[u,v]|$ in $G|_{\re R_i}$. Let $G|_{\re R_k}$ be this last face, and let $P_k$ be a shortest path in $G|_{\re R_k}$ connecting two opposite vertices, or if the cycle $\Perim(R_k)$ has odd length, then let $P_k$ be a shortest path connecting two vertices cutting $\Perim(R_k)$ into paths whose lengths differ by $1$. The corresponding path $\pl P_k$ divides $\re R_k$ into two new faces.

Since $(1+\eps)|P_i|<|\wf R_i[u,v]|$ for $i<k$, we get that
\[|\wf R_{i+1}|=|\wf R_i|-|\wf R_i[u,v]|+|P_i|<|\wf R_i|-\eps|P_i|.\]
Let $P_{k-1}$ be the last shortcut added by the procedure. Summing the previous
inequality for all $i\leq k-1$ gives
\[
\eps\sum_{i=1}^{k-1} |P_i| < \sum_{i=1}^{k-1} (|\wf R_i|-|\wf R_{i+1}|)
=|\wf R_1|-|\wf R_{k}|.
\]
Since $|P_k|<|\wf R_{k}|$, we get that $\eps\sum_{i=1}^{k} |P_i|<|\wf R_1|$ and thus 
\begin{equation}\label{eq:skeletonlen}
|\skel|= |\wf R_1|+\sum_{i=1}^k |P_i| = O(|\wf R_1|/\eps),
\end{equation}
which proves \ref{it:sk_northsouthbound}.

To prove the complexity bound on $\cA(\skel)$, notice first that the regions $\re R_{i+1}\setminus \re R_i$ are pairwise interior-disjoint and in particular the sets $V(\cA(G))\cap(\inter (\re R_{i}\setminus \re R_{i+1}))$ are pairwise disjoint subsets of $V(\cA(G))$. Thus iterating \eqref{eq:compl_P_i} implies
\[
\compl(\wf R_j)\leq \compl(\wf R) + \sum_{i=1}^{j-1} 8|V(\cA(G))\cap(\inter (\re R_{i}\setminus \re R_{i+1}))|\leq O(n^2)\cdot \compl(\cA(G)) + 8\cdot \compl(\cA(G))\]
for all $j=1,\dots,k$. By \eqref{eq:skeletonlen} we have that $k=O(n/\eps)$, thus 
\begin{align*}
\compl(\skel)&= \compl(\wf R_1) + \sum_{i=1}^{k-1} \compl(\pl P_i)\\
&< \sum_{i=1}^{k} \compl(\wf R_i)\\
&< k\cdot O(n^2) \cdot \compl(\cA(G))\\
&= O(n^3/\eps)\cdot \compl(\cA(G)),
\end{align*}
which concludes the proof of the complexity bound.

Let $\re F_i=\Clo(\re R_i\setminus \re R_{i+1})$ denote the face of $\skel$ created in step~$i$ of the above iteration, and let $\re F_{k}$ and $\re F_{k+1}$ be the final faces obtained by the cutting of $\re R_{k}$. For $i\leq k-1$ define $\No(\wf F_i)=\pl P_i$, and $\So(\wf F_i)$ to be the path given by the rest of the edges on the face. For the final faces $\re F_{k}$ and $\re F_{k+1}$ we define $\No(\wf F_{k})=\No(\wf F_{k+1})=\pl P_k$, and $\So(\wf F_k)$ and $\So(\wf F_{k+1})$ are the paths formed by the remaining edges on the face $\re F_k$ and $\re F_{k+1}$, respectively.

Observe that \ref{it:sk_compass} 
holds. Let us denote by $\So(F_i)$ and $\No(F_i)$ the paths of $\Perim(G|_{\re F_i})$ corresponding to $\So(\wf F_i)$ and $\No(\wf F_i)$, respectively.

Since $\skel$ was constructed by repeatedly adding shortcuts to an initial cycle, the constructed graph remains $2$-connected throughout the iteration, and in particular, Property~\ref{it:sk_2connected} holds. As the initial cycle contained all terminals Property~\ref{it:sk_terminals} also holds.

Notice that the faces $\re F_1,\dots \re F_{k+1}$ of $\skel$ are created using a path $P_i$ that is a shortest path in $G|_{R_i}$, consequently, it is also shortest in $G|_{F_i}$. Observe that no proper subpath of $\So(\wf F_i)$ admits a shortcut in $R_i$, as in case of $i<k$ that would give a shortcut with a shorter reach than what was chosen, or for $i=k$ it contradicts the non-existence of a shortcut in $\re R_k$. Hence 
\[|\So(\wf F_i)[\dot a, \dot b]|\leq (1+\eps) \dist_{G|_{R_i}} (a,b) \leq (1+\eps) \dist_{G|_{F_i}}(a,b)\]
for any pair $a,b$ on $V(\So(F_i))$ that are not both endpoints of $\So(F_i)$. This concludes the proof of Property~\ref{it:sk_approximatepaths}.
\end{proof}

\begin{proof}[Proof of \Cref{thm:spanner}]
Let $\skel$ be the skeleton for $(G,T)$ constructed according to \Cref{lem:skeletonprops}. 
To construct our spanner $\cS$, for each vertex $v$ of $\skel$ we add $p(\dot v)$ to the set $S$. Next, we use the following procedure on each relevant face $\re R$ of $\skel$. We apply the bipartite spanner construction of \Cref{lem:bipartitespan} on
$V(\No(R))$, and $V(\So(R))$ in $R$, and for each vertex $v$ of this spanner we add $v^\orig$ (i.e., the original object of $V(G)$ that $v$ is a restriction of containing $v$) to $S$.
Once the above procedure is done, the final spanner consists of the cliques of $\cP$ intersecting $S$ or having a vertex within distance $1/\eps$ from the parent of some vertex of $\skel$. Formally, we set $\cS=(V(S)\cup
N_G(p(V(\skel)),1/\eps))_\cP$. It remains to bound $|\cS|$ and to show that it has the desired spanner property.

\paragraph{Bounding $|\cS|$.} 

Since for any vertex $v$ there are $O(r^2)$ cliques within distance $r$
from~$v$, we can bound $(N(p(V(\skel)),1/\eps))_\cP$ by $O(|\skel|/\eps^2)
=O(\smt(G,T)/\eps^3)$ by \eqref{eq:skeletonlen}.

To bound $S_\cP$, we observe that on the cycle corresponding to each face of $\skel$  we applied the bipartite
spanner once (except for the irrelevant face). Consequently, the total length of the cycles on which we
have applied the bipartite spanner is at most $2|\skel|$, and thus
by Property~\ref{it:sk_northsouthbound} we get $|S|=O(|\skel|/\eps^3)=O(|\smt(G,T)|/\eps^4)$.
Therefore $|S_\cP|=O(|\smt(G,T)|/\eps^4)$ and $|\cS|=O(|\smt(G,T)|/\eps^4)$, as required.

\paragraph{The spanning property for $\cS$.} Consider a shortest path $Q$ in $G$ connecting two vertices of $T$. Notice that if $V(Q)\subseteq N_G(p(V(\skel)),1/\eps)$, then clearly $Q$ is contained in $G[\bigcup \cS]$ and the claim follows. Suppose the contrary, that there is some vertex $x\in V(Q)$ outside $N_G(p(V(\skel)),1/\eps)$. Notice that $x$ must be outside $F_\textup{irr}$ as all of its vertices are in the $1$-neighborhood of $p(\skel)$. Let $\re F$
be the relevant face of $\skel$ that contains $x$, and let $F=G|_{\re F}$. We find a wireframe $\pl Q$ corresponding to the shortest path $Q$ using \Cref{cor:shortestpathframe}, and let $\dot x$ be the vertex of $\pl Q$ whose parent is $x$. Fix an orientation of $Q$, and let $u\in V(Q)$ be the last vertex intersecting $\bd \re F$ before $Q$ arrives in $x$. Let $u'\in \Perim(F)$ be some neighbor of $u$ (more precisely, some vertex of $N_B(u\cap \re F)$). Such a vertex exists since $\Perim(F)$ covers $\bd \re F$. Analogously, let $v\in V(Q)$ be the first vertex on $Q$ after $x$ whose object intersects $\bd \re F$, and let $v'\in \Perim(F)\cap N_B(v\cap \re F)$.
Observe that $u'$ is either a neighbor of $u$ or $u'$ is a connected component of $u\cap \re F$. The analogous observation holds for $v'$ and $v$.

It is sufficient to show that we can replace $Q[u,v]$ with a path $Q'$ in
$G[\bigcup \cS]$ such that $|Q'|\leq
(1+O(\eps))|Q[u,v]|$. Since the internal vertices of $Q[u,v]$ are inside the skeleton face, they cannot be terminals, thus $Q[u,v]$ is a shortest path, and in particular, $|Q[u,v]|=\dist_G(u,v)$.
Let $Q'=(u,u')\cup P \cup (v',v)$, where $P$ is the shortest
$u' \to v'$ path in $G[\bigcup \cS]|_{\re F}$. 
We claim that $|Q'|\leq (1+\eps)\dist_G(u,v)+4$.  We distinguish three cases.

\begin{description}
\item[Case 1.] Both $u'$ and $v'$ are in $\No(F)$.\\ Then by \ref{it:sk_approximatepaths} we know that $\No(F)$ is a shortest path in $F$. Thus the triangle inequality implies 
\[|Q'|\leq |\No(F)[u',v']|+ 2=\dist_F(u',v')+2\leq |Q| + 4 = \dist_G(u, v) +4,\]
where the second inequality uses that the internal vertices of $Q$ are inside $F$.
\item[Case 2.] Both $u'$ and $v'$ are outside $\No(F)$, i.e., they are internal vertices of $\So(F)$.\\
Consider the path $\bar{Q}:=(u',u)\cup Q[u,v] \cup (v,v')$. By \ref{it:sk_approximatepaths} we have that $\So(F)$ is a internally $(1+\eps)$-approximate shortest path in $F$, thus $(1+\eps)|\bar{Q}| >|\So(F)[u',v']|$.
Therefore 
\[|Q'|\leq |\So(F)[u',v']| < (1+\eps)|\bar{Q}| = (1+\eps)(\dist_G(u,v)+2) < (1+\eps)\dist_G(u,v)+4.\]

\item[Case 3.] Exactly one of $u'$ and $v'$ lies in $\No(F)$.\\
 By the spanning property of the bipartite spanner, we
have that \[|Q'|\leq (1+\eps)\dist_G(u',v')\leq (1+\eps)(\dist_G
(u,v)+2)\leq (1+\eps)\dist_G(u,v)+4.\]
\end{description}

Recall that $\dist_G(u,v)\geq 2/\eps-2$ since $x$ is on the
shortest path from $u$ to~$v$ and has distance at least $1/\eps$ from both
$u'$ and~$v'$. Since $\eps<0.5$ implies $\dist_G(u,v)>1/\eps$, we get
\[|Q'|\leq (1+\eps)\dist_G(u,v)+4 < (1+5\eps)\dist_G(u,v).\]

This concludes the proof.
\end{proof}

\subsection{Algorithm for Subset TSP}

It can be shown that contracting an edge can decrease the optimum value of Subset TSP only by at most 2: a closed walk in the modified instance can be extended to a closed walk of the original graph by traversing the contracted edge at most twice. Thus consolidating a set $X$ of vertices decreases the optimum by at most $2|X|$. To obtain more efficient algorithms, we show that contracting a bounded number of cliques has only a bounded effect on the value of the optimum solution.

\begin{lemma}\label{lem:TSPUncontractCost}
Let $T\subset V(G)$ be a set of terminals and $X\subset V(G)$ an arbitrary
set, and fix a clique partition $\cP$ of $G$. Then $\tsp(G,T)\leq \tsp
(G\vcon X,T_X)+ |T\cap X|+ O(|\cX|)$, where $T_X=V(G\vcon X)|_T$ is the set of vertices in $G\vcon X$ whose
preimage in $G$ contains a terminal from $T$, and $\cX$ is the set of cliques
in $\cP$ that contain some vertex of $X$. Moreover, given an optimum tour for
$(G\vcon X,T_X)$, a tour of length $\tsp(G,T) + O(|\cX|)$ that is feasible for $
(G,T)$ can be constructed in $\poly(|V(G)|)$ time.
\end{lemma}

\begin{proof}
Let $W$ be an optimum subset TSP walk for the instance $(G\vcon X,T_X)$, and fix a
cyclic orientation on $W$. We will show how to build a subset TSP walk $W_G$ for $(G,T)$ that is potentially longer than required, and shorten it later. First, we add all arcs of $W$ to $W_G$ where
at least one endpoint is outside $X$. Consider a connected component $Y$
induced by $X$ and let $z\in V(G\vcon X)$ be the vertex that results from
contracting this component; suppose moreover that $z$ appears on $W$. Then
for each consecutive arc pair $uz$ and $zv$ of $W$ we add a shortest
$u\rightarrow v$ path to $W_G$ whose internal vertices are in $Y$.
Additionally, for any terminal $t\in T$ that appears in $Y$, we extend the
tour $W_G$ with a path that goes to $t$ and back from some already visited
point in $Y$. After we do this for all consecutive arc pairs of $W$ whose
shared vertex is $z$, the walk $W_G$ correctly traverses $Y$. We repeat the
above procedure for all connected components of $G[X]$ whose contraction
appears on $W$. Notice that the resulting tour $W_G$ is a feasible subset TSP
tour of $(G,T)$.

We now shorten the walk $W_G$. We can apply the simplifications described in the proof of
\Cref{lem:subsetTSPsolTerminals} to ensure that in each clique of $\cX$
the new walk $W'_G$ contains at most $O(1)$ non-terminal vertices, and each of them at most $O(1)$ times. Consequently,
\begin{equation}\label{eq:W'bound}
|W'_G|\leq \tsp(G\vcon X,T_X)+|T\cap X|+O(|\cX|).
\end{equation}

Let $W^*$ be an optimum tour of $(G,T)$. We claim that \[\tsp(G,T) = |W^*| \geq \tsp(G\vcon X,T_X)+|T\cap X| -O(|\cX|).\]
The techniques of \Cref{lem:subsetTSPsolTerminals} imply that without loss of generality, $W^*$ is a tour where from each clique $C$ of $\cP$ the tour contains $O(1)$ non-terminal vertices, and each with $O(1)$ multiplicity. Moreover, all but $O(1)$ among the terminals $(T \cap X) \cap C$ appear consecutively on $W^*$. Now we contract $X$, and as a result, get a walk $W^*\vcon X$ that is a feasible solution to $(G\vcon X,T_X)$. Since in each clique $C$ intersecting $T\cap X$ we have contracted at least $|(T\cap X)\cap C|-O(1)$ edges, we have
\[\tsp(G\vcon X,T_X) \leq |W^*\vcon X| \leq |W^*| - \sum_{C\in \cP} \big(|T\cap X \cap C| + O(1)\big) \leq |W^*| - |T\cap X| + O(|\cX|),\]
as claimed. Substituting in~\eqref{eq:W'bound} now yields the desired bound:
\[|W'_G|\leq \tsp(G\vcon X,T_X)+|T\cap X|+O(|\cX|) \leq \tsp(G,T) + O(|\cX|).\]
Moreover, the steps described above can be achieved
within a running time that is polynomial in $|V(G)|$.
\end{proof}

We are now ready to present our approximation scheme for \stsp.

\TSPalg*

\begin{proof}
We apply \Cref{thm:atalakitas} to obtain an $\alpha$-standard graph $G\in \cI_\alpha$ whose representation has polynomial complexity. Let $\cP$ be the corresponding cell partition.
By \Cref{lem:sparseTSP} we may assume without loss of generality that
for each $C\in \cP$ we have $|C\cap T|\leq 1/\eps$ and $|C|\leq 1/\eps^6$. 
Let $\opt$ denote the length of the
optimal tour for $(G,T)$. Let us now apply
\Cref{thm:spanner} to get a spanner consisting of the
cliques $\cS\subseteq \cP$, with $|\cS|=O(\frac{1}{\eps^4}|V(A)|)=O(\frac{1}
{\eps^4})\cdot \opt$. This therefore gives the desired spanner~$G'$ of size $\frac{1}{\eps^{10}}\cdot \opt$.

If we denote by $\bigcup \cS$ the set of vertices in
$\cS$, then by the bound on the size of cliques we get that
$|\bigcup\cS|\leq \opt/\eps^{10}$. By the spanner property, it is sufficient
to compute a $(1+ O(\eps))$-approximate tour in the spanner $G^*\eqdef
G[\bigcup \cS]$ for the terminal set~$T$.

We apply
\Cref{thm:contractdecomp} on $G^*$ and the cell partition $\cS$ for $k=1/\eps^5$, which yields a collection $\cX_1,\dots,\cX_k\subset \cS$ and corresponding vertex sets $X_1,\dots,X_k$. Note that $\ell=\max_{C\in \cS} |C|\leq 1/\eps^6$.

Let us now compute an
optimum tour for each of the instances $(G^*\vcon X_i,T_{X_i})$, where $T_
{X_i}=V(G^*\vcon X_i)|_T$. By \Cref{thm:contractdecomp} each $X_i$ satisfies $\tw(G^*\vcon X_i)= O(kl)= O(1/\eps^{11})$. We can therefore apply
the $2^{O(\tw)} \poly(n)$ algorithm~\cite{single-exponential} to solve these instances in $2^
{O(1/\eps^{11})}\poly(n)$ time. To get a faster algorithm, we can use that \Cref{thm:contractdecomp}(\textit{iii}) gives $\tw_{\cS\vcon X_i}(G^*\vcon X_i)=O(kl)=O\big(
(1/\eps^5) \cdot \log(1/\eps)\big)$, so the framework of~\cite{BergBKMZ20} yields a running time of $2^{O(1/\eps^5\log(1/\eps))}\poly(n)$. Note that $\tsp(G^*\vcon X_i,T_{X_i})\leq \opt^*$ because contractions can only decrease the optimum cost of
a subset TSP instance.

By \Cref{lem:TSPUncontractCost} in polynomial time we can construct tours
$\tau_i$ for $(G^*,T)$ of length at most
\[\tsp(G^*,T)+O(|\cX_i|)
\leq \opt^*+O(|\cX_i|).\]
Since \Cref{thm:contractdecomp} guarantees $\sum_{i=1}^k |\cX_i|=O(|\cS|)$, there exists some $i$ where $|\cX_i|=O
(|\cS|/k)= O(\eps)\cdot\opt$. As a result, the shortest tour among the tours
$\tau_i$ has length at most $\opt^*+O(\eps)\cdot\opt=(1+O(\eps))\cdot\opt$.
\end{proof}

\section{The mortar graph}
\label{sec:mortar}

This section aims to construct a so-called mortar graph for a given intersection graph $G$ and terminal set $T$. Mortar graphs were introduced by Borradaile~\etal~\cite{BorradaileKM09} for the planar Steiner tree problem. We also only need them for our Steiner tree algorithms. Our construction is much more involved as we need to avoid unwanted crossings between columns.

\begin{definition}[Mortar graph and bricks] A $2$-connected wireframe $\pl M$ of $G$ is a \emph{mortar graph} if $\pl M$ has a unique \emph{irrelevant} face with region $\re F_\textup{irr}\subset \Reals^2$, such that any object of $V(G)$ intersecting $\re F_\textup{irr}$ also intersects $\Reals^2 \setminus \re F_\textup{irr}$. We call all other faces of $\pl M$ (as closed subsets of $\Reals^2$) \emph{brick regions}, and each brick region $\re B$ defines a corresponding \emph{brick} $B:=G|_\re B$.
\end{definition}

\begin{definition}[Compass]
Given a mortar graph $\pl M$, a \emph{compass} of $\pl M$ is a partition of the boundary cycle $\wf B$ for each brick region $\re B$ into six paths: \[\So(\wf B), \Ea^-(\wf B), \Ea^|(\wf B), \No(\wf B), \We^|(\wf B), \We^-(\wf B),\] so that these subpaths of $\wf B$ are in anticlockwise order if $\re B$ is a bounded brick region and in clockwise order if $\re B$ is the unique unbounded brick region. We further define 
\[\We(\wf B):= \We^-(\wf B)\cup\We^|(\wf B)\quad \text{ and } \quad
\Ea(\wf B):=\Ea^-(\wf B)\cup\Ea^|(\wf B).\]
\end{definition}

For a mortar graph with a compass and a brick $B$, let \[ \dot v^B_{SSE}, \dot v^B_{SE},\dot v^B_{NE},\dot v^B_{NW},\dot v^B_{SW},\dot v^B_{SSW}\] denote the shared vertices among the consecutive paths 
\[\So(\wf B), \Ea^-(\wf B), \Ea^|(\wf B), \No(\wf B), \We^|(\wf B), \We^-(\wf B),\] where $\dot v^B_{SSW}$ is shared between $\We^-(\wf B)$ and $\So(\wf B)$. We set $v^B_{SSE}:=p|_B(\dot v^B_{SSE})$ to denote the object of $B$ corresponding to $\dot v^B_{SSE}$, and we define the objects $v^B_{SE}$, $v^B_{NE}$, $v^B_{NW}$, $v^B_{SW}$, and $v^B_{SSW}$ analogously.

An important property we will need of the columns is that they behave similarly to shortest paths toward the northern path of their bricks. Let $G$ be a graph, and suppose that $P$ is a path and $Q$ is a path or vertex set that contains exactly one endpoint $y$ of $P$. We say that $P$ is \emph{$c$-approaching} $Q$ if for any vertex $x$ of $P$ we have that $|P[x,y]|\leq c\cdot \dist_G(x,Q)$.

Recall that the object duplication introduced on Page~\pageref{sec:objectduplication} means that we can use the correspondence of $\wf B$ with $\Perim(B)$ to define  paths in $B$ corresponding to the compass of $\wf B$: let 
$\So(B)$, $\Ea^-(B)$, $\Ea^|(B)$, $\No(B)$, $\We^|(B)$, $\We^-(B)$, as well as $\Ea(B)$ and $\We(B)$ be these paths.

Given $\eps>0$, we set $\kappa=1/\eps^3$ and $\lambda=1/\eps^{7.5}$. Let $\re B$ be a brick region with corresponding graph $B$. If $B$ has no vertex $v$ where $\dist(v,\Perim(B))>\lambda$ and for each $v\in V(B)$ where $\dist_B(v,\So(B))\leq \lambda$ and $\dist_B(v,\No(B))\leq \lambda$ it holds that $\dist_B(v,\{v^B_{SE},v^B_{SW}\})\leq 2\lambda/\eps^2$, then $\re B$ and $B$ are called \emph{small}, otherwise they are called \emph{large}.
Our goal for the rest of this section is to construct a mortar graph with the following properties.

\begin{theorem}[Mortar graph properties]\label{thm:mortarprops}
Given the intersection graph $G$, a terminal set $T$, and a number $\eps\in (0,0.1)$, we can construct a mortar graph $\pl M$ of complexity $\compl(\pl M)=O(n^4/\eps^2)\cdot \compl(\cA(G))$ with a compass such that the following properties hold.
\begin{enumerate}[label=\textbf{M\arabic*}]
\item For each $v\in T$ the mortar graph $\pl M$ contains at least one vertex $\dot v \in V(\wf B)$ with $p(\dot v) = v$ for some brick region $\re B$ of $\pl M$.\label{it:M_terminals}
\item In each brick $B$ the path $\No(B)$ is a shortest path in $B$, and $\We^-(B) \cup \So(B)\cup \Ea^-(B)$ is a internally $(1+\eps)$-approximate shortest path in $B$. \label{it:M_approximatepaths}
\item There is a constant $c$ depending only on $\alpha$ such that in each brick $B$ the path $\We^|(B)$ and $\Ea^|(B)$ are $c$-approaching $\No(B)$. Moreover, $\We^|(B)$ and $\Ea^|(B)$ are internally $c$-approximate shortest paths in~$B$.\label{it:M_eastwestsuffixes}
\item The total length of the paths $\No(B), \So(B), \We^-(B), \Ea^-(B)$ for all bricks $B$ is $O(1/\eps)\smt(G,T)$.\label{it:M_northsouthbound}
\item The total length of the eastern and western paths of all bricks is $\bigcup_B |\Ea(B)\cup \We(B)|=O(\eps)\smt(G,T)$.\label{it:M_eastwestbound}
\item For each large brick $B$ and for some $k\leq \kappa$ there is a set of vertices $s_0,\dots,s_k$ ordered from west to east on $\So(B)$ such that for any vertex $x$ of $\So(B)[s_i,s_{i+1})$ we have $|\So(B)[s_i,x]|\leq \eps \dist_B(x,\No(B))$.\label{it:M_thickbrick}
\end{enumerate}
Given $G,T$, and $\eps$, such a mortar graph $\pl M$ and its compass can be computed in polynomial time.
\end{theorem}

We start proving \Cref{thm:mortarprops} by invoking \Cref{lem:skeletonprops} to construct a skeleton $\skel$. Consider now a relevant face of $\skel$ with region  $\re R\subset \Reals^2$, and let $\wf R$ be the cycle of $\skel$ bounding this face, consisting of north and south boundary paths $\No(\wf R)$ and $\So(\wf R)$, respectively. We will define so-called \emph{bridges} and \emph{columns} on the graph $R:=G|_{\re R}$ to slice the region $\re R$ into smaller brick regions. Let $\No(R)$ denote the walk in $R$ corresponding to $\No(\wf R)$, that is, $\No(R)$ is the subpath of $\Perim(R)$ that corresponds to $\No(\wf R)$. We define the path $\So(R)$ analogously based on $\So(\wf R)$. Recall that the northern paths are always shortest paths in $R$ by \ref{it:sk_approximatepaths}.

\subsection{Defining bridges}
A shortest path of $R$ is a \emph{bridge} if it is a shortest path of length at most $7$ connecting $\gamma(\So(\wf R))$ to $\gamma(\No(\wf R))$, that is, a
shortest path in $R$ with vertices $v_1,\dots,v_k$ where $k\leq 8$, $v_1$ intersects $\gamma(\So(\wf R))$, and $v_k$ intersects $\gamma(\No(\wf R))$. We note that the parents of the endpoints of $\No(\wf R)$ and $\So(\wf R)$ are
considered to be length-$0$ bridges. Consider a bridge $P$. Since it intersects the northern and southern boundary of $\re R$, its first vertex must also
intersect some $u\in V(\No(R))$ and its last vertex must intersect some $u'\in V(\So(R))$ (where $v_1=u$ and $v_k=u'$ are possible), creating a path of length at most $9$ connecting $u$ and $u'$. Let $\dot u, \dot u'$ be the vertices of
$\wf R$ such that $p(\dot u)=u$ and $p(\dot u')=u'$. We invoke \Cref{cor:shortestpathframe} to make a shortest path connecting $u$ and $u'$ as terminals such that the corresponding objects of the object frame contain $\dot u$ and $\dot u'$. The corresponding wireframe is a \emph{bridge path} connecting
$\dot u$ and $\dot u'$, and it splits the region $\re R$ into two parts: a western and eastern subregion. (Note that the wireframe $\wf R$ could now be extended with the edges and new vertices of the bridge path.)

Next, we define a greedy collection of bridges as follows. Consider the vertices of $\So(R)$ in a west-to-east order, denoting its first and last vertex by $v$ and $w$, and let $P_1$ be the trivial bridge made up by $v$. Suppose that $P_i$ is a bridge already defined, then let $P_{i+1}$ be a bridge whose starting point is the first on $\So(R)[v,w]$ whose objects are disjoint from the objects of $P_i$. This defines some greedy collection of pairwise disjoint bridges indexed according to their west-to-east order with the property that if $u\in \So(R)$ is between the starting points of $P_i$ and $P_{i+1}$, then any bridge starting at $u$ has an object intersecting either some object of $P_i$ or some object of $P_{i+1}$. For each of these bridges let us fix corresponding bridge paths $\pl P_i$, which subdivides $\re R$ into a sequence of smaller regions. In such a smaller region $R'$ we  define $\So(\wf R')$ and $\No(\wf R')$ to be the parts of $\wf R'$ that fall on $\So(\wf R)$ and $\No(\wf R)$, respectively.

A vertex $v$ of $R'=G|_{\re R'}$ is called \emph{boundary-distant} if $\dist_{R'}(v,V(\So(R')\cup \No(R'))>\lambda$. (Recall that $\lambda=1/\eps^{7.5}$). Consider now two consecutive small regions $\re R'$ and $\re R''$ divided by some bridge path $\pl P_i$. If neither of $R'$ and $R''$ have a boundary-distant vertex, then we remove the bridge path $\pl P_i$. We claim that the resulting region $\re R'\cup \re R''$ still does not contain a boundary-distant vertex. First, if a vertex $v\in V(R')$ disjoint from $\gamma(\pl P_i)$ has distance $k_v$ to $V(\So(R')\cup \No(R')$ in $R'$, then its distance to the same set of vertices is not larger in $R''':=G|_{\re R'\cup \re R''}$, and notice that $V(\So(R')\cup \No(R')\subset V(\So(R''')\cup \No(R''')$. The analogous property holds for any $v\in V(R'')$ disjoint from $\gamma(\pl P_i)$. Thus if $R'''$ has a boundary-distant vertex $v$, then it must intersect $\gamma(\pl P_i)$. This however implies that $v\in N_{R'''}p(\pl P_i)$ and thus $v$ is within distance $10$ from the starting vertex of $P_i$ in $R'''$. Since $\lambda>10$, we conclude that $R'''$ cannot have a boundary-distant vertex.

We continue eliminating bridges between consecutive regions lacking a boundary-distant vertex. As a result, we have a collection of regions $\re L_1,\dots, \re L_t$ indexed in west-to-east order such that the following properties hold.
\begin{enumerate}[label=\textbf{Bdg\arabic*},leftmargin=1.3cm]
\item If $L_i=G|_{\re L_i}$ has a boundary-distant vertex, then any bridge in $L_i$ intersects some object of the eastern or western bridge of $R_i$.\label{it:bridge_willdivert}
\item For any $i\in [t-1]$ at least one among $L_i$ and $L_{i+1}$ contains a boundary-distant vertex.\label{it:bridge_lensalter}
\end{enumerate}
 
\begin{lemma}\label{lem:bridgecount}
If $\re R$ is decomposed with $t-1$ bridges into sub-regions $\re L_1,\dots,\re L_t$ as described above, then the total length of the bridges separating them is $O(t)=O\left(\frac{|\wf R|}{\sqrt{\lambda}}\right) = O(\eps^2|\So(\wf R)|)$.
\end{lemma}

\begin{proof}
Let us denote by $\No(\wf L_i)$ the portion of $\wf L_i$ that falls on $\No(\wf R)$. We define $\So(\wf L_i)$ analogously. Let $L_i=G|_{\re L_i}$, and the  paths of $\Perim(L_i)$ corresponding to $\No(\wf L_i)$ and $\So(\wf L_i)$ are denoted by $\No(L_i)$ and $\So(L_i)$, respectively.

If there are no boundary-distant vertices in $\re R$, then $t=1$ and the claim holds; suppose now that this is not the case.
Notice that there are at least $\lfloor t/2 \rfloor$ regions among $\re L_i$ that have boundary-distant vertices. We claim that each region $\re L_i$ with a boundary-distant vertex has two points on its boundary of (Euclidean) distance at least $\Omega(\sqrt{\lambda})$. Indeed, if $v$ is boundary-distant in $L_i$, then let $P$ be a shortest path from $v$ to $V(\No(L_i)\cup\So(L_i))$. Then the vertices that are at an odd integer distance from $v$ among the internal vertices of $P$ are pairwise disjoint, thus by fatness they must cover an area of at least $\Omega(\lambda)$. All of these objects are contained in the bounded area defined by the curve $\bd \re L_i$, thus by the isodiametric inequality of the Euclidean plane we have that $\re L_i$ has diameter at least $\Omega(\sqrt{\lambda})$, as claimed.

Recall that objects of $\Perim(L_i)$ are subpolygons of objects of $V(G)$, thus each object has diameter at most $2\alpha=O(1)$. Moreover, the bridges on the east and west of $L_i$ contribute at most $O(1)$ edges to $\wf L_i$, so the diameter of a bridge path is at most $\beta = O(\alpha)=O(1)$.
Since the vertices of $\No(\wf L_i)$ and $\So(\wf L_i)$ correspond to disjoint objects of $\Perim(L_i)$, each of diameter at most $2\alpha$, and by the definition of wireframes these objects cover $\No(\wf L_i)\cup \So(\wf L_i)$, we have that there must be at least $\frac{\diam(\bd L_i)-2\beta}{2\alpha}=\Omega(\sqrt{\lambda})$ vertices on $\No(\wf L_i) \cup \So(wf(L_i))$. Let $I_{bdv}$ denote the set of indices $i$ where $\re L_i$ has a boundary-distant vertex. We have that $|I_{bdv}|\geq \lfloor t/2 \rfloor$. The total bridge length is $O(t)$, while $|\wf R|$ can be bounded as follows:
\[|\wf R|=|\No(\wf R)\cup \So(\wf R)|>\sum_{i\in I_{bdv}}|\No(\wf L_i)\cup \So(\wf L_i)|=\lfloor t/2 \rfloor\cdot \Omega(\sqrt{\lambda}).\]
Thus we have $O(t)=O\left(\frac{|E(\wf R)|}{\sqrt{\lambda}}\right)$.

Recall that $|E(\wf R)|=|\No(\wf R)|+|\So(\wf R)|$ and denoting $R=G|_{\re R}$ we have $|\No(\wf R)|=|\No(R)|$ and $|\So(\wf R)|=|\So(R)|$ by definition. Moreover, $\No(R)$ is a shortest path in~$R$ by~\ref{it:sk_approximatepaths}, while $\So(R)$ is a path between the same endpoints in $R$. Therefore $|\So(\wf R)|\geq |\No(\wf R)|$, and consequently $E(|\So(\wf R)|)\geq |E(\wf R)|/2$. Now substituting $\lambda=1/\eps^{7.5}$ yields that
\[O(t)=O\left(\frac{|\wf R|}{\sqrt{\lambda}}\right)=O\left(\eps^{3.75}|\So(\wf R)|\right).\qedhere\]
\end{proof}

\subsection{Region slicing with columns}

Consider now some region $\re L=\re L_i$ with the wireframe $\wf L$ surrounding $\re L$, which consists of a part of $\No(\wf R)$ denoted by $\No(\wf L)$, a part of $\So(\wf R)$ denoted by $\So(\wf L)$, and the western and eastern bridge paths denoted by $\We(\wf L)$ and $\Ea(\wf L)$, respectively, so that the anticlockwise\footnote{when $\re L$ is the unbounded face, then clockwise} order traverses the paths in the order $\So(\wf L),\Ea(\wf L),\No(\wf L),\We(\wf L)$. We define $L=G|_{\re L}$ and define the subpaths $\So(L),\Ea(L),\No(L),\We(L)$ from $\Perim(L)$ as the paths corresponding to $\So(\wf L)$, $\Ea(\wf L)$, $\No(\wf L)$, $\We(\wf L)$, respectively. Note that $\No(L),\Ea(L),\We(L)$ are in fact shortest paths by \ref{it:sk_approximatepaths} and by the definition of bridges. A region $\re L$ that contains a boundary-distant vertex is called a \emph{lens}, but first we will slice non-lens regions into small bricks as follows.

\subparagraph*{Slicing non-lens regions into small bricks}

Consider a region $\re L$ with graph $L=G|_{\re L}$ that does not contain any boundary-distant vertex. Recall that a brick region $\re B$ with a compass is small if $B=G|_{\re B}$ has no vertex $v$ where $\dist_B(v,\Perim(B))>\lambda$ and for each $v\in V(B)$ where $\dist_B(v,\So(B))\leq \lambda$ and $\dist_B(v,\No(B))\leq \lambda$, it holds that $\dist_B(v,\{v^B_{SE},v^B_{SW}\})\leq 2\lambda/\eps^2$. We prove the following lemma.

\begin{lemma}\label{lem:slice_into_small_bricks}
Let $\re L$ be a region where for each $v\in L$ it holds that $\dist_L(v,\Perim(L))\leq \lambda$. Then there is a collection of vertex-disjoint wireframe paths $\pl P_0=\We(\wf L), \pl P_1,\pl P_2,\dots, \pl P_k=\Ea(\wf L)$ from $\So(\wf L)$ to $\No(\wf L)$, indexed in their West-to-East order such that the brick regions $\re B_i$ enclosed between $\pl P_{i-1}$ and $\pl P_i$ (for $i=1$ to $k-1$) can be assigned a compass with the following properties:
\begin{itemize}
\item $\So(\wf L)$ is the concatenation of $\So(\wf B_1),\dots,\So(\wf B_k)$, and $\No(\wf L)$ is the concatenation of $\No(\wf B_1),\dots,\No(\wf B_k)$.
\item $\We(\wf B_0)=\We(\wf L)$, $\Ea(\wf B_k)=\Ea(\wf L)$, and for each $i\in [k]$ we have $\Ea(\wf B_{i-1}) = \We(\wf B_i)$.
\item  Each $\re B_i$ is small.
\item The paths $\pl P_i$ for $i=1$ to $k-1$ have total length $O(\eps^2|\So(L)|)$.
\end{itemize}
Moreover, given $G$ and $\wf L$, the paths $\pl P_i$ (and the bricks $B_i$) can be computed in polynomial time.
\end{lemma}

\begin{proof}
We will greedily build a collection of paths $\cL\cP$ connecting some vertices of $\So(L)$ to some vertex of $\No(L)$; they will later be used to slice $\re L$ into small bricks. Initially set $\cL\cP=\{\We(L),\Ea(L)\}$, and for a path $P\in \cL\cP$ let $s_P$ denote its starting object on $\So(L)$ and $t_P$ its ending object on $\No(L)$. We will maintain throughout the procedure that (i) the objects of distinct paths of $\cL\cP$ are pairwise disjoint and their staring points on $\So(L)$ have distance at least $\lambda/\eps^2$, and (\textit{ii}) for each $P\in \cL\cP$ we have $|P|\leq 2\lambda$. Both conditions hold initially: recall that $\We(L)$ and $\Ea(L)$ are built from bridges and thus have length at most $9$.

Suppose that $v\in V(L)$ is within distance $\lambda$ from both $\So(L)$ and $\No(L)$, but its distance is at least $\lambda/\eps^2 + \lambda$ from every vertex $\{s_P\mid P\in \cL\cP\}$. Let $s\in \So(L)$ and $t\in \No(L)$ be endpoints of shortest paths from $v$ to $\So(L)$ and $\No(L)$, respectively. Consequently, $\dist_L(v,s)\leq \lambda$ and $\dist_L(s,t)\leq 2\lambda$. Let $Q$ be a shortest path in $L$ connecting $s_Q=s$ to $t_Q=t$. Notice that $s_Q$ is within distance $\lambda$ from $v$, and $\dist_L(s_Q,t_Q)\leq 2\lambda$ implies that $|Q|\leq 2\lambda$. Let $P\in \cL\cP$ be arbitrary. Then by the triangle inequality
\[\dist_L(s_Q,s_P)\geq \dist_L(s_P,v) - \dist_L(v,s_Q) \geq \lambda/\eps^2 +\lambda - \lambda = \lambda/\eps^2.\]
Notice that since both $P$ and $Q$ are of length at most $2\lambda$ and their starting points have distance at least $\lambda/\eps^2$, they must be disjoint.
Thus the objects of $Q$ are disjoint from the objects in paths of $\cL\cP$, and we can insert $Q$ into $\cL\cP$ and maintain properties (i) and (\textit{ii}).

The above greedy procedure can be applied on a given $v\in V(L)$ at most once and can be executed in polynomial time. Since $L$ has less than $n$ vertices, the procedure terminates in polynomial time. Let $\cL\cP$ denote the obtained path collection. Let $k=|\cL\cP|-1$ and index the paths of $\cL\cP$ from $0$ to $k$ according to their starting points in $\So(L)$ in West-to-East order. We apply~\Cref{cor:shortestpathframe} on each path $P_i$ and the fixed points $\dot s_{P_i}\in \So(\wf L)$ and $\dot t_{P_i}\in \No(\wf L)$ to gain pairwise disjoint wireframe paths $\pl P_i$ that can extend the wireframe $\wf L$.

Given the paths $\pl P_i$ we can define the regions $\re B_i$ for $i=1,\dots,k$ Define
\begin{align*}
\So(\wf B_i)&:=\So(\wf L)\cap \wf B_i,\\
\No(\wf B_i)&:=\No(\wf L)\cap \wf B_i,\\
\We(\wf B_i)&:=\We^|(\wf B_i):=\pl P_{i-1},\\
\Ea(\wf B_i)&:=\Ea^|(\wf B_i):=\pl P_i.
\end{align*}
The path $\We^-(\wf B_i)$ is defined to be a single-vertex edgeless path given by the shared vertex of $\So(\wf B_i)$ and $\We(\wf B_i)$.
We define $\Ea^-(\wf B_i)$ analogously as the shared vertex of $\So(\wf B_i)$ and $\Ea(\wf B_i)$.

The total length of the paths is at most $(k-1)\cdot 2\lambda<2k\lambda$. On the other hand the distance between $s_{P_i}$ and $s_{P_{i+1}}$ is at least $\lambda/\eps^2$ fore each $i=0,\dots,k-1$, thus the triangle inequality gives
\[|\So(L)|\geq \sum_{i=0}^{k-1} \dist_L(s_{P_i},s_{P_{i+1}})\geq k\cdot \lambda/\eps^2.\] 
Consequently, the total length of the paths in question is at most $2k\lambda \leq 2\eps^2|\So(L)|$, as required.
It remains to show that the region $\re B_i$ is small. 

We will start with the following claim.
\begin{claim}\label{cl:LdisttoBdist}
For any $v \in V(B_i)$ and any $j\in \{0,1,\dots,k\}$ we have $\dist_{B_i}(v,v^{B_i}_{SW})\leq \dist_{L}(v_L,s_{P_{j}}) + 2\lambda + 2$ if $j<i$ and $\dist_{B_i}(v,v^{B_i}_{SE})\leq \dist_{L}(v_L,s_{P_{j}}) + 2\lambda + 2$ if $j\geq i$, where $v_L$ is the object of $V(L)$ that creates the subpolygon $v$ when restricted to $\re B_i$.
\end{claim}

\begin{claimproof}
Recall that $v^{B_i}_{SE}$ is a subpolygon of $s_{P_{i}}$ and $v^{B_i}_{SW}$ is a subpolygon of $s_{P_{i-1}}$. We will show that $\dist_{B}(v,v^{B_i}_{SE})\leq \dist_{L}(v,s_{P_j}) +2\lambda + 1$ when $j\geq i$; the other statement can be proven analogously.

Let $Q$ be a shortest path from $v$ to $s_{P_j}$ in $L$, and let $w$ be the first vertex of $Q$ (starting from $v$) whose object intersects the curve $\gamma(\pl P_i)$. (Such an object exists as $j\geq i$ implies that either $j=i$ and $s_{P_j}$ thus intersects the curve, or $j>i$ and the curve separates $v$ from $s_{P_j}$ inside $\re L$.

Then there is some $w'\in V(B_i)$ such that $w'$ is a  subpolygon of $w$ and there is a path $Q_B$ from $v$ to $w'$ of length $|Q[v,w]|$. Moreover, $w'$ is a neighbor of some object $x'\in V(\We(B_i))$. Let $x\in V(L)$ be the polygon that $x'$ is a subpolygon of. Note that $\We(B_i)$ is a shortest path in $B_i$ as $P_i$ is a shortest path in $L$ and by \Cref{lem:restrictionmetric}(\textit{ii}) with $c=1$ implies the same for $\We(B_i)$ in $B_i$.
Consequently, \[\dist_{B_i}(x',v^{B_i}_{SE})=|\We(B_i)[x',v^{B_i}_{SE}]|=|P_i[x,s_{P_i}]|=\dist_L(x,s_{P_i}).\]
Consider the concatenation of $Q[v,w']$, the edge $w'x'$, and the path $\We(B_i)[x',v^{B_i}_{SE}]$.
Thus by the triangle inequality and the bound $|P_i|=|\We(B_i)|\leq 2\lambda$ we have
\begin{align*}
\dist_{B}(v,v^{B_i}_{SE})&\leq |Q[v,w']| + 1 + |\We(B_i)[x',v^{B_i}_{SE}]|\\
&=|Q| + 1 + |\We(B_i)|\\
& = \dist_{L}(v,s_{P_j}) +2\lambda + 1,
\end{align*}
which concludes the proof of the claim.
\end{claimproof}

To show that the region $\re B_i$ is small, notice first that $B_i$ has no vertex whose distance from $\Perim(B_i)$ is more than $\lambda$, as such a vertex would have distance more than $\lambda$ from $\Perim(L)$.

To show the second condition of smallness, let $v\in V(B_i)$ be any vertex such that 
\[\dist_{B_i}(v,\So(B_i))\leq \lambda \quad \text{and} \quad \dist_{B_i}(v,\No(B_i))\leq \lambda.\]
We will show that $\dist_B(v,\{v^B_{SE},v^B_{SW}\})\leq 2\lambda/\eps^2$. Let $v_L\in V(L)$ denote the object such that $v$ is a subpolygon of $v_L$. Since paths in $B_i$ can be realized with (potentially shorter) paths in $L$ and $\So(B_i)$ is a subpath of $\So(L)$, we have that  $\dist_L(v_L,\So(L))\leq \dist_{B_i}(v,\So(B_i))\leq \lambda$, and similarly $\dist_L(v_L,\No(L))\leq \lambda$. Moreover, $v$ was not picked in the the greedy procedure, thus $\dist_L(v,s_{P_j})< \lambda/\eps^2+\lambda $ for some $j\in \{0,1,\dots,k\}$. Thus \Cref{cl:LdisttoBdist} implies that \[\min\big(\dist_{B_i}(v,v^{B_i}_{SW}), \dist_{B_i}(v,v^{B_i}_{SE})\big)< \lambda/\eps^2 +2\lambda +1<2\lambda/\eps^2,\] since $\lambda>10$ and $\eps<1/2$.
\end{proof}

\subparagraph*{Columns and distinguished columns in lens regions}

Consider a lens region $\re L$ and the corresponding graph $L$.
We define a set of vertices $s^0_i$ in $\So(L)$ to serve as the starting points of columns. Let $C_0$ be the path $\We(L)$ and let $s^0_0\in V(\So(L))$ its starting vertex. Then for $i=1,2,\dots,$ let $s^0_i$ be the earliest vertex on $\So(L)$ (in west-to-east order) where
\[\dist_{\So(L)}(s^0_{i-1},s^0_i) > \eps(\dist_{L}(s^0_i,\No(L)).\]
Let $C^0_i$ be a shortest path from $s^0_i$ to $\No(L)$ in $L$. 
When we get to the end of $\So(L)$, notice that we will add the path $\Ea(L)$ as the final column $C^0_i$. This is because $\Ea(L)$ is either of length $0$ or it is a bridge of length at most $10$, thus when $s^0_i$ is the starting point of $\Ea(L)$, then $1>\eps(\dist_{L}(s^0_i,\No(L))$. We call these paths the \emph{columns}.

Now let us consider the column set $\cC^0_j=\{C^0_j,C^0_{j+\kappa},C^0_{j+2\kappa},\dots\}$ for $j\in [\kappa]$ where $\kappa=\kappa(\eps)$ will be defined later. Let $j^*\in [\kappa]$ be the index that minimizes the total length of the columns in $\cC^0_{j^*}$, i.e., minimizes $\sum_t |C^0_{j^*+(t-1)\kappa}|$. Notice that the columns in $\cC^0_{j^*}$ have total length $O(|\wf L|/(\eps \kappa))$. We call the columns of $\cC^0_{j^*}$ \emph{distinguished}.

\begin{lemma}\label{lem:supercollen}
The total number of edges in the union of all distinguished columns in the lens region $\re L$ is $O(\eps^2 |\So(\wf L)|)$.
\end{lemma}

\begin{proof}
First, we claim that the total length of all the columns defined in $L$ is at most $O(|\So(\wf L)|/\eps)$.

By definition we have that
\[|C^0_i|=\dist_{L}(s^0_i,\No(\wf L))<\frac{1}{\eps}\dist_{\So(\wf L)}(s^0_{i-1},s^0_i).\] Summing this for all $i$ and using the fact that the last column $\Ea(\wf L)$ is either empty or a bridge (and thus constant length), we get that
\[\sum_i |C^0_i| < \frac{1}{\eps}(|\So(\wf L)|) +  |\Ea(\wf L)|< \frac{1}{\eps}|\wf L|+O(1)=O(|\So(\wf L)|/\eps).\]

Thus the total length of the distinguished columns in $\cC_{j^*}$ is at most \[O\left(\frac{|\So(\wf L)|}{\eps\kappa}\right)=O(\eps^2|\So(\wf L)|).\qedhere\]
\end{proof}

\subparagraph*{Uncrossing the distinguished columns}

Unfortunately, the distinguished columns may still cross each other, so we will eventually need to define different \emph{uncrossed} columns. Let $C^1_i=C^0_{j^*+(i-1)\kappa}$, and we denote by $s^1_i$ the starting vertex of $C^1_i$. Let $\cC^1=\cC^0_{j^*}$, and let $S^1$ denote the corresponding set of starting points. Consider now the last subpath of $C^1_i$ obtained by removing all objects that intersect $\bd \re L$.
If this path has length at most $5$, then we say that $C^1_i$ is \emph{short}; otherwise it is \emph{long}. Moreover, if $C^1_i$ intersects $V(\We(L)) \cup V(\Ea(L))$, then it is called \emph{diverting}. In particular, any short column is a bridge, so \ref{it:bridge_willdivert} implies that all short columns are diverting. If $C^1_i$ is non-diverting (and thus long), then let $\Top(C^1_i)$ denote its final maximal subpath of length at least $5$ that is disjoint from $\bd L$.

Suppose that $C^1_i$ is non-diverting. Let $x_i$ and $y_i$ denote the first and last vertex on $\Top(C^1_i)$, respectively. Observe that $x_i$ has a neighbor in $C_i^1$ that intersects $\So(\re R)$, and $y_i$ has a neighbor that intersects $\No(\re R)$. (Note that $x_i$ and $y_i$ is undefined when $C^1_i$ is diverting.)

Now consider the $x_i$ vertices. For each $x_i$, there is a path of length at most two connecting $x_i$ to some vertex in $V(\So(L))$.

We introduce the object $x^\So_i$ that consists of the union of the objects on a shortest path connecting $x_i$ to $V(\So(L))$.
Crucially, the objects $x^\So_i$ are still fat with a slightly worse constant, and they are still similarly sized as the original objects with a worse constant (fatness and diameter both change by a factor of $3$).
Let $X^\So$ denote the objects created here. Similarly, we introduce the objects $Y^\No$ extending the $y_i$ vertices to objects $y^\No_i$ with shortest paths to $\No({\re L})$, having the same fatness and diameter ratio guarantees as $X^\So$.
We note that since $C^1_i$ is not short, we have that $x^\So_i$ cannot intersect any $y^\No_j$. Indeed, if such an intersection would occur, then $x_i$ would be at distance at most $5$ from $\No(L)$, contradicting the fact that the suffix of $C^1_i$ from $x^1_i$ is a shortest path to $\No(L)$ of length at least $|\Top(C^1_i)|+1 \geq 6$.

Let $\MIS(X^\So)$ be a fixed maximal independent set in $X^\So$ and let $\MIS(Y^\No)$ be a fixed maximal independent set in $Y^\No$. We denote by $I$ the set of indices $i$ such that $s^\So_i \in \MIS(X^\So)$. Let $\Cols_L$ denote the set of objects $v\in V(L)$ that appear on the top of columns from index set $I$, or more precisely, we set
\[\Cols_L:= \bigcup_{i\in I} \big(V(C^1_i[x_i,y_i]) \setminus \{x_i,y_i\}\big).\]
Recall that objects of $\Top(C^1_i)=C^1_i[x_i,y_i]$ are disjoint from $\bd \re L$, thus we have $\Cols_L\subset V(G)$. Let $\Cols = \Cols_L \cup \MIS(X^\So) \cup \MIS(Y^\No)$. Notice that $\Cols$ consists of similarly sized fat objects.
We apply \Cref{lem:polygontoholefree} to obtain a fair polygon collection $\Cols'$ that induces the same intersection graph, consists of similarly sized connected fat objects, is $\alpha'$-standard after scaling, and for any object $p\in \Cols$ we have that the corresponding object $p'\in \Cols'$ is a subpolygon of $p$. It follows that $\Cols'$ is an object frame of the intersection graph induced by $\Cols$.

We apply \Cref{thm:Lipschitz_wireframe} on $\Cols'$, which results in a wireframe $\pl D'$; notice that due to the subpolygon property and \Cref{obs:perturb} it can be perturbed into a wireframe of $\Cols$. For a vertex $v\in \Cols$ let $\dot v\in V(\pl D)$ be its picture under the Lipschitz embedding of the intersection graph of $\Cols$ to $\pl D$. Let $D^\No$ denote the vertices of $V(\pl D)$ that correspond to $\MIS(Y^\No)$, that is, $D^\No=\{\dot u \mid u\in \MIS(Y^\No)\}$. Consider a vertex $v\in V(C^1_i(x_i,y_i])\cup \{x^\So_i\}$, and let $v_L=v$ if $v\in V(C^1_i(x_i,y_i])$, and let $v_L$ be the parent of $\dot v$ in $\wf L$ when $v=x^\So_i$, i.e., the object of $\So(L)$ that is one of the objects making up $x^\So_i$. By the Lipschitz property, we have that
\begin{equation}\label{eq:selected_topcols}
\begin{aligned}
\dist_{\pl D}(\dot v, D^\No) &= O\Big(\dist_\Cols\big(v,\MIS(Y^\No)\big)\Big)\\
&= 
O\big(\dist_\Cols(v,y^\No_i)+1\big)\\
&= O\big(\dist_{L}(v_L,y_i)\big)\\
&= O\big(\dist_{L}(v_L,\No(L))\big),
\end{aligned}
\end{equation}
Here the second equality follows from the fact that when $y^\No_i\not\in \MIS(Y^\No)$ then it intersects some $y^\No_j\in \MIS(Y^\No)$. To show the third equality we need that any non-empty path on $\Cols$ can be followed with a path in $L$ that is at most constant times longer: indeed, any object of $\Cols$ is the union of at most $3$ objects of $V(L)$, so for any non-empty shortest path in $\Cols$ we can build a corresponding path in $L$ which is at most $7$ times longer.

Consider now shortest paths $Q^{\pl D}_i$ in ${\pl D}$ connecting each $\dot x^\So_i$ to $D^\No$ for each $i\in I$. Since ${\pl D}$ is planar, when the curve of $Q^{\pl D}_i$ intersects the curve of some $Q^{\pl D}_j$, then they must share a vertex. We can choose the first shared vertex $v$ of these paths, and replace the part of $Q^{\pl D}_j$ starting at $v$ with the part of $Q^{\pl D}_i$ starting at $v$. Since both $Q^{\pl D}_i$ and $Q^{\pl D}_j$ were shortest paths to $D^\No$, we have that the new $Q^{\pl D}_j$ has the same length as the original one. Using such modifications we can find paths $Q^{\pl D}_i$ that are pairwise \emph{non-crossing}, that is, when $Q^{\pl D}_i$ and $Q^{\pl D}_j$ share a vertex, then all edges after the first shared vertex are also shared. Let $\Gamma$ denote the planar curves given by the edges of $Q^{\pl D}_i$ for each $i\in I$. Notice that for each curve $\gamma\in \Gamma$ we have $\gamma\subset \inter(\re L)$.

Let ${\re L}^{XY}$ denote the unique connected component of the region 
\[{\re L}^{XY}:={\re L}\setminus \bigcup \{v \mid v\in \MIS(Y^\No) \cup \MIS(X^\So)\}\]
that is adjacent to each object in $\MIS(X^\So)$ and $\MIS(Y^\No)$; indeed such a component exists as $\MIS(X^\So)$ and $\MIS(Y^\No)$ all intersect $\bd \re L$ and the objects in $\MIS(X^\So) \cup \MIS(Y^\No)$ are pairwise disjoint: if some $x^\So\in \MIS(X^\So)$ would intersect some $y^\No\in \MIS(Y^\No)$, then they would form a bridge in $\re L$, contradicting the definition of $\re L$. Notice that when walking $\bd \re L^{XY}$ from the Southeast corner of ${\re L}^{XY}$ in a counter-clockwise direction, then the objects $\MIS(X^\So)$ appear along $\bd \re L^{XY}$ in the order of their indices, followed by the objects $\MIS(Y^\No)$ in reverse order of their indices, and each of these objects is encountered exactly once.

The region ${\re L}^{XY}$ cuts each $\gamma\in \Gamma$ into subcurves. Let $\mu_i$ denote the subcurve among these that connects some point of $\bd x_i^\So$ to some $\bd y^\No$ for some $y^\No\in \MIS(Y^\No)$ if such a subcurve of $\gamma$ exists. If there are multiple curves from a fixed $\bd x^\So_i$, then we dispose of all but one of them. Let $I_\mu\subseteq I$ denote the indices $i$ where $\mu_i$ is defined. We remark that when $j\in I\setminus I_\mu$ then the path $Q^{\pl D}_j$ intersects some $x^\So_i$ for some $i\in I_\mu$. \skb{esetleg abra?}

Notice that the curves $\mu_i$ for $i\in I_\mu$ are pairwise non-crossing (although some suffix of $\mu_i$ and $\mu_j$ may be equal), and when $\mu_i$ ends at $\bd y^\No_{i'}$ and $\mu_j$ ends at $\bd y^\No_{j'}$ for some $i<j$, then $i'\leq j'$ due to the order in which these boundaries appear along $\bd \re L^{XY}$. Let us now extend both ends of each $\mu_i$ inside $\bigcup \MIS(Y^\No)$ and inside $ \bigcup \MIS(X^\So)$ to $\bd \re L$.
Let $\dot v$ denote the vertex of the wireframe $\wf L$ when $v\in V(\So(L))\cup V(\No(L))$, that is, $\dot v \in \bd \re L$. When $\mu_i$ connects $x^\So$ to $y^\No$, then there is some $\dot u\in V(\So(\wf L))\cap x^\So$ and $\dot v\in V(\No(\wf L))\cap y^\No$ where we will connect the ends of $\mu_i$ to.
These extensions can be done without introducing any new intersections among $\mu_i$: that is, distinct $\mu_i$ remain disjoint except when extending some $\mu_i$ with a shared suffix inside some $y\in \MIS(Y^\No)$ that get extended with the same suffix inside $y$.
This is possible due to the fact that the polygons of $ \MIS(X^\So) \cup \MIS(Y^\No)$ are connected polygons, each of which covers some subcurve of $\bd \re L$.
Let $\mu^*_i$ denote the resulting curves (for each $i\in I_\mu$). In fact, in each of these extensions we can introduce at most two vertices (for the objects making up $x^\So$ or $y^\No$ to create a valid wireframe path $\pl Q^*_i$.
Since $\mu^*_i$ stays inside $\re L$, it corresponds to a walk of $L$ of length at most
\[|\pl Q^*_i|=|Q^{\pl D}_i|+O(1) =\dist_{\pl D}(\dot x^\So_i, D^\No) +O(1) =  O(\dist_{L}(x_i,y_i))\]
by \eqref{eq:selected_topcols}.

Next, for later convenience we need to ensure that the paths $\pl Q^*_i$ have no shared suffixes other than perhaps their endpoints. If some non-endpoint vertex $\dot v$ is shared between $b$ paths $(\pl Q^*_{i_1},\dots,\pl Q^*_{i_b})$ (whose indices $i_1,\dots,i_b$ are necessarily consecutive in $I^\mu$), then we replace $\dot v$ with $b$ points $\dot v_{i_1},\dots,\dot v_{i_b}$ in a small neighborhood of $\dot v$. An edge $\dot v \dot w$ that has been shared between $\pl Q^*_{i_1},\dots,\pl Q^*_{i_b}$ can now be represented by $b$ curves that are disjoint, but each of the new edges run in an infinitesimally small neighborhood of the original edge curve of $\dot v \dot w$. Finally, when $(\pl Q^*_{i_1},\dots,\pl Q^*_{i_b})$ share their last edge $\dot v \dot w$, then $\dot w$ is not duplicated, but it becomes the center of a star, with edges $\dot v_{i_j} \dot w$ for $j\in [b]$, where the edge curves go in a small neighborhood of the original edge curve of $\dot v \dot w$. One can verify that these modifications can be done while maintaining the wireframe property for each $\pl Q^*_i$.

Let $\pl{D\re L}$ denote the wireframe realized by $\wf L$ and the paths $\pl Q^*_i$ (as modified above to avoid shared suffixes), that is, the vertices are $V(\wf L)$, the new vertices inside the extensions, and the vertices of $\pl D$ that fall on $\pl Q^*_i$. The edges are realized by the drawing $\bd \re L \cup \bigcup_i \mu^*_i$.  We denote by $\pl Q^*_0$ and $\pl Q^*_{|\cC|+1}$ the paths $\We(\wf L)$ and $\Ea(\wf L)$, respectively. Consequently, the paths $\pl Q^*_i$ are defined only for the index set $\bar I_\mu := I_\mu \cup \{0,|\cC^1|+1\}$. For each $j\in \bar I_\mu$ we denote by $\dot v_j$ the southern starting vertex of $\pl Q^*_j$.

We say that a column $C^1_i$ \emph{leans west} if one of the following hold:
\begin{itemize}
\item $i\in I_\mu$ and $\dot v_i$ is to the west of (or equal to) $\dot s^1_i$
\item $i\in I\setminus I_\mu$ and the curve of $Q_i^{\pl D}$ intersects some $x^\So_j$ with $j\in I_\mu$ where $\dot v_j$ is to the west of (or equal to) $\dot s^1_i$,
\item $i\not \in I$, because $C^1_i$ is a non-diverting column where $x^\So_i$ intersects some $x^\So_{i'}$ for $i'\in I$, and the curve of $Q_{i'}^{\pl D}$ intersects\footnote{We note that it is possible that $i'\in I_\mu$ and consequently $j'=i'$, so $Q_{i'}^{\pl D}$ intersects $x^\So_{i'}$ by the fact that its starting point $\dot x^\So_{i'}$ is inside $x^\So_{i'}$.} some $x^\So_{j'}$ with $j'\in I_\mu$ where $\dot v_{j'}$ is to the west of (or equal to) $\dot s^1_i$,
\item $i\not \in I$ because $C^1_i$ is diverting, and it intersects some $v\in \We(L)$.
\end{itemize}
If $C^1_i$ does not lean west then it \emph{leans east}. It follows that if $C_1^i$ leans east, then one of the following hold:
\begin{itemize}
\item $i\in I_\mu$ and $\dot v_i$ is to the east of (and not equal to) $\dot s^1_i$
\item $i\in I\setminus I_\mu$ and the curve of $Q_i^{\pl D}$ intersects some $x^\So_j$ with $j\in I_\mu$ where $\dot v_j$ is to the east of (and not equal to) $\dot s^1_i$,
\item $i\not \in I$, because $C^1_i$ is a non-diverting column where $x^\So_i$ intersects some $x^\So_{i'}$ for $i'\in I$, and the curve of $Q_{i'}^{\pl D}$ intersects some $x^\So_{j'}$ with $j'\in I_\mu$ where $\dot v_{j'}$ is to the east of (and not equal to) $\dot s^1_i$,
\item $i\not \in I$ because $C^1_i$ is diverting, and it intersects some $v\in \Ea(L)$.
\end{itemize}

For a fixed $\dot s^1_i$, let $j,k\in \bar I_\mu$ be the consecutive indices of $\bar I_\mu$ where $\dot s^1_i\in \So(\wf L)[\dot v_j,\dot v_k]$ and $\dot s^1_i \neq \dot v_k$. We define the paths $\pl P_i$ according to increasing indices $i$. We set $\pl P_0=\We(\wf L)$. For general $i$, if $\pl P_{i-1}$ contains $\dot s^1_i$, then let $\pl P_i$ be the suffix of $\pl P_{i-1}$ starting at $\dot s^1_i$. Otherwise, let $\pl P_i$ be the concatenation $\So(\wf L)[\dot s^1_i,\dot v_j] \cup Q^*_j$ if $C^1_i$ leans west and $\So(\wf L)[\dot s^1_i,\dot v_k] \cup Q^*_k$ if it leans east. Finally, let $\pl P_{|\cC^1|+1}=\Ea(\wf L)$. Observe that under this definition two paths $\pl P_i, \pl P_j$ are pairwise non-crossing, but they may have overlaps along some subpath of $\So(\wf L)$. Moreover, $\pl P_i$ is a subpath of $\pl P_{i-1}$ if and only if $\dot s^1_i\in V(\pl P_{i-1})$ if and only if there is no face of $\pl{D\re L}$ between $\pl P_{i-1}$ and $\pl P_i$. We call the resulting paths $\pl P_i$ the \emph{uncrossed columns}.

\subparagraph*{Uncrossed column properties}

For the lens $\re L$ let $\prec$ denote the west-to-east order of vertices along $\So(\wf L)$ and $\So(L)$. 
We prove the following lemmas for later use.

\begin{lemma}\label{lem:leaning}
Suppose that $\dot s^1_i\not \in V(\pl P_{i-1})$, and let $\pl P_i = \So(\wf L)[\dot s^1_i,\dot v_j] \cup \pl Q^*_j$ for some $j\in \bar I_\mu$. Then there is a path $\tau$ in $L$ from $s^1_i \in \So(L)$ to $\No(L)$ such that
$\tau$ is $c$-approaching $\No(L)$ for some constant $c=O(1)$ and it intersects $\pl Q^*_j$.
\end{lemma}

\begin{proof}
Suppose without loss of generality that $C^1_i$ leans west and thus $\dot v_j \preceq \dot s^1_i$; the east leaning case with $s^1_i\preceq \dot v_j$ can be handled analogously.
For the rest of the proof we will follow closely the notations in the definition of a west-leaning column.

Notice that when $C^1_i$ is diverting then $\tau:=C^1_i$ intersects $\We(L)$, and since $Q^*_j$ separates $\We(L)$ and $\dot s^1_i$ in $L$ we have that $C^1_i$ must intersect $Q^*_j$. If $i\in I_\mu$ then $j=i$ and $C^1_i$ contains $x^i$, thus $\tau:=C^1_i[s^1_i,x_i]$ followed by a constant length diversion to $v_i$ and the suffix of $C^1_i$ from $x_i$ intersects $Q^*_i$, and it is $c$-approaching as $C^1_i$ is a shortest path to $\No(L)$ and the constant length diversion from $x_i$ to $v_i$ and back weakens this to a $c$-approaching path for some constant~$c$.

If $i\in I\setminus I_\mu$, then $Q_i^{\pl D}$ intersects $x_j^\So$, thus it has a vertex within distance $3$ of $v_j$ (as $x_j^\So$ is a union of at most $3$ objects, one of which is $v_j$). Let $\tau$ be the concatenation of $C_1^i[\dot s^1_i,x_i]$, the path corresponding to the wireframe path $Q_i^{\pl D}$ from $x_i$ to its vertex intersecting $x_j^\So$, the constant length path connecting this neighbor of $x_j$ to $v_j$, and the path corresponding to $Q^*_j$. Notice that $C^1_i$ is a shortest path to $\No(L)$, and $Q_i^{\pl D}$ is $c$-approaching $\No(L)$ by \ref{eq:selected_topcols}, the same holds for $Q_j^{\pl D}$, and the concatenation we defined adds constant length paths to these in constantly many places, which still leaves a $c$-approaching path to $\No(L)$ for some constant~$c$.

If $i\not \in I$ and $C^1_i$ is not diverting, then we again define $\tau$ as a concatenation. From $s^1_i$ we use $C^1_i$ to get to $x_i$, then in constant steps we get to $x_{i'}$ for some $i'\in I$, then continue on $Q_{i'}^{\pl D}$ to $x_{j'}$ for $j' \in I_\mu$, then go along $Q^*_j$ to $\No(L)$. Again, this is a concatenation of $c$-approaching paths to $\No(L)$ with constantly many constant-length paths between them, so it is $c$-approaching $\No(L)$ for some $c=O(1)$.
\end{proof}

\begin{lemma}\label{lem:PiMinimized}
Suppose that $\dot s^1_i\not \in V(\pl P_{i-1})$.
Then $\pl P_i$ defined above satisfies
\[
|\pl P_i|=
\begin{cases} \begin{displaystyle}
O\Big(\min_{\substack{k\in \bar I_\mu\\ \dot v_k \preceq \dot s^1_i}} |\So(\wf L)[\dot s^1_i,\dot v_k] \cup Q^*_k|\Big) \end{displaystyle}
&\text{when $C^1_i$ leans west,}\\
\begin{displaystyle}
O\Big(\min_{\substack{k\in \bar I_\mu\\ \dot v_k \succ \dot s^1_i}} |\So(\wf L)[\dot s^1_i,\dot v_k] \cup Q^*_k|\Big)\end{displaystyle}
&\text{when $C^1_i$ leans east.}
\end{cases}
\]
\end{lemma}

\begin{proof}
Assume without loss of generality that $C^1_i$ leans west; the east-leaning case can be handled symmetrically.
Fix some $k\in \bar I_\mu$ with $\dot v_k \preceq s^1_i$, and we will show that $|\pl P_i|= O(\So(\wf L)[\dot s^1_i,\dot v_k] \cup Q^*_j|)$. Let $\pl P_i=\So(\wf L)[\dot s^1_i,\dot v_j] \cup \pl Q^*_j$. See \Cref{fig:brick_compass}(i). The claim clearly holds if $k=j$. Otherwise $\So(\wf L)[\dot s^1_i,\dot v_k]$ must pass through $\dot v_j$. Since $Q^*_{j}$ is a shortest path from $\dot x^\So_{j}$ to $D^\No$ in $\pl D$, we get by~\eqref{eq:selected_topcols} that $|Q^*_{j}| = \dist_{\pl D}(\dot x^\So_{j}, D^\No) = O(\dist_L(v_{j},\No(L)))$, where the last step uses that $v_{j}$ and the objects making up $x^\So_j$ are within constant distance of each other in $L$.
Using the Lipschitz property from~\Cref{cor:getplanargraph}, we can write:
\begin{align*}
|\pl P_i| &\leq |\So(\wf L)[\dot s^1_i,\dot v_{j}]| + |Q^*_{j}|\\
&\leq |\So(\wf L)[\dot s^1_i,\dot v_{j}]| + O(\dist_L(v_{j},\No(L)))\\
&\leq |\So(\wf L)[\dot s^1_i,\dot v_{j}]| + O(|\So(L)[v_{j},v_k]|+|Q^*_{k}|)\\
&\leq O(|\So(\wf L)[\dot s^1_i,\dot v_k] \cup Q^*_k|),
\end{align*}
which concludes the proof.
\end{proof}

\begin{lemma}\label{lem:eastwestpathsbound}
Let $\pl P_i$ be the path defined above. Then $|\pl P_i|=O(|C^1_i|)$.
\end{lemma}

\begin{proof}
Recall that $v_i\in \So(L)$ is some fixed vertex within distance $2$ of $x_i$ for each $i$. We will assume without loss of generality that $C^1_i$ leans west.

We distinguish several cases based on the index $i$.
\begin{description}
\item[Case 1.] $i\in I_\mu$\\
Since $v_i\preceq s^1_i$, Lemma~\ref{lem:PiMinimized} gives that $|\pl P_i| = O(|\So(\wf L)[\dot s^1_i,\dot v_i]|+|Q^*_i|)$. Using~\ref{it:M_approximatepaths} and that $x_i$ is within constant distance of $v_i$ in $L$, we get
\begin{equation}\label{eq:Pifirstpart}
\begin{aligned}
|\So(\wf L)[\dot s^1_i,\dot v_i]| &\leq (1+\eps) \dist_L(s^1_i,v_i)\\
&\leq (1+\eps)(|C^1_i[s^1_i,x_i]| + \dist_L(x_i,v_i))\\
&= O(|C^1_i[s^1_i,x_i]|).
\end{aligned}
\end{equation}
On the other hand, \eqref{eq:selected_topcols} implies that $|Q^*_i| = O(\dist_L(x_i,\No(L))$. Since $C^1_i$ is a shortest path in $L$, this gives
\begin{equation}\label{eq:Piwrap}
|P_i| \leq O(|C^1_i[s^1_i,x_i]|) + O(\dist_L(x_i,\No(L))= O(|C^1_i|)
\end{equation}
as desired.
\item[Case 2.] $i\in I\setminus I_\mu$\\
It follows that the curve of $Q^{\pl D}_i$ intersects some $x^\So_j$ where $j\in I_\mu$, and $v_j\preceq s^1_i$. Thus there is some $\dot x'_j \in V(Q^{\pl D}_i)$ such that $x'_j:=p(\dot x'_j)$ intersects $x^\So_j$, i.e., $x_j$ and $x'_j$ are within constant distance of each other in $L$. By \Cref{lem:PiMinimized} and the triangle inequality we have that
\begin{equation}\label{eq:Picase2}
|P_i| \leq O(|\So(\wf L)[\dot s^1_i,\dot v_j]| +|Q^*_j|)\leq  O(|\So(\wf L)[\dot s^1_i,\dot v_i]|+ |\So(\wf L)[\dot v_i,\dot v_j]| +|Q^*_j|),
\end{equation}
where the first term of the right hand side expression can be bounded as in \eqref{eq:Pifirstpart} by $O(|C^1_i[s^1_i,x_i]|)$. Using the fact that $x_i$ and $v_i$ are neighbors as well as $x_j,v_j$ and $x'_j$ are within constant distance of each other, we get by the Lipschitz property of~\Cref{cor:getplanargraph} that these distances are distorted by a constant factor in $\pl D$ for the corresponding vertices $\dot x^\So_i$ and $\dot v_i$, as well as among $\dot x^\So_j$, $\dot x'_j$ and $\dot v_j$. Thus we can easily switch among such vertices and incur only a constant length cost.

Let $\dot y^\No$ denote the endpoint of $Q^{\pl D}_i$.
The second and third terms of~\eqref{eq:Picase2} can be bounded as follows:
\begin{equation}\label{eq:Picase2latterparts}
\begin{aligned}
|\So(\wf L)[\dot v_i,\dot v_j]| +|Q^*_j|&\leq (1+\eps)\dist_L(x_i,x_j)+O(1) + \dist_{\pl D}(\dot x^\So_j, D^\No)\\
&\leq O(|Q^{\pl D}_i[\dot x^\So_i,\dot x'_j]|) + |Q^{\pl D}_i[\dot x'_j,\dot y^\No]|+O(1)\\
&=O(|Q^{\pl D}_i[\dot x^\So_i,\dot y^\No]|\\
&=O(\dist_L(x_i,\No(L))
\end{aligned}
\end{equation}
where line 1 used \ref{it:M_approximatepaths}, line 2 used the Lipschitz property of~\Cref{cor:getplanargraph}, and line 4 used~\eqref{eq:selected_topcols}. We can thus conclude this case again with \eqref{eq:Piwrap}.

\item[Case 3.] $i\not \in I$ and $C^1_i$ is not diverting.\\
It follows that $x^\So_i$ was not included in $I$ because it intersected some $x^\So_{i'}$ where $i'\in I$. Let $j'\in \bar I_\mu$ be the index such that $P_{i'}$ is using $Q^*_{j'}$, i.e., it is preceding or succeeding $i'$. We have $v_{j'}\preceq s^1_i$.
By \Cref{lem:PiMinimized}, the triangle inequality and the fact that $|\So(\wf L)[\dot v_i,\dot v_{i'}]|=O(1)$ we can write
\[|P_i| = O(|\So(\wf L)[\dot s^1_i,\dot v_i]| + |\So(\wf L)[\dot v_{i'},\dot v_{j'}]| +|Q^*_{j'}|).\]
The first term can again be bounded as in Case 1 formula \eqref{eq:Pifirstpart} by $O(|C^1_i[s^1_i,x_i]|)$, and the last two terms can be bounded by $O(\dist_L(x_{i'},\No(L))$ using~\eqref{eq:Picase2latterparts}, concluding with \eqref{eq:Piwrap} and the fact that $\dist_L(x_i,x_{i'})=O(1)$.

\item[Case 4.] $i\not \in I$ and $C^1_i$ is diverting.\\
Suppose that $C^1_i$ intersects some object of $\We(L)$ (the proof is analogous when it intersects $\Ea(L)$. By \Cref{lem:PiMinimized}, property \ref{it:M_approximatepaths}, and the fact that $Q^*_0 = \We(\wf L)$ has constant length, we have
\begin{align*}
|P_i| &= O(|\So(\wf L)[\dot s^1_i,\dot v_0]|+|Q^*_0|) \\
&= O(|\So(L)[s^1_i, v_0]|) \\
&= O(\dist_L(s^1_i,\We(L))) \\
&=O(|C^1_i|).\qedhere
\end{align*}
\end{description}
\end{proof}

\subsection{Defining bricks and constructing the mortar graph}

The \emph{mortar graph} $\pl M$ is the wireframe that consists of $\skel$, and in each skeleton face its bridge decomposition, and the union of all graphs $\pl{D\re L}$ in the lenses of the decomposition. Notice that each relevant face $\re R$ of $\skel$ is now sliced into smaller regions with south-north paths that are either bridges or the paths $Q^*_i$ in the lenses of $\re R$. We also include the eastern and western corner of $\re R$ as length-0 south-north paths.

\begin{figure}
\centering
\includegraphics{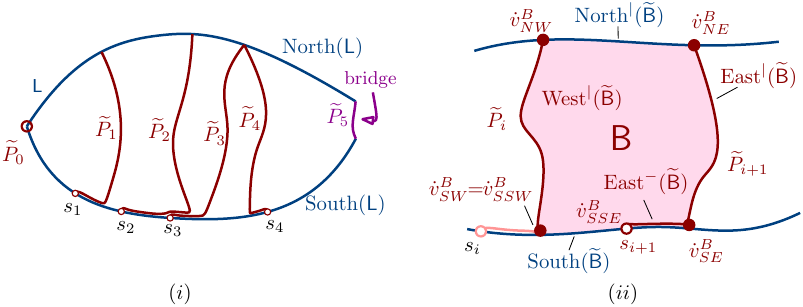}
\caption{(i) A lens $\re L$ sliced into bricks with paths $\pl P_i$. (ii) A brick and its compass. The path $\We^-(\wf B)$ consists of a single vertex $\dot v^B_{SW}=\dot v^B_{SSW}$.}
\label{fig:brick_compass}
\end{figure}

Let $\re B$ be a face in the drawing of $\pl M$  that is located inside a relevant face $\re R$ of $\skel$, and it is bounded by consecutive south-north paths inside $\re R$. (See \Cref{fig:brick_compass}(ii)). Recall that the brick $B$ corresponding to $\re B$ is the intersection graph $G|_{\re B}$.

Next, we define the compass for each brick region $\re B$. Observe that in the above definition, either $\re B$ is inside some lens $\re L\subseteq \re R$, or $\re B$ is not contained in any lens. If $\re B$ is not contained in any lens, then we use the definitions guaranteed by \Cref{lem:slice_into_small_bricks}.

If $\re B$ is inside the lens $\re L$, then it is some region between the paths $\pl P_i,\pl P_{i+1}$ for some $i\in \{0,\dots,|\cC^1|\}$.
We let $\dot v^B_{SE}$ denote the southeast corner of $\wf B$, defined as the easternmost vertex of $\So(\wf L)$ that is on $\wf B$. Similarly, we define $\dot v^B_{SW}$, the southwest corner as the westernmost vertex of $\So(\wf L)$ that is on $\wf B$. Let $\dot v^B_{NE}$ and $\dot v^B_{NW}$ denote the easternmost and westernmost vertices of $\pl B$ that are on $\No(\wf L)$, respectively. Finally, let $\dot v^B_{SSW}:=\dot s^1_i$ and $\dot v^B_{SSE}:=\dot s^1_{i+1}$ denote the first vertex on $\pl P_i$ and $\pl P_{i+1}$.
We can now define the compass of $\wf B$ as follows.
\begin{align*}
\No(\wf B)  &:=\No(\wf L)[\dot v^B_{NE},\dot v^B_{NW}]\\
\We^|(\wf B)&:=\pl P_i[\dot v^B_{SW},\dot v^B_{NW}]\\
\Ea^|(\wf B)&:=\pl P_{i+1}[\dot v^B_{SE},\dot v^B_{NE}]\\
\We^-(\wf B)&:=\So(\wf L)[\dot v^B_{SW},\dot v^B_{SSW}]\\
\Ea^-(\wf B)&:=\So(\wf L)[\dot v^B_{SSE}, \dot v^B_{SE}]\\
\So(\wf B)  &:=\So(\wf L)[\dot v^B_{SSW}, \dot v^B_{SSE}]\\
\We(\wf B)  &:=\We^-(\wf B)\cup \We^|(\wf B)\\
\Ea(\wf B)  &:=\Ea^-(\wf B)\cup \Ea^|(\wf B)
\end{align*}
Note that each of these paths may only consist of a single vertex.

Let $\So(B),\No(B),\Ea(B)$ and $\We(B)$ denote the paths of $\Perim(B)$ corresponding to $\So(\wf B),\No(\wf B),\Ea(\wf B)$ and $\We(\wf B)$, respectively. 

We are now ready to prove \Cref{thm:mortarprops}.

\begin{proof}[Proof of \Cref{thm:mortarprops}]
The graph $\pl M$ is $2$-connected because, according to \ref{it:sk_2connected}, its construction started from a $2$-connected wireframe $\skel$, and $\pl M$ was created by adding paths between vertices of $\skel$. The irrelevant face of $\skel$ is inherited by $\pl M$ and has the desired property by \ref{it:sk_irrelevant}. Thus, $\pl M$ is a mortar graph of $G$. It is routine to check that the compass of $\pl M$ is well-defined 
by starting from \ref{it:sk_compass}, 
and following the construction of the eastern and western brick sides in \Cref{lem:slice_into_small_bricks} and above.

Property \ref{it:M_terminals} holds because by \ref{it:sk_terminals} we have that $\pl M$ is a supergraph of $\skel$ where all but a single face are brick regions.

Next, we prove property \ref{it:M_approximatepaths}. Notice first that if $\re B$ is inside a relevant face $\re R$ of $\skel$, then the distances of vertices in $B=G|_{\re B}$ are greater or equal than the distance of the corresponding vertices in $G|_{\re R}$ by Lemma~\ref{lem:restrictionmetric}(\textit{i}). It is therefore sufficient to prove that these paths are $t$-approxiamte shortest paths in $R=G|_{\re R}$, or possibly just in some lens graph $L=G|_{\re L}$ where $\re L$ is a lens of $\re R$ containing $\re B$.

The path $\We^-(\wf B) \cup \So(B)\cup \Ea^-(\wf B)$ is a subpath of $\So(\wf R)$, and $\No(\wf B)$ is a subpath of $\No(\wf R)$. Switching to objects of $\Perim(B)$ and $\Perim(R)$, the path $\So(B)$ is a path consisting of some subpolygons on a subpath of $\So(R)$, and either $\No(B)$ is a single vertex, or it is a is a path consisting of some subpolygons on $\No(R)$. Consequently, \ref{it:sk_approximatepaths} implies that $\We^-(B) \cup \So(B)\cup \Ea^-(B)$ is a internally $(1+\eps)$-approximate shortest path in $R$, and $\No(B)$ is a shortest path in $R$.

Next, we prove \ref{it:M_eastwestsuffixes}; first we consider only $\We^|(\wf B)$, and later we will extend the proof to $\We(\wf B)$. (The proofs for $\Ea^|(\wf B)$ and $\Ea(\wf B)$ are analogous). Consider a path $\We^|(\wf B)$. Such a path is either a single vertex edgeless path (and thus satisfies \ref{it:M_eastwestsuffixes}), or it came about as the walk inside some lens $\re L$ corresponding to a path $\pl Q^*_i$  in $\pl D'$ where $i\in \bar I_\mu$. When $i=0$ or $i=|\cC^1|+1$, then this is a constant length bridge path and the claim follows. Othewise, since $\pl Q^*_i$ extends some path $Q^{\pl D}_i$ with at most two vertices on both ends, we have that the corresponding path $Q^*_i$ in $L$ is $c$-approaching $Q^{\pl D}_i$ and is a internally $c$-approximate shortest path if and only if the same holds for the path given by the internal vertices of $Q^{\pl D}_i$. Then \eqref{eq:selected_topcols} implies that for any internal vertex $\dot v\in V(\pl Q^{\pl D}_i)$ we have
\begin{equation}\label{eq:Qstar}
\pl Q^{\pl D}_i[\dot v,\dot v^B_{NW}] = \dist_{\pl D}(\dot v, D^\No)+O(1) = O(\dist_L(v,\No(L)),
\end{equation}
where $v=p|_{\re L}(\dot v)$ is an object in $V(L)$. Since the hidden constant above only depends on $\alpha$, the claim of \ref{it:M_eastwestsuffixes} about $c$-approaching follows. Simiarly, to show that it is internally $c$-approximate shortest, we can similarly pick any vertex pair $u,v\in V(\We^|(B))$ and corresponding vertices $\dot u, \dot v\in V(Q^{\pl D}_i)$ and use the Lipschitz property:
\[
\pl Q^*_i[\dot u,\dot v] = \dist_{\pl D}(\dot u, \dot v) = O(\dist_{\Cols}(u,v))= O(\dist_{L}(u,v)),
\]
which concludes the proof of \ref{it:M_eastwestsuffixes}.


Property \ref{it:M_northsouthbound} follows from \ref{it:sk_northsouthbound} and the fact that the edges of $\No(\wf B),\So(\wf B), \We^-(\wf B)$, and $\Ea^-(\wf B)$ correspond to edges of $\skel$, and each edge of $\skel$ is listed in at most two bricks.

We can now prove the complexity bound on $\compl(\pl M)$. Notice that $\pl M$ is obtained from $\skel$ by adding at most $O(|E(\pl M)|)$ shortest paths to various faces $\wf F$ of $\skel$ using \Cref{cor:shortestpathframe}. In each case the addition is based on some path $P$ of objects where $\cA(P)$ can be bounded by $O(\compl(\cA(G))+\compl(\wf F)) = O(\compl(\cA(G))+\compl(\skel)) = O(n^3/\eps)\cdot \compl(\cA(G))$, where the last step used \Cref{lem:skeletonprops}. By \ref{it:M_northsouthbound} add at most $O(|E(\pl M)|)=O(n/\eps)$ such paths, so we get $\compl(\pl M) = O(n^4/\eps^2)\cdot \compl(\cA(G))$, as required.

Next, we prove property \ref{it:M_eastwestbound}. By \Cref{lem:bridgecount}, \Cref{lem:slice_into_small_bricks}, and \Cref{lem:supercollen} we have that the total length of all bridges, slicing paths and columns used for the relevant face region $\re R$ of $\skel$ is $O(\eps^2|\So(\wf R)|)$. Since the paths $\So(\wf R)$ cover each skeleton edge at most twice, this is at most $O(\eps^2|E(\skel)|)$ edges in total. Notice that all edges on these paths appear on the boundary of at most two bricks (once on the east, once on the west), so the same bound holds for the total length of all $\We(B)$ and $\Ea(B)$ together. Finally, \ref{it:sk_northsouthbound} implies that $O(\eps^2|E(\skel)|)=O(\eps \cdot \opt)$, which concludes the proof of~\ref{it:M_eastwestbound}.

It remains to prove \ref{it:M_thickbrick}. The large brick in question is part of some lens $\re L$; suppose that it is between the paths $\pl P_i,\pl P_{i+1}$ for some $i\in \{0,\dots,|\cC^1|\}$. It follows that $\So(\wf B)\subset \So(\wf L)[\dot s^1_i,\dot s^1_{i+1}]$. Consequently, there are at most $\kappa$ vertices $s^0_j,s^0_{j+1},\dots,$ on $\So(B)$ that are starting points of the consecutive columns $C^0_j,\dots,C^0_{j+1}$. By the definition of these columns, we have that $s^0_j,s^0_{j+1},\dots,$ are at most $\kappa$ points ordered west to east on $\So(B)$ such that for any vertex $x$ of $\So(B)[s^0_k,s^0_{k+1})=\So(L)[s^0_k,s^0_{k+1})$ we have $|\So(L)[s^0_k,x]|\leq \eps \dist_L(x,\No(L))$. Notice however that (i) $|\So(L)[s^0_k,x]|=|\So(B)[s^0_k,x]|$, (\textit{ii}) $\No(B)\subset \No(L)$, and (\textit{iii}) the distance of vertices in~$B$ is greater or equal to the distance of corresponding vertices in~$L$ by \Cref{lem:restrictionmetric}(\textit{i}). Thus
\[|\So(B)[s^0_k,x]|\leq \eps \dist_L(x,\No(L)) \leq \eps \dist_B(x,\No(B)),\]
which concludes the proof of \ref{it:M_thickbrick}.

Finally, we note that all of the construction works on polynomial-sized graphs and the construction steps can be followed by a polynomial algorithm.
\end{proof}

\section{The structure theorem and spanner for Steiner tree}\label{sec:SteinerStruct}

Let $G\in \cI_\alpha$ be a graph, and suppose $T\subset V(G)$ and $\eps>0$ are given. Suppose that we have constructed a mortar graph $\cM$ in $G$ using \Cref{thm:mortarprops}. We will set $G^\orig:=G$, and then extend $G$ with all restricted objects in all bricks $B$ of $\cM$. As a result, the optimum solution in $G$ is unchanged (see also the discussion on Page~\pageref{sec:objectduplication} around object duplication), but the size of $G$ might be much bigger than the original graph $G^\orig$, and we occasionally need to refer to only original fat objects, as $G$ now typically contains non-fat subpolygons of objects.

\subsection{Simplifying a forest in a large brick}

Recall that every large brick is a part of some lens.
We call the vertices of $N_B(V(\bd B)),\lambda)$ \emph{outer} vertices, and vertices of $V(B)\setminus N_B(V(\bd B),\lambda)$ will be called \emph{inner} vertices of $B$. Recall that small bricks do not have boundary-distant vertices and thus they do not have inner vertices.

Suppose that $A$ is a vertex set in a thick brick $B$. Based on $A$, we will define a wireframe $\pl Q_A$, whose outer face is the concatenation of the paths $\So(\pl Q_A), \Ea(\pl Q_A), \No(\pl Q_A), \We(\pl Q_A)$. The graph $\pl Q_A$ will consist of the paths $\So(\wf B)$ (equal to $\So(\pl Q_A)$), $\No(\wf B)$ (equal to $\No(\pl Q_A)$), and a forest $\pl F_A$. We define the wireframe $\pl Q_A$ such that it satisfies the following conditions.
\begin{description}
\item[(P1)] The path $\No(\pl Q_A)$ is a shortest path in $\pl Q_A$.
\item[(P2)] For some $k\leq \kappa$ there is a set of vertices $s_0,\dots,s_k$ ordered from west to east on $\So(\pl Q_A)$ such that for any vertex $x$ of $\So(\pl Q_A)[s_i,s_{i+1})$ we have \[|\So(\pl Q_A)[s_i,x]|\leq \eps \dist_{\pl Q_A}(x,\No(\pl Q_A)).\]
\item[(P3)] For any pair of vertices $a,b$ of $\So(\pl Q_A)$ where $a,b$ are not both endpoints of $\So(\pl Q_A)$, the path along $\So(\pl Q_A)$ is an approximate shortest path, that is, $|\So(\pl Q_A)[a,b]| \leq (1+\eps) \dist_{\pl Q_A}(a,b)$.
\end{description}
We can therefore apply the following theorem in a black-box fashion.
Given a subgraph $\pl F$ of $\pl Q_A$ a \emph{joining vertex} of $\pl F$ is a vertex that has an incident edge both in $\pl F$ and $\pl Q_A\setminus \pl F$.

\begin{theorem}[Borradaile~\etal, Thm.~10.7 \cite{BorradaileKM09}]\label{thm:blackboxplanar}
Let $\pl Q_A$ be a plane graph satisfying conditions\footnote{In the setting of Borradaile~\etal~\cite{BorradaileKM09} there is one more condition, namely that every terminal of $\pl Q_A$ is either on $\So(\pl Q_A)$ or $\No(\pl Q_A)$. Observe that this condition is immaterial to their Theorem~10.7, as terminals do not appear in its statement.} (P1)-(P3). Let $\pl F_A = E(\pl Q_A) \setminus E(\So(\pl Q_A)\cup\No(\pl Q_A))$. Then there is a forest $\pl F$ of $\pl Q_A$ with the following properties.
\begin{description}
\item[(F1)] If two vertices of $\No(\pl Q_A)\cup \So(\pl Q_A)$ are connected in~$\pl F_A$ then they are connected in~$\pl F$.
\item[(F2)] The number of \emph{joining vertices} of $\pl F$ on $\No(\pl Q_A)\cup \So(\pl Q_A)$ is $o(1/\eps^{5.5})$.
\item[(F3)] The total length of $\pl F$ is at most $(1+O(\eps))|\pl F_A|$.
\end{description}
\end{theorem}

\begin{remark}\label{rem:intervalcount}
The forest $\pl F$ in \Cref{thm:blackboxplanar} joins the northern and southern boundary of $\wf B$ at only $o(1/\eps^{5.5})$ vertices. Assume without loss of generality that we remove any component from $\pl F$ whose intersection with $\No(\pl Q_A)\cup \So(\pl Q_A)$ is disjoint from the $\pl F_A\cap \No(\pl Q_A)\cup \So(\pl Q_A)$. This preserves the properties required of $\pl F$. Consider a maximal subpath $\pl I$ of $\No(\pl Q_A)\cup \So(\pl Q_A)$ that is also a subpath of $\pl F$. Either there is a joining vertex of $\pl F$ on $\pl I$, or the component of $\pl F$ containing $I$ must be $I$ itself. Moreover, if two such disjoint maximal intervals $I_1,I_2$ of $\pl F\cap E(\So(\pl Q_A)\cup \No(\pl Q_A))$ are in different components of $\pl F$, then no vertex pair $u\in I_1,v\in I_2$ can be connected in $\pl F_A$: indeed, such a connection would also have to be present in $\pl F$. Thus the number of maximal intervals of $\pl F$ in $\pl Q_A$ is at most $o(1/\eps^{5.5})$ plus the number of connected components of $\pl F_A$.
\end{remark}

In what follows, we are generally concentrating on a single brick $B$. Let $\bd_0 = V(\So(B)\cup\No(B))$ and we use the shorthand $\bd_1$ for the set of objects in $V(B)$ that intersect the curve $\gamma(\No(\wf B))$ or the curve $\gamma(\So(\wf B))$. Let $\bd_2 = N_B(\bd_1\cup \bd_0)$. Consequently, $\bd_0\subseteq \bd_1 \subseteq \bd_2$.

\subsection{\texorpdfstring{The wireframe $\pl Q_A$.}{The wireframe Q sub A.}}\label{sec:customplanar}

Let $A$ be a vertex set in $V(B)\setminus \bd_1$ where each component of $A$ has size at least $\lambda$, and $A$ contains the east and west of $B$, or more precisely, we require that 
\[\big(V(\Ea(B))\cup V(\We(B))\big) \setminus \bd_1 \subset A\subset V(B)\setminus \bd_1.\]
Let $a \in V(A)\cap \bd_2$ be a vertex neighboring to some vertex of $\bd_1$. Then there is a path $P_a$ of length at most $2$ connecting $a$ to $\bd_0$ in $B$; we fix one such path for each such vertex $a$. Let $A^\bd$ denote the set of vertices on these paths $P_a$.
In each component $K$ of $B[A\cup A^\bd]$, we greedily select a maximal number of vertices from $V(K)\cap \So(B)$ whose pairwise distance is at least $1/\eps$ in $B$, and do the same in $V(K)\cap \No(B)$. Moreover, for each selected vertex $y$, there is some path $P_a$ ending there. For each selected $y$ we add the vertices of one such path $P_a$ to a set $Z_K$. Note that $Z_K\cap \So(B)$ consists of vertices that have pairwise distance at least $1/\eps$ in $B$, and all vertices of $K\cap \So(B)$ are within distance $1/\eps$ from some vertex in $Z_K\cap \So(B)$. The analogous property holds for $Z_K\cap \No(B)$. Let $Z$ be the union of the sets $Z_K$ for each component $K$ of $B[A\cup A^\bd]$.

We apply \Cref{thm:tracespanningtree} to each connected component of $B[A\cup Z]$ where the point terminals are the vertices of $\wf B$ corresponding to the objects of $Z_K\cap \bd_0$. Let $\pl F_A$ be the resulting plane forest with its leaves on $\So(\re B)\cup \No(\re B)$. We add to $\pl F_A$ the edges of $\wf B$. This results in a wireframe that we denote by $\pl Q_A$.

In order to simplify a forest in a brick via \Cref{thm:blackboxplanar}, we will need the following.

\begin{lemma}\label{lem:BFconditions}
Let $A$ be a vertex set in $B\setminus \bd_1$. Then the wireframe $\pl Q_A$ defined above satisfies the conditions (P1)-(P3).
\end{lemma}

\begin{proof}
Condition (P1) can be easily proven using the construction of $\pl Q_A$. Recall that $\pl Q_A$ is a (boundaried) wireframe of $B$, so by \Cref{lem:restrictionmetric}(\textit{iii}) the distance between two vertices of $\pl Q_A$ are at least as big as the distance of their parents in $B$. On the other hand, $\No(\pl Q_A)=\No(\wf B)$ is a path that corresponds to a shortest path in $B$ by \ref{it:M_approximatepaths}, thus $\No(\wf B)$ is a shortest path in any wireframe inside $\re B$ that contains it. 

We can analogously argue that Condition (P2) follows from \ref{it:M_thickbrick} and Condition (P3) follows from~\ref{it:M_approximatepaths}.
\end{proof}

We will also need the following simple lemma.

\begin{lemma}\label{lem:planegraphsize}
The forest $\pl F_A = E(\pl Q_A)\setminus E(\So(\pl Q_A)\cup\No(\pl Q_A))$ has $(1+O(\eps))|A|$ vertices. 
\end{lemma}
\begin{proof}
By \Cref{thm:tracespanningtree} the forest $|\pl F_A|$ has $|A| + O(t)$ vertices, where $t$ is the number of point terminals in the construction of $|\pl Q_A|$.
Thus it is sufficient to show that $t=O(\eps |A|)$.

Recall from the definition of $\pl Q_A$ that the point terminals of a connected component $K$ of $B[A\cup A^\bd]$ are placed on the boundaries of polygons in $Y$ that have pairwise distance at least $1/\eps$ in $B$. Consequently, if $u,v\in \So(B)\cap K$ then $\dist_{\So(B)}(u,v)\geq 1/\eps$. If $u,v$ have corresponding point terminals $\dot u, \dot v \in V(\pl F_A)$, then $d_{\pl F_A}(u,v) \geq 1/\eps - 2 = \Omega(1/\eps)$, as otherwise there would be a valid shortcut in $B$ between $u$ and $v$ formed by the parent objects along the shortest $\pl F_A$ path from $\dot u$ to $\dot v$, contradicting the definition of $\skel$. (In case of $\No(B)$, any path along $\No(B)$ and thus along $\No(\pl Q_A)$ is a shortest path.)

Suppose that $K$ has $t_K$ point terminals attached, with $t^\So_K$ on $\So(B)$ and $t^\No_K$ on $\No(B)$. We will show that $t_K\leq 2+O(\eps|K|)$. Without loss of generality assume that  $t^\No_K\ge 2$. Consider an Euler tour of the doubled spanning tree of $K$, i.e., an Euler tour of the graph where each edge of some spanning tree of $K$ is doubled.

Mark one occurrence of each object that contains a point terminal in this Euler tour. The Euler tour has length at least $\Omega(t_K/\eps)$, since it has distance at least $\Omega(1/\eps)$ between at least $t^\So_K -1 + t^\No_K -1 = t_K-2$ pairs of consecutive marked objects. As the length of the Euler tour is $2|K|-2= \Omega(t_K/\eps)$, we have $t_K = O(\eps|K|)$. Summing over all components, we get that
\[t = \sum_K t_K \leq \sum_K \big( 2 + O(\eps|K|)\big) = O(\eps|A\cup A^\bd|) + 2\cc(A\cup A^\bd)= O(\eps|A|) + O(\cc(A)),\]
as $|A^\bd|=O(|A|)$ and removing $A^\bd$ from $B[A\cup A^\bd]$ only increases the number of connected components. Notice that $\cc(A)\leq \eps|A|$ since each component of $A$ has size at least $\lambda=1/\eps^{7.5}$. This concludes the proof.
\end{proof}

\subsection{Stating and proving the structure theorem}

We will now state and prove our structure theorem for Steiner trees. Roughly, our theorem states that there exists a $(1+\eps)$-approximate solution where in each brick either no inner points are used, or when they are used, then they are from an optimum Steiner forest of some portal family of the brick. We need some definitions first.

\begin{definition}[Portals]
Given $G, \eps>0$ and a mortar graph $\pl M$, we designate $2/\eps^{7.5} = 2\lambda$ vertices of $\Perim(B)$ as \emph{portals} in each brick $B$ of $\pl M$. The portals $p_1,p_2,\dots$ on $\Perim(B)$ are placed at (almost) equal distances along the cycle $\Perim(B)$, that is, the distances between consecutive portals along $\Perim(B)$ differ by at most $1$.
\end{definition}

It follows that in any brick $B$ and any vertex $v$ of $\Perim(B)$ there is a portal $p_i$ such that $\dist_{\Perim(B)}(v,p_i)\leq \frac{1}{\lambda}|\Perim(B)|$. We will usually denote by $P_B$ the set of portals assigned in $B$.

\begin{definition}[Portal-respecting generator]
Given $G^\orig \in \cI_\alpha$, a mortar graph $\pl M$ and portals $P_B$ in each brick $B$ of $\pl M$, a portal-respecting generator is a pair $(S_0,\cF)$ where $S_0 \subset N_G(p(V(\pl M)),\lambda)$, and $\cF$ is a sequence of portal families $\cF_B$, one for each brick $B$ of $\pl M$, such that each $\cF_B$ consists of pairwise disjoint subsets of $P_B$ of size at least $2$. The pair  $(S_0,\cF)$ generates all sets $S$ that arise as the union of $S_0$ and for each brick $B$ and each portal set $U\in \cF_B$ an arbitrary optimum Steiner tree of $(B,U)$. The \emph{cost} of the generator $(S_0,\cF)$ is 
\[|S_0|+\sum_{B} \sum_{U\in \cF_B} \smt(B,U).\]
\end{definition}

\begin{definition}[Feasible portal-respecting generator]
Given $G^{\orig}\in \cI_\alpha$, a mortar graph $\pl M$, and a terminal set $T\subset V(G)^\orig$, the \emph{feasibility graph} of the portal-respecting generator $(S_0,\cF)$ is a graph $G(S_0,\cF)$ that consists of $G[S_0\cup \bigcup_{B \text{ brick}}\bigcup_{U\in \cF_B} U]$ together with the complete graphs on each $U$ (where $U\in \cF_B$ and $B$ is a brick). The portal-respecting generator $(S_0,\cF)$  is \emph{feasible} if $G(S_0,\cF)$ is connected and its vertex set (more precisely,\footnote{Recall that for a set $U$ of objects in $V(G)$ possibly containing subpolygons, the set $U^\orig$ consists of the corresponding original objects of $V(G^{\orig})$.} $V(G(S_0,\cF))^{\orig}$) contains $T$.
\end{definition}

We note that in the above definition the size of every generated set is at most the cost of the generator. We are now ready to state our main structure theorem for Steiner Trees.

\begin{theorem}[Structure theorem for Steiner Trees]\label{thm:steinerstruct}
Let $G$ be an intersection graph of similarly sized connected fat objects, and suppose a terminal set $T\subset V(G)$ and $\eps>0$ are given. Then in polynomial time we can construct a mortar graph $\pl M$ of $G$ for the terminal set $T$ and corresponding portal set, such that there exists a feasible portal-respecting generator $(S_0,\cF)$ of cost at most $(1+O(\eps))\smt(G,T)$.
\end{theorem}

\begin{proof}
We use \Cref{thm:mortarprops} to compute a mortar graph $\pl M$ that satisfies the properties \ref{it:M_terminals}-\ref{it:M_thickbrick}.
Let $\Sopt$ denote the vertex set of an optimum Steiner tree for $(G,T)$, and let $\Sopt^{\pl M}$ denote the vertices in $\Sopt$ whose object intersects the drawing of $\pl M$. Let us now fix a brick $B$. We denote by $\Sopt^B$ the set of vertices in $\Sopt$ whose object lies in the interior of $\re B$, i.e., those that intersect $\re B$ but are disjoint from the boundary curve $\bd \re B$. Let $\Ain$ denote the vertices of the connected components of $\Sopt^B$ that contain at least one inner vertex, and we denote by $\Aout$ the other connected components.

Recall that $\bd_1$ denotes the set of objects in $V(B)$ that intersect the curve $\gamma(\No(\wf B))$ or the curve $\gamma(\So(\wf B))$.
If $\Ain\neq \emptyset$, then for each connected component $K$ of $\Ain$ that has a neighboring vertex $v$ in $\Sopt$ which intersects the eastern boundary $\gamma(\Ea(\wf B))$ but it is disjoint from $\gamma(\No(\wf B))\cup \gamma(\So(\wf B))$ there exists some $v'\in V(B)\setminus \bd_1$ such that $v'$ connects $K$ to $\Ea(B)$. Let $\Ain'\subset V(B)\setminus \bd_1$ denote the set that we get by adding to $\Ain$ such vertices, i.e., connecting components of $\Ain$ to $\Ea(B)$ and $\We(B)$ in this manner whenever possible. Observe that $|\Ain'|\leq |\Ain|+2\cc(\Ain)$.

Define the planar graph $\pl Q_A$ as above as defined in \Cref{sec:customplanar} based on the set $A:=\Ain'\cup V(\Ea(B)) \cup V(\We(B)) \setminus \bd_1$. By \Cref{lem:BFconditions} the graph $\pl Q_A$ satisfies the properties (P1), (P2) and (P3), therefore by \Cref{thm:blackboxplanar} applied to $\pl Q_A$ and $\pl F_A = \pl Q_A \setminus (\So(\pl Q_A)\cup\No(\pl Q_A))$, there is a forest $\pl F$ satisfying the properties (F1), (F2), and (F3).
For each connected component $\pl K$ of $\pl F$ and each joining vertex $\dot v\in V(\pl K) \cap N(V(\No(\pl Q_A))\cup V(\So(\pl Q_A)))$, we add the portal of $\Perim(B)$ closest to $p(\dot v)$ along $\Perim(B)$ to a set $P_{\pl K}$. Finally, we add the set $P_{\pl K}$ to the family $\cF_B$. Note that if $\Ain=\emptyset$, then $\pl F$ is the empty graph and $\cF_B=\emptyset$. Let $\cF$ be the sequence of the created portal families.

Next, we add the following sets to $S_0$.
\begin{enumerate}
\item \label{it:boundaryneighbor} We add all vertices of $\Sopt^{\pl M}$ to $S_0$.
\item \label{it:outercomp} For each brick $B$, if $X$ is the vertex set of a connected component of $\Sopt^B$ where $X\subseteq A_{out}$, then we add $X$ to $S_0$.
\item \label{it:interextend} For each brick $B$, if $X$ is the vertex set of a maximal subpath of $\So(B)$ or $\No(B)$ such that the corresponding path $\pl X$ is a subpath of $\pl F$, then we extend $X$ with $6/\eps$ edges in both directions along $\So(B)$ or $\No(B)$, and add all of these vertices (i.e., both $X$ and the extension) to~$S_0$.
\item \label{it:eastwest} For each brick $B$, we add $V(\We(B))$ and $V(\Ea(B))$  to $S_0$.
\item \label{it:toportal} For each brick $B$ and for each joining vertex $\dot v \in V(\pl F)\cap V(\So(\wf B))\cup V(\No(\wf B))$ we add to $S_0$ the subpath of $\Perim(B)$ connecting $p(\dot v)$ to the nearest portal.
\end{enumerate}

We claim that $(S_0,\cF)$ is a feasible portal-respecting generator of the desired cost.

\subsubsection*{Proving that $(S_0,\cF)$ is feasible.}

To prove the feasibility, we must ensure that the feasibility graph $G(S_0,\cF)$ is connected and contains a subpolygon of each object in $T$.

Since terminals have corresponding vertices in $\pl M$, a path connecting two (subpolygons of) terminals in $G(S_0,\cF)$ can be decomposed into subpaths whose internal vertices are in the interiors of individual bricks, each of which starts and ends with an object that intersects the curve of the brick boundary $\bd \re B$. Thus it would be sufficient to prove that the endpoint pairs of such subpaths remain connected in in $G(S_0,\cF)$. In fact, in order to show connectivity, within $B$, we simply consider all the subpolygons of each object of $G(S_0,\cF)$ that fall in the region $\re B$ as part of $G(S_0,\cF)$: this does not influence the connectedness of $G(S_0,\cF)$. Let $B(S_0,\cF)$ denote the resulting graph; we note that $B(S_0,\cF)$ contains a clique on each $U\in \cF_B$ . (Note that we will not use these objects when we will bound the cost of $(S_0,\cF)$.)

Consider a maximal subpath $W$ of $\Sopt$ whose internal vertices do not intersect the boundary curve. Let $W=(a=a_0,a_1,a_2,\dots,a_{t},a_{t+1}=b)$ be the corresponding path in $\re B$ where we have restricted the first and last vertex to $\re B$, i.e., $a,b\in V(B)$ intersect $\bd B$ while $\{a_1,\dots, a_t\}\subseteq \Sopt^B$. Our goal now is to prove that $a$ and $b$ are connected by $B(S_0,\cF)$. We note here that the subpolygons $a,b$ are included in $B[S_0]$ as the original $S_0$ contains the original polygons containing $a$ and $b$ by Addition~\ref{it:boundaryneighbor}. If $\{a_1,\dots, a_t\}$ is contained in a connected component of $\Sopt^B$ that has no inner vertex, then by Addition~\ref{it:outercomp} we have $W\subseteq S_0$, thus $a$ and $b$ are connected in $S$. Otherwise the component of $B[A]$ containing $\{a_1,\dots, a_t\}$ must contain some inner vertex. Let $K_W$ be this a component of $\Sopt^B$.

If $a$ intersects $\gamma(\Ea(\wf B))$ but it is disjoint from $\gamma(\No(\wf B))\cup \gamma(\So(\wf B))$, then by Addition~\ref{it:eastwest} it is connected to the brick corners on $\Ea(B)$ in $S$, and by definition, it is also connected to these corners in its component of $A$. Thus by extending $W$ to a brick corner, we will henceforth assume that $a\in \bd_1$. The case where $a$ intersects $\gamma(\We(\wf B))$ is handled analogously. Similarly, we will assume $b \in \bd_1$. It remains to show that $a$ and $b$ are connected in $B(S_0,\cF)$.

Suppose now that $a$ intersects the curve $\gamma(\So(\pl B))$. Recall from the definition of $\pl Q_A$ that $A$ is extended by connections to $\bd_0$ with the set $A^\bd \subseteq \bd_1$. Let $K$ be the connected component of this extended set which contains $K_U$. By the definition of $\pl Q_A$, there is a vertex $z_a\in Z_K\cap \So(B)$ within distance at most $1+1/\eps$ from $a$ in $B$, as otherwise the connection from $a$ to a neighbor in $\So(B)$ could have been added to $Z_K$ to obtain a larger independent set. Analogously, if $a$ intersects the curve $\gamma(\No(\pl B))$, then there is a vertex $z_a\in Z_K\cap \No(B) \cap N_B(a,1+1/\eps)$. We define a vertex $z_b$ in the same manner within distance $1+1/\eps$ from $b$.

Consequently, $z_a$ and $z_b$ are connected in $K$. Recall that they have corresponding vertices $\dot z_a,\dot z_b \in V(\pl F_A)$ which are connected in $\pl F_A$. 
By property (F1) in \Cref{thm:blackboxplanar}, we have that $\dot z_a$ and $\dot z_b$ are also connected also in $\pl F$, i.e., they are in the same component $\pl K$ of $\pl F$. If $\pl K$ is a subpath of $\So(\wf B)$ or $\No(\wf B)$, then $z_a$ and $z_b$ are connected in $S_0$ (and thus in $S$) by Addition~\ref{it:interextend}.
Otherwise, there is a joining vertex on the path of $\wf B\cap \pl F$ containing $\dot z_a$, thus $z_a$ is connected to some portal in $P_{\pl K}$ by additions~\ref{it:interextend} and~\ref{it:toportal}; similarly, $z_b$ is connected to some portal in~$P_{\pl K}$.
These portals are contained in the set $P_{\pl K}\in \cF_B$ and thus they are directly connected by the clique on $P_{\pl K}$ inside $B(S_0,\cF)$. We conclude that $z_a$ and $z_b$ are connected in $B(S_0,\cF)$.

It remains to show that $z_a$ is connected to $a$ in $B(S_0,\cF)$. An analogous argument shows that $z_b$ is connected to $b$. Suppose again that $a$ intersects the curve $\gamma(\So(\pl B))$, as the case when it intersects $\gamma(\No(\pl B))$ can be proven with the same arguments. Since $a\in S$, it is sufficient to show that it has a neighbor $\bar a\in \So(B)$ such that $z_a, \bar a$ are connected in $S$. Notice that $\dot z_a$ is on some maximal subpath of $\pl F\cap \So(\wf B)$. Then by Addition~\ref{it:interextend} we have that any vertex of $\Perim(B)$ that is within distance $4/\eps$ from $z_a$ along $\So(B)$ is connected to $z_a$ in $S$. In particular, we have that $z_a$ is connected to $\bar a$, since by property~\ref{it:M_approximatepaths} we have $\dist_{\So(B)}(z_a, \bar a)\leq (1+\eps)\dist_B(z_a,\bar a) \leq 2\cdot (2+1/\eps)\leq 6/\eps$. This concludes the proof of feasibility.

\subsubsection*{Bounding the cost of $(S_0, \cF)$.}

We prove that the cost of $(S_0, \cF)$ is $(1+O(\eps))\opt$.
In the brick $B$ we have
\begin{equation}\label{eq:Sopt}
|V(\Sopt^B)| = |\Ain| + |\Aout|.
\end{equation}
Let $\mathrm{opt}(\cF_B)$ denote the total number of vertices in the optimal trees for the portal family $\cF_B$, and let $\bar{F} = \pl F\cap (\So(\wf B)\cup \No(\wf B))$. We denote by $\mathrm{join}(\pl F)$ the number of joining vertices of~$\pl F$ to $\So(\wf B)\cup \No(\wf B)$.
We can bound the cost of $(S_0,\cF_B)$ as follows.
\begin{align*}
\mathrm{cost}(S_0,\cF) \leq \; |\Sopt^{\pl M}| + \sum_B \big(&\mathrm{opt}(\cF_B) + |S_0\cap B|\big)\\
< |\Sopt^{\pl M}| + \sum_B \Big( &|V(\pl F\setminus \bar{F})| + \join(\pl F)\cdot \eps^{7.5}|\Perim(B)|\\
& + |A_\myout|\\
& + 2\cdot \frac{6}{\eps}\cdot \cc(\bar F) + |V(\bar F)|\\
& + |V(\We(B))| + |V(\Ea(B))|\\
& + \join(\pl F)\cdot \eps^{7.5}|\Perim(B)| \Big)
\end{align*}
Let $\Sigma_B$ denote the summand for $B$ above.
Notice that $|V(\pl F\setminus \bar{F})| + |V(\bar F)| = V(\pl F)+\join(\pl F)$, and $\join(\pl F)\leq \min(\cc(\bar F),o(1/\eps^{5.5}))$, thus we have 
\begin{align*}
\Sigma_B & < |\Aout|\\
& \quad + |V(\pl F)| + \frac{13}{\eps}\cdot \cc(\bar F)\\
& \quad + |V(\We(B))| + |V(\Ea(B))|\\
& \quad + O(\eps^2 |\Perim(B)|)
\end{align*}

\begin{claim}
$|V(\pl F)| + \frac{13}{\eps}\cdot \cc(\bar F) \leq (1+O(\eps))(|\Ain|+ |V(\We(B)|+|V(\Ea(B))|)$.
\end{claim}

\begin{claimproof}
The claim is trivial if $\Ain=\emptyset$, as that implies that $\pl F$ is the empty graph and $\bar{F}$ is also empty.
Suppose now that $\Ain\neq \emptyset$. To show the claim, consider first a component $K_A$ of $\pl F_A$. Then any pair of vertices on $\So(\wf B)\cup\No(\wf B)$ that are connected by $\pl F_A$ are also connected by some component $\pl K$ of $\pl F$. We assign to $K_A$ the component $\pl K$; notice that this is an injective assignment. Since a component of $\bar{F}$ can contain at most one joining vertex, we have that
\[\cc(\bar F) \leq \cc(\pl F) + \join(\pl F) \leq \cc(\pl F_A) + o(1/\eps^{5.5}) \leq \cc(\Ain') + 2 + o(1/\eps^{5.5}) \leq \cc(\Ain) + 2 + o(1/\eps^{5.5}),\]
since $\cc(\Ain')\leq \cc(\Ain)$ and $\pl F_A$ contains $\Ea(B)$ and $\We(B)$ which may form connected components disjoint from those of $\Ain'$. Recall that each connected component of $\Ain$ has size at least $\lambda - O(1)$, since each of them must connect an inner vertex to $\bd_1$. Thus each connected component of $\Ain$ has size $\Omega(\lambda)=\Omega(1/\eps^{7.5})$, and we get 
\begin{equation}\label{eq:ccFbound}
\frac{13}{\eps}\cdot \cc(\bar F) \leq \frac{13}{\eps}\cdot(\cc(\Ain) + 2 + o(1/\eps^{5.5}))\leq O(\eps^{6.5}|\Ain|) + o(\eps|\Ain|)\leq O(\eps|\Ain|).
\end{equation}
Next, we bound $|V(\pl F)|= |E(\pl F)| + \cc(\pl F)$. By \Cref{thm:blackboxplanar} and \Cref{lem:planegraphsize} we have
\[|E(\pl F)| \leq (1+O(\eps))|E(\pl F_A)| < (1+O(\eps))\cdot (|\Ain'|+ |V(\Ea(B)|+|V(\We(B))|)+O(\cc(\Ain).\]
Recall that $|\Ain'|\leq |\Ain|+2\cc(\Ain)= (1+O(\eps^{7.5})) |\Ain|)$, thus $(1+O(\eps))\cdot |\Ain'| \leq (1+O(\eps))\cdot |\Ain|$.
Since \eqref{eq:ccFbound} implies $\cc(\pl F) \leq \cc(\bar F) = O(\eps|\Ain|)$, we can bound $|V(\pl F)|$ as
\[|V(\pl F)| = |E(\pl F)| + \cc(\pl F) = (1+O(\eps))(|\Ain| + |V(\Ea(B)|+|V(\We(B))|).\]
This concludes the proof of the claim.
\end{claimproof}

Substituting this bound, we get
\begin{align*}
\Sigma_B & < |\Aout| + (1+O(\eps))|\Ain|\\
& \quad + (2+O(\eps))\cdot(|V(\We(B))| + |V(\Ea(B))|)\\
& \quad + O(\eps^2 |\Perim(B)|)
\end{align*}
By Equation~\eqref{eq:Sopt} and simplifying further, we have
\begin{equation}\label{eq:Sonebrick}
\Sigma_B < (1+O(\eps))|\Sopt^B|+ O\big(|V(\We(B))|+|V(\Ea(B))| + \eps^2 |\Perim(B)|\big).
\end{equation}
Since $\smt(G,T)=|\Sopt|= |\Sopt^{\pl M}|+ \sum_B |\Sopt^B|$, we get the following for the cost:
\begin{align*}
\mathrm{cost}(S_0,\cF) &\leq \; |\Sopt^{\pl M}| + \sum_B \Sigma_B\\
&\leq (1+O(\eps))\smt(G,T) + \sum_B O\big( |V(\We(B) \cup \Ea(B))| + \eps^2|\Perim(B)| \big).
\end{align*}
Property \ref{it:M_eastwestbound} bounds the length of the eastern and western paths, that is, $\sum_B |V(\We(B) \cup \Ea(B))| = O(\eps)\smt(G,T)$. Together with Property \ref{it:M_northsouthbound} they bound the perimeter of $B$ by
\[\sum_B \eps^2|\Perim(B)| \leq \eps^2 \sum_B |V(\We(B) \cup \No(B) \cup \Ea(B) \cup \So(B))| = O(\eps)\cdot \smt(G,T).\] Thus $\mathrm{cost}(S_0,\cF)\leq (1+O(\eps))\smt(G,T)$, which concludes the proof.
\end{proof}

Next, we simplify our structure theorem further by slicing the Steiner trees generated by our portal-respecting generator into shortest paths whose vertices come from a restricted vertex set. First, we define the variant of a generator based on shortest paths.

\begin{definition}[Path generator]
Given $G^\orig \in \cI_\alpha$ and a mortar graph $\pl M$ a path generator is a pair $(S_0,\cW)$ where $S_0 \subset N_G(p(V(\pl M)),\lambda)$, and $\cW$ is a sequence of triples $(B,u,v)$ where $B$ is a brick of $\pl M$ and $u,v\in B$. The pair $(S_0,\cW)$ generates all sets $S$ that arise as the union of $S_0$ and for each triplet $(B,u,v)$ an arbitrary shortest path from $u$ to $v$ in $B$. Let $V_\cW$ denote the set of vertices appearing in $\cW$.
The \emph{cost} of the generator $(S_0,\cQ)$ is 
\[|S_0|+|V_\cW|+\sum_{(B,u,v)\in \cQ} (\dist_B(u,v)-1).\]
\end{definition}

We observe that the size of a set $S$ generated from $(S_0,\cW)$ is at most the cost of the generator.

\begin{definition}[Feasible path generator]
Let $(S_0,\cW)$ be a path generator for $G^{\orig}\in \cI_\alpha$, a mortar graph $\pl M$, and let $T\subset V(G^\orig)$ be a given terminal set.  The \emph{feasibility graph} of the path generator $(S_0,\cW)$ is a graph $G(S_0,\cW)$ that consists of $G[S_0\cup V_\cW]$ together with the edges $uv$ for each $(B,u,v)\in \cQ$. The path generator $(S_0,\cW)$ is \emph{feasible} if $G(S_0,\cW)$ is connected and its vertex set (more precisely,\footnote{Recall that for a set $U$ of objects in $V(G)$ possibly containing subpolygons, the set $U^\orig$ consists of the corresponding original objects of $V(G^{\orig})$.} $V(G(S_0,\cW))^{\orig}$) contains $T$.
\end{definition}

With this definition we can state a modified structure theorem for path generators.

\begin{theorem}[Structure theorem with path generators]\label{thm:structure_pathgen}
Let $G$ be an intersection graph of similarly sized connected fat objects, and suppose a terminal set $T\subset V(G)$ and $\eps>0$ are given. Then in $2^{O(1/\eps^{7.5})}\poly(n)$ time we can construct a mortar graph $\pl M$ of $G$ for the terminal set $T$ and a vertex set $W\subset V(G)$ with the following properties:
\begin{enumerate}[label=(\roman*)]
\item $|W \cap B|=2^{O(1/\eps^{7.5})}$ for each brick $B$ of $\pl M$
\item For each $w\in W\cap B$ we have $\dist_B(w,\Perim(B))\leq |\Perim(B)|$.
\item there exists a feasible path generator $(S_0,\cW)$ with
\begin{enumerate}[label=(\alph*),nosep]
\item cost at most $(1+O(\eps))\smt(G,T)$
\item $V_\cW\subset W$
\item $|V_\cW \cap B|=O(1/\eps^{7.5})$ for each brick $B$ of $\pl M$.
\item For each brick $B$ there are at most $O(1/\eps^{7.5})$ triplets in $\cW$ involving $B$.
\end{enumerate}
\end{enumerate}
\end{theorem}

\begin{proof}
Let $\pl M$ be a mortar graph of $G$ for the terminal set $T$ that is computed in polynomial time using \Cref{thm:mortarprops}. For each brick $B$, let $P_B$ denote its portal set, which we compute in polynomial time based on $\pl M$. For any subset $Q\subset P_B$ let $\mathrm{Can}(Q)$ be a \emph{canonical tree} which is an optimal Steiner tree for the instance $(B,Q)$ that we compute using the algorithm of Dreyfus and Wagner~\cite{dw71}. Note here that $\Perim(B)$ is a feasible solution to the instance $(B,Q)$ and thus $|\mathrm{Can}(Q)|\leq |\Perim(B)|$. The algorithm of Dreyfus and Wagner~\cite{dw71} finds a solution in $3^{|Q|}\poly(|V(B)|) = 2^{O(\lambda)}\poly(|V(B)|)$ time. Computing all the canonical trees for each of the $2^{O(\lambda)}$ portal sets in each brick thus takes $2^{O(\lambda)}\poly(n)$ time.

Next, for each canonical tree $\mathrm{Can}(Q)$ fix a spanning tree $SP(Q)$ of $B[\mathrm{Can}(Q)]$. Let $V_{\neq 2}(Q)$ denote the set of vertices in $SP(Q)$ whose degree is not $2$ (i.e., this is the set of leaves and branching vertices of the spanning tree). Let $W$ be the union of the sets $V_{\neq 2}(Q)$ for all sets $Q\subset P_B$ with $|Q|\geq 2$ over all bricks $B$ of $\pl M$. Since each tree $SP(Q)$ has at most $|Q|\leq |P_B| \leq 2\lambda$ leaves, we have that $|V_{\neq 2}(Q)|\leq 2\lambda + 2\lambda-1 <4\lambda$. Since there are at most $2^{2\lambda}$ subsets $Q$, there are altogether $4\lambda \cdot 2^{2\lambda} = 2^{O(\lambda)}$ vertices of $B$ that have been added to $W$, i.e., $|W\cap B|= 2^{O(\lambda)}$ and $W$ satisfies (\textit{i}). Notice moreover that $\mathrm{Can}(Q)$ has a subpath connecting each $w\in W\cap B$ to $\Perim(B)$, so $\dist_B(w,\Perim(B))\leq |\mathrm{Can}(Q)|\leq |\Perim(B)|$, thus (\textit{ii}) holds.

\Cref{thm:steinerstruct} implies that there exists a feasible portal-respecting generator $(S_0,\cF)$ of cost at most $(1+O(\eps))\smt(G,T)$. Let $S$ be a set generated by $(S_0,\cF)$. Recall that $|S|$ is at most the cost of $(S_0,\cF)$ and its connectivity directly follows from the fact that $G(S_0,\cF)$ is feasible, as the canonical trees added to $S$ ensure the connections among the vertices of each set $U\in \cF_B\in \cF$. Moreover, since $S^\orig\supseteq V(G(S_0,\cF))^\orig$, we have that $T\subset S^\orig$. Thus $S^\orig$ is a feasible solution to the Steiner tree instance $(G,T)$ of cost at most $|S^\orig|\leq |S| = (1+O(\eps))\smt(G,T)$.

Consider now a canonical tree $\mathrm{Can}(Q)$ for some $Q\in \cF_B$ in a brick $B$. We can decompose $SP(Q)$ into maximal paths $\pi^Q_1,\pi^Q_2,\dots $ whose internal vertices have degree $2$. Consequently, the endpoints of each $\pi^Q_i$ are from $V_{\neq 2}(Q)$. The number of paths created is $|V_{\neq 2}(Q)|-1$ and recall that $|V_{\neq 2}(Q)|<4\lambda$, thus there are less than $4\lambda$ paths. Notice moreover that each path $\pi^Q_i$ is a shortest path in $B$: indeed, otherwise it could be exchanged for a strictly shorter path in $SP(Q)$ and would create a Steiner tree for the instance $(B,Q)$ with fewer vertices, contradicting the minimality of $\mathrm{Can}(Q)$. For each $\pi^Q_i$, we add the triple $(B,u_i,v_i)$ to the set $\cW$ where $u_i,v_i\in V_{\neq 2}(Q)$ are the end vertices of $\pi^Q_i$. Observe that the construction implies (\textit{ii})(b),  (\textit{ii})(c) and (\textit{ii})(d).

To prove (\textit{ii})(a), notice that the path of $SP(Q)$ connecting $u_i$ to $v_i$ contains $\dist_B(u_i,v_i)$ internal vertices. Consequently, 
\[\smt(B,Q)=|V(\mathrm{Can}(Q)|= |V_{\neq 2}(Q)| + \sum_i (\dist_B(u_i,v_i)-1).\]
Thus the cost of the constructed path generator $(S_0,\cW)$ is
\begin{align*}
|S_0|+|V_\cW|&+\sum_{(B,u,v)\in \cQ} (\dist_B(u,v)-1)\\
&\leq |S_0|+ \sum_B \sum_{Q\in \cF_B} \Big(|V_\cW\cap V_{\neq 2}(Q)| + \sum_i (\dist_B(u_i,v_i)-1) \Big)\\
&=|S_0|+ \sum_B\sum_{Q\in \cF_B} |V(\mathrm{Can}(Q))|\\
&= (1+O(\eps))\smt(G,T),
\end{align*}
where the inequality comes from the fact that each vertex of $V_\cW$ appears as a vertex of $V_{\neq 2}(Q)$ for at least one set $Q$ (possibly more than one), and the last step follows because it is the cost of the portal-respecting generator $(S_0,\cF)$ guaranteed by \Cref{thm:steinerstruct}. This concludes the proof.
\end{proof}

%

As a corollary, we can find an approximate solution in the union of the mortar graph, of the set of outer vertices, and some fixed set of shortest paths.

\begin{theorem}[Steiner tree spanner for intersection graphs]~\label{thm:basicbanyan}
Let $G$ be an intersection graph of similarly sized connected fat objects with terminal set $T$ and $\eps>0$ given. Then in time $2^{O(1/\eps^{7.5})}\poly(n)$ we can compute an induced  subgraph $G'$ of $G$ such that $G'$ has a Steiner tree spanning $T$ of size at most $(1+\eps)\smt(G,T)$, and $G'$ consists of at most $2^{O(1/\eps^{7.5})}\smt(G,T)$ cell cliques of $G$.
\end{theorem}

\begin{proof}
Let $\pl M$ be a mortar graph of $G$ for the terminal set $T$ that is computed in polynomial time using \Cref{thm:mortarprops}. Use Theorem~\ref{thm:structure_pathgen} to construct a vertex set $W$ where $|W\cap B|=2^{O(1/\eps^{7.5})}$. For all vertex pairs $u,v\in W\cap B$ we consider a shortest path from $u$ to $v$ in $B$. By Theorem~\ref{thm:structure_pathgen}(\textit{ii}), we have that there exist vertices $u',v'\in \Perim(B)$ such that $\dist_B(u,u')\leq |\Perim(B)|$ and $\dist_B(v,v')\leq |\Perim(B)|$. Moreover, we have $\dist_B(u',v')\leq |\Perim(B)|$, so by the triangle inequality, we have $\dist_B(u,v)\leq 3|\Perim(B)|$. Thus every shortest path added inside $B$ has length at most $3|\Perim(B)|$.

Let $G'$ be the subgraph of $G$ consisting of $N_G(V(p(\pl M)),\lambda)\cup \bigcup_B\bigcup_{\emptyset \neq Q\subseteq P_B} \mathrm{Can}(Q)$. Observe that by Theorem~\ref{thm:structure_pathgen}(\textit{iii}) $G'$ contains a connected subgraph inducing all vertices of $T$ of size $(1+O(\eps))\opt$. Moreover, in each brick we have added $2^{O(1/\eps^{7.5})}$ shortest paths, and they have  been computed in $2^{O(1/\eps^{7.5})}\poly(n)$ time per brick. Since $
\pl M$ has $\poly(n)$ bricks, the computation takes $2^{O(1/\eps^{7.5})}\poly(n)$ time.

To bound the number of cell cliques of outer vertices, recall that the number of cells within distance $k$ of a vertex of $G$ is $O(k^2)$, thus $|N_G(p(V(M)),\lambda)_\cP| = O\left(\frac{|V(\pl M)|}{(\eps^{7.5})^2}\right)$. Since $\pl M$ is connected planar, we have $|V(\pl M)|=\Theta(|E(\pl M)|)$, and applying \ref{it:M_northsouthbound} and \ref{it:M_eastwestbound} yields $|N_G(p(V(\pl M)),\lambda)_\cP|=O(E(\pl M)/\eps^{15}) = O(\smt(G,T)/\eps^{16})$. To bound the number of cliques required to cover the shortest paths, observe that the sum of brick perimeters is twice the edge count of $\pl M$, so by~\ref{it:M_northsouthbound} and~\ref{it:M_eastwestbound} we have
\begin{align*}
\sum_B\sum_{u,v \in B\cap W} (\dist_B(u,v)-1) &\leq \sum_B 2^{O(\lambda)}\cdot 3|\Perim(B)|\\
&= 2^{O(\lambda)}\cdot 6|E(\pl M)| \\
&= 2^{O(\lambda)}\smt(G,T) \qedhere
\end{align*}
\end{proof}

\subsection{Algorithm for Steiner Tree via contraction decomposition}

We start by proving the following lemma, which is analogous to \Cref{lem:TSPUncontractCost}.

\begin{lemma}\label{lem:SteinerUncontractCost}
Let $T\subset V(G)$ be a set of terminals and $X\subset V(G)$ an arbitrary
set, and fix a clique partition $\cP$ of $G$. Then $\smt(G,T)\leq \smt
(G\vcon X,T_X)+|T\cap X|+ O(|\cX|)$, where $T_X=V(G\vcon X)|_T$ is the set of vertices in $G\vcon X$ whose
preimage in $G$ contains a terminal from $T$, and $\cX$ is the set of cliques
in $\cP$ that contain some vertex of $X$. Moreover, given an optimum Steiner tree for
$(G\vcon X,T_X)$, a Steiner tree of size $\smt(G,T)+ O(|\cX|)$ that is feasible for $
(G,T)$ can be constructed in $\poly(|V(G)|)$ time.
\end{lemma}

\begin{proof}
Let $S$ be the vertex set of an optimum Steiner tree of $(G\vcon X,T_X)$. We will show how to build a Steiner tree walk $S_G$ for $(G,T)$ that is potentially longer than required, and shorten it later. First, we add all vertices of $S$ to $S_G$ that also appear in $G$ (i.e., the vertices that were not contracted). Consider a connected component $Y$
induced by $X$ and let $z\in V(G\vcon X)$ be the vertex that results from
contracting this component; suppose moreover that $z$ appears on $S$. We select a vertex $y\in Y$ arbitrarily. Then
for each edge $uz$ of $S$ we add a shortest
$u\rightarrow y$ path to $S_G$ whose internal vertices are in~$Y$.
Additionally, for any terminal $t\in T$ that appears in~$Y$, we extend $S_G$ with the vertices of a path that connects $t$ and $y$ inside $G[Y]$. Note that the resulting graph contains all terminals in $Y$. We repeat the
above procedure for all connected components of $G[X]$ whose contraction
appears in $S$. Notice that the resulting vertex set $S_G$ induces a connected graph containing all terminals, thus any spanning tree of $S_G$ is a feasible solution for $(G,T)$.

We now decrease the size of $S_G$. 
 We can apply the simplifications described in
\Cref{lem:SteinSolTerminals} to ensure that in each clique of $\cX$
the new graph $S'_G$ contains at most $O(1)$ non-terminal vertices. Consequently,
\begin{equation}\label{eq:S'bound}
|S'_G|\leq \tsp(G\vcon X,T_X)+|T\cap X|+O(|\cX|).
\end{equation}

Let $S^*$ be the vertex set of an optimum Steiner tree of $(G,T)$. We claim that $\smt(G,T) = |S^*| \geq \smt(G\vcon X,T_X)+|T\cap X| -O(|\cX|)$. \Cref{lem:SteinSolTerminals} implies that without loss of generality, $S^*$ is a set where from each clique $C$ of $\cP$ the set contains $O(1)$ non-terminal vertices. Now we contract $X$, and as a result, get a set $S^*\vcon X$ that is a feasible solution to $(G\vcon X,T_X)$. Since in each clique $C$ intersecting $T\cap X$ we lose $|(T\cap X)\cap C|-1$ vertices, we have
\[\smt(G\vcon X,T_X) \leq |S^*\vcon X| = |S^*| - \sum_{C\in \cP} \big(|T\cap X \cap C| + 1\big) = |S^*| - |T\cap X| + O(|\cX|),\]
as claimed. Substituting in~\eqref{eq:S'bound} now yields the desired bound:
\[|W'_G|\leq \smt(G\vcon X,T_X)+|T\cap X|+O(|\cX|) \leq \smt(G,T) + O(|\cX|).\]
Moreover, the steps described above can be achieved
within a running time that is polynomial in~$|V(G)|$.
\end{proof}

\SteinerSlow*

\begin{proof}
We apply \Cref{thm:atalakitas} to obtain an $\alpha$-standard graph $G\in \cI_\alpha$ whose representation has polynomial complexity. Let $\cP$ be the corresponding cell partition.
By \Cref{lem:sparseSteiner} we may assume without loss of generality that
for each $C\in \cP$ we have $|C\cap T|\leq 1/\eps$ and $|C|\leq 2^{O((1/\eps)\log(1/\eps))}$ (after spending $2^{O((1/\eps)\log(1/\eps))}\poly(n)$ time to compute this).

Let $\opt$ denote the number of vertices in the
optimal Steiner tree of $(G,T)$. We use \Cref{thm:basicbanyan} to compute a subgraph $G^*$ with at most $2^{\gamma\lambda}\opt$ cell cliques for some constant $\gamma$ in time $2^{O(\lambda)}\poly(n)$ that contains a $(1+\eps)$-approximate solution. This is a Steiner tree spanner of the desired size, and thus proves the first part of the theorem.

Let $\cA$ denote the set of at most $2^{\gamma\lambda}\opt$ cell cliques in $G^*$. We have that the optimum $\opt^*$ of $G^*$ satisfies $\opt^*\leq (1+O(\eps))\opt$. It is now sufficient
to compute a $(1+ O(\eps))$-approximate tour in $G^*$ for the terminal set $T$.

We apply
\Cref{thm:contractdecomp} on $G^*$ and the cell partition $\cA$ for $k=\frac{2^{\gamma\lambda}}{\eps}=2^{O(\lambda)}$, which yields a collection $\cX_1,\dots,\cX_k\subset \cS$ and corresponding vertex sets $X_1,\dots,X_k$. Note that $\ell=\max_{C\in \cS} |C|\leq 2^{O((1/\eps)\log(1/\eps))}$.

Let us now compute an
optimum Steiner tree for each of the instances $(G^*\vcon X_i,T_{X_i})$, where $T_
{X_i}=V(G^*\vcon X_i)|_T$. By \Cref{thm:contractdecomp} each $X_i$ satisfies $\tw(G^*\vcon X_i)= O(kl)= 2^{O(\lambda)}$.
We can therefore apply
the $2^{O(\tw)} \poly(n)$ algorithm~\cite{single-exponential} to solve these instances in $2^{2^{O(\lambda)}}\poly(n)$ time. Note that $\smt(G^*\vcon X_i,T_{X_i})\leq \opt^*$ because contractions can only decrease the optimum cost of
a subset TSP instance.

By \Cref{lem:SteinerUncontractCost} in polynomial time we can construct trees
$\tau_i$ for $(G^*,T)$ of length at most 
\[\smt(G^*,T)+O(|\cX_i|)
= \opt^*+O(|\cX_i|).\]
\Cref{thm:contractdecomp} guarantees that
$\sum_{i=1}^k |\cX_i|=O(|\cA|)$, and therefore there exists some $i$ where
$|\cX_i|=O(|\cA|/k)= O(\eps)\cdot\opt$.
As a result, the shortest among the trees
$\tau_i$ has length at most $\opt^*+O(\eps)\cdot\opt=(1+O(\eps))\cdot\opt$.
\end{proof}

\section{A more efficient approximation scheme for Steiner Tree}\label{sec:SteinerFast}

The main goal of this section is to prove a theorem that transforms our input instance into vertex-weighted instance that is easier to handle algorithmically.

Formally, the vertex-weighted variant is defined as follows.

\begin{quote}
\wstein \\
\textbf{Input:} A graph $G$, a vertex weighting $\omega:V(G)\rightarrow \Nats$ and a vertex set $T\subset G$ called \emph{terminals}.\\
\textbf{Output:} The minimum $k\in \Nats$ such that there exist a vertex set $S \subseteq V(G)$ of weight $\omega(S)\leq k$ where $T\subset S$ and $G[S]$ is connected.
\end{quote}

Most of this section is devoted to proving the following theorem.

\begin{theorem}\label{thm:advanced_struct}
Given $G\in \cI_\alpha$, $T\subset V(G)$ and $\eps>0$, in $2^{O(1/\eps^{7.5})}\poly(n)$ time we can construct a graph $G'$, a terminal set $T'\subset V(G')$, a vertex weighting $\omega:V(G')\rightarrow \Nats$, and a partition $\cP'$ of $V(G')$ such that the following hold.
\begin{enumerate}[label=(\roman*)]
\item For each class $C'\in \cP'$ we have $|C'|\leq 2^{O(1/\eps^{7.5})}$.
\item $\tw(G'\parcon \cP')=O(1/\eps^{18})$, and a tree decomposition of $G'\parcon \cP'$ of width $O(1/\eps^{18})$ can be constructed in polynomial time.
\item Given a Steiner tree of $(G',\omega,T')$ of weight $\kappa$, we can construct a Steiner tree of $(G,T)$ of size $\kappa+O(\eps)\smt(G,T)$ in polynomial time.
\item There exists a Steiner tree $S'$ of $(G',\omega,T')$ with the following properties:
\begin{enumerate}[label=(\alph*)]
\item The weight of $S'$ is at most $(1+O(\eps))\smt(G,T)$.
\item For each $C'\in \cP'$ we have $|S'\cap C'|\leq O(1/\eps^{7.5})$.
\end{enumerate}
\end{enumerate}
\end{theorem}

To prove \Cref{thm:advanced_struct}, we start by applying \Cref{lem:sparseSteiner} to get an instance that is equivalent to $(G,T)$ where all classes in $\cP$ have size at most $2^{O(1/\eps\log(1/\eps))}$, and each class of $\cP$ contains at most $O(1/\eps)$ terminals from $T$.


Our proof proceeds throughout this section. For our input graph $G^\orig\in \cI_\alpha$ recall that $\cP$ denotes its standard partition into cliques and $G_\cP$ denotes the graph obtained by contracting each partition class of $\cP$. We apply \Cref{lem:sparseSteiner} to get an instance that is equivalent to $(G,T)$ where all classes in $\cP$ have size at most $2^{O(1/\eps\log(1/\eps))}$; we will keep using the notation $(G,T)$ for this modified instance. Let $\opt=\smt(G,T)$.

Let $H$ denote the plane graph obtained via \Cref{thm:getplanarclique} of vertex set $\cP$ that\footnote{In this section we will have the plane graph $H$ have the same vertex set $\cP$ as $G_\cP$; we will use $\dot \cP$ only when we specifically need to refer to the corresponding realization.} is $c_\Lip$-Lipshitz with $G_\cP$ under the identity map, where $c_\Lip=O(\alpha^{32})$.


We say that a clique $C\in \cP$ and the corresponding vertex $C\in V(H)$ are \emph{$\tau$-outer} if there is some $v\in C$ such that $\dist_G(v,p(V(\pl M)))\leq \tau$. Non-$\tau$-outer cliques of $\cP$ are called \emph{$\tau$-inner} cliques or vertices, respectively. Let $\Inn(\tau)$ denote the set of $\tau$-inner vertices.

We define $\tau=\lambda+\max(\beta^*,6c^2_\Lip)$ where $\beta^*$ is set in the next lemma (\Cref{lem:Ybrick}), and we define $Y_\tau= \Inn(\tau)$ as the set of $\tau$-inner cliques. 
Let $\cY_\tau$ be the partition of $Y_\tau$ where two connected components of $H[Y_\tau]$ belong to the same class of $\cY_\tau$ if and only if they are located in the same face of the plane graph $H-Y_\tau$. By this definition, the classes of $\cY_\tau$ do not necessarily induce connected graphs in $H$. However, it will be crucial for our arguments that $\cY_\tau$ is defined this way: we want to ensure that, after consolidating the classes of the partition $\cY_\tau$, at most one consolidated vertex appears in a face of the graph induced by the vertices outside $Y_\tau$, hence the existence of these consolidated vertices change treewidth only by constant factor.

\subsection{Relating the classes of $\cY_\tau$ to brick regions}

\begin{lemma}\label{lem:Ybrick}
There exists a number $\beta^*=O(c_\Lip^{2.1})$ such that for any $\tau>\beta^*$ and the above definition of $Y_\tau$ and $\cY_\tau$ the following holds. Each class of $\cY_\tau$ consists of cliques that are located in the interior of a single brick region.
\end{lemma}

\begin{proof}
The lemma is implied by the following: if $u,v\in Y_\tau$ are in the same face of $H-Y_\tau$, then the vertices of $G$ inside the cliques $u$ and $v$ are in the same brick of~$G$.

Suppose for the sake of contradiction that $u,v \in Y_\tau$ are in the same face of $H-Y_\tau$ but in distinct bricks $B_u$ and $B_v$. Since the vertices of $G$ inside $u$ and $v$ are far from the brick boundaries (at distance more than $\tau$), we have that all objects in $u$ are in the interior of the brick region $\re B_u$ and similarly all objects in $v$ fall in the interior of $\re B_v$.

Without loss of generality, assume that $\re B_u$ is a bounded face of $\pl M$ (as otherwise the proof can work on $\re B_v$, since $\pl M$ has exactly one unbounded face). Let $\gamma_B=\gamma(\wf B_u)$ be the Jordan curve given by the cycle $\wf B_u$ of $\pl M$ around the face $\re B_u$. Notice that $\gamma_B$ has a unique bounded region ($\re B_u$) containing $u$ and its unbounded region contains $v$. We will need the following claim.

In what follows, we will think of $G_\cP$ as an intersection graph whose objects are given by the unions of objects in each class of $\cP$. Our strategy will be to consider $\gamma$ and gradually change it into a closed walk of $H-Y_\tau$ ---another closed curve--- that ``follows'' $\gamma$. We will show that in each gradual change of $\gamma$, the current closed curve must still separate $u$ and $v$. We will need the following two claims.

\begin{claim}\label{cl:shortcurve_unbound}
Let $\rho$ be a closed curve that can be covered by the objects of some closed walk $W$ of $G_\cP$ of length $|W|\leq \frac{\tau}{5\alpha}$ and suppose that there is some vertex $v\in V(W)\cap (p(V(\wf B_u))_\cP$. Then both $u$ and $v$ are in the unbounded component of $\rho$.
\end{claim}

\begin{claimproof}
We will show that $u$ is in the unbounded component; the proof for $v$ is analogous. Notice first that $u$ cannot be intersected by any object of $W$, since all vertices of $W$ are within distance $|W|$ from $\Perim(B)$ in $B$ and thus they are $3|W|$-outer; in particular, they cannot be $\tau$-inner as $\tau> 3|W|$.

Suppose for the sake of contradiction that $u$ is in some bounded component $\re R$ of $\Reals^2\setminus \rho$, and consider a shortest path from $u$ to $\Perim(B_u)$. Since $u$ is $\tau$-inner, the shortest path has at least $\tau-3|W|$ vertices that are $3|W|$-inner, thus they are disjoint from the objects of $W$. Consequently, there are at least $(\tau-3|W|)/2$ vertices on this shortest path that are pairwise disjoint, and all of them are in a bounded component of $\Reals^2\setminus \rho$. Since these are original objects of $G$, each of them contains a disk of diameter at least $\sqrt 2$, thus each of these objects has area at least $\pi/2>1$. Consequently, $\re R$ has to cover these pairwise disjoint objects must have area more than $(\tau-3c_W)/2$.

On the other hand, observe that $\diam(\rho)\leq 3\alpha|W|$ (as each object of $V(W)$ has diameter at most $3\alpha$). By the isodiameteric inequality, we have $\Area(\re R)\leq \diam(\re R)^2\cdot \frac{\pi}{4}$, thus
\[
\frac{\tau-3|W|}{2}\leq 3\alpha |W| \quad \Leftrightarrow \quad
|W|\geq \frac{\tau}{3\alpha+3/2}.
\]
On the other hand, we have
\[c_W\leq \frac{\tau}{5\alpha}<\frac{\tau-1}{3\alpha+3/2},\]
as $\alpha>1$. Therefore, we have arrived at a contradiction.
\end{claimproof}

Consider a closed curve $\gamma$ that is potentially self-intersecting. Notice that $\Reals^2\setminus \gamma$ has one or more connected components, with exactly one unbounded component and $0$ or more bounded components.
Let $\gamma:[0,1]\rightarrow \Reals^2$ be a closed curve and let $\sigma$ be a continuous subcurve of $\gamma$, i.e., $\sigma=\gamma|_{[a,b]}$. We denote by $\gamma-\sigma$ the remaining closed curve including the endpoints, i.e., $\gamma-\sigma=\gamma|_{[0,a]}\cup\gamma|_{[b,1]}$. 

\begin{claim}\label{cl:stillseparates}
Let $\rho$ be a closed curve whose unbounded component contains $u$ and $v$, and let $\gamma$  be a closed curve where the unbounded component contains $v$ and some bounded component contains $u$. Assume moreover that there is some continuous subcurve $\sigma$ of $\gamma$ such that $\sigma$ also appears as a continuous subcruve of $\rho$. Suppose that $\rho$ and $\gamma$ are both piecewise linear curves of finite complexity. Then the closed curve $\gamma'$ obtained by concatenating $\rho-\sigma$ and $\gamma-sigma$ has $u$ in one of its bounded components, and $v$ in its unbounded component.
\end{claim}

\begin{claimproof}
Observe that the intersection of the unbounded components of $\gamma$ and $\sigma$ is disjoint from both curves and thus it is covered by the unbounded component of $\gamma'$. Since $v$ is in this intersection, it must be in the unbounded component of $\gamma'$. It remains to show that $u$ is in some bounded component of~$\gamma'$. Suppose the contrary: $u$ is in the unbounded component of~$\gamma'$.

Then $u$ and $v$ are both in the intersection of the unbounded components of both $\rho$ and $\gamma'$, thus they can be connected by a curve inside the intersection of these components. In particular, the curve connecting $u$ and $v$ is disjoint from $\gamma'$ and disjoint from $\rho$. Notice that since $\gamma$ can be covered by $\gamma'\cup \rho$, we have that the curve connecting $u$ and $v$ is disjoint from $\gamma$. This contradicts the fact that $u$ and $v$ are in different components of $\gamma$.
\end{claimproof}

Let $\gamma_0=\gamma_B$ and we extend $\gamma_0$ in several steps using \Cref{cl:stillseparates}.
For each vertex $\dot x$ of $\wf B$ let $\gamma_{\dot x}$ be a closed piecewise linear curve connecting $\dot x$ to $\dot x_\cP$ and back inside $x:=p(\dot x)$. (Recall that $x_\cP$ is the clique of $\cP$ containing $x$ and $\dot x_\cP$ is the grid point assigned to $x_\cP$ that stabs all objects in $x_\cP$.) Clearly we can choose $\gamma_{\dot x}$ to have complexity at most $2\compl(x)$. The curve $\gamma_{\dot x}$ can be covered by a single object $x_\cP$ of $G_\cP$ and thus a trivial walk. We will ensure that $\tau>5\alpha$, thus by \Cref{cl:shortcurve_unbound} we have that $\gamma_{\dot x}$ does has both $u$ and $v$ in its unbounded face. Consequently, if we extend $\gamma_0$ to also include $\gamma_{\dot x}$ (notice that this can be done as $\dot x$ appears on both curves), then by \Cref{cl:stillseparates} the resulting curve $\gamma'_0$ still has the property that $u$ is in one of its bounded components while $v$ is in its unbounded component.

We repeat the above procedure for each $\dot x\in V(\wf B)$, to obtain a curve $\gamma_1$ with $u$ in one of its bounded components and $v$ in its unbounded components.

Now for each edge $\dot x \dot y$ of $\wf B$, consider the subcurve $\sigma_{\dot x\dot y}$ of $\gamma_1$ given by the curve from $\dot x_\cP$ to $\dot x$ (a part of $\gamma_{\dot x}$), the
edge $\dot x \dot y$ of $\wf B$, and the path from $\dot y$ to $\dot y_\cP$ (a part of $\gamma_{\dot y}$).

Consider the shortest path in $H$ connecting $\dot x_\cP$ and $\dot y_\cP$ in $H$. Notice that there is such a path: indeed, $\dot x_\cP,\dot y_\cP \in V(H)$; if $\dot x_\cP=\dot y_\cP$, then a path of length $0$ connects them.
Suppose that $\dot x_\cP \neq \dot y_\cP$. Then $x$ and $y$ are neighboring by the wireframe property of $\wf B$, thus $x_\cP$ and $y_\cP$ are neighboring cliques, i.e., $\dot x_\cP\dot y_\cP$ is an edge of $G_\cP$.
By the Lipschitz property of $H$, there must be a path of at most $c_\Lip$ edges connecting  $\dot x_\cP$ and $\dot y_\cP$.

Let $P_{\dot x \dot y}$ be this path in $H$, and we set  $\gamma_{\dot x\dot y}$ to be the concatenation of $P_{\dot x \dot y}$ and $\sigma_{\dot x\dot y}$, which results in a closed curve. We will now bound the geometric diameter of $\gamma_{\dot x\dot y}$ and show that it can be covered by constantly many cliques of $G_\cP$.

\skb{esetleg abra?}

Note that if $\dot w$ is a vertex of $H$ on $P_{\dot x \dot y}$, then it is within $H$-distance $c_\Lip$ from $\dot x$ and $\dot y$, thus it is within $G_\cP$-distance $c^2_\Lip$ from $x_\cP$ and $y_\cP$ by the Lipschitz property. 
By \Cref{thm:getplanarclique} we have that each edge of $P_{\dot x \dot y}$ can be covered by at most $O(\alpha^{34})=O(c_\Lip^{1.1})$ objects of $G_\cP$. Thus $P_{\dot x \dot y}$ can be covered by cliques that are within $G_\cP$-distance $c^{G_\cP}_{\gamma}=O(c_\Lip^{2.1})$ from $x_\cP$. Since $\sigma$ can be covered by two objects of $G_\cP$, we conclude that $\gamma_{\dot x \dot y}$ can be covered by $c^{\mathrm{cover}}=O(c_\Lip^{2.1})$ objects of $G_\cP$.

We can now set
\[\beta^*:= c^{\mathrm{cover}}\cdot 5\alpha +1 =O(c_\Lip^{2.1}).\]

Observe that the vertices of $H$ (and of $G_\cP$) participating in this cover are $\beta^*$-outer, and in particular, we have that $V(P_{\dot x \dot y})$ does not have any $\tau$-inner vertices, thus it is disjoint from $Y_\tau$.

Consequently, \Cref{cl:shortcurve_unbound} implies that  $\gamma_{\dot x \dot y}$ has both $u$ and $v$ in its unbounded component.
Since the curve  $\sigma_{\dot x \dot y}$ is a shared subcurve of $\gamma_1$ and  $\gamma_{\dot x \dot y}$, we can use \Cref{cl:stillseparates} to show that the curve $\gamma'_1$ obtained as the concatenation of $\gamma_{\dot x \dot y}- \sigma_{\dot x \dot y}$ and  $\gamma_1-\sigma_{\dot x \dot y}$ has $u$ in one of its bounded components and $v$ in its unbounded component.
Notice that the subcurves $\gamma_{\dot x \dot y}$ from a partition of $\gamma_1$, thus we can apply the above procedure repeatedly on each edge $\dot x \dot y$ of $\wf B$.

The resulting curve $\gamma_2$ does not contain any subcurve $\sigma_{\dot x \dot y}$, i.e., it is the concatenation of the paths $P_{\dot x \dot y}$, which are paths of $H-Y_\tau$. Thus, $\gamma_2$ is a walk of $H-Y_\tau$ that has $u$ in a bounded face and $v$ in the unbounded face. In particular, there is some cycle in this walk (and in $H-Y_\tau$) where $u$ is inside the cycle and $v$ is outside. This contradicts the fact that $u,v$ are in the same face of $H-Y_\tau$, concluding the proof.
\end{proof}

\subsection{A contraction decomposition that avoids $Y_\tau$}

We will need a bound on the treewidth of certain plane graphs. The \emph{radial graph} $\Radi(G)$ of a plane graph $G$ is the bipartite graph where vertices are either in $V(G)$ or they are obtained by adding a vertex $v_F$ to each face $F$ of $G$. A vertex $v_F$ is then connected to $v\in V(G)$ in $\Radi(G)$ if $v$ appears on the face $F$. Let $V_F$ denote the set of vertices corresponding to faces. The dual graph $G^\mathrm{dual}$ of $G$ is a plane multigraph on $V_F$ where two vertices are connected by $k$ edges if and only if the corresponding faces have $k$ shared edges in $G$.
We note that for a connected graph $G$ we have that $(G^\mathrm{dual})^\mathrm{dual}$ is
isomorphic to $G$ and $\Radi(G^\mathrm{dual})$ is isomorphic to $\Radi(G)$. A result of Demaine~\etal~\cite{DemaineHK09} gives the following bound.

\begin{lemma}[Corollary of~\cite{DemaineHK09}]\label{lem:radial_tw}
For any plane graph $G$ we have that $\tw(G\cup \Radi(G))=O(\tw(G))$.
\end{lemma} 

\begin{proof}
Demaine~\etal~\cite{DemaineHK09} proved that $\tw(H^\mathrm{dual}\cup \Radi(H))=\tw(H^\mathrm{dual}\cup \Radi(H^\mathrm{dual}))=O(\tw(H^\mathrm{dual}))$ for any plane graph $H$. For each connected component $G_i$ of $G$ let $H=G_i^\mathrm{dual}$, and apply the above result to get $\tw(G_i\cup \Radi(G_i))=O(\tw(G_i))$. Now a tree decomposition of $G\cup \Radi(G)$ of width $O(\max_{i}\tw(G_i))= O(\tw(G))$ can be obtained as follows. Remove the vertex of $\Radi(G_i)$ corresponding to the outer face of $G_i$ from each bag, and add the vertex $v_G$ corresponding to the outer face of $G$ to every bag; let $\cT_i$ be the obtained tree decomposition. Finally, a singleton bag containing only $v_G$ is connected to an arbitrary bag of each tree decomposition $\cT_i$. This is a valid tree decomposition of $G\cup \Radi(G)$ of width $1+O(\max_{i}\tw(G_i))=O(\tw(G))$.
\end{proof}

We have by \Cref{thm:mortarprops} \ref{it:M_northsouthbound} and \ref{it:M_eastwestbound} that there are $O(\opt/\eps)$ vertices of $G_\cP$ containing vertices of $p(V(\pl M))$. Recall that all vertices of $\cP-Y_\tau$ are within distance $O(\lambda)=O(1/\eps^{7.5})$ to some vertex of $p(V(\pl M))$. Since $G_\cP$ connects grid points within constant distance of each other, the $r$-neighborhood of a vertex has size at most $O(r^2)$, so in particular  we have that
\[|V(H-Y_\tau)|=|\cP-Y_\tau|\leq |N_{G_\cP}(p(V(\pl M))\parcon \cP,\lambda)|\leq O(\lambda^2)\cdot \opt/\eps=O(\opt/\eps^{16}).\]

Next, apply \Cref{thm:contractdecomp_planar} on $H-Y_\tau$ to get a partition $E_1,E_2,\dots,E_{k_X}$ consisting of $k_X:=1/\eps^{18}$ edge sets. Notice that since $V(E_i)\leq 2|E_i|$, the planarity of $H-Y_\tau$ implies that
\[\sum_{i=1}^{k_X} |V(E_i)|\leq 2 \sum_{i=1}^{k_X} |E_i| \leq 6|V(H-Y_\tau)|.\]
In particular, there is some index $i$ such that $X:=V(E_i)$ has size at most $6|V(H-Y_\tau)|/k_X$.
Moreover, we have 
\begin{equation}\label{eq:Hcontractdecomp}
\tw((H-Y_\tau)\vcon X)\leq \tw((H-Y_\tau)/E_i) = O(k_X)=O(1/\eps^{18}).
\end{equation}
We can bound the size of $|X|$ as follows:
\begin{equation}\label{eq:Xbound}
|X|\leq |6|V(H-Y_\tau)|/k_X=O(\opt/\eps^{16})/k_X=O(\eps^2 \opt).
\end{equation}

Notice that $X$ is disjoint from $Y_\tau$, and we denote by $\cX$ the partition of $X$ to its connected components induced in the graph $H - Y_\tau$ (or, equivalently in $H$).

We now turn to the consolidation $H\parcon(\cX,\cY_\tau)$.

\begin{lemma}\label{lem:contractionisgood}
The graph $H \parcon (\cX,\cY_\tau)$ is planar and $\tw(H \parcon (\cX,\cY_\tau))=O(\tw((H-Y_\tau)\vcon X))=O(1/\eps^{18})$. Moreover, a tree decomposition of $H \parcon (\cX,\cY_\tau)$ of width $O(1/\eps^{18})$ can be constructed in polynomial time.
\end{lemma}

\begin{proof}
Consider some face region $\re F$ of the plane graph $H-Y_\tau$ that contains some point from $\dot Y_\tau$, and let $W(\re F,Y_\tau)$ denote the set of vertices along $\bd \re F$ that are connected to $Y_\tau$ in $H$. Then $H\parcon \cY_\tau$ is the planar graph where each such face $\re F$ has a new vertex $w_{\re F}$ whose neighborhood is $W(\re F,Y_\tau)$. Thus  $(H\parcon \cY_\tau)\vcon X$ is also planar. On the other hand, by the definition of $\cX$ and the disjointness of $X$ and $Y_\tau$ we have $(H\parcon \cY_\tau)\vcon X=(H\parcon \cY_\tau)\parcon \cX = H\parcon(\cX,\cY_\tau)$, which concludes the proof of the planarity claim.

Notice above that $H\parcon(\cX,\cY_\tau)=(H\parcon \cY_\tau)\parcon \cX$, where $X$ does not contain any vertices $w_{\re F}$. Thus $H\parcon(\cX,\cY_\tau)$ can be obtained from $(H-Y_\tau)\vcon X=(H-Y_\tau)\parcon \cX$ by adding new vertices $w_{\re F}$ to the corresponding faces of $H\parcon \cX$ and connecting it to some of the vertices along the face cycle, in particular, $H\parcon(\cX,\cY_\tau)$ is a subgraph of the graph $(H-Y_\tau)\vcon X \cup \Radi((H-Y_\tau)\vcon X)$. By \Cref{lem:radial_tw} and \eqref{eq:Hcontractdecomp} we get that 
\[\tw(H \parcon (\cX,\cY_\tau)) \leq \tw\big((H-Y_\tau)\vcon X \cup \Radi((H-Y_\tau)\vcon X)\big)= O(\tw((H-Y_\tau)\vcon X))=O(1/\eps^{18}).\]

Finally, using the algorithm of Kammer and~Tholey~\cite{KammerT12} we can compute a tree decomposition of the planar graph $H \parcon (\cX,\cY_\tau)$ of width $O(\tw(H \parcon (\cX,\cY_\tau)))=O(\tw((H-Y_\tau)\vcon X))$ in polynomial time (i.e., polynomial both in terms of the graph size and treewidth; recall that the treewidth is at most~$n$).
\end{proof}

Define $Y:=N_{H-X}(Y_\tau,c_\Lip)$, thus, by definition, $Y$ is disjoint from $X$.
Let $\xi_0: \cP \rightarrow \cP\parcon (\cX,\cY_\tau)$ denote the consolidation map. Set $\cY$ to be the partition of $Y$ that is defined based on $\cY_\tau$ as follows: we put $u,v\in Y$ in the same partition class of $\cY$ if and only if $\xi_0(u)$ and $\xi_0(v)$ are in the same connected component of $\big(H \parcon (\cX,\cY_\tau)\big)[\xi_0(Y)]$. In other words, $u,v\in Y$ are in the same class of $\cY$ if, after consolidating the classes of $Y_\tau$, there is a $u-v$ path using only the (images) of the vertices of $Y$. Recall that a class of $\cY_\tau$ does not necessarily induce a connected subgraph of $H$, thus this definition is {\em not} the same as being in the same component of $H[Y]$.

\begin{lemma}\label{lem:Yinner}
We have $Y\subset \Inn(\lambda+3c^2_\Lip)$, and for any $u,v\in Y$ if $u,v$ are in the same class of $\cY$, then they are in the interior of the same brick of $\pl M$.
\end{lemma}

\begin{proof}
Notice that $Y=N_{H-X}(Y_\tau,c_\Lip)$ is a subset of $N_H(Y_\tau,c_\Lip)$. So, by the Lipschitz property we have:
\[Y\subset N_H(Y_\tau,c_\Lip) \subset N_{G_\cP}(Y_\tau,c^2_\Lip).\]
Recall that $Y_\tau=\Inn(\tau)$, and since each edge of $G_\cP$ can be simulated in $G$ by a path of length at most $3$, we have that each vertex of $Y$ is $(\tau-3c^2_\Lip)$-inner, thus ---as $\tau\geq \lambda+6c^2_\Lip$--- we have that $Y\subset \Inn(\lambda+3c^2_\Lip)$.

To prove the second claim, let $P$ be a path in $(H \parcon \cY_\tau)[\xi_0(Y)]$ connecting $\xi_0(u)$ to $\xi_0(v)$. It is sufficient to prove that for each edge $\xi_0(a)\xi_0(b)$ of this path the vertices $a,b$ fall in the same brick of $B$. Consider the original edge $ab\in E(H[Y])$. By the Lipschitz property, there is a path $P_{ab}$ of length at most $c_\Lip$ connecting $a$ to $b$ in $G_\cP$. Moreover, since  $Y\subset \Inn(\lambda+3c^2_\Lip)$, we have that the vertices of this shortest path $P_{ab}$ are in $\Inn(\lambda+3c^2_\Lip-3c_\Lip)$. Thus the path $P_{ab}$ can never leave its starting brick, as it cannot reach its boundary, i.e., $a$ and $b$ are in the same brick.
\end{proof}

\subsection{The graphs $G^*_\cP$, $G^*$, and their properties}

Notice that $X$ and $Y$ are disjoint, thus the definition of $\cY$ can be stated equivalently as follows: $u,v\in Y$ are in the same partition class of $\cY$ if and only if $\xi_0(u)$ and $\xi_0(v)$ are in the same connected component of $\big(H \parcon \cY\big)[\xi_0(Y)]$.
Consequently, for each $Y_i\in \cY_\tau$ there is a unique $\phi(Y_i) \in \cY$ such that $Y_i\subseteq \phi(Y_i)$. (That is, $\cY$ is defined on the larger set $Y\supseteq Y_\tau$, but within $Y_\tau$ it can be thought of as a coarsening of $\cY_\tau$.) We also observe that each connected component of $H[Y]=(H-X)[Y]$ is contained in a single class of $\cY$.
Let $\xi:\cP \rightarrow \cP\parcon (\cX,\cY)$ denote the new consolidation map.

We define $Z:= N_{G_\cP}(X,c^2_\Lip)$. Now~\eqref{eq:Xbound} and the fact that $G_\cP$ connects cliques corresponding to integer grid points that are within constant distance, we have
\begin{equation}\label{eq:Zbound}
|Z|= O(|X|c^4_\Lip)=O(|X|)=O(\eps^2 \opt).
\end{equation}

\begin{lemma}\label{lem:GcsillagP}
In polynomial time we can construct a  graph $G^*_{\cP}$ and a tree decomposition of $G^*_{\cP}\parcon(\cX,\cY)$ of treewidth $O(1/\eps^{18})$ with the following properties.

\begin{enumerate}[label={(\roman*)}]
\item $V(G^*_{\cP})=\cP$.
\item For each edge $uv$ of $G_{\cP}$ either $uv$ is in $E(G^*_\cP)$ or there exists a path of length at most $c_\Lip$ in $G^*_{\cP}$ from $u$ to $v$ whose vertices are in $Z$.
\item For each edge $u^*v^*$ of $G^*_\cP$ either $u^*v^*$ is in $E(G_\cP)$ or there exists a path of length at most $c_\Lip$ in $G_\cP$ from $u$ to $v$ whose vertices are in $Z$.
\end{enumerate}
\end{lemma}

\begin{proof}
We construct a tree decomposition of $G^*_{\cP}\parcon(\cX,\cY)$ based on the tree decomposition $\cT$ of width $O(1/\eps^{18})$ given by Lemma~\ref{lem:contractionisgood}, and we will use this to define the graph $G^*_\cP$ on the vertex set $\cP$. The two tree decompositions will have the same tree structure, but the contents of the bags will differ.

Let $B$ be a bag of the tree decomposition $\cT$ of $H\parcon(\cX,\cY_\tau)$; we will construct a corresponding bag $B^*$ for the tree decomposition of $G^*_{\cP}\parcon(\cX,\cY)$ as follows.

If $v\in B$ is a consolidated vertex, i.e., $v\in \xi_0(X\cup Y)$, then we add the corresponding vertex of $\cP\parcon(\cX,\cY)$ to $B^*$, or formally, we include the vertex $\xi(\xi_0^{-1}(v))$ to $B^*$. We observe that $\xi(\xi_0^{-1}(w))$ is always a single vertex: indeed, $\xi_0^{-1}(w)$ can only be of size more than one if $w\in \xi_0(X)$ or $w\in \xi_0(Y)$.
In the former case we have $\xi(\xi_0^{-1}(w))=w$ as $\xi_0$ and $\xi$ are the same over $X$. In the latter case, we have that there is a unique $Y_i\in \cY$ containing $Y_j=\xi_0^{-1}(w)$ thus $\xi(\xi_0^{-1}(w))=\xi(Y_i)$ is also a single vertex.

Now let $v\in B$ be some non-consolidated vertex of $H\parcon(\cX,\cY_\tau)$, i.e., let $v\in B\setminus \xi_0(X\cup Y)$. Consider the set of paths $\cQ_v$ in $H\parcon(\cX,\cY)$ starting at $v$ that have length at most $c_\Lip$ whose internal vertices are not consolidated vertices. (That is, the starting vertex $v$ and every internal vertex of these paths is from $\cP\setminus(X\cup Y)$, while the end vertex of the path may or may not be consolidated.)
For each vertex $w\in \xi_0(\cP)$ that appears as an endpoint of some path in $\cQ_v$ we add $\xi(\xi_0^{-1}(w))$ to the bag $B^*$. As seen earlier, $\xi(\xi_0^{-1}(w))$ is always a single vertex. Let $\cT^*$ be the constructed tree decomposition.

Let us now bound the size of the bag $B^*$ based on $|B|$. Notice first that $B^*$ contains at most one vertex for each consolidated vertex of $B$. Additionally, for each non-consolidated vertex $v\in B\setminus \xi_0(X\cup Y)$ we consider all paths of length $c_\Lip$ starting at $v$ whose internal vertices are also non-consolidated. Recall that $H$ was obtained via \Cref{thm:getplanarclique}, which guarantees that the maximum degree of $H$ is bounded by a constant $c_{\Delta(H)}$. Consequently, the number of paths starting at $v$ of length at most $c_\Lip$ in $H$ can be bounded as $c_{\Delta(H)}^{c_\Lip}=O(1)$. Notice that the same bound holds for $|\cQ_v|$ as we allowed only the last vertex of a path of $\cQ_v$ to be consolidated. Thus, for each non-consolidated vertex $v$ of $B$ we have added at most $O(1)$ vertices to $B^*$. We conclude that $|B^*|=O(|B|)$ and the width of $\cT^*$ is at most constant times the width of $\cT$, thus the width of $\cT^*$ is $O(1/\eps^{18})$. Notice that the construction of $G^*$ and $\cT^*$ can be obtained in polynomial time when a tree decomposition of $H\parcon(\cX,\cY_\tau)$ of width $O(1/\eps^{18})$ is given, which we have also obtained in polynomial time using \Cref{lem:contractionisgood}.

Next, we verify that each vertex $w$ of $\xi(\cP)$ appears in a non-empty connected subtree of $\cT^*$.

A vertex $w\in\xi(\cP\setminus Y)$ will appear in a bag $B^*$ if and only if $B$ contained an un-consolidated vertex $v$ of $H\parcon(\cX,\cY_\tau)$ that is reachable on a path of length at most $c_\Lip$ from $w$ in $H\parcon(\cX,\cY_\tau)$ whose internal vertices are not consolidated. (Notice that since $\xi$ and $\xi'$ are the same on $X$, the vertex $w$ is necessarily a vertex of $H$.) Since the set of such vertices $v$ form a connected subgraph of $H\parcon(\cX,\cY_\tau)$, and connected subgraphs appear in a connected subtree of $\cT$, the vertex $w$ will also appear in a connected subtree of~$\cT^*$.

When $w\in \xi(Y)$, then let $W=\xi_0(\xi^{-1}(w))$ be the set of corresponding vertices in $H\parcon(\cX,\cY_\tau)$. Note that $W$ induces a connected subgraph of $H\parcon(\cX,\cY_\tau)$ by the definition of $\cY$. The vertex $w$ will appear in a bag $B^*$ if and only if either $B$ contained an unconsolidated vertex $v$ of $H\parcon(\cX,\cY_\tau)$ that is reachable on a path of length at most $c_\Lip$ from $W$ in $H\parcon(\cX,\cY_\tau)$ whose internal vertices are not consolidated, or if there is some consolidated $v\in \xi_0(Y_\tau)\cap W$ in $B$. The set of such vertices $v$ again form a connected subgraph of $H\parcon(\cX,\cY_\tau)$ as $W$ induces a connected subgraph of $H\parcon(\cX,\cY_\tau)$. Thus $w$ will also appear in a connected subtree of $\cT^*$, which concludes the verification. 


We are now ready to define the edges of $G^*_\cP$ with vertex set $\cP$ as follows:

\begin{itemize}
\item If $uv$ is an edge in $G_\cP$ and $\xi(u),\xi(v)$ appear together in some bag of $\cT^*$, then include the edge $uv$ in $G^*_\cP$.
\item If there are $u,v\in Z$ and there is a path $Q_{uv}$ in $G_\cP[Z]$ of length at most $c_\Lip$ such that $\xi(u),\xi(v)$ appear in some bag of $\cT^*$, then include the edge $uv$ in $G^*_\cP$.
\end{itemize}

Notice that the definition of $G^*_\cP$ implies properties (i) and (\textit{iii}). It is also straightforward to verify that $\cT^*$ is a valid tree decomposition of $G^*_\cP\parcon(\cX,\cY)$ as the edges of $G^*_\cP$ were defined in such a way to ensure that they are represented in some bag $B^*$ of $\cT^*$.

It remains to prove Property (\textit{ii}). Let $uv \in E(G_\cP)$, and $uv\not\in E(G^*_\cP)$, and let $Q$ be a path of length at most $c_\Lip$ in $H$ connecting $u$ and $v$. We need to prove that there exists a path of length at most $c_\Lip$ in~$G^*_{\cP}[Z]$.

We distinguish three cases based on the interaction of the vertices of~$Q$ with~$X$ and~$Y$.
\begin{itemize}
\item The path $Q$ has a vertex $x\in X$.\\
We will show that each edge $ab$ of $Q$ is also an edge of $G^*_\cP$, which will prove that $Q$ has the desired properties. Note that $ab\in E(H)$ implies that either $\xi_0(a)=\xi_0(b)$ or there is an edge $\xi_0(a)\xi_0(b)\in H\parcon(\cX,\cY_\tau)$, thus $\xi_0(a)$ and $\xi_0(b)$ appear in some bag $B$ of $\cT$. Moreover, there is a path $Q_{ab}$ of length at most $c_\Lip$ in $G_\cP$ connecting $a$ and $b$ by the Lipschitz property. The concatenation of these paths for each edge $ab$ of $Q$ is a walk $Q^{G_\cP}$ of $G_\cP$ of length at most $c^2_\Lip$ that still contains some $v\in V(Q)\cap X\subseteq V(Q^{G_\cP})\cap X$. Recall that $Z=N_{G_\cP}(X,c^2_\Lip)$, thus $V(Q^{G_\cP})\subset Z$, and in particular, the vertices of $Q_{ab}$ are in $Z$. The definition of $G_\cP$ now yields that $ab\in E(G^*_\cP)$, as claimed.

\item The path $Q$ has a vertex $y \in Y$, but no vertex from $X$.\\
It follows that $V(Q)\subseteq N_{H-X}(y,c_\Lip)$, and thus it is in a single connected component of $(H-X)[Y]$. This connected component of $(H-X)[Y]$ falls in a single class $Y_i$ of $\cY$ by the definition of $\cY$. Thus $G^*_\cP$ contains the edge $uv$ because $\xi(u)=\xi(v)=\xi(Y_i)$. The edge $uv$ is therefore a single-edge path with the desired property.

\item The path $Q$ has no vertex from $X\cup Y$.\\
If $u=\xi_0(u)$ appears in some bag $B$, then both $\xi(u)$ and $\xi(v)$ appear in $B^*$ as $Q\in \cQ_u$. It follows that the edge $uv$ is in $G^*_\cP$ and serves as a single-edge path with the desired property.\qedhere
\end{itemize}
\end{proof}

For a vertex $u\in V(G)$ recall that $u_\cP$ denotes the corresponding class of $\cP$ containing $u$. Let $G^*_\cP$ be the graph from \Cref{lem:GcsillagP}.

\begin{definition}\label{def:Gcsillag}
We define $G^*$ to be the graph on the vertex set $V(G^\orig)$ where we include an edge $uv$ if and only if:
\begin{enumerate}[label=\textit{(\roman*)}]
\item $u_\cP=v_\cP$ or
\item $u_\cP v_\cP\in E(G^*_\cP)$ and either a) $uv\in E(G)$ or b) $u_\cP v_\cP\not\in E(G_\cP)$
\end{enumerate} 
\end{definition}

\begin{observation}\label{obs:Gcsillagosszehuz}
$G^*_\cP=G^*\parcon \cP$.
\end{observation}

\begin{proof}
First, we show that $G^*_\cP$ is a subgraph of $G^*\parcon \cP$.
Suppose that $u_\cP v_\cP$ is an edge in $G^*_\cP$. We need to show that there is an edge between some vertex $u\in u_\cP$ and $v\in v_\cP$ in $G^*$. Notice that if there is some $uv\in E(G)$, then $uv$ is included in $G^*$ by part (\textit{ii}) of the definition. On the other hand, if there is no such edge, then $u_\cP$ and $v_\cP$ are not connected in $G_\cP$, thus  $u_\cP v_\cP\not\in G_\cP$ implies that all edges between any $u\in u_\cP$ and $v\in v_\cP$ have been added to $G^*$.

To show that $G^*\parcon \cP$ is a subgraph of $G^*_\cP$, notice that any edge $uv$ to be included in $G^*$ where $u_\cP\neq v_\cP$ requires that $u_\cP v_\cP\in E(G^*_\cP)$. Thus $G^*\parcon \cP$ can only contain edges that are present in~$G^*_\cP$.
\end{proof}

We can define an analogous consolidation on $V(G^*)$ via a new notation. For a set $A\subset \cP$ define $X^\cup:= \bigcup_{C\in X} C$. For a subpartition $\cZ$ over $\cP$ (i.e., a partition of a subset of $\cP$) we define $\cZ^\cup=\{Z^\cup \mid Z\in \cZ\}$ to be the corresponding subpartition of $V(G)$. We will now consider the consolidation $G^*\parcon (\cX^\cup, \hat \cY^\cup)$, and we denote by $\psi:V(G)\rightarrow V(G)\parcon (\cX^\cup, \hat \cY^\cup)$ the consolidation map.


\begin{lemma}\label{lem:PPcsillag}
Let $u,v\in V(G)$ whose objects intersect the same brick region $\re B$ of $G$, and let $P$ be a path from $u$ to $v$ in $G$. Then in polynomial time we can construct a walk $P^*$ in $G^*$ from $u$ to $v$ with the following properties:
\begin{enumerate}
\item $V(P^*)\subset N_{G^*}(V(P_B),3c_\Lip)$.
\item For any $a,b\in V(P)$ we have $V(P^*[a,b]\setminus Z^\cup)\subset V(P[a,b]) \subset V(P^*[a,b])$.
\end{enumerate}
\end{lemma}

\begin{proof}
Consider the edges of $P$ in order from $u$ to $v$. Let $pq$ be the next edge of $P$ in this order. If $pq\in E(G^*)$, then we add $pq$ to $P^*$. Suppose that $pq\not\in E(G^*)$.
By the definition of $G^*$ (\Cref{def:Gcsillag}) this implies that $p$ and $q$ are in different cliques of $\cP$, i.e., $p_\cP\neq q_\cP$, and $p_\cP q_\cP\not\in E(G^*_\cP)$. By \Cref{lem:GcsillagP}(\textit{ii}) we have that $p_\cP$ and $q_\cP$ can be connected by some path of length at most $c_\Lip$ in $G^*_\cP[Z]$. Such a path corresponds to a path of length at most $3c_\Lip$ in $G^*[Z^\cup]$ from $p$ to $q$ by \Cref{obs:GtoGP}. We add this path to $P^*$. Once all edges of $P$ are processed, the resulting walk $P^*$ satisfies the above properties.
\end{proof}


\subsection{Finishing the proof of \Cref{thm:advanced_struct}}

Next, we define a new partition $\hat \cY$ of $Y$ that is a coarsening of $\cY$. The classes $Y_i,Y_j\in \cY$ will be added to the same class of $\hat \cY$ if and only if they are in the same connected component of $G^*_\cP[Y]$, that is, we unify two classes of $\cY$ while there is an edge in $G^*_\cP$ connecting them. 

\begin{lemma}\label{lem:holehasuniquebrick}
If $\psi(a)=\psi(b)$ for some $a,b\in Y^\cup$ then $a,b$ are in the same brick.
\end{lemma} 

\begin{proof}
\Cref{lem:Yinner} implies that $Y\subset \Inn(\lambda+3c_\Lip^2)$, thus any edge of $G_\cP$ induced by $Y$ is an edge connecting a pair of inner cliques; in particular, any edge of $G_\cP$ incident to $Y$ has both its endpoints in the inner clique of the same brick. Now if the corresponding cliques $a_\cP$ and $b_\cP$ are in the same class of $\cY$, then they are in the same brick by the second part of \Cref{lem:Yinner}. Suppose that this is not the case, but $a_\cP$ and $b_\cP$ are in the same class of $\hat \cY$.

Recall that the definition of $\hat \cY$ was based on a coarsening of $\cY$ where two partition classes were unified if and only if they were connected by some edge in $G^*_\cP$. Consider now some edge $uv$ of $G^*_\cP$. By property (\textit{iii}) of \Cref{lem:GcsillagP} we have that either $uv$ is also an edge of $G_\cP$ ---and by the argument above we automatically have that $u$ and $v$ are $\lambda$-inner cliques of the same brick---, or there exists a path $P_{uv}$ of length at most $c_\Lip$ in $G_\cP$ from $u$ to $v$. By \Cref{obs:GtoGP} every vertex of this path are within distance $3c_\Lip$ from $u$, thus they are $(\lambda+3c^2_\Lip-3c_\Lip)$-inner vertices of the brick containing $u$. Thus $u$ and $v$ must be in the same brick. Since each class of $\hat \cY$ can be connected by such paths, we have that each class of $\hat \cY$ is contained in a single brick, and in particular, the objects in these cliques are original objects in the same brick of~$\pl M$.
\end{proof}

\begin{proof}[Finishing the proof of \Cref{thm:advanced_struct}]
Recall that we have applied \Cref{lem:sparseSteiner}, thus all classes in $\cP$ have size at most $2^{O(1/\eps\log(1/\eps))}$. 

We apply \Cref{thm:structure_pathgen} on $(G,T)$ to obtain a vertex set $W$ where $|W\cap B|=2^{O(1/\eps^{7.5})}$ for each brick $B$ of $\pl M$. Consider the consolidation $G^*_{\cX\hat \cY}:=G^*\parcon (\cX^\cup, \hat \cY^\cup)$ defined above. We 
will modify $G^*_{\cX\hat \cY}$ around the vertices of $\psi(Y^\cup)$ in order to construct $G'$.

\subparagraph*{The construction of $G'$.}

As seen in \Cref{sec:SteinerStruct} we will occasionally need to refer to a subpolygon $u\in B$ and the corresponding original object $u^\orig$ from the original graph $G^\orig$.
 
Consider a vertex $y\in \psi(Y^\cup)$, and let $B$ be the brick containing the class $A_y:=\psi^{-1}(y)\in \hat \cY^\cup$ (there exists such a brick by \Cref{lem:holehasuniquebrick}). We will add new vertices and edges to $G^*_{\cX\hat \cY}$ related to $y$. 
Next, for each vertex pair $u,v\in W\cap B$ we fix a shortest path in $B$ from $u^\orig$ to $v^\orig$; let $P_{uv}$ be the corresponding path in $G$ (using original objects of $V(G)$, from $u^\orig$ to $v^\orig$). We will now work exclusively on $V(G^\orig)$, so for the sake of brevity we will use $u,v$ to refer to $u^\orig,v^\orig$ throughout the definition of $G'$. Our goal is to model the behavior of this path in each set $A_y$ it passes through; in order to do that, we will identify entry and exit points $\bar a^{uv}_i,\bar b^{uv}_i$ for each path that pass through some sets $A_y$. Due to the small difference between $G$ and $G^*$ we need to define several paths. 
For each path $P_{uv}$ we define the corresponding walk $P^*_{uv}$ in $G^*$ as in \Cref{lem:PPcsillag}.

Next, we slice $P^*_{uv}$ as we walk from $u$ to $v$ as follows. Let $a_1$ be the first vertex on $P^*_{uv}\in Y^\cup$ when starting from $u$ such that $a_1\in Y^\cup$. We set $y_1=\psi(a_1)$, and let $b_1$ be the last vertex from $A_{y_1}$ in $P_{uv}$; consequently, $V(P^*_{uv})\setminus V(P^*_{uv}[a_1,b_1])$ is disjoint from $A_{y_1}$. See \Cref{fig:path_slicing}. (We remark that when $u,v\in A_{y_1}$, then $a_1=u $ and $b_1=v$, and no slicing will take place.) 
Assuming some sequence of vertices $a_1,b_1,\dots,a_{i-1},b_{i-1}$ along $P^*_{uv}$ is already defined, let $a_i$ be the first vertex of $P^*_{uv}$ after $b_{i-1}$ that falls into $Y^\cup$. We define $y_i=\psi(a_i)$, and let $b_i$ be the last vertex of $P^*_{uv}$ that falls into $A_{y_i}$. The procedure stops once the suffix of $P^*_{uv}$ after $b_i$ is disjoint from $Y^\cup$.

\begin{figure}
\centering
\includegraphics{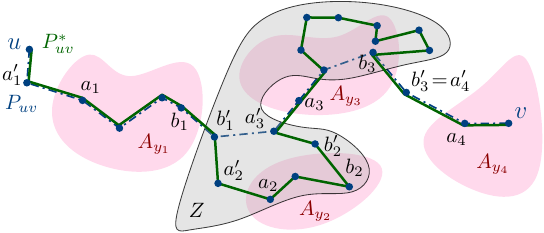}
\caption{Slicing a path $P^*_{uv}$ that can behave differently than $P_{uv}$ inside $Z$ (gray shaded region). We slice $P^*_{uv}$ according to its interaction with $A_{y_i}\subset Y^\cup$ (red shaded regions).}\label{fig:path_slicing}
\end{figure}

For a given $a_i\neq u$ let $(a')^{uv}_i=a'_i$ be the vertex preceding $a_i$ along $P_{uv}$. Similarly, let $(b')^{uv}_i=b'_i$ be the vertex succeeding $b_i$. We note that it may happen that $b'_i=a'_{i+1}$. Observe that by the definition of $G^*_{\cX\hat \cY}=G^*\parcon (\cX^\cup,\hat \cY^\cup)$ and the fact that $a'_i a_i$ and $b_i b'_i$ are edges of $G^*$ we have that $\psi(a'_i)y_i$ and $y_i\psi(b'_i)$ are edges in $G^*_{\cX\hat \cY}$. Moreover, we have that $a'_i,b'_i$ are not in $Y^\cup$. 
To see why, consider first the edge $a'_i a_i\in E(G^*)$. Since $a'_i$ is not in $A_{y_i}$, we have that $a'_i$ is not in the same class of $\hat \cY^\cup$ as $a_i$, and in particular, they are in different classes of $\cP$. Thus there is a corresponding edge $(a'_i)_\cP(a_i)_\cP\in E(G^*_\cP)$ where $v_\cP$ denotes the class of $\cP$ containing $v\in V(G)$.
Recall that by the definition of $\hat \cY$ we have that no two classes of $\hat \cY$ can be adjacent in $G^*_\cP$ (as such classes would have been unified), thus $(a_i)_\cP\in Y$ implies $(a'_i)_\cP\not \in Y$ and thus $a'_i\not\in Y^\cup$.
As a result, each path $P^*_{uv}$ is sliced into several subwalks using the vertices $a'_i,b'_i$. For a subwalk $Q^*$  whose endpoints are consecutive among $u,a'_1,b'_1,\dots,a'_k,b'_k,v$ we say that it is a \emph{$Y$-slice} if its internal vertices are in $V(Q^*)\subset Y^\cup$ and otherwise (when $V(Q^*)\subset B\setminus Y^\cup$) it is called an \emph{ordinary slice}. For a fixed brick $B$ let $\cQ_B^Y$ and $\cQ_B^\ord$ denote these collections of $Y$-slices and ordinary slices, and let $\cQ^Y$ and $\cQ^\ord$ denote the set of all slices in all bricks.

With the above definitions at hand, we are ready to define $G'$ based on $G^*_{\cX\hat \cY}$. As a result of the above definitions, each $\psi(a'_i),\psi(b'_i)$ are vertices of $G^*_{\cX\hat \cY}-\psi(Y^\cup)$, and the edges $\psi(a'_i)y_i$ and $y_i\psi(b'_i)$ exist in $G^*_{\cX\hat \cY}$.
Consider now some $y\in \psi(Y^\cup)$. Let $\AW(y)$ denote a collection of $|A_y\cap W|$ new vertices, and fix an arbitrary bijection $\nu:\AW(y)\rightarrow A_y\cap W$.

For any $y\in  \psi(Y^\cup)$, we define $\NAW(y)\subset V(G^*_{\cX\hat \cY})$ the following way.
Consider every brick $B$ and every pair $u,v\in B\cap W$ and, as above, define the sequences $y_i$, $a'_i$, $b'_i$ based on $P_{uv}$. For every $i$ such that $y=y_i$, let us introduce $\psi(a'_i),\psi(b'_i)\in V(G^*_{\cX\hat \cY})$ into $\NAW(y)$.
We extend $\nu$ to $\NAW(y)$ by setting $\nu=\psi^{-1}$ on this set, i.e., $\nu(\psi(a'_i))=a'_i$. Notice that as a result, $\nu$ may not be a bijection, but it remains a bijection on $W\cap \bigcup_{y\in \psi(Y^\cup)} A_y$.

Consider the vertex set $\AW(y)\cup \NAW(y)$, and for each pair $a,b\in \AW(y)\cup \NAW(y)$ such that there exists some $Y$-slice $Q\in \cQ^Y_B$ between $a$ and $b$, we add an an edge from $a$ to $b$ and  subdivide it with a new vertex denoted by~$m_{Q}$. Let $\Jum(y)$ denote the resulting graph. The role of this vertex will be to represent traversing the $Y$-slice $Q$.

We can now define $G'$ as $G^*_{\cX\hat \cY}-\big(\psi(Y^\cup)\big) \cup \bigcup_{y\in \psi(Y^\cup)} \Jum(y)$	, i.e., the graph that we get by replacing each vertex $y\in \psi(Y^\cup)$ with the corresponding graph $\Jum(y)$. Notice that $\Jum(y)$ is glued to the existing graph along the shared vertices $\NAW(y)=V(\Jum(y))\cap V(G^*_{\cX\hat \cY})$.

Let $\psi':(V(G)\setminus Y^\cup)\cup W \rightarrow V(G')$ be defined as follows: 
\[\psi'(v)=\begin{cases}
\psi(v) & \text{if } v\in V(G)\setminus Y^\cup \\
\nu^{-1}(v) & \text{otherwise, i.e., if }v\in Y^\cup \cap W.
\end{cases}.\]

\subparagraph*{Defining $\cP'$, $\omega$ and the proof of (\textit{i}).}
We define $\omega(t)$ for each $t\in V(G')$ as follows.

\[\omega(t):=
\begin{cases}
0 & \text{if } t\in \psi(Z^\cup \setminus Y^\cup)\\
|(V(Q[a,b])\setminus \{a,b\})\setminus Z^\cup| & \text{if $t=m_Q$ for some $Y$-slice $Q$ from $a$ to $b$}\\
1 & \text{otherwise, (i.e.} v\in (V(G)\setminus (Z^\cup \cup Y^\cup))\cup \bigcup_{y\in Y^\cup} \AW(y)\text{)}
\end{cases}\]

We can observe by \Cref{lem:Yinner} that $T$ is disjoint from $Y^\cup$, thus we can set $T':=T\parcon \cX^\cup$ as the desired terminal set in $G'$.

All vertices in $\Jum(y)$ can be associated with vertices or vertex pairs of $|W\cap B|$. By \Cref{thm:structure_pathgen} we have $|W\cap B|\leq 2^{O(1/\eps^{7.5})}$ so there are at most $O(|W\cap B|^2)= 2^{O(1/\eps^{7.5})}$ vertices in $\Jum(y)$.

To define the partition $\cP'$, for each clique $C\in \cP\setminus Y$ we add $\psi(C)$ to $\cP'$, that is, $\cP'$ and $\cP$ are the same when restricted to $V(G)\setminus (X^\cup \cup Y^\cup)$, and the vertices of $\psi(X^\cup)$ appear as singletons in $\cP'$. Moreover, for each $y\in \psi(Y^\cup)$ we add the set $\cP'(y):=V(\Jum(y))\setminus \NAW(y)$ to $\cP'$. This concludes the definition of $\cP'$. Notice that $\cP'$ is indeed a partition of $V(G')$, i.e., it consists of pairwise disjoint sets and its union covers $V(G')$. Moreover, a partition class of $\cP'\cap \cP$ has size at most $2^{O(1/\eps\log(1/\eps))}$, and any other class in $\cP'$ has size at most $2^{O(1/\eps^{7.5})}$, thus (\textit{i}) holds.

\begin{figure}
\centering
\includegraphics[scale=1]{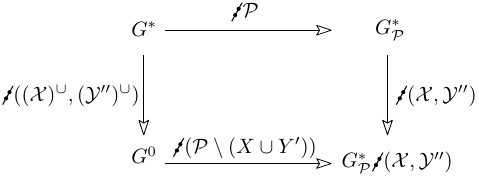}
\caption{Consolidations from $G^*$.}\label{fig:cont_diag}
\end{figure}

\subparagraph*{Bounding the treewidth, the proof of (\textit{ii}).}

In order to bound $\tw(G'\parcon \cP')$, first we show that $G'\parcon \cP'$ is a subgraph of $G^*_\cP\parcon(\cX,\hat \cY)$. Recall that $G'$ was based on a modification of $G^*_{\cX\hat \cY}$ around vertices of $\psi(Y^\cup)$, see \Cref{fig:cont_diag} for an illustration of relevant consolidations. First, we establish a bijection between vertices: if $v\in V(G')$ is the contraction of $A_v\in \cP'$ where also $A_v\in \cP$, then we simply identify $v$ with $\xi(A_v)\in V(G^*_{\cX\hat \cY})$. If $A_v\in \cX^\cup$, then it corresponds to $(A_v\parcon \cP)\parcon \cX\in V(G^*_\cP\parcon(\cX,\hat \cY))$. 
Otherwise, it must be a contraction of some $V(\Jum(y))\setminus \NAW(y)\in \cP'$ and we identify it with $(A_y\parcon \cP)\parcon \hat \cY \in V(G^*_\cP\parcon(\cX,\hat \cY))$. Now, we consider the edges. Recall that no edges of $G^*_{\cX\hat \cY}$ have been modified that are not incident to $Y\cup$, and moreover, $\cP'$ consolidates each set $\cP'(y)$, so the only edges of $G'$ that need to be considered are from $\cP'(y)$ to $\NAW(y)$. Recall however that the edge we add corresponds to some edge $y\psi((a')^{uv}_i)$ or $y\psi((b')^{uv}_i)$ of $G^*_{\cX\hat \cY}$, which ---after consolidating with $\cP\setminus(X\cup Y)$--- corresponds to an edge of $G^*_\cP\parcon(\cX,\hat \cY)$. Thus $G'\parcon \cP'$ is a subgraph of $G^*_\cP\parcon(\cX,\hat \cY)$ and $\tw(G'\parcon \cP')\leq \tw(G^*_\cP\parcon(\cX,\hat \cY))$. 

Now recall that $\hat \cY$ arises from $\cY$ by unifying classes that are connected in $G^*_\cP\parcon (\cX,\cY)$, so $G^*_\cP\parcon (\cX,\hat \cY)$ arises from $G^*_\cP\parcon (\cX,\cY)$ via the contraction of some edges. Thus $\tw(G^*_\cP\parcon(\cX,\hat \cY))\leq \tw(G^*_\cP\parcon(\cX,\cY))$. \Cref{lem:GcsillagP} implies that $\tw(G^*_\cP\parcon(\cX,\cY))=O(1/\eps^{18})$. The same lemma guarantees that a tree decomposition of $G^*_\cP\parcon(\cX,\cY)$ of width $O(1/\eps^{18})$ can be constructed in polynomial time, and in polynomial time we can adapt it into a tree decomposition of $G'\parcon \cP'$ as the latter is obtained from the former via a series of contractions and deletions.

\subparagraph*{Relating Steiner trees in $G'$ to $G$, the proof of (\textit{iii}).}

Let $S'$ be a Steiner tree for the instance $(G',\omega, T'=T\parcon \cX^\cup)$ of weight $\kappa=\omega(S')$. We first remove vertices from $S'$ in order to make it a minimal feasible solution. Notice that each vertex in $\cP'(y)$ that is subdividing some edge can only be included in $S'$ if both of its neighbors are in $S'$: indeed, $\cP'(y)$ contains no terminals and including a non-terminal vertex of degree $2$ without both of its neighbors  would contradict the minimality of $S'$.

Now, for each vertex $v\in S'\cap V(G)$ we add $v$ to the set $S$. Consider a vertex $x\in S'\cap \psi(X^\cup)$ and the corresponding vertex set $A_x \in \cX^\cup$ where $\psi(A_x)=x$. We add the set $A_x$ to $S$. Now we process the remaining vertices, which is $V(\Jum(y)\setminus \NAW(y)) \cap S'$ for each $y\in \psi(Y^\cup)$.
For each vertex $v\in S'\cap (\AW(y))$ we simply add the corresponding $\nu(v)\in V(G)$ to $S$. If $v=m_Q$ for some $Y$-slice $Q$ from $a$ to $b$, then we add all vertices of $Q(a,b)$ into $S$. Notice that $S$ gains at most $\omega(m_Q)$ internal vertices that are not in $Z^\cup$, and all weight 1 vertices of $S'$ grew $S$ by at most one vertex.

Finally, we extend $S$ with all vertices of $Z^\cup$. Notice that $G[S]$ is not necessarily connected, and it may be very large, as $Z^\cup$ is significantly larger than $Z$. However, we observe that $|S\setminus Z^\cup|\leq \omega(S')$. Next, we will show that there is a single connected component of $G[S]$ that contains all of $T$, and later we will modify $S$ so that feasibility is maintained and the size is decreased below the desired upper bound.

For now, we need to prove that all pairs of terminals are connected in $G[S]$. Notice that we can decompose any terminal-to-terminal path in $G'[S']$ into edges that are outside of the subgraphs $\Jum(y)$ as well as minimal subpaths inside $\Jum(y)$ that start and end in $\NAW(y)$. Consequently, the connectivity of terminals in $G[S]$ can be proven via the following statements:
\begin{enumerate}[label= (\arabic*) ]
\item For every edge $st$ of $G'[S']$ outside the subgraphs $\Jum(y)$ there is a path connecting some vertex of  $G[\psi^{-1}(s)]$ to some vertex of $G[\psi^{-1}(t)]$ in $G[S]$.
\item For every $s\in S'$ that is not in any set $\cP'(y)$ (i.e., $s\in(V(G')-Y^\cup)\parcon \cX^\cup$) we have that $G[\psi^{-1}(s)]$ is connected in $G[S]$.
\item If $G'[S']$ connects $s,t \in \NAW(y)$ inside $\Jum(y)$, then $G[S]$ connects some vertex of $\psi^{-1}(s)$ and some vertex of $\psi^{-1}(t)$.
\end{enumerate}

To prove (1), notice that $st$ is an edge of $G^*_{\cX\hat \cY}$, and consider the original edge $s^*t^*\in E(G^*)$ that gave rise to the edge $st$ after the consolidation of $G^*$ into $G^*_{\cX\hat \cY}=G^*\parcon (\cX^\cup,\hat \cY^\cup)$. (Notice here that $s^*\in \psi^{-1}(s)$ and $t^*\in \psi^{-1}(t)$). Let $s_\cP,t_\cP$ be the cliques of $\cP$ containing $s^*$ and $t^*$, respectively.
Observe that each of $s_\cP$ and $t_\cP$ can be in $X$ or $\cP\setminus (X\cup Y)$, and by \Cref{obs:Gcsillagosszehuz} either $s_\cP t_\cP\in E(G^*_\cP)$, or $s_\cP=t_\cP$. If they coincide, then $s^*,t^*$ are in the same clique of $G$ and thus they are connected in $G[S]$. Otherwise, we have that $s_\cP t_\cP\in E(G^*_\cP)$, so by \Cref{lem:GcsillagP} (\textit{iii}) either $s_\cP t_\cP\in E(G_\cP)$ or there is a path of length at most $c_\Lip$ connecting them in $G_\cP[Z]$. In the latter case, we have that $s^*$ and $t^*$ are connected in $G[Z^\cup]$ and since $Z^\cup \subset S$ they are connected in $G[S]$. 

To prove (2), if $s_\cP\in \cP\setminus (X\cup Y)$, then $\psi^{-1}(s)=s$ and it is trivially connected, so assume that $s$ is the consolidation of some set $A_x\in \cX^\cup$. Let $A_\cP$ be the corresponding class in $\cX$, i.e., the one where $A_x=A_\cP^\cup$. Recall that by the definition of $\cX$ we have that $A_\cP$ is connected in $H$. Thus, by the Lipshitz property, each edge of $H$ can be represented as a path of length at most $c_\Lip$ in $G_\cP$. Since both endpoints of this path are in $X$ (as $A_\cP\subset X$), all path vertices must be in $N_{G_\cP}(X,c_\Lip)\subset Z$. Thus $A_\cP$ is connected in $G_\cP[Z]$ and therefore $A_x$ is connected in $G[Z^\cup]$ as well as in $G[S]$.

To prove (3), consider some path in $\Jum(y)$ connecting $s$ to $t$, and without loss of generality assume that the internal vertices of this path are not from $\NAW(y)$. Such a path has the following form due to the structure of $\Jum(y)$:
\[s:=v_0,m_{Q_1},v_1,m_{Q_2},v_2,\dots, v_{k-1},m_{Q_k}, v_k,s'=v_k\]
where $v_1,\dots,v_{k-1}\in \AW(y)$.
For each subdividing vertex $m_{Q_i}$ there is a corresponding vertex pair $\nu(v_i),\nu(v_{i+1})$ in $G$ that we can connect with $Y$-slice $Q$. The concatenation of these slices gives a walk from $\nu(s)$ to $\nu(t)$ in $G^*$. Each edge of this walk that is not an edge of $G$ can be substituted with a path of length at most $3c_\Lip$ in $G^*$. Moreover, recall that by the definition of $\nu$ we have $\nu(s)\in \psi^{-1}(s)$ and analogously $\nu(t)\in \psi^{-1}(t)$. This concludes the proof of connectivity.

We now know that $G[S]$ has a connected component that includes $T$. Remove all other connected components from $S$, and apply the simplification of \Cref{lem:SteinSolTerminals} to obtain a solution $\hat S$ that is a subset of $S$ and contains at most $O(1)$ vertices from each clique of $\cP$. It follows that $|\hat S\cap Z^\cup|\leq |T\cap Z^\cup|+O(|Z|)$. We apply \eqref{eq:Zbound} and recall that our application of \Cref{lem:sparseSteiner} at the beginning of \ref{sec:SteinerFast} ensured that each clique of $\cP$ contains at most $O(1/\eps)$ terminals of $T$. Therefore,
\[|\hat S|\leq \omega(S')+|T\cap Z^\cup|+O(|Z|)=\omega(S')+O(\eps)\cdot \opt + O(\eps^2)\cdot \opt,\]
concluding the proof of (\textit{iii}).

\subparagraph*{Relating Steiner trees in $G$ to $G'$, the proof of (\textit{iv}).}

Consider the path generator $(S_0,\cW)$ with $V_\cW\subset W$ for $(G,T)$ guaranteed by \Cref{thm:structure_pathgen}, with cost at most $(1+O(\eps))\smt(G,T)$.
Let $S$ be the solution that $(S_0,\cW)$ generates where for each $(B,u,v)\in \cW$ we use the shortest path $P_{uv}$ that we fixed earlier in the definition of $G'$. Recall that $|S|$ is at most the cost of $(S_0,\cW)$ and thus $|S|\leq (1+O(\eps))\smt(G,T)$. Let $W_0\subset W$ denote the vertices of $V_\cW$ that are not $\lambda$-outer, i.e., $S_0\cup W_0$ is disjoint from $\Inn(\lambda)^\cup$. (More precisely, $W\subset V(G^\orig)$ consists of the original objects in $G$ rather than any of their subpolygons in $B$.)

Consider now $S_0\cup W_0$ as a vertex set in $G^*$, and extend it with $Z^\cup\setminus Y^\cup$ to get the vertex set $S^*_0$. We will show that edges induced by $G[S_0\cup W_0]$ remain connected in $G^*[S^*_0]$. The subgraph $G^*[S_0\cup W_0]$ is not necessarily connected, but all edges of $G[S_0\cup W_0]$ induced by some clique of $\cP$ are also present in $G^*[S^*_0]$ by the definition of $G^*$ (see \Cref{def:Gcsillag}). We claim that for any edge $ab\in E(G[S_0\cup W_0])$ where $a,b\in S_0 \cup W_0$ it holds that $a,b$ are in the same component of $G^*[S^*_0]$.

Consider the cliques $a_\cP$ and $b_\cP$ If $a_\cP= b_\cP$, then they are directly connected in $G^*$. Otherwise, notice that $a_\cP b_\cP \in E(G_\cP)$. If $a_\cP \not \in Z$ or $b_\cP \not \in Z$, then by \Cref{lem:GcsillagP} (\textit{ii}) we have that $a_\cP b_\cP \in E(G^*_\cP)$, so by the definition of $G^*$ (see \Cref{def:Gcsillag}) we have that $ab\in E(G^*)$. Suppose now that $a,b \in (S\cap Z^\cup) \setminus Y^\cup$ and thus $a_\cP,b_\cP \in Z$. \Cref{lem:GcsillagP} (\textit{ii}) we have that 
$a_\cP,b_\cP$ are connected in $G^*_\cP[Z]$ with a path $R^*$ of length at most $c_\Lip$. This corresponds to a path $R$ of length at most $c_\Lip^2$ in $G_\cP$ where $V(R^*)\subset V(R)$, and thus by \Cref{obs:GtoGP} a path of length at most $3c^2_\Lip$ in $G$. Since both endpoints of this path are in $B\setminus \Inn(\lambda)$, we have that $V(R^*)$ (and thus also $V(R)$) is disjoint from $\Inn(\lambda+3c_\Lip^2)$. In particular, by \Cref{lem:Yinner} we conclude that $V(R)$ is disjoint from $Y$, and thus $a,b$ can be connected within $G^*[(S_0 \cup W_0 \cup Z^\cup)\setminus Y^\cup]$ as claimed.
 
Next, we build the desired solution $S'$ to $(G',\omega,T'=T\parcon \cX^\cup)$. Include all vertices of $\psi'(V_\cW\cup S^*_0)=\psi'(V_\cW\cup (S_0\cup W_0 \cup Z^\cup) \setminus Y^\cup)$ in $S'$. For each triplet $(B,u,v)\in \cW$ let $P_{uv}$ be the shortest path in the brick $B$ of $G$ we have fixed; more precisely, $P_{uv}$ is the path where we replace each subpolygon on this shortest path with their orignal object from $G^\orig$. As before, we will use $u$ and $v$ to refer to $u^\orig,v^\orig$. For each $Y$-slice $Q$ of the walk $P^*_{uv}$ we add $m_Q$ to $S'$. Similarly, for each ordinary slice $Q$ we add $\psi'(V(Q))$ to $S'$. Observe that for a given triplet $(B,u,v)$ we have added some consolidation of the corresponding $G^\orig$-path's endpoints (as we have added $\psi'(V_\cW)$ as well as consolidations of paths for each of the slices of $P^*_{uv}$. Since consolidations preserve connectivity, one can verify that $\psi'(V_\cW\cup S^*_0)$ retains its connected components in $G'[S']$, thus in particular the vertices $T'=T\parcon \cX^\cup$ are in the same connected component of $G'[S']$.

To prove the weight bound, recall that $V(P_{uv})\subset V(P^*_{uv})$, and by \Cref{lem:PPcsillag} property 2, all vertices of $V(P^*_{uv})$ that are outside $V(P_{uv})$ are in $Z^\cup$. Notice in the definition of $\omega$ that vertices of $\psi(Z^\cup\setminus Y^\cup)$ have weight $0$, and vertices of $Z^\cup$ that are also internal vertices of some $Y$-slice $Q$ within $P^*_{uv}$ are also subtracted from $\omega(m_Q)$. Thus, the only internal vertices of $P^*_{uv}$ that can contribute weight are from $V_\cW\cup V(P_{uv})$. In total, the contribution because from the ordinary and $Y$-slices of $P^*_{uv}$ as well as all internal vertices separating these is at most $|V(P_{uv}\setminus \{u,v\})\setminus Z^\cup|=\dist_B(u,v)-1$.

It follows that the weight contribution from the triplets is at most $|V_\cW|+\sum_{(B,u,v)\in \cQ} (\dist_B(u,v)-1)$. The contribution from $\psi'(Z^\cup\setminus Y^\cup)$ is $0$, and the contribution from $|S_0|$ is at most $|S_0|$. Thus we get:
\[\omega(S')\leq |S_0|+|V_\cW|+\sum_{(B,u,v)\in \cQ} (\dist_B(u,v)-1)\leq (1+O(\eps))\smt(G,T),\]
as required for (\textit{iv})(a).

To prove (\textit{iv})(b), let $S''$ be the minimum weight solution for $(G',\omega,T\parcon \cX)$ among the subsets of $S'$. Consdier a class $C'\in \cP'$. We have three cases.
\begin{description}
\item[Case 1.] $C'=\Jum(y)\setminus \NAW(y)$ for some $y\in \psi(Y^\cup)$.\\
We have by \Cref{thm:structure_pathgen}(\textit{iii})(d) that there are at most $O(1/\eps^{7.5})$ triplets for any given brick $B$ in the path generator. Since each path $P_{uv}$ will have at most one corresponding slice of $P^*_{uv}$ going through any fixed class of $\hat \cY$ (by the slicing procedure), we have that $S'$ contains at most $O(1/\eps^{7.5})$ vertices from $C'$.
\item[Case 2.] $C'$ is a singleton containing some consolidated vertex.\\
Then $|C'\cap S'|\leq 1$.
\item[Case 3.] $C'\in \cP$ is an original clique of $G$.\\
Recall that $|T\cap C'|\leq 1/\eps$ due to our usage of \Cref{lem:sparseSteiner}. Moreover, \Cref{thm:structure_pathgen}(\textit{ii})(c) implies $|V_\cW\cap C'|\leq O(1/\eps^{7.5})$. Recall that $C'$ has $O(1)$ cliques of $\cP$ that are neighboring in $G$, and each of those cliques may remain in $G'$, but they might also become part of distinct classes of $\hat \cY$. For a given class $y\in \psi(\hat \cY^\cup)$ we have already shown that there are at most $O(1/\eps^{7.5})$ vertices selected from $\Jum(y)\setminus \NAW(y)$, thus $C'$ has at most $O(1/\eps^{7.5})$ edges induced between itself and such classes. For every other neighboring class of $\cP'$ the neighboring class is a clique, thus it is again sufficient to have at most one edge induced between $C'$ and such neighbors; see \Cref{lem:sparsenSteinercliques} for the argument. By the minimality of $S'$, each vertex in $C'\cap S'$ must either be a terminal or it must induce an edge going into neighboring classes of $\cP'$. Consequently, there are at most $O(1/\eps^{7.5})$ vertices in $C'\cap S'$.
\end{description}

Thus in each case there $O(1/\eps^{7.5})$ vertices in $C'\cap S'$, concluding both the proof of (\textit{iv}) and of the theorem.
\end{proof}

\subsection{A faster EPTAS for Steiner Tree}

The following lemma presents a general way to compute Steiner trees in a graph with a partitioned vertex set. 
Let $G$ be graph and let $\cR$ be a partition of $V(G)$. We say that a vertex set $S$ is \emph{$(k,\cR)$-bounded} if for each $C\in \cR$ we have $|C\cap S|\leq k$.

\begin{lemma}\label{lem:partitioned_steiner_algo}
Let $G$ be an $n$-vertex graph with nonnegative vertex weighting $\omega$ and let $\cR$ be a partition of $V(G)$ where each class $C$ of $\cR$ has at most $r$ vertices. Let $K$ be a given set of terminals, and suppose that we are given a tree decomposition of $G\parcon \cR$ of treewidth $w$. Then we can find the minimum weight $(\kappa,\cR)$-bounded Steiner tree of $(G,K)$ in $(rw)^{O(\kappa w)}\poly(n)$ time.
\end{lemma}

\begin{proof}
We modify the standard treewidth-based dynamic programming for Steiner tree to this setting. Observe that, based on the tree decomposition, we can make a tree decomposition $\cT_0$ of $G$ by considering each bag $X_t$ of the given decomposition, and substituting each class $C\in X_t$ of $G\parcon \cR$  with all vertices inside $C$. The resulting tree decomposition is a valid decomposition of $G$ and has width $rw$. Observe that the trees we need to consider intersect each class in at most $\kappa w$ vertices.

The standard dynamic programming algorithm for Steiner tree~\cite{fptbook,CHIMANI201267} proceeds bottom-up on a (slightly tweaked) \emph{nice} tree decomposition. We will rely on the notation of \cite[Section 7.3.3.]{fptbook}. For a bag $X_t$ let $G_t$ denote the subgraph of $G$ induced by the subtree rooted at the bag $X_t$. Then, for each $X\subset X_t$ and partition $\cP$ of $X$, there is a subproblem $c[X_t,X,\cP]$, which is the minimum number of edges in a forest $F$ in $G_t$ such that:
\begin{enumerate}
\item $X_t\cap V(F)=X$.
\item Every terminal from $K\cap V(G_t)$ is in $V(F)$.
\item The forest $F$ has $|\cP|$ connected components, and for each class $X^i$ of $\cP$ there is a corresponding connected component $F^i$ of $F$ such that $X^i=V(F_i)\cap X_t$.
\end{enumerate}

First, we compute a nice tree decomposition based on $\cT_0$ in polynomial time using standard techniques~\cite{fptbook}; notice that the resulting nice tree decomposition $\cT$ has bags that are subsets of existing bags, so in particular, it still holds that for each bag $X_t$ of $\cT$ there are at most $w$ classes of $\cR$ intersecting~$X_t$.

We can use the same algorithm with two small modifications.
First, we will only iterate over sets $X\subset X_t$ where $|X\cap C|\leq \kappa w$, as any $(k,\cR)$-bounded Steiner tree will intersect $X_t$ in at most $\kappa w$ vertices. Second, instead of $c[X_t,X,\cP]$ denoting the number of edges in $F$, it will denote the minimum vertex weight of $V(F)$ satisfying the same properties.

One can verify that the same dynamic programming algorithm works to compute the minimum weight $(\kappa,\cR)$-bounded Steiner tree. The original algorithm has $|X_t|^{O(|X_t|)}$ states for the bag $X_t$, as it needs to consider all $2^{|X_t|}$ subsets $X$ and up to $|X_t|^{O(|X_t|)}$ partitions on each subset.
The dynamic programming needs to consider all pairs of states from two child bags to compute the value for a state, so the running time is $\tw^{O(\tw)}n$ for an $n$-vertex graph with a tree decomposition of width $\tw$. In our case, we only have  $|X_t|^{O(kw)}=(rw)^{O(\kappa w)}$ subsets to consider, each of which has $(\kappa w)^{O(\kappa w)}\leq (rw)^{O(\kappa w)}$ partitions. Thus the number of states is $(rw)^{O(\kappa w)}$ and the running time is  $(rw)^{O(\kappa w)}\poly(n)$.
\end{proof}

\SteinerFast*

\begin{proof}
We apply \Cref{thm:atalakitas} to obtain an $\alpha$-standard graph $G\in \cI_\alpha$ whose representation has polynomial complexity.
Apply \Cref{thm:advanced_struct} to obtain a graph $(G',\omega,T')$ and vertex partition $\cP'$. \Cref{thm:advanced_struct}(\textit{ii}) gives a tree decomposition $\cT'$ of $G'\parcon \cP'$ of width $w=O(1/\eps^{18})$.
By \Cref{thm:advanced_struct}(\textit{i}) we have that each class of $\cP'$ has size at most $r=2^{O(1/\eps^{7.5})}$.
Finally, by  \Cref{thm:advanced_struct}(\textit{iv})(\textit a),(\textit b) there exists a solution $S'$ to the instance $(G',\omega,T')$ of weight $(1+O(\eps))\opt$ that is $(\kappa,\cP')$-bounded where $\kappa=O(1/\eps^{7.5})$.
Thus by \Cref{lem:partitioned_steiner_algo} we can find a solution $S'$ of weight at most $(1+O(\eps))\opt$ in $(rw)^{O(\kappa w)}\poly(n)$ time, where the first term can be bounded as:
\[(rw)^{O(\kappa w)}=\Big(2^{O(1/\eps^{7.5})} \cdot O(1/\eps^{17})\Big)^{O(1/\eps^{7.5})\cdot O(1/\eps^{18})}=2^{O(1/\eps^{33})}.\]
Using \Cref{thm:advanced_struct}(\textit{iii}) we now construct a solution of $(G,T)$ based on $S'$ whose weight can be bounded as
\[|S|=\omega(S')+ O(\eps)\cdot \opt= (1+O(\eps))\opt\]
where we used that there are at most $O(1/\eps)$ terminals in each clique.
The dominant term in the running time is the algorithm of \Cref{lem:partitioned_steiner_algo}, which takes $2^{O(1/\eps^{33})}\poly(n)$ time.
\end{proof}

\section{Lower bounds and connection to planar problems}
\label{sec:lower}

\subsection{APX-hardness results}

The goal of this subsection is to prove the following theorem.

\APXhard*

A \emph{split graph} is a graph whose vertex set can be partitioned into $A\dot\cup B$ such that $A$ induces a clique and $B$ induces an independent set.
Our proof of \Cref{thm:apxhard} will rely on the fact that two of the above graph classes are able to represent subdivisions of bounded degree graphs, while the other two are able to represent arbitrary split graphs. 

\begin{lemma}\label{lem:splitapxhard}
\textsc{Steiner Tree} and \textsc{Subset TSP} are APX-hard in split graphs.
\end{lemma} 

\begin{proof}
The traveling salesman problem (TSP) is APX-hard in discrete metric spaces where all distances are $2$ or $3$~\cite{Trevisan00}. One can make an equivalent instance of \textsc{Subset TSP} in a split graph: let $M=(X,\dist)$ be a metric space on $n$ points where distances of distinct points are either $2$ or $3$. Then let $G$ be a split graph with a clique $C$ of size $\binom{n}{2}+n$ indexed by edges of the complete graph on $X$ plus a dummy vertex $d_x$ for each $x\in X$. The independent set $I$ of $G$ is of size $n$ and indexed by $X$. We have an edge from vertex $v_x\in I$ to $w_{yz}\in C$ if and only if $x\in {y,z}$ and $\dist(y,z)=2$, and we connect each $v_x\in I$ to $d_x\in C$. Since the dummy vertices are neighbors and together with $I$ they form a perfect matching, we have that any pair of vertices in $I$ has distance at most $3$. On the other hand, $v_x,v_y\in I$ has a common neighbor in $G$ if and only if $\dist(x,y)=2$, and as a result, $\dist_G(v_x,v_y)=2$ holds. Thus when considering the terminal set $T=I$ in $G$, the metric of $G$ for the terminals is isomorphic to $M$. As a result, a subset TSP tour of $(G,T)$ is in a direct one-to-one correspondence with a TSP tour of $M$ and they have the same length. Since the reduction is polynomial, this shows that \textsc{Subset TSP} is APX-hard in split graphs.

For Steiner Tree our reduction is based on the fact that \textsc{Vertex Cover} is APX-hard on connected graphs of maximum degree $3$~\cite{AlimontiK00}. In an instance of vertex cover, we are given some graph $G$ and the goal is to find a minimum subset of vertices $S$ such that each edge of $G$ is incident to some vertex in $S$.
Notice that in a connected graph of maximum degree $3$ each vertex can cover at most $3$ edges, and there are at least $n-1$ edges to cover, thus the minimum vertex cover has size at least $(n-1)/3$.

Based on $G$, we can build a split graph $S=(I,C)$ where the independent set $I$ has a vertex $u_{xy}$ for each edge $xy$ of $G$, and the clique $C$ has a vertex $v_z$ for each vertex $z$ of $G$. We add an edge from $v_z\in I$ to $w_{xy}\in C$ if and only if $z$ is incident to the edge $xy\in E(G)$, that is, $z=x$ or $z=y$. Let $T=I$ be the set of terminals, and consider the \textsc{Steiner Tree} instance $(S,T)$. It is easy to see that a vertex cover of size $k$ in $G$ corresponds to a Steiner tree of size $|I|+k$ in $X$: namely, we take the terminals $T=I$ and the vertices in $C$ that correspond to vertices of the vertex cover. On the other hand, any Steiner tree must contain all terminals ($I$) and each vertex $w_{xy}$ must be connected to the rest of the tree either by having $v_x$ or $v_y$ in the tree, thus the vertices of $G$ corresponding to tree vertices of $C$ form a vertex cover. This establishes a bijection between vertex covers of $G$ and Steiner trees of $G$ with an additive size difference of $|I|$. Notice moreover that $n-1<|I|<3n/2$ and $(n-1)\leq k \leq n$ implies that this is a valid L-reduction. This concludes the proof that \textsc{Steiner Tree} is APX-hard in split graphs.
\end{proof}

Our next goal is to prove the following lemma.

\begin{lemma}\label{lem:deg3APXhard}
\textsc{Steiner Tree} and \textsc{Subset TSP} are APX-hard in graphs of maximum degree $3$.
\end{lemma}

We note that the claim about \textsc{Subset TSP} follows from the APX-hardness of \textsc{Graphic TSP} on the same graph class~\cite{KarpinskiS15}; we give an independent proof specifically for \textsc{Subset TSP}.

We will give a PTAS reduction from the same problems in arbitrary graphs.  Notice that the optimum of both \textsc{Subset TSP} and \textsc{Steiner Tree} in an $n$-vertex graph is less than $2n$. In particular, a spanning tree of $G$ is a feasible Steiner tree on at most $n$ vertices, and the walk that travels this tree twice is a feasible subset TSP tour of length $2n-2$. Thus for any $\eps\leq\frac{1}{2n-2}$ a solution is $(1+\eps)$-approximate if and only if it is optimal. Thus, for $\eps<\frac{1}{2n-2}$, we set $\eps:=\frac{1}{2n}$ and for the rest of the reduction, we assume without loss of generality that $\eps\geq\frac{1}{2n}$.

An \emph{$\ell$-subdivision} of a graph $G$ is the graph $G_t$ obtained by replacing its edges with paths of length $\ell$. Let $S'_{uv}$ be the length-$\ell$ path of $G_\ell$ that replaced the edge $uv$ of $G$. Notice that there is a natural one-to-one correspondence between minimal Steiner trees of $G$ and Steiner trees of $G_\ell$, where the terminals in $G$ and $G_\ell$ are the same. If $X$ is a Steiner tree for $(G,T)$ of size $k$, then consider any spanning tree $E_X$ of $X$, and let $X'$ be the set of vertices occurring on paths corresponding to $E_X$ in $G'$. Then $X'$ spans a feasible Steiner tree of $(G_\ell,T)$ of size $|X'|=\ell\cdot (k-1) + k$. Similarly, a minimal Steiner tree of $(G',T)$ consists of some collection of paths $S'_{uv}$ for some $uv\in E(G)$; if it contains $k-1$ such paths, then its size is $\ell(k-1)+k$, and there is clearly a corresponding set $X\supset T$ in $V(G)$ that consists of $k$ vertices.

We say that a subset TSP walk $W$ of the instance $(G,T)$ is minimal if it can be decomposed into $|T|$ shortest paths of $G$ that connect $T$ in some cyclic permutation. (Note that one can change any feasible walk $W$ into a minimal feasible walk that is not longer than $W$ in polynomial time.) We claim that there is a simple bijection between minimal subset TSP walks of $(G,T)$ and $(G',T)$: one can consider the length-$\ell$ paths corresponding to the edges of a minimal walk $W$ of $(G,T)$ of length $k$: this leads to a minimal walk $W'$ of $(G',T)$ of length $\ell k$. Similarly, any minimal walk of $(G',T)$ consists of shortest paths between terminal pairs, which in turn are concatenations of length-$\ell$ paths.

We can now prove the following lemma.

\begin{lemma}\label{lem:subdivisionAPXhard}
Let $\cG$ be a graph class where \textsc{Subset TSP} and \textsc{Steiner Tree} are APX-hard. Then for any positive integer $\ell$ they are also APX-hard in the graph class $\cG_\ell$ consisting of the $\ell$-subdivisions of the graphs in $\cG$.
\end{lemma}

\begin{proof}
Given the above polynomial-time-computable bijection between minimal feasible Steiner trees and minimal feasible subset TSP paths of $G\in \cG$ and $G_\ell\in \cG_\ell$, it remains to show that the approximation ratio is preserved. In case of subset TSP paths the approximation ratios are identical as the corresponding walk lengths differ by a multiplicative factor of $\ell$. For Steiner tree, the function $f(k)=\ell(k-1)+k$ is monotone, thus an optimum tree of $(G,T)$ corresponds to an optimum tree in $G_\ell,T$. If the optimum tree of $(G_\ell,T)$ has $\kappa-1$ paths, then a minimal $(1+\eps)$-approximate tree of $(G_\ell,T)$ has size at most $(1+\eps)(\ell(\kappa-1)+\kappa)$. Suppose for the sake of contradiction that this approximate tree has $k-1$ paths where $k>(1+\eps)\kappa$ paths. Then
\[(1+\eps)(\ell(\kappa-1)+\kappa)\geq f(k) > f((1+\eps)\kappa)=(1+\eps)(\ell\kappa-1)+\kappa),\]
which simplifies to $-(1+\eps)\ell>-(1+\eps)$, which is a contradiction as $\ell\geq 1$. Thus the approximate tree has at most $(1+\eps)\kappa-1$ paths, and the corresponding tree has at most $(1+\eps)\kappa$ vertices, concluding the proof.
\end{proof}

In order to prove \Cref{lem:deg3APXhard} we will need to decrease the vertex degrees of graphs using the following method.

A \emph{subcubic dilation} of a vertex set $S$ in the graph $G$ replaces each vertex $v\in S$ with a path~$H_v$ on at least $\deg(v)$ vertices, and connects the vertices of $N_G(v)$ to distinct vertices of~$H_v$. Consequently, if~$G'$ is a subcubic dilation of~$S$ in~$G$ where $S$ contains all vertices of $G$ of degree at least $4$, then~$G'$ is a graph of maximum degree $3$ where all vertices of degree $3$ are on the paths~$H_v$, and contracting the paths~$H_v$ results in the graph $G$. We say that a subcubic dilation~$G'$ of~$S$ in~$G$ is~$f(n)$-bounded if all paths~$H_v$ have at most $f(n)$ vertices.

\begin{lemma}\label{lem:reductiontosubdivdilation}
Let $G'$ be the graph obtained as a subcubic dilation of $V(G)$ in a subdivision $G_\ell$ of $G$. Given $T\subset V(G)$, for each $v\in T$ let $v'$ be an arbitrary vertex of $H'_v$, and let $T'$ be the terminal set in $G'$ containing the points $v'$ for each $v\in T$. Assume moreover that $G'$ is $b=b(n)$-bounded where $b\leq \frac{\ell}{4n^2}$ and $\ell=n^{O(1)}$. Then the following hold.
\begin{enumerate}[label=(\roman*)]
\item If $(G,T)$ has Steiner tree of size $s$, then $(G',T')$ has a Steiner tree of size at most $(s-1)\ell+\eps\ell$. If $(G,T)$ has  a subset TSP walk of size $s$, then $(G',T')$ has a subset TSP walk of size at most $s\ell+\eps\ell$.
\item
If $X'$ is a Steiner Tree for $(G',T')$ where $|X'| < s\ell$, then in polynomial time we can construct a Steiner tree $X$ for $(G,T)$ of size at most $s$.
If $W'$ is a subset TSP walk for $(G',T')$ of size $|W'|<(s+1)\ell$, then in polynomial time we can construct a subset TSP walk $W$ of~$(G,T)$ of size at most $s$.
\end{enumerate}
\end{lemma}


\begin{proof}
Let $G$ be an arbitrary graph on $n$ vertices, and subdivide each edge $uv$ of $G$ into a path $S'_{uv}$ of length $6n^2$. Replace each original vertex $v$ of $G$ by a path $H'_v$ with $\deg(v)$ vertices $v_1,\dots,v_{\deg{v}}$, so that each edge incident to $v$ is now incident to distinct edges of the new path. Let $G'$ denote the resulting graph. Notice that $G'$ is constructed in polynomial time and has maximum degree $3$. Let $T$ be the set of terminals in $G$, and let $T'$ be the vertex set in $G'$ given by the vertices $v_1$ for each $v\in T$.

Using the correspondence between vertices of $G$ and the paths $H'_v$ as well as between edges of $G$ and the length-$\ell$ paths $S'_{uv}$ it is easy to construct the desired tree and walk. Let $X$ be a Steiner tree  of size $|X|=k$. For each vertex $v$ in of $G$ take all vertices of $H'_v$ into a set $X'$, and for each edge $uv$ in some fixed spanning tree of $G[X]$ we take all vertices of the path $S'_{uv}$ into $X'$. Then $X'$ induces a connected subgraph of $G'$ containing all terminals of $T'$. It has $(s-1)\ell$ vertices that are internal vertices of some $S'_{uv}$ and it has at most $sb$ vertices from the sets $H'_u$ as each such set has at most $b$ vertices. Since $s<n$, $b\leq \frac{\ell}{4n^2}$ and $\eps\geq 1/(2n)$, we have that $X'$ has at most
\[(s-1)\ell+sb\leq (s-1)\ell+ \frac{\ell n}{4n^2}<(s-1)\ell+\eps\ell\]
such vertices.

Similarly, for a walk $W$ we take the paths $S'_{uv}$ corresponding to each edge of the walk into a walk $W'$. For each vertex $v$ between consecutive edges $uv$ and $vw$ of $W$ we take the walk along $H'_v$ from the entry point of $S'_{uv}$ to the leftmost point of $H'_v$, and then to the exit point of $S'_{vw}$. This defines a closed walk $W'_0$ in $G'$. Notice however that some edges of a path $H'_v$ may be traversed by the walk several times in the same direction. Every such double traversal can be eliminated while retaining the closed walk: suppose that $ab$ is traversed twice from $a$ to $b$ by $W'_0$. Then the directed walk $W'_0$ can be written as $x_1,\dots,x_a, a, b, y_1,\dots, y_b, a,b$. We can then change this to the walk $x_1,\dots,x_a, a, y_b,y_{b-1}, \dots, y_1, b$ that is shorter but visits the same set of vertices. Using such simplifications exhaustively, we get a walk $W'$ that traverses each edge of each $H'_v$ at most once in each direction, so altogether each edge of $H'_v$ is traversed at most twice. Now $W'$ has the desired size: in addition to the $s$ edge paths of length $\ell$, we included at most $2nb$ edges incident to the sets $H'_v$. Thus $b\leq \frac{\ell}{4n^2}$ and $\eps\geq 1/(2n)$ implies:
\[|W'|\leq  s\ell + 2nb \leq s\ell + \frac{2n}{4n^2} \leq s\ell+\eps\ell.\]

(\textit{ii}) Suppose that $X'$ has vertices from $k$ distinct sets $H'_u$; let $X$ contain all such vertices $u$, which is clearly feasible and of size $k$. Notice that $X'$ must contain enough paths $S'_{uv}$ to connect these vertices, i.e., at least $k-1$ such paths. It follows that $X'$ has size at least $|X|\geq (k-1)\ell$. Consequently $|X'| < s\ell$ implies that $k-1<\ell$ and thus the constructed Steiner tree satisfies $|X|=k\leq \ell$, as required.

Similarly, we may assume without loss of generality that $W'$ is a minimal feasible walk in the sense that it can be decomposed into $n$ interior-disjoint paths connecting the terminals according to some cyclic permutation. Recall that as a result, if $W'$ contains some edge of $G'$ incident to some edge path $S'_{uv}$ then it must traverse the entire edge path. For each path $P'$ of the decomposition we can create an analogous path $P$ in $G$ by simply taking the edges $uv$ of $G$ such that the path $S'_{uv}$ representing the edge $uv$ are traversed by $P'$, and we set $W$ to be the concatenation of these paths. If $W'$ traverses $k$ edge paths, then it has length at least $k\ell$, thus $|W'|<(s +1)\ell$ implies $|W|=k \leq \ell$ as required.
\end{proof}

We can now easily prove \Cref{lem:deg3APXhard}.

\begin{proof}[Proof of~\Cref{lem:deg3APXhard}]
Let $G$ be an aribtrary graph, and let $G'$ be a subcubic dilation of the subdivision $G_{4n^3}$, where each vertex $v$ has been dilated into a path $H'_v$ of length at most $n-1$. Consequently, $G'$ is $b=(n-1)$-bounded, and satisfies the conditions of \Cref{lem:reductiontosubdivdilation}. Assuming that the optimum Steiner tree of $(G,T)$ has size $s$, a minimum Steiner tree of $(G',T')$ has size at most $(s-1)\ell+\eps\ell$. Thus a $(1+\eps)$-approximate solution has size at most $(1+\eps)((s-1)\ell+\eps\ell)<(1+\eps)s\ell$ as $\eps<1$ without loss of generality. Now (\textit{ii}) implies that in polynomial time we can create a Steiner tree for $(G,T)$ of size at most $(1+\eps)s$.

Similarly, if $(G,T)$ has an optimum subset TSP walk of size $s$, then the optimum subset TSP walk of $(G',T')$ has length at most $s\ell+\eps\ell$, thus a $(1+\eps)$-approximate solution is of size at most $(1+\eps)(s\ell+\eps\ell)< (1+\eps)(s+1)\ell$. Consequently, (\textit{ii}) constructs a subset TSP walk of size at most $(1+\eps)(s+1)-1<(1+2\eps)s$. This concludes the PTAS reduction.
\end{proof}

We can also wrap up the proof of \Cref{thm:apxhard} by realizing split graphs and (subdivisions of) graphs of degree at most $3$ as intersection graphs.

\begin{proof}[Proof of \Cref{thm:apxhard}.]
We claim that in cases (a) and (d) one can realize some subdivision of any graph of degree at most $3$ that has at least $10$ vertices. If we can do this, then \Cref{lem:deg3APXhard} and \Cref{lem:subdivisionAPXhard} implies that both \textsc{Steiner Tree} and \textsc{Subset TSP} are APX-hard on these graph classes.

Similarly, we claim that in cases (b) and (c) we can realize any split graph, so by \Cref{lem:splitapxhard} we have that both problems are APX-hard on these intersection graphs.

Let us now define these realizations.

\textbf{(a)} We show that a $28n$-subdivision of any graph of maximum degree $3$ can be realized as the intersection graph of the following type of objects. Each object consists of two disks of diameter $1/2$, whose centers are at distance $1$ from each other. Moreover, each disk center's coordinates are integral or half-integral. These are clearly intersection graphs of similarly sized fat but disconnected objects.

Our strategy is to first construct a drawing of $G$ in a grid (with crossings) where all edges are represented by paths of equal lengths. We will then replace these paths with the objects in such a way that crossings in the drawing will be ``jumped'' over by the corresponding objects. Let $G$ be an arbitrary graph of maximum degree $3$. Let $v_1,\dots,v_n$ denote its vertices, and let $e_1,\dots,e_m$ denote its edges where $m\leq 3n/2$. To construct a grid drawing, we first assign each vertex $v_i$ to the grid point $(6i,0)$. When $v_i$ has degree $\deg(v_i)\leq 3$, then $v_i$ will be represented by a horizontal segment of length $2(k-1)$ whose left endpoint is $(6i-2,0)$, and the grid points $(6i-2,0),(6i,0),(6i+2,0)$ will be used as starting points of grid paths representing the edges.

For each $j=1,\dots,m$, and corresponding edge $e_j=v_{j_1}v_{j_2}$ we create a path from one of the grid points $(6j_1-2,0),(6j_1,0),(6j_1+2,0)$ that are all on the segment of $v_{j_1}$. Similarly, the path is assigned to one of the grid points   $(6j_2-2,0),(6j_2,0),(6j_2+2,0)$ from the segment of $v_{j_2}$. Since the maximum degree is $3$, we can assign distinct grid points from the segment of each $v_i$ to distinct edges incident to $v_i$. Let $x_{j_1}<x_{j_2}$ denote the $x$-coordinates of these assigned grid points. Then the grid path $P_j$ corresponding to $e_j$ is defined as the union of the following axis-parallel grid path $P_j$ through the following sequence of grid points:
\[(6j_1,0),(x_{j_1},0),(x_{j_1},y_j),(x_{j_1}+1,y_j),(x_{j_1}+1,j),(x_{j_2},j),(x_{j_2},0),(6j_2,0),\]
where $y_j\in [2n, 14n]$ is set so that the total length of the path is exactly $14n$. See \Cref{fig:disconn_realize}. We claim that such a value $y$ exists. The points $(x_{j_1},0)$ and $(x_{j_2},0)$ are both even-coordinate points, and their horizontal distance is at least $2$ and at most $6n$, and $j\leq m\leq \frac32 n$. Thus the path
\[(6j_1,0),(x_{j_1},0)(x_{j_1},j),(x_{j_2},j),(x_{j_2},0),(6j_2,0)\]
has even length at least $4$ (when $j=1$ and $x_{j_2}-x_{j_1}=2$) and at most $2j+6n+4\leq 9n+4<10n$. The length of $P_j$ is $2(y_j-j)$ longer than this path, therefore the desired $y_j\in [2n, 14n]$ exists such that the length of $P_j$ is exactly $14n$.

\begin{figure}
\centering
\includegraphics[width=\textwidth]{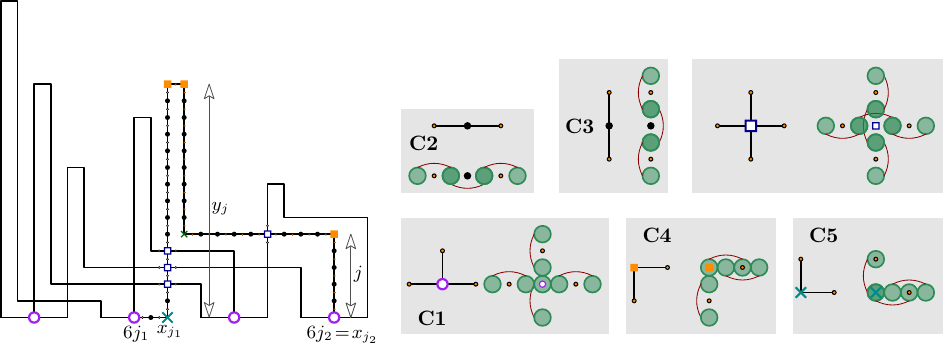}
\caption{Realizing the subdivision of a $3$-regular graph ($K_{3,3}$) using disconnected objects. Each object is a pair of disks, depicted with a connected red line that is not part of the object. The different types of vertices are realized with appropriately placed object pairs.}
\label{fig:disconn_realize}
\end{figure}

Next, we scale the grid drawing by a factor of $2$, so that each edge of the original drawing is realized by two grid edges in the new drawing. Let $D$ denote this new drawing. Consequently, $D$ is a drawing of the $28n$-subdivision of the original graph where crossings can only happen at a point $(a,b)$ when $a,b$ are both even and $b\geq 2$.

To realize the subdivision $G_{28n}$ with the disconnected objects (the disk pairs), notice first that the vertices of the $D$ can be categorized as follows:
\begin{enumerate}[label=\textbf{C\arabic*}]
\item vertices of degree $3$ (with $2$ horizontal edges and one vertical edge on the top) \label{it:deg3}
\item vertices of degree $2$ where both incident edges are horizontal \label{it:deg2h}
\item vertices of degree $2$ where both incident edges are vertical\label{it:deg2v}
\item vertices with a vertical incident edge on the top and a horizontal incident edge\label{it:turnbott}
\item vertices with a vertical incident edge on the bottom and a horizontal incident edge\label{it:turntop}
\end{enumerate}

For the vertex $(a,b)$ of type~\ref{it:deg3} or~\ref{it:turnbott} we realize it with the object whose disk centers are $(a,b)$ and $(a-1,b)$. If $(a,b)$ is of type~\ref{it:turntop}, then we realize it with the disk pair of centers $(a,b)$ and $(a+1,b)$. If $(a,b)$ has type~\ref{it:deg2h}, then it is realized with the disks $(a-\frac12,b),(a+\frac12,b)$, and finally, if it is of type~\ref{it:deg2v}, then it is realized with the disks $(a,b-\frac12),(a,b+\frac12)$. It is routine to check that all adjacencies of $G_{28n}$ are realized, while every crossing of $D$ is ``jumped over'' by our disks, i.e., there is no intersection with the horizontal and vertical paths of $G_{28n}$ at the crossings of $G$.

\textbf{(d)} We can realize the subdivision $G_{14n}$ with balls of diameter $1$. In fact, we will realize $G_{14n}$ as an induced subgraph of the $3$-dimensional grid. Clearly placing unit diameter balls at the vertices of such a grid graph gives the desired ball graph. The idea is similar to that seen in case (a): we use the same assignment of vertices $v_i$ to the point $(6i,0,0)$ on the $x$-axis, and we also use the points $(6i-2,0),(6i,0),(6i+2,0)$ as starting points of grid paths representing the edges of $G_{14n}$. Then the edge $e_j=v_{j_1}v_{j_2}$ will have its ends assigned to $(x_{j_1},0,0)$ and $(x_{j_2},0,0)$. The path corresponding to $e_j$ is then realized by the following axis-parallel grid path $P_j$:
\[(6j_1,0,0),(x_{j_1},0,0),(x_{j_1},2j,0),(x_{j_1},2j,z_j),(x_{j_2},2j,z_j),(x_{j_2},2j,0),(x_{j_2},0,0),(6j_2,0,0),\]
where $z_j\in [2,14n]$ is chosen so that $P_j$ has length exactly $14n$. We can again notice that the path
\[(6j_1,0,0),(x_{j_1},0,0),(x_{j_1},2j,0),(x_{j_2},2j,0),(x_{j_2},0,0),(6j_2,0,0)\]
has length at least $6$ and at most $2\cdot 2j+6n+4\leq 12n+4<14 n$. Thus there exists some $z_j$ such that the extra length of $2z_j$ compared to the above paths will ensure that $P_j$ has length exactly $14n$.

It is routine to check that the described realization with the paths $P_j$ is an induced grid graph.

\textbf{(b) and (c)}: For a given split graph $G=(C\cup I, E)$ where $C$ is a clique of size $n_C$ and $I$ is an independent set of size $n_I$, we will assign each vertex $v_j$ of $I$ to the complex points $p_j:=e^{i\cdot \frac{j}{4n_I\pi}}$ on the open upper half-circle of the unit circle. The vertex $w\in C$ where $W:=N(w)\cap I$ is the set of vertices in $I$ neighboring $w$ is assigned to the convex hull (in the planar sense) of the points $\{1,-1,-i\},\cup \{p_j\mid v_j\in W\}$. Thus each vertex $w\in C$ has been assigned to a convex fat object of diameter $2$, as the objects are inside the unit disk, and they all contain the fat triangle $\{1,-1,-i\}$. It also follows that all objects assigned to the vertices of $C$ are pairwise intersecting. One can verify that the triangle with vertices $\{1,-1,-i\}$ is $0.18$-fat.
In (b), we can finish the realization by assigning to each $v_j$ a segment of length $2$ of radial orientation (that is, whose line passes through $0$) whose closest point to the origin is $p_j$. These segments are clearly connected and have diameter $2$. In (c), we can assign to each vertex $v_j$ a small disk (say, of radius $\frac{1}{10n}$) that touches the unit disk from the outside at $p_j$. Clearly disks are connected and fat. It is routine to check that the resulting intersection graph in both cases is exactly $G$.
\end{proof}

\begin{remark}
We crucially leave open the case of intersection graphs of disks (of arbitrary size) and pseudodisks: neither our algorithm nor our lower bounds extend to these graph classes.
\end{remark}

\subsection{Reducing planar approximation problems to unit disk graphs}

In this section we will prove that there is a PTAS-reduction from \textsc{Planar Subset TSP} to \textsc{Unit Disk Subset TSP} and from \textsc{Planar Steiner Tree} to \textsc{Unit Disk Steiner Tree}.

Although our reduction is from the unweighted variant, we note that one can easily construct a PTAS reduction form these problems in the case when the edges have non-negative weights. In this sense, our algorithms for \textsc{Subset TSP} and \textsc{Steiner Tree} in intersection graphs also implies (E)PTAS algorithms for the unweighted and weighted variants of these problems in planar graphs.

To prove a formal PTAS reduction, it suffices to prove the following~\cite{DBLP:conf/coco/Crescenzi97}. Given $\eps>0$, a planar graph $G$ with terminals $T\subset V(G)$, we need to construct a unit disk graph $G'$ and terminal set $T'$ in $\poly(n)$ time, such that given a $(1+\eps)$-approximate solution $S$ to $(G',T')$, we can generate a $(1+c\eps)$-approximate solution to $(G,T)$ in polynomial time, where $c$ is a constant.

\paragraph{Construction of $(G',T')$}

Given $\eps, G, T$, we start by computing a so-called bar visibility layout of $G$. In such a layout the vertices of $G$ are represented as horizontal segments whose endpoints have integer coordinates, and edges are represented as vertical segments whose endpoints have integer coordinates, and each edge starts and ends on the segments corresponding to its end vertices, and the segments is disjoint from all other vertex-segments of the drawing. 

We will use the following algorithm to compute the desired layout:

\begin{theorem}[Tamassia and Tollis~\cite{tamassia1986unified}, Rosenstiehl and Tarjan~\cite{RosenstiehlT86}]
Given a planar graph $G$, a bar visibility layout of G inside the bounding box $[1,n]\times [1,2n-4]$ can be computed in $O(n)$ time.
\end{theorem}

\begin{figure}
\centering
\includegraphics{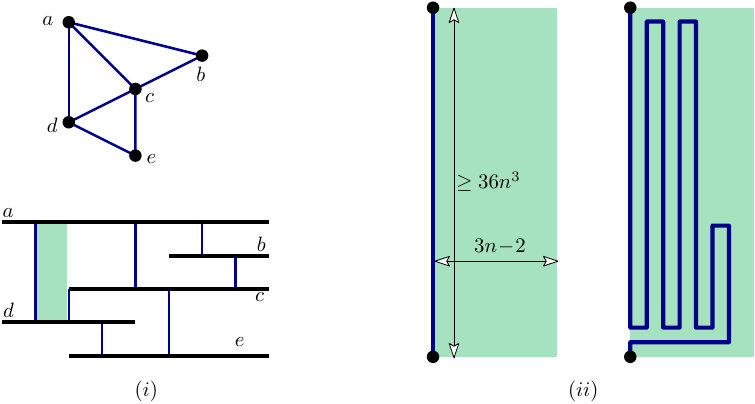}
\caption{(i) Realizing a planar graph as a bar visibility graph. (ii) After scaling, the space between neighboring vertical paths allows us to create a grid path of length exactly $36n^4$.}
\label{fig:kukacutak}
\end{figure}

We apply the transformation $f(x,y)=(3n\cdot x,36n^3\cdot y)$ on the bar layout to get a representation in the bounding box. As a result, the vertex-representing segments have length between $3n$ and $6n^2-12n$, and the edge-representing vertical segments have length between $36n^3$ and $36n^4$, while their horizontal distance is at least $3n$, see \Cref{fig:kukacutak}. Using the horizontal space of width of $3n$ between consecutive vertical segments, we can change each vertical segment into an \textit{induced} grid path of length exactly $36n^4$ connecting the same pair of endpoints: indeed, a grid rectangle of height $36n^3$ and width $3n-2$ does contain such a path connecting two of its corners.
Since the result is an induced grid graph, it can be represented as a unit disk graph where the disks have diameter 1 and the disk centers are the grid vertices of the drawing. Finally, for each vertex $v\in T$ we define $v'\in T'$ to be the unit disk centered leftmost vertex of the segment corresponding to $v$ in the drawing. For $v\in V(G)$ let $H'_v$ denote the set of vertices on the horizontal segment of $G'$ representing $v$. Similarly, for each edge $uv$ of $G$ let $S'_{uv}$ denote the \emph{internal} vertices of the grid path representing the edge $uv$ in $G'$. We can now observe the following.

\begin{observation}\label{obs:reductiontosubdivapplies}
Given $(G,T)$, the instance $(G',T')$ can be constructed in $O(n^6)$ time and the resulting graph has size $O(n^6)$. Moreover, $G'$ is a subcubic dilation of a $36n^4$-subdivision of $G$ that is $(6n^2-12n)$-bounded. In particular, $(G',T')$ satisfies the conditions of \Cref{lem:reductiontosubdivdilation}.
\end{observation}

Using our construction and the above observation, we can wrap up the section with the following theorem.

\udgreduction*

\begin{proof}
Let $k$ be the number of vertices of the minimum Steiner tree in a planar graph $G$ with terminals $T$, and let $X'$ be the vertex set of a $(1+\eps)$-approximate solution for the instance $(G',T')$ constructed above. Then in polynomial time we construct $X\subset V(G)$. Notice that by \Cref{lem:reductiontosubdivdilation}(\textit{i}) we have that the minimum Steiner tree for $(G',T')$ has at most $36(k-1)n^4+\eps\cdot 36n^4$ vertices, thus a $(1+\eps)$-approximate solution has size at most $|X'|\leq (1+\eps)(36(k-1)n^4+\eps\cdot 36n^4)<(1+\eps)k\cdot 36n^4$. Thus by \Cref{lem:reductiontosubdivdilation}(\textit{ii}) we can construct a Steiner tree of size at most $(1+\eps)k$ for $(G,T)$.

Now let $k$ be the length of the minimum subset TSP tour of $G$ for the terminal set $T$, and let $W'$ be a $(1+\eps)$-approximate closed walk for the instance $(G',T')$ constructed above. Then in polynomial time we construct the closed walk $W$ using \Cref{lem:reductiontosubdivdilation}(\textit{ii}). Notice that by \Cref{lem:reductiontosubdivdilation}(\textit{i}) we have that the minimum subset TSP for $(G',T')$ has at most $36kn^4+\eps\cdot 36n^4$ edges, thus a $(1+\eps)$-approximate solution has size at most $|W'|\leq (1+\eps)(36kn^4+\eps\cdot 36n^4)\leq(1+\eps)(k+1)\cdot 36n^4$. Thus \Cref{lem:reductiontosubdivdilation}(\textit{ii}) gives that $|W|<(1+\eps)(k+1)-1<(1+2\eps)k$, concluding the reduction.
\end{proof}

\addcontentsline{toc}{section}{References}
\bibliographystyle{plainurl}
\bibliography{udgapprox}

\appendix

\section{\texorpdfstring{Converting into $\alpha$-standard intersection graphs}{Converting into alpha-standard intersection graphs}}\label{sec:graf_konvert}

The following lemmas allow us to turn an intersection graph of similarly sized (e.g. unit) disks or arbitrary connected similarly sized fat polygons into an isomorphic intersection graph that is $\alpha$-standard.

\begin{lemma}\label{lem:disktopolygon}
Let $G$ be an intersection graph of $\delta$-similarly-sized disks given via its representation. Then in polynomial time we can build an intersection graph $G'$ of $\delta$-similarly sized connected $(1/\sqrt{2})$-fat polygons with polynomial complexity.
\end{lemma}
 
\begin{proof}
For each disk $d\in V(G)$ let $S(d)$ be the maximum axis-parallel square inscribed into $d$.
For each $d,d'\in V(G)$ where $d\cap d'\neq \emptyset$ let $p_{d,d'}$ be the midpoint of the arc $\bd d \cap d'$. Let $p(d)$ be the convex hull of the set $\sigma(d)\cup \bigcup_{d'\in N(d)\setminus \{d\}} p_{d,d'}$. Notice that the vertices of $p(d)$ are on $\bd d$. Consider now the intersection graph $G'$ given by the convex polygons $\{p(d) \mid d\in V(G)\}$. We claim that $p(d),p(d')$ intersect in $G'$ if and only if $d,d'$ intersect in $G$. Since $p(d)\subset d$ for each $d\in V(G)$, no new intersections are created. To check the maintenance of edges, let $c(x)$ denotes the center of $x\in D_2$. For each $d'\in N(d)$ the intersection between $d,d'$ is maintained as the segment $c(d)p_{d,d'}$ intersects the segment $p_{d',d}c(d')$ as both segments are on the line $c(d)c(d')$, and the ordering of these four points on the line is $c(d),p_{d',d},p_{d,d'},c(d')$, where $p_{d',d}$ and $p_{d,d'}$ might coincide.
It remains to show the fatness and similar size bound. Notice that $\sigma(d)\subset p(d)$ and $\sigma(d)$ contains an inscribed disk of radius $r_d/\sqrt{2}$ where $r_d$ is the radius of $d$, on the other hand $p(d)\subset d$, which proves the fatness bound. 
Finding the points $p_{d,d'}$ takes $O(n^2)$ time. Using a standard convex hull algorithm~\cite{grahamscan}, the construction of the polygons $p(d)$ can be executed in $O(n^2\log n)$ time, as each $p(d)$ is the convex hull of at most $4+n-1$ vertices.
\end{proof}

We say that a pair of intersecting segments have a \emph{true crossing} if they intersect in a single point, and neither segment contains any endpoint of the other segment.

\begin{lemma}\label{lem:fairness2}
If $p$ and $q$ are simple polygons where any edge $e$ of $p$ and any edge $f$ of $q$ are either disjoint or have a true crossing, then $\{p,q\}$ satisfies the second fairness condition, that is, any component of their intersection is a simple polygon.
\end{lemma}

\begin{proof}
Let $x$ be some connected component of $p\cap q$. The lemma follows if $p\subset q$ or $q\subset p$ as then $x=p$ or $x=q$, and $p$ and $q$ are simple polygon; assume that this is not the case, and let $A= \cA(p,q)$ be the arrangement given by $p$ and $q$.

If $s$ is a segment of $A$ in $x$, then at least one of the two faces of $A$ incident to $s$ must also be in $x$, as otherwise $s$ would be a segment of positive length in the intersection of some edge of $p$ and some edge of $q$. Moreover, if $\dot v\in \bd x$ is a vertex of $\bd x$, then either $\dot v\in V(p)\cap \inter q$, $v\in V(q)\cap \inter p$, or $\dot v$ is a true crossing. Observe that in each of these cases $\dot v$ has degree $2$ in $\bd x$ and has some incident face of $A$ that is both in $p$ and $q$. Finally, since $p$ and $q$ are closed, $x$ is closed. We conclude that $x$ is a non-self-intersecting polygon; it remains to show that it has exactly one boundary component.

Suppose that $h$ is a hole of $x$, and let $f$ be some edge of $h$. Then $f$ cannot be on an edge of both $\bd p$ and $\bd q$, so assume without loss of generality that $f\subset \bd p$. Since the other side of $\bd h$ (namely, $x$) is in $p$, it follows that the face $F_f\subset \inter h$ of $A$ on the other side of $f$ is not in $p$. On the other hand, $\bd h\subset x \subset p$, thus $F_f$ is a bounded region outside $p$ that is in the bounded region of the Jordan curve $\bd h \subset p$ and thus in the bounded region of $\bd p$. Thus the bounded region of $\bd p$ is not covered by $p$, which contradicts the fact that $p$ is a simple polygon.
\end{proof}

\begin{lemma}\label{lem:polygontoholefree}
Let $G$ be an intersection graph of $\delta$-similarly-sized $\beta$-fat non-self-intersecting polygons, possibly with holes, given via its representation of polynomial complexity. Then in polynomial time we can build an intersection graph $G'$ of $\delta$-similarly-sized $(\beta/4)$-fat simple polygons. Moreover, each polygon of $G'$ is a subpolygon of the corresponding polygon of $G$, and if $G$ is $\alpha$-standard, then scaling the representation of $G'$ by $4$ results in an $4\alpha$-standard graph.
\end{lemma}

\begin{proof}
Let $W$ be the collection of all vertices in the arrangement of the polygons of $V(G)$, and let $X=\{x_1,\dots,x_N)$ such that $x_1\leq x_2\dots\leq x_N$ be the set of $x$-coordinates in $W$. Let $P=V(G)$ be the initial collection of polygons, and fix in each polygon a large inscribed disk $d_p$ of radius at least $\frac{\beta}{4} \diam(P)$ we will modify $P$ in several steps as follows.

Consider some polygon $p\in P$ and let $h$ be a hole of $p$. Take the vertical slab between the vertical lines $x=x_{p,h}$ and $x=x'_{p,h}$ that intersects $h$ such that $x_{p,h}$ and $x'_{p,h}$ subdivide the distance between two consecutive numbers of $X$ into three equal parts, i.e., $x_{p,h}=\frac{2x_i+x_{i+1}}{3}$ and $x'_{p,h}=\frac{x_i+2x_{i+1}}{3}$ for some $i\in [N-1]$. Note that since $h$ is a simple polygon whose vertices project into some subset of $X$, there exists such a slab. Notice that the slab $s:=\{(x,y) \in \Reals^2 \mid x_{p,h}\leq x \leq x'_{p,h}\}$ does not contain any polygon vertices. Consider the intersection $s$ and $\Reals^2\setminus p$. Each region of $\Reals^2\setminus p$ belongs to either some hole $h^p_i$ of $p$ or the unbounded $\re U_p$ region of $\Reals^2\setminus p$; see \Cref{fig:lyuktalanito}.

Since there are no vertices of $p$ in $s$, each edge of $p$ cuts the slab with a segment, and thus there is a top-to-bottom ordering of the arrangement $\cA(s\cap \bd p)$, where the regions of $p$ and $\Reals^2 \setminus p$ alternate, and the top and bottom unbounded regions are in $\re U_p$. Moreover, each region is a trapezoid and thus convex, so there is a unique region that intersects the disk $d_p$. Since $s$ intersects $h$, consider the region of $p$ before the first occurrence of a region of $h$ as well as the region of $p$ after the last occurrence of a region of $h$ in this top-to-bottom order. At least one of these regions is disjoint from $d_p$; let $\re R_{p,h}$ be this region. Notice that and exactly one of the regions neighboring $\re R_{p,h}$ is in $h$, and the other is either in the unbounded face of $p$ or in a different hole. Let $p'$ be the closure of $p\setminus \re R_{p,h}$. As a result, either $h$ is merged into the unbounded face of $p'$ or some other hole of $p'$ via $\re R_{p,h}$.

If $q$ is a polygon intersecting $p$, then it must have some region $\re A$ of the arrangement $\cA(P)$ that is inside $p\cap q$. Notice that $s$ contains no vertices of $\re A$, thus $\re A \setminus s$ is non-empty, and thus $p'$ intersects $q$. Moreover, $\diam(p)=\diam(p')$ (as the diameter of a polygon is realized between some pair of vertices, and no vertices of $p$ have been removed), and $d_p\subset p'$. Next, we replace the polygon $p$ in $P$ with $p'$, and also update $X$ to include the new points $x_{p,h}, x'_{p,h}$.

\begin{figure}
\centering
\includegraphics{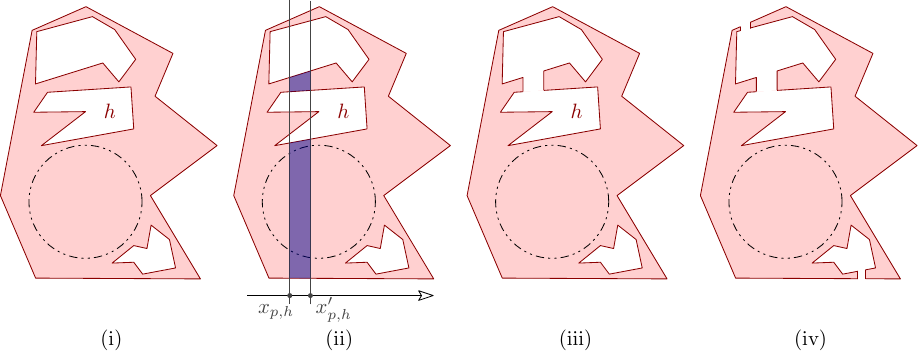}
\caption{Puncturing a polygon with holes to make it simple. (i) A polygon with holes and an inscribed disk. (ii) Finding a narrow slab of the arrangement without any vertices of the arrangement. The region after the last occurrence of the hole $h$ cannot be removed as that would hurt fatness. (iii) Removing a part of $s\cap p$ neighboring $h$ decreases the number of holes. (iv) A simple fat polygon $p'$ with the same diameter and similar fatness.}
\label{fig:lyuktalanito}
\end{figure}

We repeat this step exhaustively until we get a polygon collection $P$ without holes. The step retains the connectedness of $p$: indeed, the top and bottom edge of the trapezoid $\re R_{p,h}$ remains connected in $p$ via $\bd h$, as each hole $h$ is a simple polygon (and thus has a single boundary component) that intersects $\bd \re R_{p,h}$ in a segment. Notice that each step reduces the total number of holes in the polygons of $P$ by at least one. The polygon diameters remain unchanged, and we maintain a ball of radius $\frac{\beta}{4}  \diam(p)$ in each polygon~$p$, thus the final polygon set $P$ is a collection of $\delta$-similarly-sized $(\beta/4)$-fat simple polygons. Moreover, the intersection graph is maintained in each step as the polygon $p'$ replacing $p$ is a subset of $p$ (thus no new intersections can be created) but all intersections with any other $q\in P$ are maintained by $p'$. Thus the resulting intersection graph $G'$ is isomorphic to $G$.

If $G$ was an $\alpha$-standard graph, then observe that no edges or vertices have been introduced to $G'$ that could create intersecting edges that do not make a true crossing, and no new concurrent edge triples are created as new edges avoid existing edges of the arrangement. \Cref{lem:fairness2} implies that fairness is maintained. The newly introduced polygons are subpolygons of earlier polygons, and the maximum inscribed disks decreased by a factor of at most $4$. Thus, after scaling the new representation by $4$, the resulting representation is $4\alpha$-standard.

It remains to bound the running time of this construction. Notice that covering a polygon $p$ with a disk requires radius at least $\diam(p)/2$, thus by the fatness bound each polygon contains a disk $d'_p$ of radius at least $\frac{\beta}{2}\diam(p)$. To find the desired disk $d_p$, we can consider a square grid of side length $\frac{\beta}{4\sqrt 2}\diam(p)$. Among the grid points inside $p$ there must be one whose center is within distance $\sqrt 2 \cdot \frac{\beta}{4\sqrt 2}\diam(p)=\frac{\beta}{4}\diam(p)$ to the center of $d'_p$; the disk of radius $\frac{\beta}{4}\diam(p)$ centered at this point is covered by $d'_p$ and thus covered by $p$, so it is a suitable choice for $d_p$. Clearly the grid has constantly many points inside the bounding box of $p$ and we can check whether a disk is contained in $p$ in linear time. Finally, notice that the initial arrangement $\cA(P)$ is polynomial and has $\poly(n)$ holes; and each step adds at most $O(|\cA(P)|)$ new elements to the arrangement as there are only $2$ new vertical segments introduced. Thus all configurations have polynomial complexity (and the same holds for their arrangements). We conclude that the construction of the representation of $G'$ can be executed in polynomial time.
\end{proof}

\begin{lemma}\label{lem:simplepolygontoorganized}
Let $G$ be an intersection graph of $\delta$-similarly-sized $\beta$-fat simple polygons given via its representation of polynomial complexity. Then in polynomial time we can build the representation of an $\alpha$-standard intersection graph $G'$ that is isomorphic to $G$ where $\alpha=9/(\delta\beta)$.
\end{lemma}

\begin{proof}
Let $P_0=V(G)$ be the initial collection of simple polygons. We find a $1/4$-approximation $d_\myin(p)$ of the maximum enclosed disk of all polygons $p\in P_0$ in polynomial time using the same method as seen in the proof of~\Cref{lem:polygontoholefree}. We modify $P_0$ in several steps. First, we scale $P_0$ to ensure each polygon contains a disk of radius $1$ (thus diameter $2$), i.e., let $P_1$ be the set of polygons after $P_0$ is scaled by $\frac{2}{\min_{p\in P_0} \diam(d_\myin(p))}$. As a result, every polygon of $P_1$ has diameter at most $8/(\beta\delta)$. Consider the arrangement $\cA(P_1)$, and let  $E_1$ and $V_1$ denote its set of edges and vertices. We denote by $N$ the total number of edges in the polygons of $P_1$, and let $\mu_0$ the minimum distance among $V_1$ and the minimum distance between non-incident vertices of $V_1$ and $E_1$. Let $\mu_1:=\min\{\mu_0/(100N^2),1/(100N^2)\}$.

We will now modify the polygons so that any pair of edges from distinct polygons intersect in a point, and no three polygons edges are concurrent. We say that an edge $f$ of the polygon $p$ is \emph{bad} if either (a) there is some edge $f'$ of some polygon $p'$ where $p\neq p'$ and $f\cap f'$ contains some segment of positive length (i.e., they overlap) (b) there is some vertex $v\in V(p')$ on $f$ from some polygon $p'\neq p$, or (c) there is some vertex $v\in V_1$ of degree at least $5$ on $f$.

If $a,b$ are the endpoints of $f$, then consider the outer angle bisector rays of $p$ at $a$ and $b$. For each $i=1,\dots,N$ let $a_i$ and $b_i$ be points on the angle bisector from $a$ and $b$ at distance $i \mu_1$ from $a$ and $b$, respectively, where $i\leq 12N^2$. \skb{TODO:abra} Let $p(i)$ be the polygon where $a$ and $b$ are replaced with $a_i$ and $b_i$. Notice that $p_f(i)$ stays within distance $\mu_0/4$ from $p$, thus if $p,q\in P_1$ are disjoint, then $p_f(i),q_{f'}(j)$ are also disjoint for any $i,j\in [12N^2]$. Moreover, $p\subset p_f(i)$ thus if $p$ intersects $q$ then $p_f(i)$ intersects $q_{f'}(j)$ for all $i\neq j$. Finally, note that $p$ and $p_f(i)$ differ on three edges, and the curves $\bd p_f(i)\setminus \bd p$ are pairwise disjoint for different $i$. Notice also that the polygon $p$ remains simple after the modification.

We now exchange bad edges of $P_1$ one by one so that whenever we replace some polygon $p$ by $p(i)$, we ensure that all edges on the curve $\bd p_f(i)\setminus \bd p$ are non-bad. We claim that choosing such an index $i\in [N^2]$ is always possible. First, we observe that the total number $N$ of edges in the polygon collection is unchanged after such a modification. Consequently, the number of vertices in the arrangement of the $N$ total polygon edges is always at most $\binom N 2$, thus any arrangement of this many edges produces a planar graph with at most $3\binom N 2 -6 < 2N^2$ edges. We conclude that there is always a choice of $i\in [12N^2]$ such that the $3$ new edges given by the curve $\bd p_f(i)\setminus \bd p$ are disjoint from all vertices of the current arrangement and they do not overlap with any edges of the current arrangement. Observe that a modification eliminates at least one bad edge, and the edges neighboring the modified edge cannot be bad either. In particular, if a vertex of $p$ is moved, then the incident edges remain good, so no vertex of $p$ is moved more than once. After at most $N$ steps the resulting polygon collection $P_2$ is free of bad edges.

By definition, each polygon of $P_2$ includes a disk of radius $1$, and they have diameter less than $8/(\beta\delta)+1/2$. Moreover, the intersection graph is maintained throughout the modifications, so the resulting graph $G'$ is isomorphic to $G$.

By \Cref{lem:fairness2} that $G'$ satisfies all properties of being $\alpha$-standard with $\alpha=8/(\beta\delta)+1/2<9/(\beta\delta)$. 
\end{proof}

\Atalakitas*

\begin{proof}
In case of (a), use \Cref{lem:disktopolygon} to obtain an isomorphic intersection graph of simple polygons of fatness $1/\sqrt{2}$ that are $\delta$-similarly sized with representation complexity $O(n^2)$. In case of (b), use \Cref{lem:polygontoholefree} to obtain an intersection graph of $\beta/4$-fat $\delta$-similarly-sized objects. On both of these graphs, we use \Cref{lem:simplepolygontoorganized} to get an $\alpha$-standard graph of with $\alpha=9/(\delta\sqrt{2})$ and $\alpha=36/(\beta\delta)$, respectively. Notice that each step increases the complexity of the representation polynomially, thus the final graph $G'$ also has $\poly(n)$ complexity.
\end{proof}

\end{document}